%% file: Ha_PhD_thesis_INRS.tex
\documentclass[phd, noindex, hyperref, entetes, prelimtm, anglais]{theseINRS} 
\usepackage[utf8]{inputenc} 
\usepackage[T1]{fontenc} 
\usepackage{lmodern} 
\usepackage[square, numbers, comma, sort&compress]{natbib}
\usepackage[letterpaper, margin=25mm]{geometry} 
\usepackage{graphicx} 
\usepackage[figuresright]{rotating} 
\usepackage{tabularx} 
\usepackage{nomencl} 
\usepackage[font={singlespacing,footnotesize,bf}]{caption} 
\usepackage[font={singlespacing,footnotesize}]{subcaption} 
\usepackage{textcomp} 
\usepackage{emptypage} 

\usepackage{amssymb,flushend}
\usepackage[cmex10]{amsmath}
\usepackage{amsthm}
\usepackage{array}
\usepackage{mdwmath}
\usepackage{mdwtab}
\usepackage{eqparbox}
\usepackage{url}
\usepackage{algorithm}
\usepackage{caption}
\makeatletter
\@addtoreset{algorithm}{chapter}
\makeatother

\usepackage{algorithmic}
\usepackage{enumerate}

\usepackage{etaremune}

\newcommand{\argmin}{\operatornamewithlimits{argmin}}
\newcommand{\argmax}{\operatornamewithlimits{argmax}}

\newtheorem{remark}{\it Remark}[chapter]
\newtheorem{theorem}{Theorem}[chapter]
\newtheorem{lemma}{Lemma}[chapter]
\newtheorem{proposition}{Proposition}[chapter]
\newtheorem{definition}{Definition}[chapter]

\newcommand{\mb}[1]{\mathbf{#1}}

\newtheorem{prop}{Proposition}[chapter]
\newcommand{\beq}{\begin{equation}}
\newcommand{\eeq}{\end{equation}}
\newcommand{\beqn}{\begin{eqnarray}}
\newcommand{\eeqn}{\end{eqnarray}}
\newcommand{\beqno}{\begin{eqnarray*}}
\newcommand{\eeqno}{\end{eqnarray*}}
\newcommand{\bma}{\begin{displaymath}}
\newcommand{\ema}{\end{displaymath}}
\newcommand{\bnu}{\begin{enumerate}}
\newcommand{\enu}{\end{enumerate}}
\newcommand{\bce}{\begin{center}}
\newcommand{\ece}{\end{center}}
\newcommand{\btb}{\begin{tabular}}
\newcommand{\etb}{\end{tabular}}

\begin{document}
\shorthandoff{:} 
\shorthandoff{;} 
\shorthandoff{?} 
\shorthandoff{!} 

\pagesprelim
\PrenomNom{VU N. HA}
\titre{RADIO RESOURCE MANAGEMENT FOR HIGH-SPEED WIRELESS CELLULAR NETWORKS} 
\programme{t\'el\'ecommunications}
\grade{\textit{Doctorat en philosophie}, Ph.D.} 
\jury{Examinateur externe &  Prof. Fran\c{c}ois Gagnon\\
	& \textit{\'{E}cole de Technologie Sup\'{e}rieure} \\[0.2cm]
	& Prof. Fei Richard Yu \\
	& \textit{Carleton University} \\ [0.5cm]
	
	Examinateur interne & Prof. Andr\'e Girard \\
	& \textit{INRS-\'EMT}  \\[0.5cm]
	Directeur de recherche & Prof. Long Bao Le \\
	& \textit{INRS-\'EMT} \\[0.5cm]}
\annee{2016}

\maketitle

\cleardoublepage
\begin{flushright}
{~}\\[1in]
{\emph{To my Parents}} \\
{\emph{To my wife Nguy\~{\^{e}}n Ng\d{o}c Qu\`{y}nh Duy\^{e}n}} \\
{\emph{To my son H\`{a} Ki\'{\^{e}}n V\u{a}n}}
\end{flushright}

\include{Acknowledgments/Acknowledgments} 
\include{abstract/abstract}

\cleardoublepage
\phantomsection
\addcontentsline{toc}{chapter}{\contentsname} 
\tableofcontents

\cleardoublepage
\phantomsection
\renewcommand*\listfigurename{List of Figures} 
\addcontentsline{toc}{chapter}{\listfigurename} 
\listoffigures 

\cleardoublepage  
\phantomsection                                   
\addcontentsline{toc}{chapter}{\listtablename} 
\listoftables 

\cleardoublepage
\phantomsection
\renewcommand*\listalgorithmname{List of Algorithms} 
\addcontentsline{toc}{chapter}{\listalgorithmname} 
\listofalgorithms

\cleardoublepage 
\phantomsection
\renewcommand{\nomname}{List of Abbreviations}
\addcontentsline{toc}{chapter}{\nomname} 
\makenomenclature
\setlength\nomlabelwidth{2cm} 
\printnomenclature

\cleardoublepage
\corps
\renewcommand{\tablename}{Tableau} 

\include{chap1/summary}
\include{chap2/resume}
\include{chap3/Ha_chap3}
\include{chap4/Ha_chap4}
\include{chap5/Ha_chap5}
\include{chap6/Ha_chap6}
\include{chap7/Ha_chap7}
\include{chap8/Ha_chap8}
\include{chap9/Ha_chap9}


\singlespacing
\cleardoublepage   
\phantomsection     
\renewcommand{\bibname}{Références} 
\addcontentsline{toc}{chapter}{\bibname} 
\bibliographystyle{ieeetr}
\renewcommand{\rightmark}{References}
\bibliography{biblio/Ha_Ref} 


\end{document}

%% file: Acknowledgments/Acknowledgments.tex
\chapter*{Acknowledgments}
\addcontentsline{toc}{chapter}{Acknowledgments} 

\large

I would like to gratefully acknowledge and express a sincere thank you to my supervisor, Professor Long Bao Le for giving me the opportunity to pursue doctoral study at INRS-\'{E}MT, University of Qu\'{e}bec. 
I am truly privileged to have learned from his remarkable technical knowledge and research enthusiasm. 
Since the very beginning, he has always pointed me in good research directions and encouraged me to pursue them to concrete results. 
His invaluable support and guidance during my study have certainly helped me complete this Ph.D. dissertation. 
I would like to express my gratitude to other members of my Ph.D. committee -- 
Professor Andr\'e Girard of INRS-\'{E}MT, University of Qu\'{e}bec who has regularly
reviewed and constructively commented on the progress of my doctoral study.
I would also like to thank Professor Fran\c{c}ois Gagnon of \'{E}cole de Technologie Sup\'{e}rieure and Professor Fei Richard Yu of Carleton University for serving as the external examiner to my Ph.D. dissertation.

I would like to express gratitude to all my colleagues for the wonderful and memorable time at the Networks and Cyber Physical Systems Lab (NECPHY-Lab), INRS-\'{E}MT, University of Qu\'{e}bec: 
Tan Le, Tuong Hoang, Hieu Nguyen, Dai Nguyen, Tam Tran, Thinh Tran, and Ti Nguyen. 
Additionally, many thanks to 
Dr. Quoc-Thai Ho and Mr. Badreddine Ben Nouma for helping me with the French translation
to my Ph.D. dissertation.

Finally, my deepest love and gratitude are devoted to all of my family members:
Mom and Dad, beloved wife and son, brother and sister-in-laws Quy\^{e}n, Parents-in-laws, sister-in-laws H\`{\u{a}}ng, brother-in-laws Khoa, who always support me in each and every endeavor in my life.
My Ph.D. study would not be finished without the constant and unconditional support from my family.
I thank you all and hope that I made you proud of my accomplishments.

%% file: abstract/abstract.tex
\chapter*{Abstract}
\addcontentsline{toc}{chapter}{Abstract} 

\renewcommand{\baselinestretch}{1.4}
The fifth-generation (5G) wireless cellular system, which would be deployed by 2020, is expected to deliver significantly higher
capacity and better network performance compared to those of the current fourth-generation (4G) system.
Specifically, it is predicted that tens of billions of wireless devices will be connected to the wireless network over next few years, which 
results in an exponential explosion of mobile data traffic. 
Therefore, more advanced wireless architecture, as well as radical and innovative access technologies,
must be proposed to meet this urgent increasing growth of mobile data and connectivity requirements in the coming years.
Toward this end, two important wireless cellular architectures, namely
wireless heterogeneous networks (HetNets) based on the dense deployment of small cells and the
cloud radio access networks (C-RANs) have been proposed and actively studied by
both academic and industry communities. 
Besides enabling a lot of advantages in increasing network coverage as well as end-to-end system throughput, these two novel network architectures have also raised some novel technical challenges and opened exciting research areas
for further research.

Motivated by the aforementioned technical challenges, the general
objective of this Ph.D. research is to develop efficient radio resource allocation and interference management
algorithms for the future high-speed wireless cellular networks. In particular, we 
have developed adaptive resource management techniques that can effectively control both critical interferences in wireless HetNets
and designed innovative access techniques for the C-RAN 
that efficiently exploit the radio, cloud computation resources, and fronthaul capacity.  
Our research has resulted in four major
research contributions, which are presented in four corresponding main chapters
of this dissertation.

First, we consider the joint base station association and power control design for single-carrier-based HetNets, which is presented in Chapter \ref{Ch3}. 
In particular, we have developed a generalized BSA and PC algorithm and proved its convergence if the underlying power update function satisfies the so-called two-sided-scalable property. In addition, we have proposed an hybrid power control adaptation algorithm that effectively adjusts key design parameters to maintain the SINR requirements of all users 
whenever possible while enhancing the system throughput. 

Second, we study fair resource allocation design with subcarrier assignment and power control for OFDMA-based HetNets, which is presented in Chapter \ref{Ch4}.
Specifically, we have presented a resource allocation formulation for the two-tier macrocell-femtocell network that aims to maximize the total minimum rate of all femtocell 
subject to QoS protection constraints for macrocell users. We have proposed a low-complexity distributed joint subchannel and power allocation algorithm where FBSs can make subchannel allocations for femto user equipments in the distributed manner.

Third, we design the joint cooperative transmission protocol for the downlink C-RAN, which is covered in Chapter \ref{Ch5}. 
The considered system captures the fact that fronthaul links connecting remote radio heads with cloud processing center have limited
capacity, which is translated into the new fronthaul capacity constraint involving a non-convex and discontinuous function. 
Then, we propose two low-complexity algorithms, so-called pricing-based algorithm, and iterative linear-relaxed algorithm, to duel with this difficult problem.

Finally, we consider the resource allocation for virtualized uplink C-RAN in which multiple OPs are assumed to 
share the C-RAN infrastructure and resources to serve their users under the limited fronthaul capacity and cloud computation.
The research outcomes of this study are presented in Chapter \ref{Ch8}.
In particular, we model the design into the upper-level and lower-level problems. 
The upper-level problem focuses on slicing the fronthaul capacity and cloud computing resources for all OPs. 
Then, the lower-level maximizes the operator's sum rate by optimizing users' transmission rates and quantization bit allocation 
for the compressed I/Q baseband signals. 
Then, a two-stage algorithmic framework is proposed to solve these problems.

We have developed various efficient resource allocation algorithms for reducing the transmission power 
and increasing the end-to-end network throughput for both HetNets and C-RANs.
Furthermore, extensive numerical
results are presented to gain further insights and to evaluate the performance of
our resource allocation designs. The numerical results confirm
that our proposed protocols can achieve efficient spectrum utilization, power saving, and
significant performance gains compared to existing and fast greedy designs.

%% file: chap1/summary.tex
\chapter{Extended Summary}
\section{Background and Motivation}

The fifth-generation (5G) wireless cellular system, which would be deployed by 2020, is expected to deliver significantly higher
capacity and better network performance compared to those of the current fourth-generation (4G) system. This is required to meet various practical technical
and service challenges. 
Specifically, it is predicted that tens of billions of wireless devices will be connected to the wireless network over next few years.
Together with the increasing number of connections, the amount of mobile data traffic has been exploding at an exponential rate. 
Therefore, the 5G wireless mobile network should be able to support
the data traffic volume of an order of magnitude larger than that in the current wireless network \cite{Zander13,Boccardi14,Bhushan14,Le_EU15}.
\nomenclature{5G}{Fifth-Generation}
\nomenclature{4G}{Fourth-Generation}

Therefore, more advanced wireless architecture, as well as radical and innovative access technologies, must be proposed to 
meet the exponential growth of mobile data and connectivity requirements in the coming years \cite{Dottling09, 
Lopez-Perez11, 3GPP,cisco13,cisco16}. 
Toward this end, two important wireless cellular architectures, namely wireless heterogeneous networks (HetNets) based on
the dense deployment of small cells and the cloud radio access networks (Cloud-RAN or C-RAN) have been proposed and actively 
studied by both academic and industry communities. These two terms Cloud-RAN and C-RAN
 will be used changeably in the sequel. This doctoral dissertation focuses on the radio resource management 
designs for these two potential wireless network architectures.
\nomenclature{HetNets}{Wireless Heterogeneous Networks}
\nomenclature{C-RAN}{Cloud Radio Access Network}

The wireless HetNet is typically based on the dense deployment of small cells such as femtocells and picocells in coexistence
with existing macrocells \cite{Chandrasekhar08,Claussen08,Kim09,Zhang_bk_10} where small cells
have short communications range, low power, and low cost. 
Massive deployment of small cells can fundamentally improve the indoor throughput and coverage performance of
the wireless cellular network \cite{Andrews12,Le12}. Moreover,
certain small cells such as femtocells can be randomly deployed by end user equipments (UEs) and they operate on the same
frequency band with the existing macrocell to enhance the spectrum utilization.
\nomenclature{UE}{User Equipment}
 
In addition, indoor traffic supported by femtocells can be backhauled through the IP connections
such as DSL to reduce the traffic load of the macrocells. 
The macrocell tier can thus dedicate more radio resources to better service outdoor UE.
Moreover, femtocells, which can be deployed by end users in a plug-and-play manner, usually require low capital expenditure and operating cost.
However, dense deployment of femtocells on the same frequency band with the macrocells poses various technical challenges.
First, the strong cross-tier interferences between macrocells and femtocells may occur, which may severely impact the network performance
if not managed properly \cite{Yavuz09}. Second, macro UEs (MUEs) typically have higher priority
in accessing the radio spectrum compared to the femto UEs (FUEs); therefore, the QoS requirements of MUEs must be protected and maintained. 
Therefore, the cross-tier interference induced by FUEs to the macrocell tier must be appropriately controlled. 
Lastly, dynamic and intelligent access control strategies must be developed to efficiently manage the 
network interference and traffic load balancing \cite{roche10}.
\nomenclature{IP}{Internet Protocol}
\nomenclature{DSL}{Digital Subscriber Lines}
\nomenclature{MUE}{Macro User Equipment}
\nomenclature{FUE}{Femto User Equipment}

The C-RAN architecture aims at exploiting the cloud computing infrastructure to realize various network functions and protocols
\cite{chinamobile2011,NGMN2013,maketresearch2013}. The general C-RAN architecture consists of three main 
components, namely  (i) centralized processors or baseband unit (BBU) pool, (ii) the optical 
transport network (i.e., fronthaul (fronthaul) links), and  (iii) remote radio head (RRH) access units with antennas located at remote sites. 
The cloud processing center comprising a large number of BBUs is the heart of this architecture where
BBUs operate as virtual base stations (BSs) to process baseband signals for UEs and optimize the radio resource allocation. 
In general, RRHs can be relatively simple, which can be spatially distributed over the network for more energy and cost-efficient network deployment.
\nomenclature{BBU}{Baseband Unit}
\nomenclature{FH}{Fronthaul}
\nomenclature{BS}{Base Station}
\nomenclature{RRH}{Remote Radio Head}

Additionally, the centralized processing enables to implement sophisticated physical-layer and radio resource management designs
such as the coordinated multi-point (CoMP) transmission and reception techniques proposed in the LTE wireless standard,
efficient clustering design for RRHs to balance between the network capacity enhancement and design complexity \cite{Costa-Perez_CM13, Liang_CST15}.
Other benefits of C-RAN include reduced backhaul and network core traffic, deployment and operation costs,
enhanced service quality and differentiation to effectively support different wireless applications \cite{mob_rep_13,chinamobile2011}. 
Successful deployment of C-RAN, however, requires us to resolve several major technical challenges \cite{Costa-Perez_CM13,Liang_CST15}. 
In particular, efficient techniques and solutions for advanced signal processing and radio resource management must be developed
to achieve the potential network performance enhancement while accounting for constraints and efficient utilization of fronthaul 
capacity and cloud computing resources. 
\nomenclature{CoMP}{Coordinated Multi-Point Transmission or Reception}
\nomenclature{LTE}{Long-Term Evolution}

\section{Research Contributions}

Motivated by the aforementioned technical challenges, the general
objective of this Ph.D. research is to develop efficient radio resource allocation and interference management
algorithms for the future high-speed wireless cellular networks. In particular, we 
have developed adaptive resource management techniques that can effectively control both co-tier
and cross-tier interferences in wireless HetNets and designed innovative access techniques for the C-RAN 
that efficiently exploit the radio, cloud computation resources, and fronthaul capacity.  
The contributions of this Ph.D. research are described in the following.

\subsection{Base Station Association And Power Control for Single-carrier-based Wireless HetNets}

In this contribution, we develop the joint base station association (BSA) and power control (PC) techniques for the multi-tier single-carrier-based HetNets. 
In the HetNet setting, designs of the PC and BSA techniques are challenging since UEs in different network tiers (e.g., macrocell and femtocell) 
have distinct access priorities as well as QoS requirements.
There are several existing works that proposed advanced interference management solutions for HetNets
by using dynamic PC and BSA \cite{jo09,chand09,madan10,yun11,Le12a,yanzan12}. 
However, most of these works considered the closed access mode where MUEs are not allowed to connect with femto BSs (FBSs).
In addition, joint consideration of QoS guarantees and efficient utilization of network resources were not
thoroughly treated. Our current design aims at addressing these limitations where we make the following contributions.
\nomenclature{BSA}{Base Station Association}
\nomenclature{PC}{Power Control}
\nomenclature{QoS}{Quality of Service}
\nomenclature{FBS}{Femto Base Station}
\begin{itemize}
\item We develop a generalized BSA and PC algorithm for multi-tier HetNets and prove its convergence if the power update 
function satisfies the two-sided-scalable (2.s.s.) property. \nomenclature{2.s.s.}{Two-Sided-Scalable}
\item We propose an hybrid adaptation algorithm that adjusts the key design parameters to support the differentiated signal-to-interference-plus-noise ratio (SINR) requirements of all UEs whenever possible while enhancing the system throughput. \nomenclature{SINR}{Signal-to-Interference-plus-Noise Ratio}
\item This proposed algorithm is proved to achieve better performance than some existing state-of-the-art algorithms.
\item We describe how to extend the proposed framework to enable hybrid access design in two-tier macrocell-femtocell networks.
\end{itemize}

\subsubsection{System Model}
\label{section2_r3}

We consider uplink (UL) communications in a heterogeneous wireless cellular network where $K$ BSs serve $M$ UEs on the same spectrum by using CDMA.
\nomenclature{UL}{Uplink}
\nomenclature{CDMA}{Code-Division Multiple Access}
Assume that each UE $i$ communicates with only one BS at any time, which is denoted as $b_i$.
However, UEs can change their associated BSs over time. 
The SINR of UE $i$ at BS $b_i$ can be written as \cite{Alpcan08}
\begin{equation}
\label{eq:sinr_r3}
\Gamma_i(\mathrm{p})=\dfrac{G h_{b_i i} p_i}{\sum_{j \neq i}{h_{b_ij}p_j}+\eta_{b_i}}
=\dfrac{p_i}{I_i \left( \mathrm{p}, b_i\right) }
\end{equation}
where $h_{k i}$ denotes the channel gain between BS $k$ and UE $i$, $p_i$ is the transmission power of UE $i$, $I_i \left( \mathrm{p},k\right) \triangleq  \dfrac{\sum_{j\neq i}{h_{kj}p_j}+\eta_{b_i}}{G g_{kii}}$ is the effective interference, and $\mathrm{p}=[p_1,...,p_M]$.
We will sometimes write $I_i \left( \mathrm{p}\right)$ instead of $I_i \left( \mathrm{p},b_i\right)$.
We assume that UE $i$ requires the minimum QoS in terms of a target SINR $\bar{\gamma}_i$ as
\begin{equation}
\label{equ:SINR_cond_r3}
\Gamma_i(\mathrm{p}) \geq \bar{\gamma}_i, \:\: i \in \mathcal{M}.
\end{equation}
The objective of this contribution is to develop distributed BSA and PC algorithms that can maintain the SINR 
requirements in (\ref{equ:SINR_cond_r3}) (whenever possible) while 
exploiting the multiuser diversity gain to increase the system throughput.    
The proposed algorithms, therefore, aim to support both voice and high-speed data applications where the voice UEs would typically require some fixed target SINR $\bar{\gamma}_i$ while data UEs would
seek to achieve higher target SINR $\bar{\gamma}_i$ to support their broadband applications.

\subsubsection{Generalized Base Station Association and Power Control Algorithm}
\label{sec:gnrl_PC_BAS}
We first develop a general BSA and PC algorithm.
Specifically, we will focus on a general iterative PC algorithm where
each UE $i$ performs the following power update function (puf) $p_i^{(n+1)} := J_i(\mathrm{p}^{(n)})=J'_i(I_i^{(\mathit{n})}(\mathrm{p}^{(n)}))$ where
$n$ denotes the iteration index and $J_i(.)$, $J'_i(.)$ is the puf. \nomenclature{puf}{Power Update Function}
In fact, this kind of PC algorithm converges if we can prove 
that its corresponding puf is 2.s.s. \cite{sung05, sung06}.
In particular, we propose a joint BSA and PC algorithm as summarized in Algorithm~\ref{alg:gms2_r3}. 
Under this design, each UE chooses a BS that results in minimum effective interference and employs the p.u.f to update its power.
This algorithm ensures that each UE experiences low effective interference and therefore high throughput at convergence.
In general, the performance of a PC algorithm depends on how we design the corresponding 2.s.s. puf $\mathrm{J}(\mathrm{p})$.
\renewcommand{\baselinestretch}{0.9}
\small
\begin{algorithm}
\caption{\textsc{Minimum Effective Interference BS Association and Power Control Algorithm}}
\label{alg:gms2_r3}
\begin{algorithmic}[1]
\STATE Initialization:

- $p_i^{(0)}=0$ for all UE $i$, $i \in \mathcal{M}$.

- $b_i^{(0)}$ is set as the nearest BS of UE $i$.

\STATE Iteration $n$: Each UE $i$ ($ i \in \mathcal{M}$) performs the following:

- Calculate the effective interference at BS $\mathit{k} \in D_i$ as follows:

- Choose the BS $b_i^{(n)}$ with the minimum $I_i^{(n)} ( \mathrm{p}^{(n-1)},\mathit{k})$, i.e., 
$ b_i^{(n)}=\mathrm{argmin}_{\mathit{k} \in D_i} I_i^{(n)} ( \mathrm{p}^{(n-1)},\mathit{k})$. Then, we have $ I_i^{\sf o (\mathit{n})}(\mathrm{p}^{(n-1)}) = \mathrm{min}_{\mathit{k} \in D_i} I_i^{(n)} ( \mathrm{p}^{(n-1)},\mathit{k}) = I_i^{(n)} ( \mathrm{p}^{(n-1)},b_i^{(n)})$.

- Update the transmit power for the chosen BS as follows:
\begin{equation}
\label{p}
p_i^{(n)}=J'_i(I_i^{\sf o (\mathit{n})}(\mathrm{p}^{(n-1)})). 
\end{equation}
\vspace{-0.3cm}
\STATE Increase $n$ and go back to step 2 until convergence.
\end{algorithmic}
\end{algorithm}
\renewcommand{\baselinestretch}{1.4}
\normalsize

\subsubsection{Proposed HPC Algorithm}

The hybrid PC (HPC) algorithm will be now developed to take advantages of both target-SINR-tracking PC (TPC) \cite{foschini93, yates95} and opportunistic 
PC (OPC) \cite{sung05, sung06}.
\nomenclature{HPC}{Hybrid Power Control}
\nomenclature{TPC}{Target-SINR-Tracking Power Control}
\nomenclature{OPC}{Opportunistic Power Control}
Specifically, the HPC algorithm can be described by the following iterative power update
\begin{equation} \label{hpcrule}
p_i^{\left(n+1\right)} = J^{\mathsf{HPC}}_i \left( \mathrm{p}^{(n)}\right) = \mathrm{min}\left\lbrace P^{\mathsf{max}}_i, J_i\left( \mathrm{p}^{(n)}\right) \right\rbrace,
\end{equation}
where $J_i\left( \mathrm{p}\right) \triangleq \dfrac{\alpha_i\xi_i I_i\left( \mathrm{p}\right)^{-1}+ \bar{\gamma}_i I_i\left( \mathrm{p}\right)}{\alpha_i + 1}$, $P^{\mathsf{max}}_i$ is the transmission power budget of UE $i$, and $\xi_i=P_i^{\mathsf{max}2}/\bar{\gamma}_i$. 
Here, HPC algorithm is presented in a general form with parameters $\alpha_i$. 
Specifically, for $\alpha_i=0$ the HPC becomes the TPC algorithm where each UE $i$ aims to achieve its target SINR.
When $\alpha_i \rightarrow \infty$, each UE $i$ attempts to achieve higher SINR (if it is
in favorable condition) and HPC becomes OPC. 
The convergence of the proposed HPC algorithm can be proved by showing that its puf is 2.s.s. as stated in Theorem 2 of [J1].
 
\subsubsection{Two-Time-Scale Adaptive Algorithm}
\label{sec:Gnrl_Time_Scal_Sepa}
To complete our design objective using Algorithm~\ref{alg:gms2_r3} and HPC, we develop an adaptive algorithm as follows. 
First, we name UEs whose SINRs are greater than or equal to its target SINR as supported UEs, and the others as non-supported. 
Note also that UE $i$ is non-supported if $I_i \left( \mathrm{p}\right) > I^{\mathsf{thr}}_i=P^{\mathsf{max}}_i/\bar{\gamma}_i$, and vice versa. 
In addition, $\Gamma_i(\mathrm{p}) \geq \bar{\gamma}_i$ if $\alpha_i >0$ and $I_i(\mathrm{p}^*) < I^{\mathsf{thr}}_i$. 
We refer to such a supported UE as a potential one which can reduce its power by adjusting its parameter $\alpha_i$ to help 
non-supported UEs enhance their target SINRs. We exploit this fact to develop the HPC adaptation algorithm, which is described in Algorithm \ref{alg:gms3_r3}.

\renewcommand{\baselinestretch}{0.9}
\small
\begin{algorithm}
\caption{\textsc{HPC Adaptation Algorithm}}
\label{alg:gms3_r3}
\begin{algorithmic}[1]
\STATE Initialization: Set $\mathrm{p}^{(0)}=0$, $\Delta^{(0)}$ as $\alpha_i^{(0)}=0$ for voice UE and $\alpha_i^{(0)}=\alpha_0$ ($\alpha_0 \gg 1$) for data UE, $\overline{N}^{\ast}=|\overline{\mathcal{U}}^{(0)}|$ and $\Delta^{\ast}=\Delta^{(0)}$.

\STATE Iteration $l$: 

- Run HPC algorithm until convergence with $\Delta^{(l)}$.

- If $|\overline{\mathcal{U}}^{(l)}|>\overline{N}^{\ast}$, set $\overline{N}^{\ast}=|\overline{\mathcal{U}}^{(l)}|$ and $\Delta^{\ast}=\Delta^{(l)}$.

- If $\underline{\mathcal{U}}^{(l)} = \varnothing$ or $\overline{\mathcal{U}}^{(l)} = \varnothing$, then go to step 4.

- If $\underline{\mathcal{U}}^{(l)} \neq \varnothing$ and $\overline{\mathcal{U}}^{(l)} \neq \varnothing$, then run the \textit{``updating process''} as follows:

$\;$ $\;$ a: For UE $i \in \underline{\mathcal{U}}^{(l)}$, set $\alpha_i^{(l+1)}=\alpha_i^{(l)}$.

$\;$ $\;$ b: For UE $i \in \overline{\mathcal{U}}^{(l)}$, set $\alpha_i^{(l+1)}$ so that 

$\;$ $\;$ $\;$ i) $\alpha_i^{(l)} > \alpha_i^{(l+1)} \geq 0$ if $\alpha_i^{(l)}>0$.

$\;$ $\;$ $\;$ ii) $\alpha_i^{(l+1)} = \alpha_i^{(l)}$ if $\alpha_i^{(l)}=0$.

\STATE Increase $l$ and go back to step 2 until there is no update request for $\Delta^{(l)}$. 
\STATE Set $\Delta := \Delta^{\ast}$ and run the HPC algorithm until convergence.
\end{algorithmic}
\end{algorithm}
\renewcommand{\baselinestretch}{1.4}
\normalsize

Let $\overline{\mathcal{U}}^{(l)}$ and $\underline{\mathcal{U}}^{(l)}$ be the sets of supported and non-supported UEs in iteration $l$, respectively.
We use $\overline{N}^{\ast}$ to keep the number of supported UEs during the course of the algorithm.
The \textit{``updating process''} is designed in such a way that all parameters $\alpha_i$ of supported UEs tend to zero
if all non-supported UEs cannot be saved. The proposed HPC adaptation algorithm achieve the following desirable performance
as characterized in the following theorem. 

\begin{theorem}
\label{thm04}
Let $\overline{N}_{\sf HPC}$ and $\overline{N}_{\sf TPC}$ be the numbers of supported UEs due to the proposed HPC adaptation algorithm and the TPC algorithm, respectively. Then, we have

- $\overline{N}_{\sf HPC} \geq \overline{N}_{\sf TPC}$.

- If the network is feasible (i.e., all SINR requirements can be fulfilled by the TPC algorithm), then all UEs achieve their target SINRs by using HPC adaptation algorithm; and there exist feasible UEs who achieve SINRs higher than the target values under the HPC adaptation algorithm.
\end{theorem}

An exemplifying design of the updating process is presented in Algorithm 4 of [J2] in which $\alpha_i$ is updated by both local and global 
processes, i.e., locally updating $\alpha_i$ within each cell and globally updating $\alpha_i$ when there are UEs not supported 
after performing the local updating process. 

\subsubsection{Numerical Results}

\begin{figure}[!t]
        \centering
                \includegraphics[width=0.7 \textwidth]{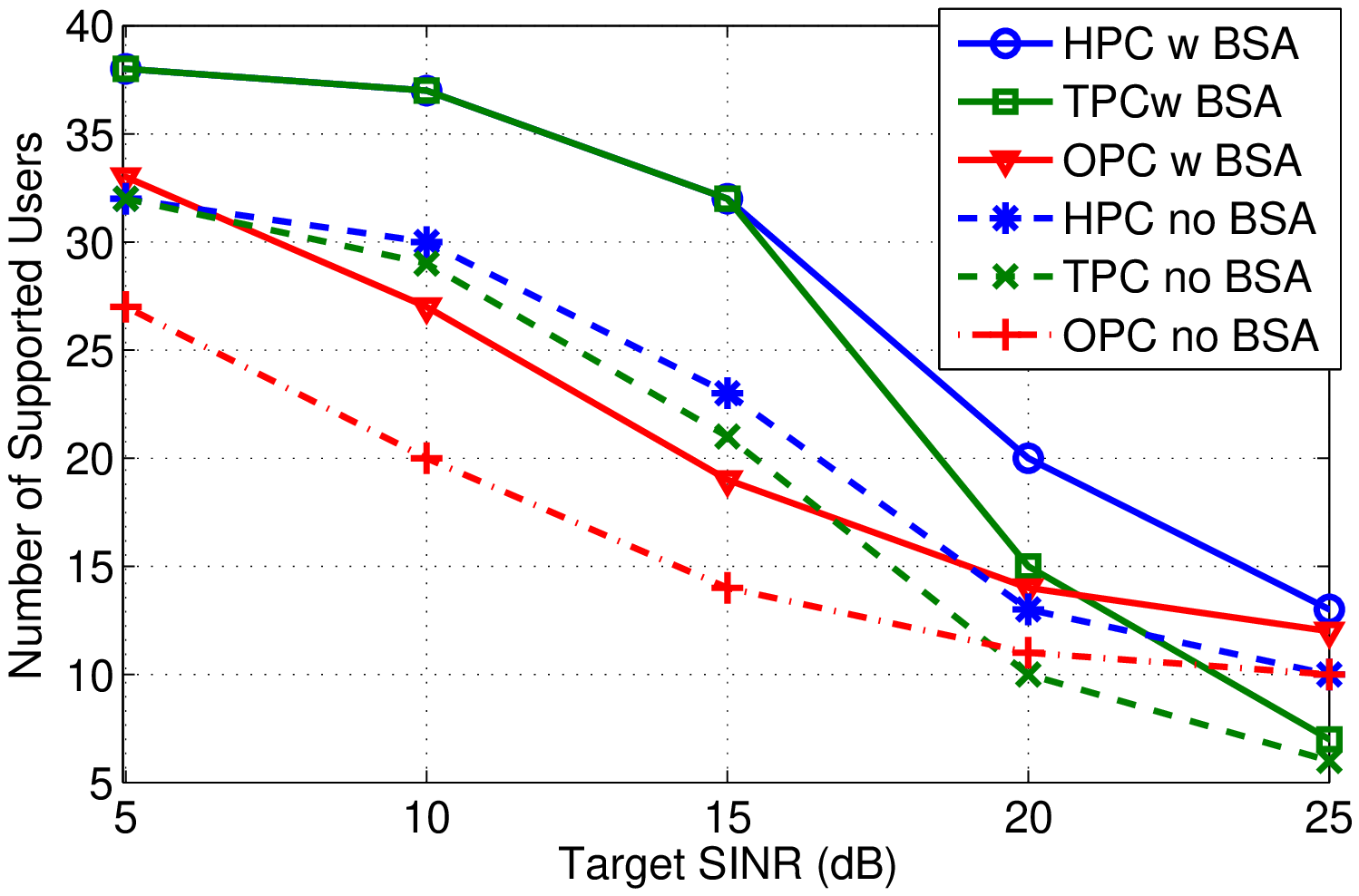}
                \caption{Number of supported UEs in HPC, TPC and OPC schemes versus target SINR.}
                \label{r3_fig1a}
\end{figure}

We now present illustrative numerical results to demonstrate the performance of the proposed algorithms. 
We study the setting where $10$ MUEs and FUEs ($1$ to $3$ for each cell) are randomly located inside circles of radii of $r_1 = 1000\:m$ and $r_2 =50\:m$, respectively. 
The channel power gain $h_{ij}$ is chosen according to the distance and Rayleigh fading. 
Other parameters are set as follows: $G=128$, $P_i^{\mathsf{max}} = 0.01 \: W$,$\eta_i=10^{-13} \: W$, $W_l=12\;dB$. Fig.~\ref{r3_fig1a} illustrates the number of supported UEs for different schemes with and without using BSA algorithm. 
It can be seen that our proposed HPC adaptation algorithm can maintain the SINR requirements for the larger number of UEs compared to TPC and OPC algorithms. In additional, the adaptive BSA algorithm helps increase the number of supported UEs when it is jointly integrated with these PC algorithms. 

\begin{figure}[!t]
        \centering        
                \includegraphics[width=0.7 \textwidth]{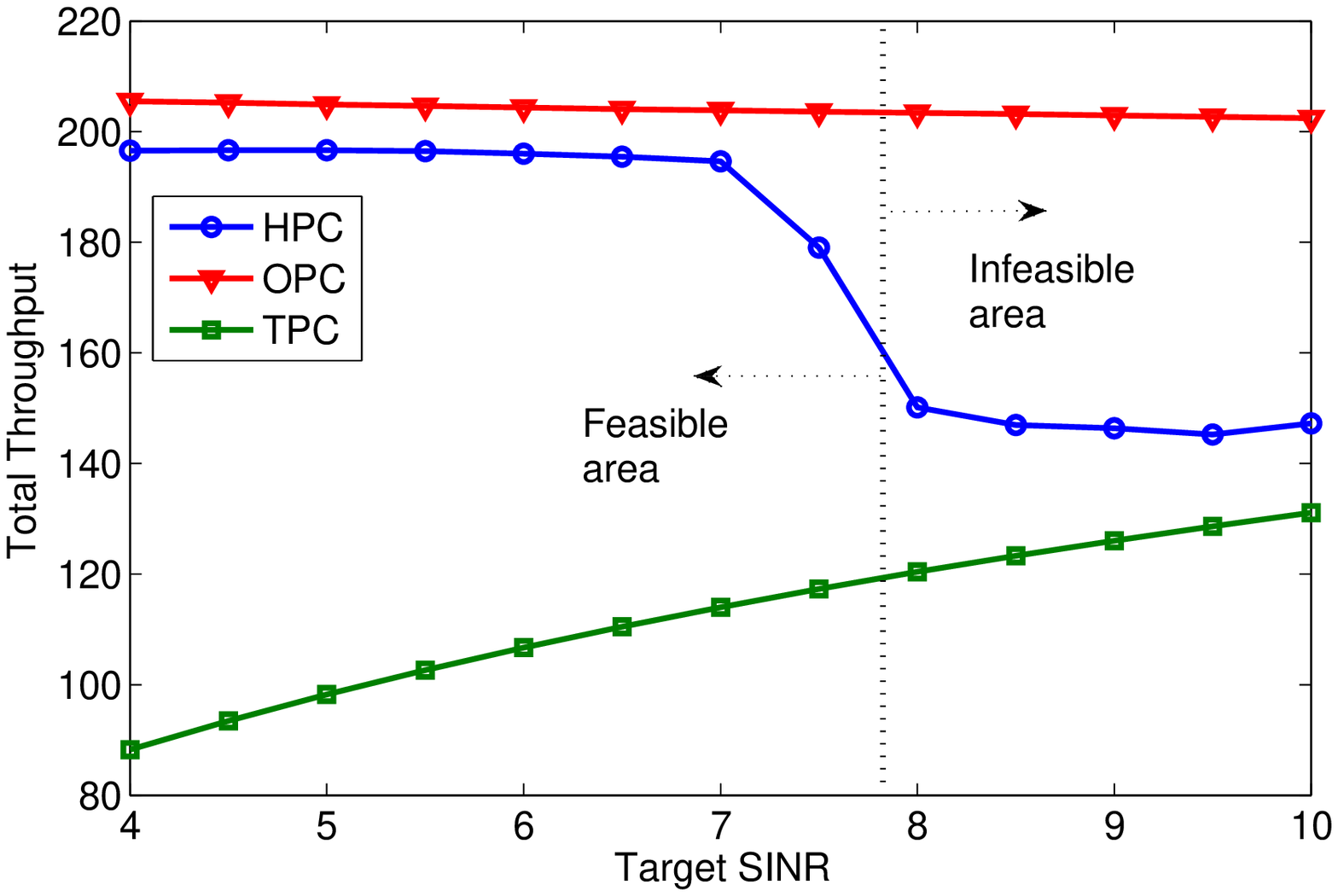}
                \caption{Average spectral efficiency vs. target SINR.}
 \label{r3_fig1b}
\end{figure}

Figs.~\ref{r3_fig1b} illustrates the average throughput achieved by different schemes versus the target SINR. 
Here, the throughput of UE $i$ is calculated as $\log_2(1+\Gamma_i)$ (b/s/Hz).
As can be observed, the average throughput achieved by our HPC adaptation algorithm is higher than that due to 
the TPC algorithm, but lower than the one due to OPC algorithm. 
Interestedly, when the network is lowly-loaded ($\bar{\gamma}_i < 8$, our proposed algorithm attains much higher 
throughput than TPC. However,the gap becomes smaller when the network load is higher ($\bar{\gamma}_i \geq 8$).
This is because Algorithm \ref{alg:gms3_r3} attempts to maintain the SINR requirements for the larger number of UEs.

\subsection{Fair Resource Allocation for OFDMA-based HetNets}

In this contribution, we consider the joint subchannel allocation and power control problem for orthogonal frequency division
 multiple access (OFDMA) femtocell networks. \nomenclature{OFDMA}{Orthogonal Frequency Division
 Multiple Access}
The design of efficient radio resource management for multi-tier OFDMA-based HetNets is an important research topic \cite{perez09}
and there have been some existing works in this area \cite{chuhan11,yanzan12,yushan12,Wangchi12,Long12,hoon11}. 
These existing works, however, do not consider PC in their resource allocation algorithms or do not provide QoS guarantees for UEs of both network tiers.
Our design aims to jointly design subchannel and power allocation for femtocells considering 
fairness for FUEs in each femtocell, QoS protection for MUEs, and the maximum power constraints. 
To the best of our knowledge, none of the existing works jointly consider all these design issues. 
In particular, we make the following contributions.

\begin{itemize}
\item We present a fair uplink resource allocation formulation for the two-tier HetNets which aims at maximizing the total minimum rate of 
all femtocells subject to different QoS and system constraints.

\item We present both optimal exhaustive search algorithm and low-complexity distributed subchannel allocation and power control algorithm.

\item We then describe how to extend the proposed design to different scenarios including downlink setting, adaptive multiple-rate transmission, 
and open access design.
\end{itemize}

\subsubsection{System Model}

We consider a system with full frequency reuse in which there are $M_{\sf f}$ FUEs served by $(K-1)$ FBSs, which 
are underlaid by one macrocell serving $M_{\sf m}$ MUEs over $N$ subchannels.
We describe the subchannel assignments (SAs) as matrix $\mathrm{\bf{A}} \in \mathfrak{R}^{M \times N}$ where
\nomenclature{SA}{Subchannel Assignment}
\begin{equation}
\label{eq:A}
\mathrm{\bf{A}}(i,n)=a^n_i=\left\lbrace \begin{array}{*{10}{l}}
1 & \text{if  subchannel $n$ is assigned for UE $i$}\\
0 & \text{otherwise}.
\end{array}
\right. 
\end{equation}
We assume that the M-QAM modulation is adopted for communications on each subchannel where the constellation size $s$. 
\nomenclature{QAM}{Quadrature Amplitude Modulation}
According to 
\cite{proakis01}, the target SINR for the $s\text{-QAM}$ modulation scheme can be calculated as
\beq \label{eq:tSINR}
\bar{\gamma}(s) = \dfrac{\left[ \mathbb{Q}^{-1}( \overline{P}_e /x_s)\right]^2}{y_s}, \: s \in \mathfrak{M},
\eeq
where $\mathbb{Q}(.)$ stands for the Q-function, $x_s=\frac{2(1-1/\sqrt{s})}{\log_2s}$, $y_s=\frac{3}{2(s-1)}$, and $\overline{P}_e$ is the target BER value. 
\nomenclature{BER}{Bit Error Rate}
Then, if $s\text{-QAM}$ modulation scheme is employed, the
spectral efficiency per one Hz of system bandwidth is ${\log_2s}/{N}$ (b/s/Hz).

Similar to the SAs, we define $\mathrm{\bf{P}}$ as an $M \times N$ power allocation (PA) matrix where $\mathrm{\bf{P}}(i,n)=p^n_i$. 
\nomenclature{PA}{Power Allocation}
Then, for a given SA and PA solution, i.e., given $\mathrm{\bf{A}}$ and $\mathrm{\bf{P}}$, the SINR achieved at
 BS $b_i$ due to the transmission of UE $i$ over subchannel $n$ can be written as
\begin{equation}
\label{sinr}
\Gamma_i^n(\mathrm{\bf{A}},\mathrm{\bf{P}})=\dfrac{a_i^n h_{b_i i}^n p_i^n}{\sum_{j \notin \mathcal{U}_{b_i}}{a_j^n h_{b_ij}^n p_j^n}+\eta_{b_i}^n}=\dfrac{a_i^np_i^n}{I_i^n(\mathrm{\bf{A}},\mathrm{\bf{P}})}.
\end{equation}

We assume that $\mathrm{\bf{A}}_1$ (SA for macrocell) is
fixed while we need to determine $\mathrm{\bf{A}}_k, \: 2 \leq k \leq K$ (SA for femtocells) and the corresponding PAs.
To protect the QoS of the licensed MUEs, we wish to maintain a predetermined target SINR $\bar{\gamma}_i^n$ for each of its 
assigned subchannel $n$. For FUEs, the spectral efficiency (bits/s/Hz) achieved by FUE $i$ on one subchannel can be written as 
\begin{equation}
\label{eq:rate_n}
r_i^n(\mathrm{\bf{A}},\mathrm{\bf{P}}) = \left\lbrace \begin{array}{*{5}{l}}
0, & \text{if } \Gamma_i^n(\mathrm{\bf{A}},\mathrm{\bf{P}}) < \bar{\gamma}_i^n,\\
r_{\sf f}, & \text{if } \Gamma_i^n(\mathrm{\bf{A}},\mathrm{\bf{P}}) \geq \bar{\gamma}_i^n,
\end{array} \right. 
\end{equation}
where $r_{\sf f}=(1/n)\log_2s^{\sf f}$ and $\bar{\gamma}_i^n = \bar{\gamma}(s^{\sf f})$.
Now, we can express the total spectral efficiency achieved by UE $i$ for given SA and PA matrices, $\mathrm{\bf{A}}$ and $\mathrm{\bf{P}}$, as $R_i(\mathrm{\bf{A}},\mathrm{\bf{P}})=\sum_{n=1}^N r_i^n(\mathrm{\bf{A}},\mathrm{\bf{P}})$.
To impose the max-min fairness for all FUEs associated with the same FBS, the uplink resource allocation problem for FUEs can be formulated as follows:
\begin{align}
\label{objfun_r4}
\mathop {\max} \limits_{(\mathrm{\bf{A}},\mathrm{\bf{P}}) } & \sum \limits_{2\leq k \leq K} \mathrm{R}^{(k)}(\mathrm{\bf{A}},\mathrm{\bf{P}}) =  \sum \limits_{2\leq k \leq K} \min_{i \in \mathcal{U}_k} R_i(\mathrm{\bf{A}},\mathrm{\bf{P}}) \\
\text{s. t.} & \sum_{i \in \mathcal{U}_k} a_i^n \leq 1, \:\:\: \forall k \in \mathcal{B} \text{ and } \forall n \in \mathcal{N}, \label{eq:c4} \\
{} & \sum_{n=1}^{N}p^n_i \leq P_i^{\texttt{max}}, \quad i \in \mathcal{U}, \label{powcon} \\
{} &  \Gamma_i^n(\mathrm{\bf{A}},\mathrm{\bf{P}}) \geq \bar{\gamma}_i^n,\:\: \text{if} \:\: a_i^n=1, \:\: \forall i \in \mathcal{U}_{\sf m}.\label{eq:rate_cond}
\end{align}
This resource allocation problem is a mixed integer program, which is, therefore, NP-hard. 
\nomenclature{NP-hard}{Non-deterministic Polynomial-time hard}

\subsubsection{Optimal Exhaustive Search Algorithm}

For each subchannel $n$, we need to maintain the SINR constraints $\Gamma^n_i(\mathrm{\bf{A}},\mathrm{\bf{P}}) \geq \bar{\gamma}_i^n$ 
for all UEs who are allocated with this subchannel. Hence, the feasibility of a particular SA can be verified by using the
Perron-Frobenius theorem as presented in Section~\ref{Ch2_sec_PC}.
Since the number of possible SAs is finite, the  optimal exhaustive
search algorithm can be developed as follows. 
For a fixed and feasible $\mathrm{\bf{A}}_1$, let $\Omega\{\mathrm{\bf{A}}\}$ be the list of all potential SA solutions
 that satisfy the SA constraints (\ref{eq:c4}) and the fairness condition: $\sum_{n \in \mathcal{N}} a_i^n=\sum_{n \in \mathcal{N}} a_j^n=\tau_k$ for all FUEs $i,j \in \mathcal{U}_k$. 
Then, we sort the list $\Omega\{\mathrm{\bf{A}}\}$ in the decreasing order of $\sum_{k=2}^K \tau_k$ and obtain the sorted list $\Omega^{\ast}\{\mathrm{\bf{A}}\}$. 
Among all feasible SA solutions, the feasible one achieving the highest value of the objective function (\ref{objfun_r4}) and its corresponding PA solution is the optimal solution.

\textit{Complexity Analysis:} By calculating the cardinality of $\Omega^{\ast}\{\mathrm{\bf{A}}\}$ and the complexity involved in the feasibility verification for each of them, the complexity of the exhaustive search algorithm can be expressed as 
$O\left(K^3 \times N \times (N!)^{(K-1)}\right)$, which is exponential in the numbers of subchannels and FBSs. This optimal exhaustive search algorithm will
be employed as a benchmark to evaluate the performance of the low-complexity algorithm presented in the following.

\subsubsection{Sub-Optimal and Distributed Algorithm}

Our sub-optimal algorithm aims to assign the maximum equal number of subchannels to FUEs
in each femtocell and to perform Pareto-optimal PA for FUEs and MUEs on each subchannel so that they meet the SINR constraints
in (\ref{eq:rate_n}) and (\ref{eq:rate_cond}). 
To achieve this design goal, we propose a novel resource allocation algorithm which is described in details in Algorithm~\ref{alg:gms1_r4}. 
The key operation in this 
algorithm is the iterative weight-based SA which is performed in parallel at all femtocells. 
The SA weight for each
subchannel and FUE pair is defined as the multiplication of the estimated transmission power and a scaling factor capturing the 
quality of the corresponding allocation.  
Specifically, each UE $i$ in cell $k$ estimates the transmission power on subchannel $n$ in each iteration $l$ of the algorithm
 by using the TPC algorithm as $p_i^{n,\texttt{min}}=  I_i^n(l) \bar{\gamma}_i^n$. 
Then, the assignment weight for FUE $i$ on subchannel $n$ in cell $k$ can be defined as $w_{i}^n = \chi_i^n p_i^{n,\texttt{min}}$
where the scaling factor $\chi_i^n$ is defined as follows: 
\begin{equation}
\label{eq:w2}
\chi_i^n  = \left\lbrace \begin{array}{*{5}{l}}
\alpha_i^{n}, & \text{if } p_i^{n,\texttt{min}} \leq \frac{P_i^{\texttt{max}}} {\tau_k}\\
\alpha_i^{n} \theta_i^{n}, & \text{if } \frac{P_i^{\texttt{max}}} {\tau_k} < p_i^{n,\texttt{min}} \leq P_i^{\texttt{max}}\\
\alpha_i^{n} \delta_i^{n}, & \text{if }  P_i^{\texttt{max}} < p_i^{n,\texttt{min}},
\end{array} \right. 
\end{equation}
where $\tau_k$ denotes the number of subchannels assigned for each FUE in femtocell $k$; 
$\alpha_i^{n} \geq 1$ is a factor helping MUEs maintain their target SINR (i.e.,
it is increased if $\mathrm{\bf{A}}(i,n)=1$ results in violation of the SINR constraint of MUE); 
$\theta_i^{n},\delta_i^{n} \geq 1$ are another factors which are set higher if $\mathrm{\bf{A}}(i,n)=1$ tends to require transmission
power larger than the average power per subchannel (i.e., $\frac{P_i^{\texttt{max}}} {\tau_k}$) and the maximum power budget (i.e., $P_i^{\texttt{max}}$), respectively.

Given the weights defined for each FUE $i$, femtocell $k$ finds its SA by using the standard Hungarian algorithm
(i.e., \textit{Algorithm 14.2.3} given in \cite{Jungnickel08}) to solve the following problem
\beq
\min \limits_{\mathrm{\bf{A}}_k} \sum \limits_{i \in \mathcal{U}_k}\sum \limits_{n \in \mathcal{N}} a_i^n w_i^n \;\; \text{s.t.} \;\; \sum_{n \in \mathcal{N}} a_i^n = \tau_k \:\: \forall i \in \mathcal{U}_k. \label{eq:lcl_prob}
\eeq

\renewcommand{\baselinestretch}{0.9}
\small
\begin{algorithm}
\caption{\textsc{Distributed Uplink Resource Allocation}}
\label{alg:gms1_r4}
\begin{algorithmic}[1]
\STATE Initialization: Set $p_i(0)=0$ $\forall i$, feasible $\mathrm{\bf{A}}_1$, $\tau_k=\lfloor \frac{N}{\vert \mathcal{U}_k \vert} \rfloor$ and $\varrho_k=0$ for all FBSs, and $\alpha_i^{n}=\theta_i^{n}=1$, $\delta_i^{n}=N$ for all FUEs and subchannels.

\STATE \textbf{For the macrocell:}
\STATE MBS estimates calculates $\left\lbrace p_i^{n,\texttt{min}}\right\rbrace $ and checks the power constraint (PCON).

\IF{PCON satisfies} \STATE Update power as $p_i^{n,\texttt{min}}$.
\ELSIF{PCON not satisfies} \STATE Update power as $p_i^{n,\texttt{min}}$ with a scale down factor.
\STATE Find SC using most power, and FUE generating most interference in this SC.
\STATE Increase $\alpha$ of this FUE over that SC.
\ENDIF

\STATE \textbf{For each femtocell $k \in \mathcal{B}_{\sf f}$:}

\STATE Each FBS $k$ estimates $\left\lbrace p_i^{n,\texttt{min}}\right\rbrace $.

\IF{$\varrho_k=1$}  \STATE Keep $\mathrm{\bf{A}}_k(l)=\mathrm{\bf{A}}_k(l-1)$ 
\ELSIF{$\varrho_k=0$}  \STATE Calculate subchannel assignment weights $\left\lbrace w_{i_u}^n\right\rbrace $ and run 
Hungarian algorithm to obtain $W_k(l)$ and $\mathrm{\bf{A}}_k(l)$.
\IF{$W_k(l) > V \sum _{i \in \mathcal{U}_k}P_i^{\texttt{max}}$} \STATE Set $\tau_k:=\tau_k-1$. 
\ENDIF
\ENDIF
\STATE Check the power constraint (PCON). 
\IF{PCON satisfies} \STATE Update power as $p_i^{n,\texttt{min}}$ and set $\varrho_{k,i}=1$. 
\ELSIF{PCON not satisfies} \STATE Update power as $p_i^{n,\texttt{min}}$ with a scale down factor.
\STATE Find SC and FUE which spent most power.
\STATE Increase $\theta$ of that FUE on that SC.
\ENDIF
\STATE Set $\varrho_{k}=\prod_{i \in \mathcal{U}_k}\varrho_{k,i}$.
\STATE Let $l:=l+1$, return to step 2 until convergence.
\end{algorithmic}
\end{algorithm}
\nomenclature{PCON}{Power Constraint}
\renewcommand{\baselinestretch}{1.4}
\normalsize

\textit{Complexity Analysis:} The complexity of our proposed algorithm is $O(K \times N^3)$ for each iteration due to the
 Hungarian SA process in step 16. However, the local SAs can be performed in parallel at all $(K-1)$ 
femtocells. Therefore, the run-time complexity of our algorithm is $O(N^3)$ multiplied by the number of required iterations.

\subsubsection{Numerical Results}

To conduct computer simulations, MUEs and FUEs are randomly located inside circles of radii of $r_1 = 1000\:m$ and $r_2 =30\:m$, respectively.  
The power channel gains $h_{ij}^n$ are generated by considering both 
Rayleigh fading and the distances. 
Other parameters: $\eta_i=10^{-13} \: W$, $W_l=12\;dB$. In Fig.~\ref{r4_fig1a}, 
we show the total minimum spectral efficiency of all femtocells (i.e., the optimal objective value of (\ref{objfun_r4})) versus
 the constellation size of FUEs ($s^{\sf f}$) for the
 small network due to both optimal and sup-optimal algorithms (Algorithm~\ref{alg:gms1_r4}).
As can be seen, for the low constellation sizes (i.e., low target SINRs), Algorithm~\ref{alg:gms1_r4} can achieve almost
 the same spectral efficiency as the optimal one while for higher values of $s^{\sf f}$, it results in just slightly lower spectral efficiency than that due to the optimal one. 
\begin{figure}[!t]
\centering
\includegraphics[width=0.7\textwidth]{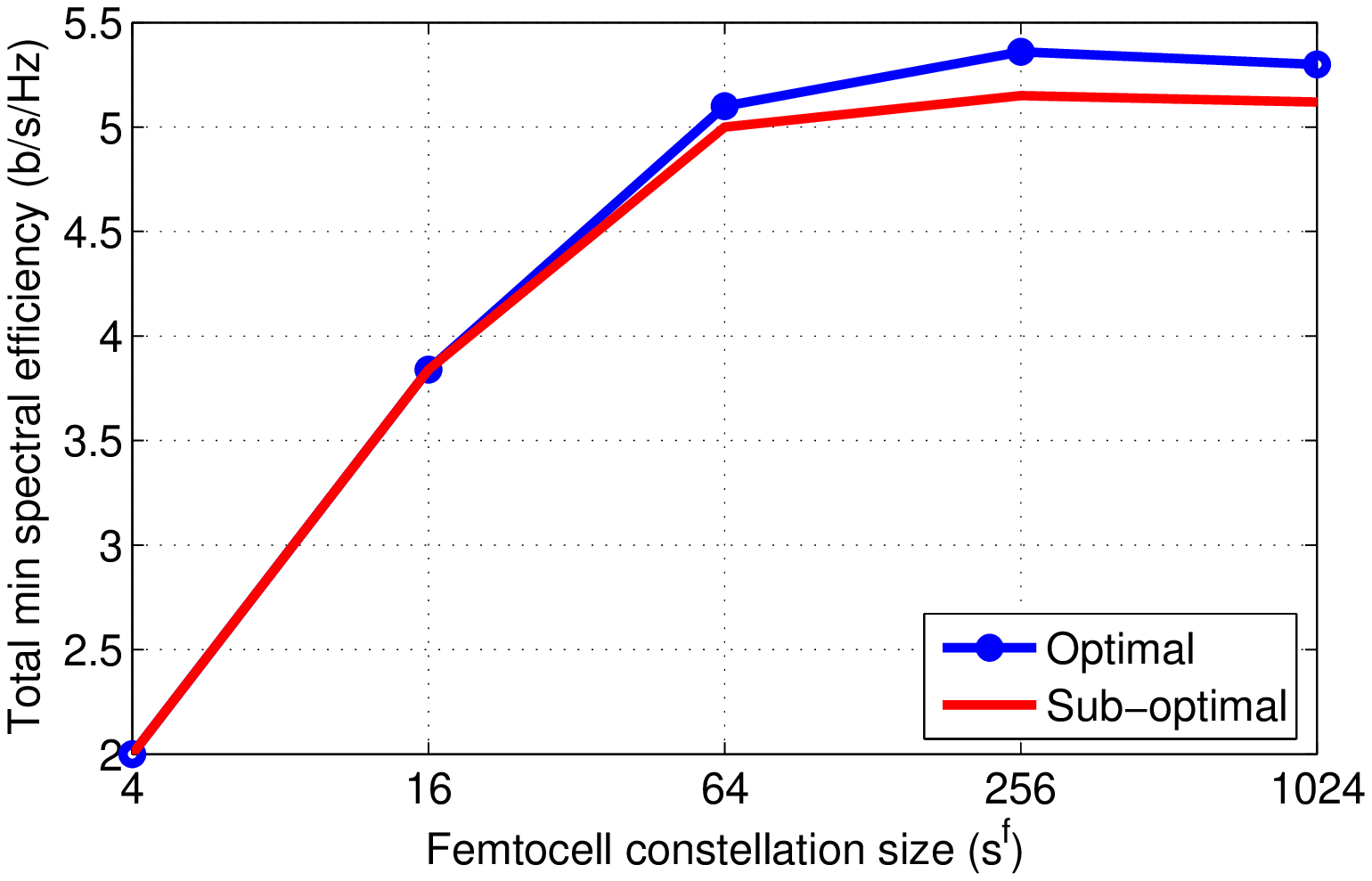}
\caption{Total minimum spectral efficiency under optimal and sup-optimal algorithms with $P_i^{\mathsf{max}} = 0.01 \: W$.}
\label{r4_fig1a}
\end{figure}

In Figs.~\ref{r4_fig:pf} and \ref{r4_fig:pm}, we plot the total femtocell minimum spectral efficiency versus the maximum power of FUEs ($P_{\sf f}^{\texttt{max}}$) and MUEs ($P_{\sf m}^{\texttt{max}}$), respectively, for different modulation levels of MUEs (the modulation scheme of FUEs is $256$-QAM). These figures show that the total 
minimum spectral efficiency  increases with the increases of maximum power budgets, $P_{\sf f}^{\texttt{max}}$ or $P_{\sf m}^{\texttt{max}}$. However, this value is
saturated as the maximum power budgets $P_{\sf f}^{\texttt{max}}$ or $P_{\sf m}^{\texttt{max}}$ become sufficiently large. In addition, as the number of MUEs increases, the total femtocell minimum spectral efficiency increases thanks to the better diversity gain offered by the macro tier.

\begin{figure}[!t]
\centering
\includegraphics[width=0.7\textwidth]{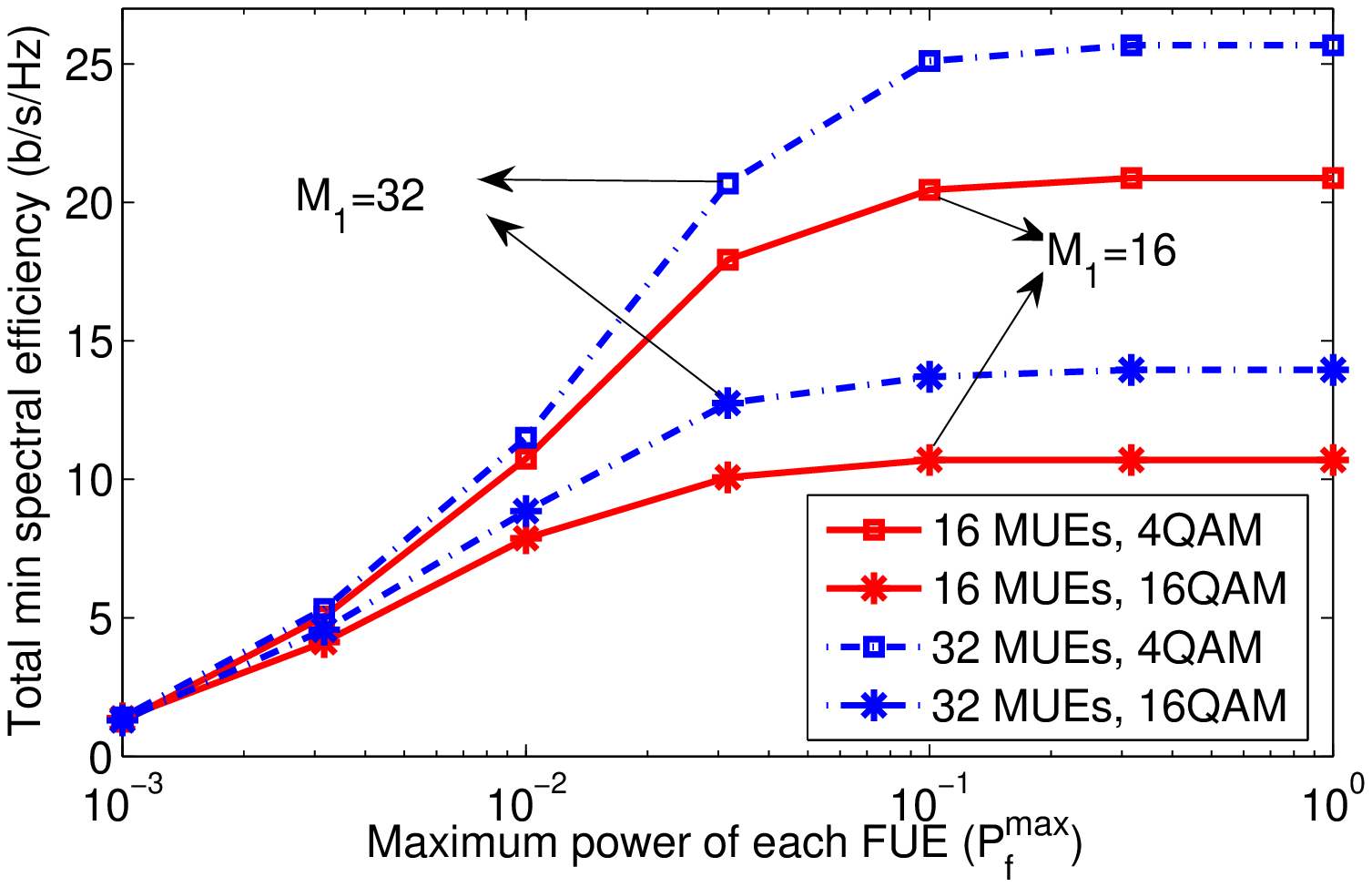}
                \caption{Total minimum spectral efficiency versus $P_{\sf f}^{\texttt{max}}$. }
                \label{r4_fig:pf}
\end{figure}

\begin{figure}[!t]
\centering
\includegraphics[width=0.7\textwidth]{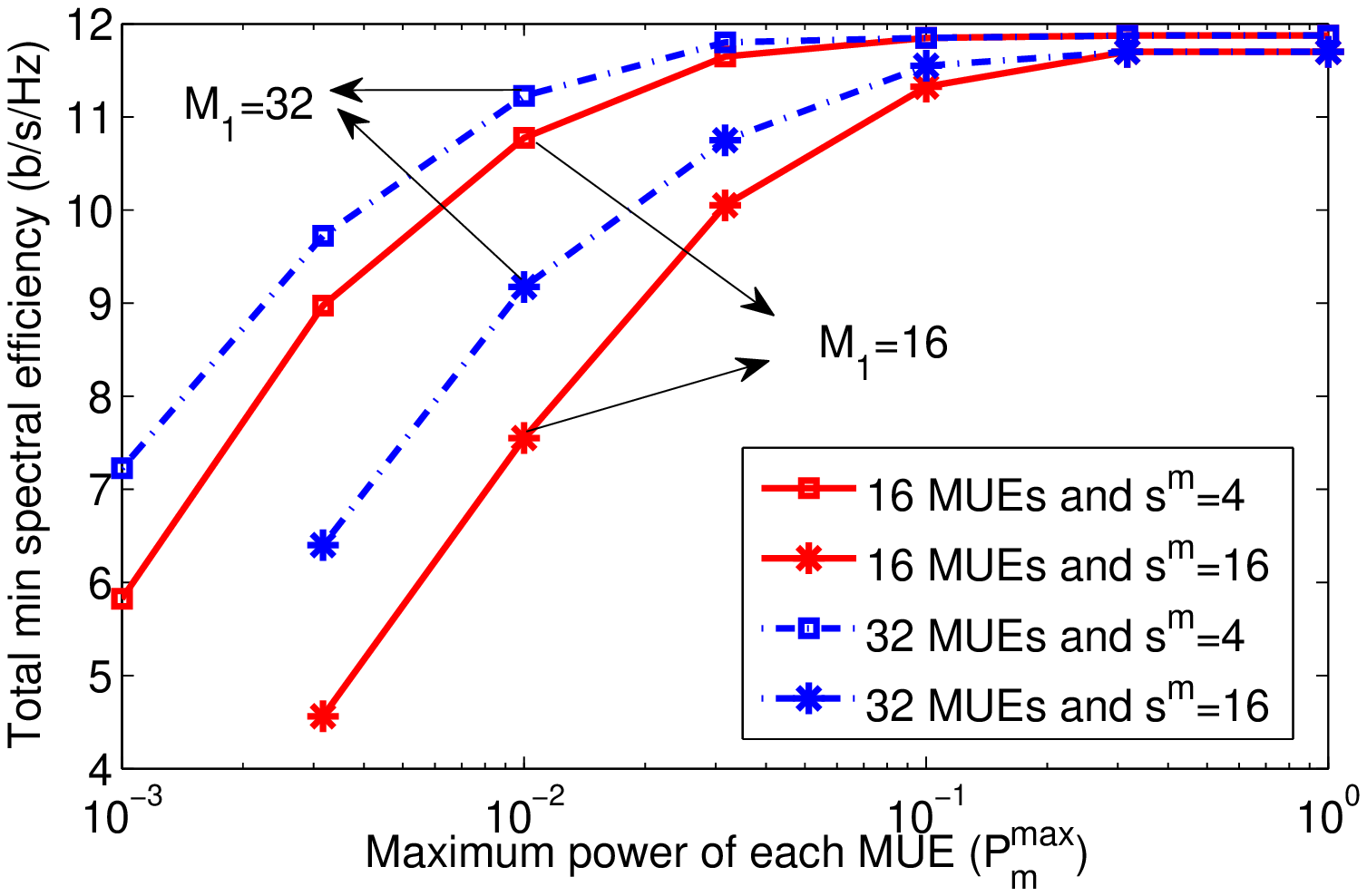}
                \caption{Total minimum spectral efficiency versus $P_{\sf m}^{\texttt{max}}$. }
                \label{r4_fig:pm}
\end{figure}

\subsection{Joint Transmission Design for C-RANs with Limited fronthaul Capacity}
In this contribution, we consider the CoMP joint transmission design for C-RAN 
that explicitly considers the fronthaul capacity as well as UEs' QoS constraints. 
In particular, we make the following contributions.

\begin{itemize}
\item
We formulate the joint transmission design problem for C-RAN which optimizes the set of RRHs serving each UE
together with their precoding and power allocation solutions to minimize the total transmission power considering
UEs' QoS and fronthaul capacity constraints. 
\item
We  develop two different low-complexity algorithms, namely \textit{Pricing-Based Algorithm} and \textit{Iterative Linear-Relaxed Algorithm},
to solve the underlying problem. 

\item 
We also study the extended settings where there are multiple individual FH capacity constraints and the multi-stream transmission scenario.
\end{itemize}

\subsubsection{System Model}

We consider the joint transmission design for CoMP downlink communications in the C-RAN with $K$ RRHs and $M$ UEs.
Let $\mathcal{K}$ and $\mathcal{U}$ be the sets of RRHs and UEs, respectively.
Suppose that RRH $k$ is equipped with $N_k$ antennas and each UE has a single antenna. 
Also, denote $\mathbf{v}_u^{k} \in \mathbb{C}^{N_k \times 1}$ as the precoding vector at RRH $k$ corresponding to the transmission to UE $u$. Then, the SINR achieved by UE $u$ can be described as
\beq \label{eq:SINR}
\Gamma_u =\dfrac{ \left|  \sum \limits_{k \in \mathcal{K}} \mathbf{h}_u^{kH} \mathbf{v}_u^{k}\right| ^2 }
{\sum \limits_{i =1, \neq u}^{M} \left| \sum \limits_{l \in \mathcal{K}}\mathbf{h}_u^{lH} \mathbf{v}_i^{l}\right|^2 + \sigma^2 }.
\eeq
Let $\mathbf{p}^k=[p^k_1 ... p^k_M]^T$ be the transmission power vector of RRH $k$, and $\mathbf{p}=[\mathbf{p}^{1 T} ... \mathbf{p}^{K T}]^T$. 
Note that $p^k_u=0$ implies that RRH $k$ does not serve UE $u$.
In contrast, if $p^k_u > 0$, the fronthaul link from RRH $k$ to UE $u$ is activated for carrying the baseband signal to serve UE $u$ 
at its required target SINR. 
Therefore, the total consumed capacity of the fronthaul links can be captured by the transmission power vector $\mathbf{p}$ and the target SINR of UEs, which can be 
written mathematically as
\beq \label{eq:Ck}
G(\mathbf{p})= \sum \limits_{k \in \mathcal{K}} \sum \limits_{u \in \mathcal{U}}  \delta( p_u^k ) R_u^{\sf{fh}}   
\eeq 
where $\delta(\cdot)$ denotes the step function, and $R_u^{\sf{fh}}$ represents the required capacity corresponding to UE $u$. In particular, $R_u^{\sf{fh}}$ can be expressed as \cite{dirk_14} $R_u^{\sf{fh}} = Q_{\sf{fh}} \log_2(1+\bar{\gamma}_u)$ where $Q_{\sf{fh}}$ denotes a multiplier factor.
Now, it is ready to state the FCPM problem as follows: \nomenclature{FCPM}{Fronthaul Constraint Power Minimization}
\begin{align}
(\mathcal{P}_{\sf{FCPM}}) \;\;\;\; \min \limits_{\lbrace \mathbf{v}_u^k \rbrace, \mathbf{p} } & \;\;\;\;\;\;\;\; \sum \limits_{k \in \mathcal{K}}\sum \limits_{u \in \mathcal{U}} \mathbf{v}_u^{kH}\mathbf{v}_u^{k}  \label{obj_1} \\
\;\;\;\;\;\;\;\; \text{s. t. } & \; \Gamma_u \geq \bar{\gamma}_u, \qquad \forall u \in \mathcal{U}, \label{eq:SINR_constraint} \\  
{} & \; \sum \limits_{u \in \mathcal{U}} p^k_u = \sum \limits_{u \in \mathcal{U}} \mathbf{v}_u^{kH}\mathbf{v}_u^{k} \leq P_k, \;\;\; \forall k \in \mathcal{K}, \label{eq:pwc} \\
{} & \; G(\mathbf{p})= \sum \limits_{k \in \mathcal{K}} \sum \limits_{u \in \mathcal{U}}  \delta( p_u^k ) R_u^{\sf{fh}}  \leq C, \label{eq:C_cons}
\end{align}
where $\bar{\gamma}_u$ denotes the target SINR of UE $u$, $P_k$ ($k \in \mathcal{K}$) denotes the maximum power of RRH $k$, and $C$ is denoted as the fronthaul capacity limit.

\subsubsection{Pricing-Based Algorithm}
The first low-complexity algorithm is developed by employing the penalty method to deal with the step-function fronthaul capacity constraint.
Specifically, we consider the so-called \textit{Pricing-Based fronthaul Capacity and Power Minimization} (PFCPM) problem, which is defined as
\nomenclature{PFCPM}{Pricing-Based fronthaul Capacity and Power Minimization}
\beq
(\mathcal{P}_{\sf{PFCPM}}) \;\;\;\; \min \limits_{\lbrace \mathbf{v}_u^k \rbrace, \mathbf{p} } \;\; \Vert \mathbf{p} \Vert_{\mathbf{1}} + q G(\mathbf{p})  \;\; \text{s. t. } \;\; \text{ constraints (\ref{eq:SINR_constraint}) and (\ref{eq:pwc}).} \label{obj_2}  
\eeq
In the following, we establish some theoretical results based on which we develop the mechanism to update the pricing parameter.
\begin{proposition} \label{R6_lm1} $G_{\sf{PFCPM}}(q)$ is a decreasing function of $q$ and lower bounded by $G_{\sf{PFCPM}}(\bar{q}) $ where $\bar{q} = \sum \limits_{k \in \mathcal{K}} P_k/\sigma_{\sf{min}}$, and $\sigma_{\sf{min}}$ is the smallest non-zero value of $\vert G(\mathfrak{a})-G(\mathfrak{a}^{\prime})\vert$ where ${\lbrace\mathfrak{a},\mathfrak{a}^{\prime}\rbrace \subset \mathcal{S}_\mathfrak{a}}$.
\end{proposition}
These results form the foundation based on which we can develop an iterative algorithm presented in Algorithm \ref{R6_alg:gms2}.
In fact, we can employ the bisection search method to update the pricing parameter $q$ until $G_{\sf{PFCPM}}(\bar{q}) = C$.
\renewcommand{\baselinestretch}{0.9}
\small
\begin{algorithm}[t]
\caption{\textsc{Pricing-based Algorithm for FCPM Problem}}
\label{R6_alg:gms2}
\begin{algorithmic}[1]
\STATE Solve PFCPM problem using Alg.~\ref{R6_alg:gms3} with $q^{(0)}=\bar{q}$.
\IF{$G_{\sf{PFCPM}}(\bar{q}) > C$} \STATE Stop, the FCPM problem is infeasible.
\ELSIF{$G_{\sf{PFCPM}}(\bar{q}) = C$} \STATE Stop, the solution is achieved.
\ELSIF{$G_{\sf{PFCPM}}(\bar{q}) < C$} 
\STATE Set $l=0$, $q_{\sf{U}}^{(l)}=\bar{q}$ and $q_{\sf{L}}^{(l)}=0$.
\REPEAT 
\STATE Set $l=l+1$ and $q^{(l)}=\left(q_{\sf{U}}^{(l-1)} + q_{\sf{L}}^{(l-1)} \right)/2$.
\STATE Solve PFCPM problem using Alg.\ref{R6_alg:gms3} with $q^{(l)}$.
\IF{$G_{\sf{PFCPM}}(q^{(l)}) > C$} \STATE Set $q_{\sf{U}}^{(l)}=q_{\sf{U}}^{(l-1)}$ and $q_{\sf{L}}^{(l)}=q^{(l)}$.
\ELSIF{$G_{\sf{PFCPM}}(q^{(l)}) < C$} \STATE Set $q_{\sf{U}}^{(l)}=q^{(l)}$ and $q_{\sf{L}}^{(l)}=q_{\sf{L}}^{(l-1)}$.
\ENDIF
\UNTIL{$G_{\sf{PFCPM}}(q^{(l)}) = C$ or $q_{\sf{U}}^{(l)} - q_{\sf{L}}^{(l)}$ is too small}.
\ENDIF
\end{algorithmic}
\end{algorithm}
\renewcommand{\baselinestretch}{1.4}
\normalsize

We now develop an efficient algorithm to solve problem $\mathcal{P}_{\sf{PFCPM}}$  based on concave approximation of the step function.
Specifically, the step function $\delta(x)$ for $x \geq 0$ can be approximated by a suitable concave function.
Denote $f_{\mathsf{apx}}^{(k,u)}(p^k_u)$ as the concave penalty function that approximates the step function $\delta(p^k_u)$ corresponding to link $(k,u)$. Then,
problem  $\mathcal{P}_{\sf{PFCPM}}$ can be approximated by the following problem
\beq
 \min \limits_{\lbrace \mathbf{v}_u^k \rbrace,\mathbf{p}} \;\; \sum \limits_{k \in \mathcal{K}} \sum \limits_{u \in \mathcal{U}} p_u^k + q \sum \limits_{k \in \mathcal{K}} \sum \limits_{u \in \mathcal{U}} f_{\mathsf{apx}}^{(k,u)}\left( p_u^k \right) R_u^{\sf{fh}} \;\; 
 \text{s. t. } \;\; \text{ constraints (\ref{eq:SINR_constraint}) and (\ref{eq:pwc}).} \label{objfun2}
\eeq
By applying the gradient method, we can solve problem (\ref{objfun2}) by iteratively solving the following problem until convergence
\beq
 \min \limits_{\lbrace \mathbf{v}_u^k \rbrace} \;\;\; \sum \limits_{k \in \mathcal{K}} \sum \limits_{u \in \mathcal{U}} \alpha_u^{k(n)} \mathbf{v}_u^{kH}\mathbf{v}_u^{k} \;\; \text{s. t.} \;\; \text{ constraints (\ref{eq:SINR_constraint}) and (\ref{eq:pwc})} \label{objfun3} 
\eeq
where 
\beq \label{eq:alp}
\alpha_u^{k(n)}=  1 + q \nabla f_{\mathsf{apx}}^{(k,u)}\left( p_u^k \right)R_u^{\sf{fh}}.
\eeq
Problem (\ref{objfun3}) is a weighted sum-power 
minimization problem, which can be transformed into the convex semi-definite program (SDP) as presented in Section~\ref{Ch2_sec_PMP}.
Then, the algorithm to solve problem $\mathcal{P}_{\sf{PFCPM}}$ is presented in Algorithm~\ref{R6_alg:gms3}.
\nomenclature{SDP}{Semi-Definite Program}
Algorithm ~\ref{R6_alg:gms2}, which is proposed to solve problem $\mathcal{P}_{\sf{FCPM}}$, is based on the solution of problem $\mathcal{P}_{\sf{PFCPM}}$, which can be 
obtained by using Algorithm~\ref{R6_alg:gms3}.
In addition, the results stated in Proposition~\ref{R6_lm1} and standard properties of the gradient method guarantee the convergence of this algorithm.
\renewcommand{\baselinestretch}{0.9}
\small
\begin{algorithm}[!t]
\caption{\textsc{SDP-Based Algorithm for PFCPM Problem}}
\label{R6_alg:gms3}
\begin{algorithmic}[1]
\STATE Initialization: Set $n=0$, and $\alpha_u^{k(0)}=1$ for all RRH-UE links $(k,u)$.

\STATE Iteration $n$:
\begin{description}
\item[a.] Solve problem (\ref{objfun3}) with $\left\lbrace \alpha_u^{k(n-1)} \right\rbrace $ to obtain $(\mathbf{p}^{(n)},\lbrace \mathbf{v}_u^k \rbrace^{(n)})$.
\item[b.] Update $\left\lbrace \alpha_u^{k(n)} \right\rbrace $ as in (\ref{eq:alp}).
\end{description} 
\STATE Set $n:=n+1$ and go back to Step 2 until convergence.
\end{algorithmic}
\end{algorithm}
\renewcommand{\baselinestretch}{1.4}
\normalsize

\subsubsection{Iterative Linear-Relaxed Algorithm}

We now present the iterative linear-relaxed algorithm to directly deal with the step-function in the fronthaul capacity constraint.
First, problem $\mathcal{P}_{\sf{FCPM}}$ can be approximated by the following problem
\begin{align}
 \min \limits_{\lbrace \mathbf{v}_u^k \rbrace, \mathbf{p}} & \;\;\; \Vert \mathbf{p} \Vert_{\mathbf{1}}  \label{obj_4} \\
 \text{s. t. } & \; \text{ constraint (\ref{eq:SINR_constraint})}  \nonumber  \\
 {} & \; \sum \limits_{k \in \mathcal{K}} \sum \limits_{u \in \mathcal{U}} f_{\mathsf{apx}}^{(k,u)}( p_u^k ) R_u^{\sf{fh}} \leq C. \label{const4}
\end{align}
We next approximate $f_{\mathsf{apx}}^{(k,u)}(p_u^k)$ by a linear function, which is based on the duality properties of conjugate 
of convex functions \cite{rockafellar70} as follows:
\beq \label{eq:g2}
f_{\mathsf{apx}}^{(k,u)}( p_u^k )= \inf \limits_{z_u^k} \left[ z_u^k p_u^k - f_{\mathsf{apx}}^{(k,u)\ast}(z_u^k) \right]  
\eeq
where $f_{\mathsf{apx}}^{(k,u)\ast}(z)= \inf \limits_{w} \left[ z w -f_{\mathsf{apx}}^{(k,u)}(w) \right]$.
It can be verified that the optimization problem on the right-hand side of (\ref{eq:g2}) achieves its minimum at
\beq \label{eq:z}
\hat{z}_u^k=\nabla f_{\mathsf{apx}}^{(k,u)}( w )\vert_{w=p_u^k}.
\eeq
Hence, for a given value of $\left\lbrace \hat{z}_u^k \right\rbrace$, the problem (\ref{obj_4})-(\ref{const4}) can be reformulated to
\begin{align}
 {} & \min \limits_{\lbrace \mathbf{v}_u^k \rbrace} \:\:\: \sum \limits_{k \in \mathcal{K}} \sum \limits_{u \in \mathcal{U}} \mathbf{v}_u^{kH}\mathbf{v}_u^{k}  \label{obj_1sdp} \\
 \text{s. t. } & \; \text{ constraint (\ref{eq:SINR_constraint})}  \nonumber  \\
  {} & \!\!\!\!\! \sum \limits_{k \in \mathcal{K}} \sum \limits_{u \in \mathcal{U}} \! \hat{z}_u^k  R_u^{\sf{fh}}  \mathbf{v}_u^{kH}\mathbf{v}_u^{k} \leq   C \!  + \! \sum \limits_{k \in \mathcal{K}} \sum \limits_{u \in \mathcal{U}} \! R_u^{\sf{fh}} f_{\mathsf{apx}}^{(k,u)\ast}(\hat{z}_u^k).  \label{const5}
\end{align}
The problem (\ref{obj_1sdp})-(\ref{const5}) is indeed the well-known sum-power minimization problem, which can be solved by 
transforming it into the SDP as described in Section~\ref{Ch2_sec_PMP}.
In summary, we can fulfil our design objectives by updating $\left\lbrace \hat{z}_u^k \right\rbrace$ iteratively based 
on which we repeatedly solve the precoding optimization problem (\ref{obj_1sdp})-(\ref{const5}).
This whole algorithm is described in Algorithm~\ref{R6_alg:gms4}. 
In addition, the properties of conjugate function can guarantee the convergence of this algorithm.
\renewcommand{\baselinestretch}{0.9}
\small
\begin{algorithm}[!t]
\caption{\textsc{Iterative Linear-Relaxed Algorithm}}
\label{R6_alg:gms4}
\begin{algorithmic}[1]
\STATE Start with a feasible solution and set $l=0$.
\REPEAT 
\STATE Calculate $\left\lbrace \hat{z}_u^{k,(l)} \right\rbrace $ as in (\ref{eq:z}) for all $(k,u)$.
\STATE Solve problem (\ref{obj_1sdp})-(\ref{const5}) with $\left\lbrace \hat{z}_u^{k,(l)}\right\rbrace $.
\STATE Update $l=l+1$.
\UNTIL Convergence.
\end{algorithmic}
\end{algorithm}
\renewcommand{\baselinestretch}{1.4}
\normalsize

\subsubsection{Numerical Results}

We consider three RRHs in the simulation setting where the distances between their centers are equal to $500 \; m$. Moreover,
UEs are randomly placed inside a circle whose center is collocated with one of these three RRHs and the radius of each circle is $125 \; m$. 
The channel gains are generated by considering both Rayleigh fading and path loss. 
Other parameters are set as follows: $\sigma^2=10^{-13} \; W$, $\epsilon=10^{-6}$, $\tau=10^{-8}$, $P_k=3 \; W$, and $N_k=4$, for all $k \in \mathcal{K}$. 

\begin{figure}[!t]
\begin{center}
\includegraphics[width=0.7\textwidth]{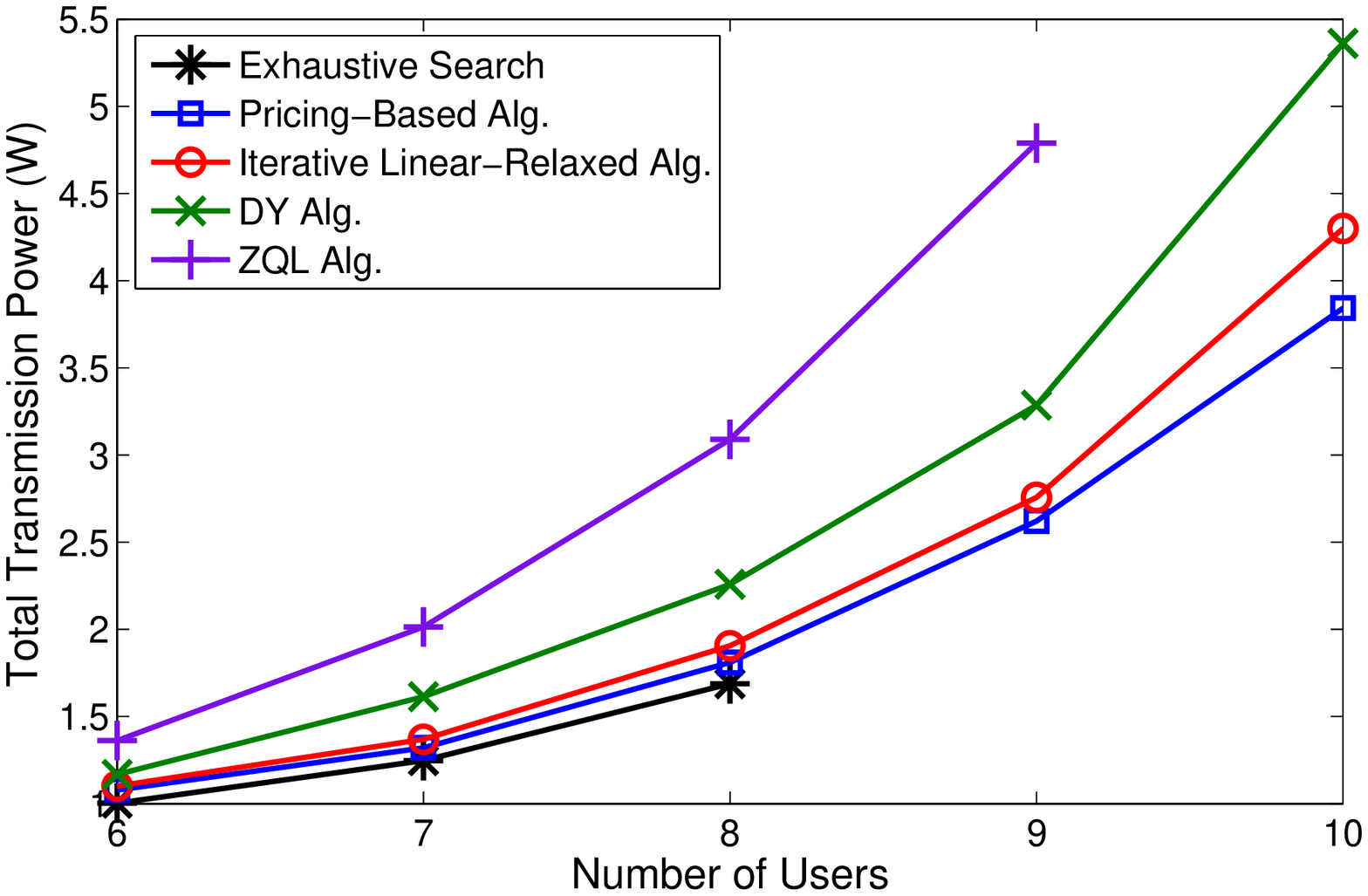}
\end{center}
\caption{Total power versus number of UEs in small network.}
\label{SMR_3BS_PvsM}
\end{figure}

\begin{figure}[!t]
\begin{center}
\includegraphics[width=0.7\textwidth]{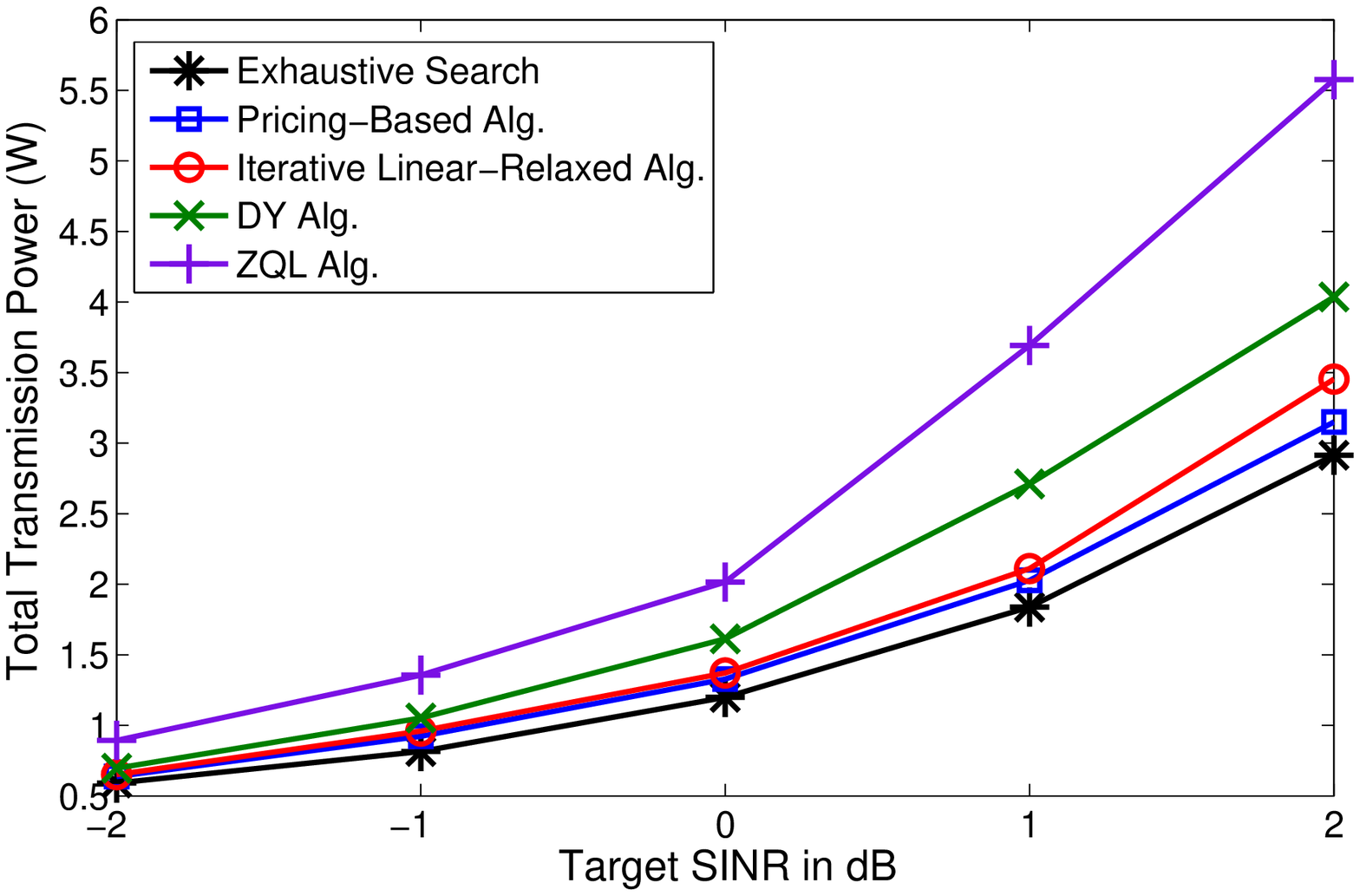}
\end{center}
\caption{Total power versus UE target SINR in small network.}
\label{SMR_3BS_PvsSINR}
\end{figure}

In Fig.~\ref{SMR_3BS_PvsM} and Fig.~\ref{SMR_3BS_PvsSINR}, we show total transmission powers of all RRHs achieved by the
exhaustive search method, our proposed algorithms, and two reference algorithms, versus the number of UEs and UE target SINR, respectively. 
As can be seen, our proposed algorithms result in lower total transmission power compared to the existing DY and ZQL algorithms.
Moreover, the pricing-based algorithm is slightly better than the iterative linear-relaxed algorithm, and 
both proposed algorithms require marginally higher total transmission power than that due to the optimal exhaustive search algorithm.
In addition, the DY algorithm outperforms the ZQL algorithm but these existing algorithms demand considerably higher powers compared to our algorithms.
Also, this figure shows that the total transmission power increases as the number of UEs becomes larger under all algorithms as expected.

\subsection{Resource Allocation for Wireless Virtualization of OFDMA-Based Cloud Radio Access Networks}

This contribution considers the resource allocation design for the OFDMA uplink C-RAN supporting multiple operators (OPs).
To our best knowledge, uplink C-RAN design considering constraints on limited fronthaul capacity and cloud computation resources has not
been studied. This line of work aims to fill this gap in the existing C-RAN literature where we make
the following contributions.
\nomenclature{OP}{Operator}
\begin{itemize}
\item
We consider the resource allocation for the virtualized OFDMA uplink C-RAN where multiple OPs share the C-RAN
infrastructure and resources owned by an infrastructure provider. The infrastructure provider and OPs are interested in
maximizing their profits, which are modelled by two different resource allocation
problems, namely upper-level and lower-level problems, respectively. The upper-level
problem focuses on slicing the fronthaul capacity and cloud computing resources for all OPs. Moreover,
the lower-level maximizes the OP's sum rate by optimizing users' transmission rates and quantization bit allocation 
for the compressed I/Q baseband signals, which must be transferred from RRHs to the cloud. Our design aims at
maximizing the weighted sum profit of both infrastructure provider and OPs considering constraints
on the fronthaul capacity and cloud computation limit.
\nomenclature{I/Q}{In-phase and Quadrature components}
\item
We develop a two-stage algorithmic framework to fulfil the design. 
In the first stage, we transform both upper-level and lower-level  problems into
the corresponding optimization problems by relaxing the underlying discrete variables 
to the continuous variables. We show that these relaxed problems are convex and therefore they can be solved optimally.
In the second stage, we propose two methods to round the solution obtained by solving the relaxed problems
to achieve a final feasible solution for the original problem.  
\end{itemize}

\subsubsection{System Model}

We consider the OFDMA uplink C-RAN system with $K$ RRHs ($K$ cells) supporting $O$ different OPs.
We assume that each cell utilizes the whole spectrum with $S$ physical resource blocks (PRBs). 
\nomenclature{PRB}{Physical Resource Block}
Denote $\mathcal{S}$ and $\mathcal{S}_k^o$ as the set of all PRBs and the set of PRBs assigned for OP $o$ in cell $k$, respectively. 
Let $x_{k}^{(s)}$ be the baseband signal transmitted on PRB $s$ in cell $k$.
Then, the signal received at RRH $k$ on PRB $s$ can be written as
\beq
y^{(s)}_{k} = \sum_{j \in \mathcal{K}} h_{k,j}^{(s)} \sqrt{p_{j}^{(s)}} x_{j}^{(s)} + \eta_k^{(s)},
\eeq
where $p_{j}^{(s)}$ denotes the transmission power corresponding to $x_{j}^{(s)}$, 
$h_{k,j}^{(s)}$ is the channel gain from UE assigned PRB $s$ in cell $j$ to RRH $k$, and $\eta_k^{(s)} \sim \mathcal{CN}\left( 0,\sigma_k^{(s)2}\right)$ denotes the Gaussian noise. Assume that RRH $k$ uses $b_{k}^{(s)}$ bits to quantize each of the I/Q parts of $y^{(s)}_{k}$ before forwarding them to cloud.
Then, the quantization noise power and the SINR of the quantized signal can be approximated as
\beqn 
q^{(s)}_{k}(b^{(s)}_{k}) \simeq   \dfrac{\sqrt{3} \pi}{2^{2 b_{k}^{(s)} +1}}  Y^{(s)}_{k}, \label{ch1_eq:q_kmn} \\
\gamma^{(s)}_{k}(b^{(s)}_{k}) = \dfrac{D_k^{(s)}}{I_k^{(s)}+ q^{(s)}_{k}(b^{(s)}_{k})} \simeq \dfrac{D_k^{(s)}}{I_k^{(s)} + \frac{\sqrt{3} 
\pi Y^{(s)}_{k}}{2^{2 b_{k}^{(s)} +1}} }, \label{ch1_eq:SINR}
\eeqn
where $Y^{(s)}_{k}$ represents the power of the received signal $y^{(s)}_{k}$, i.e., $ Y^{(s)}_{k} = \sum \limits_{j \in \mathcal{K}} \vert h_{k,j}^{(s)} \vert ^2 p_{j}^{(s)} + \sigma_k^{(s)2} = D_k^{(s)} + I_k^{(s)}$. Moreover, $I_k^{(s)}$ denotes the multi-cell interference on PRB $s$.
We assume that a quantization process satisfies $q^{(s)}_{k}(b^{(s)}_{k}) \leq \sqrt{Y^{(s)}_{k} I^{(s)}_{k}}$, which can
 guarantee the quality of the quantized signal as $\gamma^{(s)}_{k}(b^{(s)}_{k}) \geq \underline{\gamma}^{(s)}_{k} = \sqrt{\bar{\gamma}^{(s)}_{k} + 1} - 1$.
This design requirement is equivalent to $b^{(s)}_{k} \geq \lceil \underline{b}^{(s)}_{k} \rceil = \dfrac{1}{2}\left( \log_2 \left( \sqrt{3}\pi \sqrt{\dfrac{ Y^{(s)}_{k}}{I^{(s)}_{k}}} \right) -1\right)$.
Then, the total number of quantization bits corresponding to OP $o$ in cell $k$ can be expressed as
\beq
B_k^o = (2N_{\sf{RE}}) \sum_{s \in \mathcal{S}_k^o} b_{k}^{(s)},
\eeq 
where $N_{\sf{RE}}$ is the number of resource elements (REs) in one second where RE corresponds to one data symbol on each subcarrier.
We assume that the data rate $r^{(s)}_{k}$ (in \textit{``bits per channel use (bits pcu)''}, or ``bits per RE'')  for transmission over PRB $s$ 
is chosen from a pre-determined set of rates $\mathcal{R} = \lbrace R_1, R_2, ..., R_{M_R} \rbrace$. 
\nomenclature{RE}{Resource Element}
Then, the computation complexity required to successfully decode the information bits can be expressed as 
\beq \label{ch1_eq:cpt}
C^{(s)}_{k}  = \chi^{(s)}_{k}(r^{(s)}_{k},b^{(s)}_{k})  =  A r^{(s)}_{k} \!\! \left[B - 2 \log_2 \! \left( \log_2 \! 
\left( \! 1 +\gamma^{(s)}_{k}(b^{(s)}_{k}) \! \right) \! - r^{(s)}_{k} \! \right) \! \right].
\eeq
Let $\mb{r}^o$ and $\mb{b}^o$ denote the vectors representing all UEs' rates and numbers of quantization bits corresponding to OP $o$.
Then, the total computation complexity required by OP $o$ (in \textit{``bit-iteration per second (bips)''}) can be written as 
\beq
C_{\sf{tt}}^o(\mb{r}^o,\mb{b}^o) = N_{\sf{RE}} \sum_{k \in \mathcal{K}} \sum_{s \in \mathcal{S}^o_k} C^{(s)}_{k}.
\eeq

We are now ready to present the problem formulation. Let $C^o$ and $B^o_k$ be the computation complexity and fronthaul capacity corresponding to OP $o$ in cell $k$, and their prices are denoted as $\psi^o$ (\textcent$/bips$) and $\beta^o_k$ (\textcent$/bps$), respectively.
Then, OP $o$ must pay  an amount of money equal to $G^{\sf{InP}}_o=\psi^o C^o + \sum 
\limits_{k \in \mathcal{K}} \beta^o_k B^o_k$ to utilize the allocated resource slice.
Moreover, let $G^{\sf{InP}} = \sum \limits_{o \in \mathcal{O}} G^{\sf{InP}}_o$ denote the revenue obtained by the infrastructure provider.
Ignoring the operation cost of the infrastructure provider, this revenue $G^{\sf{InP}}$ can be considered as the profit of this infrastructure provider.
Note that we do not consider the optimization for the revenue or cost related to the radio spectrum, which is assumed to be fixed terms.

For notational convenience, let us introduce the vector $\mb{B}^o = [B^o_1, ..., B^o_K]$. 
Moreover, we denote $R_o(C^o,\mb{B}^o)$ as the maximum sum-rate achieved by all UEs from OP $o$,
 which is the outcome of the following lower-level problem $(\mathcal{P}^o)$ corresponding to OP $o$:
\begin{subequations} \label{ch1_obj_11}
\begin{align}
 R_o(C^o,\mb{B}^o) = \max \limits_{\mb{r}^o, \mb{b}^o} & \;\;\;\;\;  \sum_{k \in \mathcal{K}} \sum_{s \in \mathcal{S}^o_k} r^{(s)}_{k}  \label{ch1_obj_1a} \\
\text{s. t. } & \; \sum_{k \in \mathcal{K}} \sum_{s \in \mathcal{S}^o_k} C^{(s)}_{k} \leq C^o/N_{\sf{RE}}, \label{ch1_obj_1b} \\
 {} & \;  r^{(s)}_{k} \leq \log_2 \left(1 + \gamma^{(s)}_{k}(b^{(s)}_{k}) \right), \forall k \in \mathcal{K}, \forall s \in \mathcal{S}^o_k, \label{ch1_obj_1c}\\
 {} & \; \sum_{s \in \mathcal{S}^o_k} b_{k}^{(s)} \leq B^o_k/(2N_{\sf{RE}}), \forall k \in \mathcal{K}, \label{ch1_obj_1d} \\
  {} & \; b^{(s)}_{k} \geq \lceil \underline{b}^{(s)}_{k} \rceil, \forall k \in \mathcal{K}, \forall s \in \mathcal{S}^o_k, \label{ch1_obj_1e} \\
 {} & \; \text{$b^{(s)}_{k}$ is integer and } r^{(s)}_{k} \in \mathcal{M}_R, \forall k \in \mathcal{K}, \forall s \in \mathcal{S}^o_k. \label{ch1_obj_1f} 
\end{align}
\end{subequations}
Assume that the price per rate unit of OP $o$ is $\rho^o$ (\textcent$/bps$).
Then, the profit achieved by OP $o$ can be expressed as $G^{\sf{OP}}_o= \rho^o N_{\sf{RE}} R_o(C^o,\mb{B}^o) - G^{\sf{InP}}_o$.
The upper-level problem aims to maximize the weighted sum of the profits of the infrastructure provider and the OPs,
which can be stated as
\begin{subequations} \label{ch1_obj_U}
\begin{align}
 \max \limits_{\lbrace C^o, \mb{B}^o \rbrace} & \;\;\;\;\; \upsilon^{\sf{InP}} G^{\sf{InP}} + \sum \limits_{o \in \mathcal{O}} \upsilon^o G^{\sf{OP}}_o, \label{ch1_obj_Ua} \\
\text{s. t. } & \;\;\; \sum_{o \in \mathcal{O}} C^{o} \leq \bar{C}_{\sf{cloud}}, \label{ch1_obj_Ub} \\
{} & \;\;\; \sum_{o \in \mathcal{O}} B^{o}_k \leq \bar{B}_k, \forall k \in \mathcal{K}. \label{ch1_obj_Uc} 
\end{align}
\end{subequations}

\subsubsection{Problem Relaxation and Convexity}
We first approximate the lower-level problems by relaxing their discrete optimization variables into the continuous ones
and obtain the following relaxed lower-level (RLL) problem \nomenclature{RLL}{Relaxed Lower-Level}
\begin{align}
 \max \limits_{\mb{r}^o, \mb{b}^o} \;\;\;  \sum_{k \in \mathcal{K}} \sum_{s \in \mathcal{S}} r^{(s)}_{k} \;\;\; \text{s. t. } & \text{ (\ref{ch1_obj_1b})-(\ref{ch1_obj_1e}),} \label{ch1_ctb4_obj_2} \\
 {} & R_{\sf{min}} \leq r^{(s)}_{k} \leq R_{\sf{max}}, \forall k \in \mathcal{K}, \forall s \in \mathcal{S}, \label{ch1_obj_1f2}
\end{align}
where $R_{\sf{min}}$ and $R_{\sf{max}}$ are the lowest and highest rate in the rate set $\mathcal{M}_R$, respectively.
Let $\bar{R}_o(C^o,\mb{B}^o)$ be the optimal solution of this RLL problem. 
Then, the upper-level problem can be relaxed to the following relaxed upper-level problem (RUL) \nomenclature{RUL}{Relaxed Upper-Level}
\begin{align}
 \max \limits_{\lbrace C^o, \mb{B}^o \rbrace} \;\;  \sum \limits_{o \in \mathcal{O}} \Psi^o(C^o, \mb{B}^o) \;\; \text{s. t.} \;\;\; \text{(\ref{ch1_obj_Ub}), (\ref{ch1_obj_Uc})}, \label{ch1_obj_U2}
\end{align}
where $\Psi^o(C^o, \mb{B}^o)=(\upsilon^{\sf{InP}}-\upsilon^o) G^{\sf{InP}}_o + \upsilon^o \rho^o N_{\sf{RE}} \bar{R}_o(C^o,\mb{B}^o)$. 
We state some important results based on which we can develop algorithms to resolve our design problem in the following propositions and
theorems. 
\begin{theorem} \label{ch1_R8_thr1} $\chi^{(s)}_{k,u}(r^{(s)}_{k},b^{(s)}_{k})$ is a jointly convex function with respect to variables $(r^{(s)}_{k},b^{(s)}_{k})$ if $q^{(s)}_{k}(b^{(s)}_{k}) \leq \sqrt{Y^{(s)}_{k} I^{(s)}_{k}}$.
\end{theorem}

\begin{proposition} \label{ch1_R8_prt2}
The optimization problem RLL is convex.
\end{proposition}

\begin{theorem} \label{ch1_R8_thr2}
$\bar{R}_o(C^o,\mb{B}^o)$ is a concave function with respect to $C^o$ and $\mb{B}^o$.
\end{theorem}

\begin{proposition} \label{ch1_R8_prt5}
The optimization problem RUL is convex.
\end{proposition}

\subsubsection{Proposed Algorithms}

Based on the results in Propositions~\ref{ch1_R8_prt2}-\ref{ch1_R8_prt5}, we can develop an algorithm to solve the RLL problem optimally as follows. 
We first express the dual function $g^o(\lambda)$ of RLL problem corresponding to OP $o$ by relaxing the cloud computation constraints as follows:
\begin{align}
g^o(\lambda^o) \! = \! \max \limits_{\mb{r}^o, \mb{b}^o} \Phi^o \! ( \! \lambda^o \! ,\mb{r}^o \!, \mb{b}^o \! ) \;\;\; \text{ s. t. (\ref{ch1_obj_1c})-(\ref{ch1_obj_1e}) and (\ref{ch1_obj_1f2}),} \label{ch1_obj_3}
\end{align}
where $\lambda_o$ denotes the Lagrange multiplier corresponding to the cloud computation constraint, and
\begin{align}
\Phi^o(\lambda^o,\mb{r}^o, \mb{b}^o) \! = \! \sum_{k \in \mathcal{K}} \! \sum_{s \in \mathcal{S}^o_k} \! r^{(s)}_{k} \! - \! \lambda^o \!\! \left( \! \sum_{k \in \mathcal{K}} \! \sum_{s \in \mathcal{S}^o_k} \! C^{(s)}_{k} \! - \! \dfrac{C^o}{N_{\sf{RE}}} \! \right).
\end{align}
Then, the dual problem can be described as $\min \limits_{\lambda^o \geq 0} g^o(\lambda^o)$. 
Since the dual problem is convex, it can be solved by employing the standard sub-gradient algorithm as follows:
\beq \label{ch1_eq:updld}
\lambda^o_{(l+1)} = \left[\lambda^o_{(l)} + \delta^o_{(l)}\left( \sum_{k \in \mathcal{K}} \sum_{s \in \mathcal{S}^o_k} C^{(s)}_{k} - \dfrac{C^o}{N_{\sf{RE}}} \right)  \right]^{+},
\eeq
where $l$ denotes the iteration index and $\delta^o_{(l)}$ represents the step size. 

The next step is to determine the optimal solutions of $\mb{r}^o$ and $\mb{b}^o$ for a given dual point $\lambda^o$.
According to the results in Proposition~\ref{ch1_R8_prt2}, problem \ref{ch1_obj_3} is also convex. In fact, it is possible to 
determine the optimal solution for one of the two variables $\mb{r}^o_k$ and $\mb{b}^o_k$ while keeping the other fixed. 
Specifically, the optimal solution of $\mb{r}^o_k$ for a given $\mb{b}^o_k$ can be expressed as
\beq \label{ch1_eq:r_opt}
r^{(s)\star}_{k} \!\! = \max \! \left[ \! R_{\sf{min}},  \min \! \left( \! t(b^{(s)}_{k}),R_{\sf{max}}, r\vert_{\frac{\partial w(r)}{\partial r} 
= - E^{(s)}_{k}} \! \right) \! \right]
\eeq
where $w(r)= 2\lambda^o A r \log_2 \left( 1 - r/{t(b^{(s)}_{k})} \right)$. 
Moreover, the optimal solution of $\mb{b}^o_k$ for a given $\mb{r}^o_k$ can be expressed as
\beq \label{ch1_eq:b_opt}
b^{(s)\star}_k =  \max \left(\lceil \underline{b}^{(s)}_{k} \rceil, t^{-1}\left( r^{(s)}_{k}\right), b\vert_{\frac{\partial z(b)}{\partial b} = \mu} \right), 
\eeq
where $\mu$ must be determined to satisfy the following constraint $\sum_{s \in \mathcal{S}^o_k} b_{k}^{(s)\star} = B^o_k/(2N_{\sf{RE}})$.
The algorithm to solve the RLL problem is summarized in Algorithm~\ref{ch1_R8_alg:gms1}.
\renewcommand{\baselinestretch}{0.9}
\small
\begin{algorithm}[!t]
\caption{\textsc{Algorithm to Solve RLL Problem}}
\label{ch1_R8_alg:gms1}
\begin{algorithmic}[1]
\STATE Initialization: Set $r_k^{(s)} = R_{\sf{min}}$ for all $(k,s) \in \mathcal{K} \times \mathcal{S}$, $\lambda^o_{(0)}=0$ and $l=0$. Choose a tolerance parameter $\varepsilon$ for convergence.
\REPEAT
\FOR{$k \in \mathcal{K}$} 
\REPEAT
\STATE Fix $\mb{r}^o_k$ and update $\mb{b}^o_k$ as in (\ref{ch1_eq:b_opt}) with $\lambda^o_{(l)}$.
\STATE Fix $\mb{b}^o_k$ and update $\mb{r}^o_k$ as in (\ref{ch1_eq:r_opt}) with $\lambda^o_{(l)}$.
\UNTIL Convergence.
\ENDFOR 
\STATE Calculate all $C_k^{(s)}$ with current values of $\mb{r}^o_k$ and $\mb{b}^o_k$.
\STATE Update $\lambda^o_{(l+1)}$ as in (\ref{ch1_eq:updld}).
\STATE Set $l=l+1$.
\UNTIL $\vert \lambda^o_{(l)} - \lambda^o_{(l-1)} \vert < \varepsilon$.
\end{algorithmic}
\end{algorithm}
\renewcommand{\baselinestretch}{1.4}
\normalsize

We now present a sub-gradient based algorithm to solve the RUL problem.
Specifically, the sub-gradient algorithm to iteratively update the $C^o,\mb{B}^o$ can be described as follows:
\beq \label{ch1_eq:updcb}
[C^o,\mb{B}^o]_{(n+1)} = \mathcal{P} \left[[C^o,\mb{B}^o]_{(n)} + \tau^o_{(n)} \nabla \Psi^o(C^o,\mb{B}^o)  \right]
\eeq
where $\tau^o_{(n)}$ is the determined step size, $\mathcal{P}\left[*\right]$ denotes the projection operation to the feasible region
of the RUL problem, and
\beq
\nabla \Psi^o(C^o,\mb{B}^o)=\left[ \dfrac{\partial\Psi^o(C^o,\mb{B}^o)}{\partial C^o}, \dfrac{\partial \Psi^o(C^o,\mb{B}^o)}{\partial B^o_1}, \ldots, 
 \dfrac{\partial \Psi^o(C^o,\mb{B}^o)}{\partial B^o_K} \right] ^T.
\eeq

We can indeed determine $\nabla \Psi^o(C^o,\mb{B}^o)$ by employing Algorithm~\ref{ch1_R8_alg:gms1}.
We then summarize the procedure to update $[C^o,\mb{B}^o]$'s in Algorithm~\ref{ch1_R8_alg:gms2},
which solves the RUL problem.
\renewcommand{\baselinestretch}{0.9}
\small
\begin{algorithm}
\caption{\textsc{Algorithm to Solve RUL Problem}}
\label{ch1_R8_alg:gms2}
\begin{algorithmic}[1]
\STATE Initialization: Set $C^o_{(0)} = \bar{C}_{\sf{cloud}}/O$, and $B^o_{k,(0)} = \bar{B}_k/O$ for all $(o,k) \in \mathcal{O} \times \mathcal{K}$, $\nu^o_{(0)}=0$ for all $o \in \mathcal{O}$, and $n=0$. 
\REPEAT
\STATE Run Algorithm~\ref{ch1_R8_alg:gms1} to obtain $\lbrace \bar{R}_o(C^o,\mb{B}^o)\rbrace $ for given $\lbrace [C^o,\mb{B}^o]_{(n)}\rbrace$ $\forall o \in \mathcal{O}$
obtained from the previous iteration, which are used to calculate $\nabla \Psi^o(C^o,\mb{B}^o)$.
\STATE Update $[C^o,\mb{B}^o]_{(n+1)}$ for all $o \in \mathcal{O}$ as in (\ref{ch1_eq:updcb}).
\STATE Set $n=n+1$.
\UNTIL Convergence.
\end{algorithmic}
\end{algorithm}
\renewcommand{\baselinestretch}{1.4}
\normalsize

The proposed iterative algorithm converges if the step sizes $\delta^o_{(l)}$ and $\tau^o_{(n)}$ are chosen appropriately,
e.g., $\delta^o_{(l)} = 1 /\sqrt{l}$ and $\tau^o_{(n)} = 1 /\sqrt{n}$ \cite{Bertsekas99}.
After running Algorithm \ref{ch1_R8_alg:gms2}, we can obtain a feasible solution $\lbrace C^o,\mb{B}^o \rbrace$ and the
 corresponding $\left\lbrace r^{(s) \bigstar}_k \right\rbrace $ and $\left\lbrace b^{(s) \bigstar}_k \right\rbrace $ of the 
relaxed problems. 

Note that $\left\lbrace r^{(s) \bigstar}_k \right\rbrace $ and $\left\lbrace b^{(s) \bigstar}_k \right\rbrace $
are real numbers while the original variables take only discrete values.
To address this issue, we propose two rounding methods, namely \textit{Iterative Rounding (IR) Method} and \textit{One-time Rounding and Adjusting (RA) Method},
which round the continuous variables achieved from Algorithm \ref{ch1_R8_alg:gms2} to attain the corresponding discrete and feasible solutions for the original
design problem. 
\nomenclature{IR}{Iterative Rounding}
\nomenclature{RA}{Rounding and Adjusting}
\subsubsection{Numerical Results}

We consider the network setting with seven cells and three OPs ($O=3$). 
The channel gains are generated by considering both Rayleigh fading and path loss. 
Other parameters are set as follows: $\sigma^2=10^{-13} \; W$ and $p^{(s)}_k=0.1 \;W$, $T'=0.2$, $\zeta=6$, and $\epsilon_{\sf{ch}}=10\%$.
To demonstrate the impacts of the cloud computation limit and fronthaul capacity on the system performance, we show the achievable system 
sum rate, which is the outcome of Algorithm~\ref{ch1_R8_alg:gms2}) solving the RLL problem, versus the cloud computation limit
 ($\bar{C}_{\sf{cloud}}$) and the fronthaul capacity $\bar{B}_k$ in Fig.~\ref{ch1_Fig03}. 
It is evident that higher cloud computation limit and larger fronthaul capacity result in the greater system sum rate as expected.
In addition, the sum rate becomes saturated as the cloud computation limit or fronthaul capacity becomes sufficiently large.

\begin{figure}[!t]
\begin{center}
\includegraphics[width=0.7\textwidth]{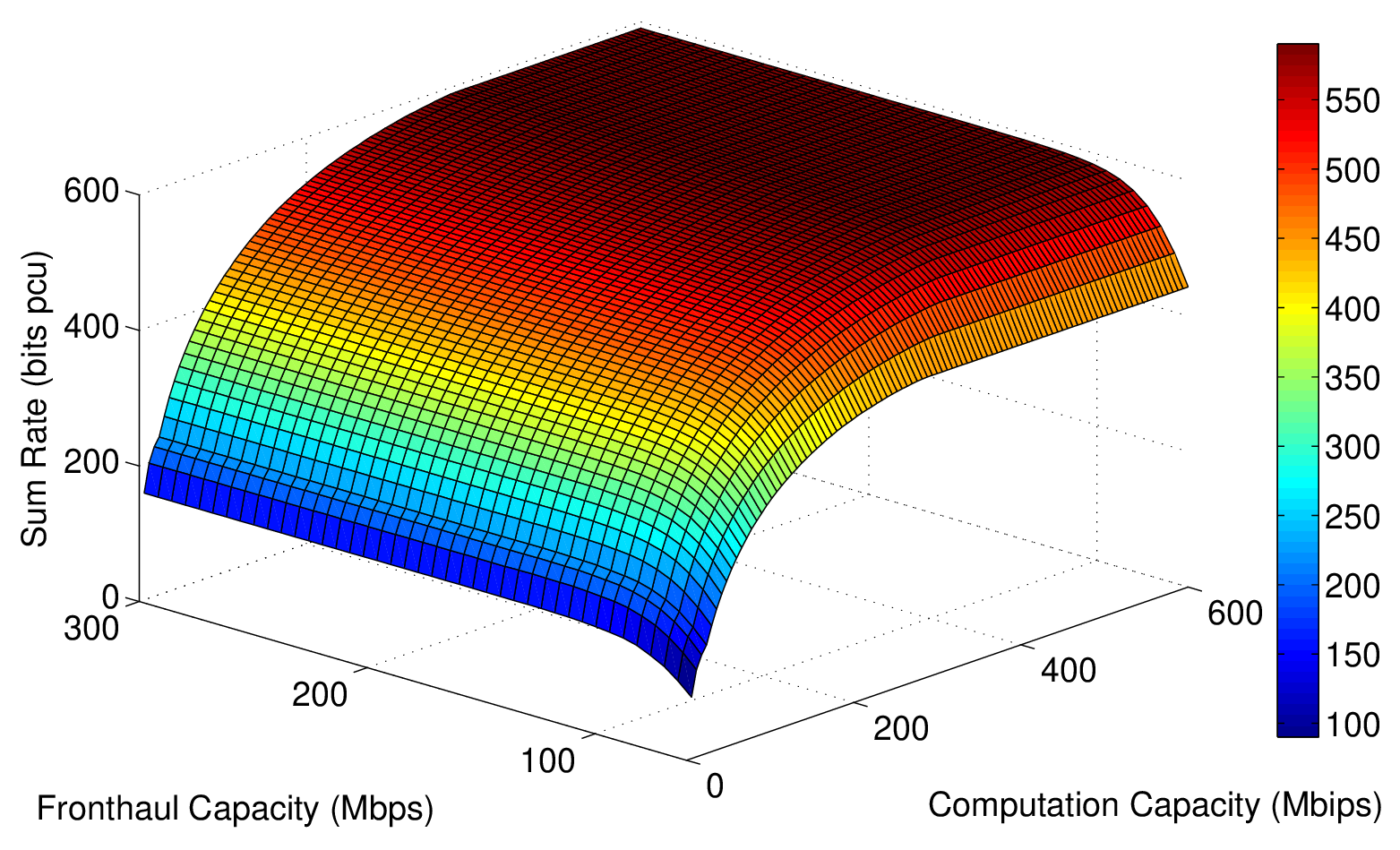}
\end{center}
\caption{Total rate versus $\bar{C}_{\sf{cloud}}$ and $\bar{B}_k$.}
\label{ch1_Fig03}
\end{figure}

In Fig.~\ref{ch1_Fig04}, we illustrate the variations of system sum rate obtained by the our proposed algorithms without rounding (Relaxed Prop. Alg.), with 
IR-rounding method (Prop. Alg. with IR), RA-rounding method (Prop. Alg. with RA), and the Fast Greedy Algorithm (Greedy Alg.), versus the number of PRBs 
in each cell. To obtain these numerical results, we sequentially add one more PRB for each OP in each cell to obtain different points on each curve.
As can been observed, our proposed algorithms significantly outperform the greedy algorithm in the studied scenarios.
In addition, the IR-rounding method results in better performance than that achieved by the RA-rounding method.
Interestingly, the system sum rate first increases then decreases as number of available PRBs in each cell becomes larger.

\begin{figure}[!t]
\begin{center}
\includegraphics[width=0.7\textwidth]{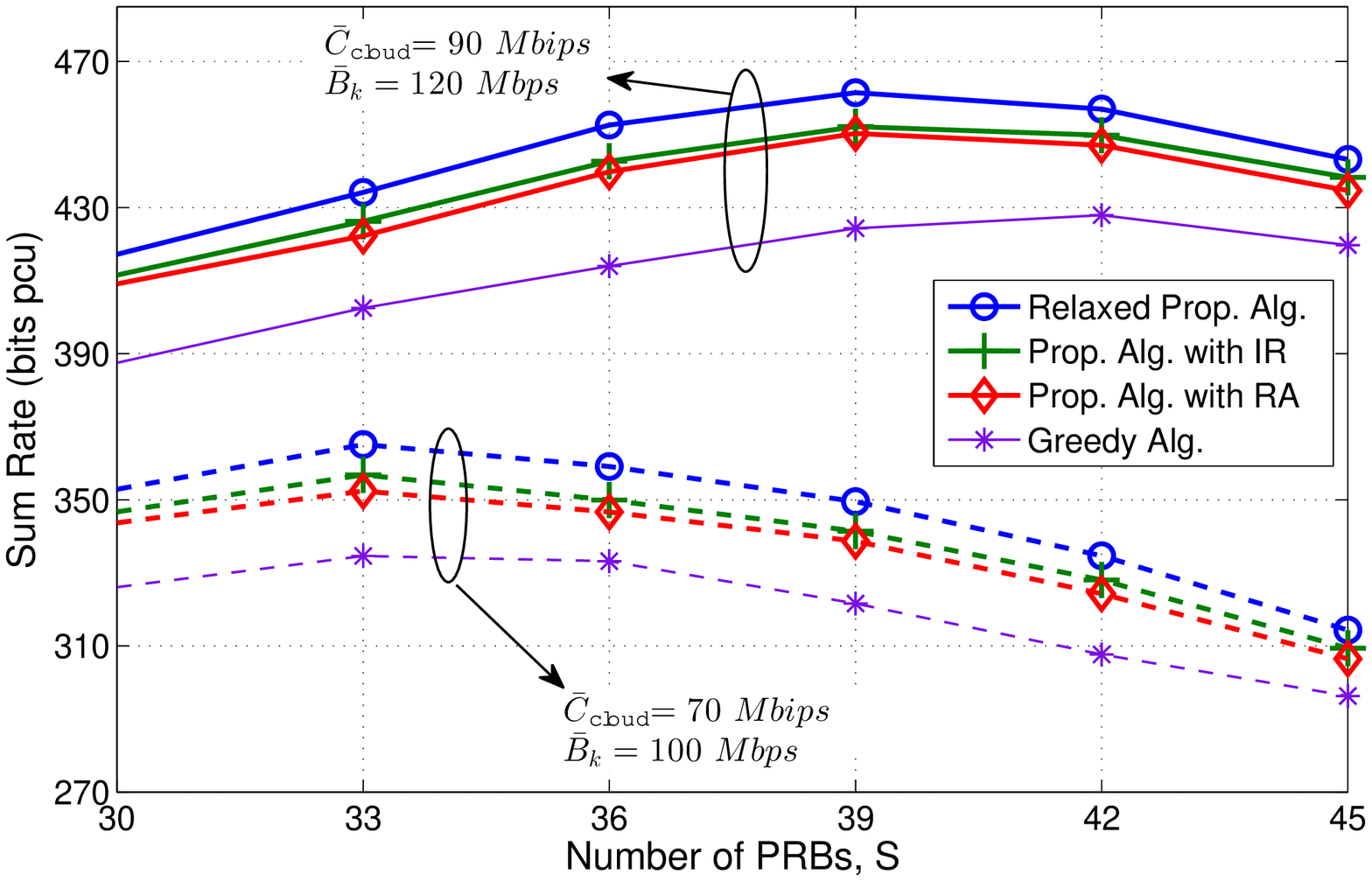}
\end{center}
\caption{Sum rate versus the number of PRBs ($S$).}
\label{ch1_Fig04}
\end{figure}

\section{Concluding Remarks}

In this doctoral dissertation, we have proposed various
novel resource management techniques and algorithms for wireless HetNets and C-RANs,
which are two potential network architectures for the future high-speed wireless cellular networks.
In particular, we have made four important research contributions.
First, we have proposed the general joint PC and BAS for the multi-tier wireless HetNets.
The developed HPC solution can efficiently balance between enhancing network throughput and achieving the
required QoSs for the maximum number of UEs.
Second, we have proposed distributed resource allocation algorithms that aim at maximizing
the total minimum spectral efficiency of the femtocell tier while ensuring fairness among FUEs
and QoS protection for all MUEs in the two-tier macrocell-femtocell wireless HetNets.

Third, we have developed efficient and low-complexity algorithms to solve the downlink joint transmission
problem in the C-RAN that aims to minimize the total transmission power subject to the constraints 
on transmission powers, fronthaul capacity, and UEs' QoS.
Finally, we have devised a novel algorithmic framework for the virtualized uplink C-RAN supporting 
multiple OPs, which captures the relevant interests and interactions between the infrastructure provider
and the OPs. Specifically, the proposed design, which accounts for the practical
 constraints on the fronthaul capacity and cloud computation limit, aims at maximizing the
profits of both infrastructure provider and the OPs through optimization the rate allocation for UEs and quantization bits allocation for the uplink baseband signals.
The research in our doctoral study has resulted in four journal publications \cite{VuHa_TVT14_BSA_PC,VuHa_TVT14_PC_SA,VuHa_TVT16,Tri_Access15}, one journal paper under submission \cite{VuHa_sTWC16} as well as nine papers on prestigious conferences \cite{VuHa_VTC12,VuHa_WCNC13,VuHa_GBCWS13,vuha_ciss_2014,VuHa_WCNC14,vuha_globecom_2014,VuHa_WCNC15,VuHa_WCNC16,VuHa_ICC16}.

%% file: chap2/resume.tex
\chapter{Résumé Long}
\section{Contexte et motivation}

Les systèmes sans fil 5G (cinquième génération), qui seraient déployés d’ici 2020, doivent offrir des performances nettement plus élevées par rapport aux systèmes 4G (quatrième génération) présentement sur le marché en termes de capacité en nombre d'utilisateurs et de gestion de réseaux. 
En effet, il est prévu que des dizaines de milliards de dispositifs sans fil seront connectés aux réseaux sans fil au cours des prochaines années. Avec le nombre croissant de connexions, la quantité de trafic des données mobiles a explosé à un rythme exponentiel. Par conséquent, le réseau mobile sans fil 5G devrait être en mesure de soutenir le volume de trafic de données d'un ordre de grandeur supérieur à celui des réseaux sans fil actuels \cite{Zander13,Boccardi14,Bhushan14,Le_EU15}.

Par conséquent, l'architecture sans fil plus avancée, ainsi que les technologies d'accès radicales et novatrices, doivent être proposées pour répondre à la croissance exponentielle des données mobiles et aux exigences de connectivité dans les années à venir \cite{Dottling09,Lopez-Perez11, 3GPP,cisco13,cisco16}.
À cette fin, deux architectures cellulaires sans fil importantes, ont été proposées et étudiées activement aux milieux académiques ainsi qu’industriels: la première est l’architecture  hétérogène (HetNets) qui est basée sur le déploiement dense de petites cellules et la deuxième est définie par les réseaux d’accès radio qui sont basés sur l’infonuagique (Cloud Radio Access Networks, Cloud-RAN ou C-RAN). Ces deux termes Cloud-RAN et C-RAN seront utilisés de manière changeable dans la suite. Cette thèse de doctorat porte sur les modèles de gestion des ressources radio pour ces deux architectures potentielles de réseau sans fil.

Le HetNet sans fil est généralement basé sur le déploiement dense de petites cellules telles que les «femtocells» et «picocells» dans la coexistence avec les  macrocells existants \cite{Chandrasekhar08,Claussen08,Kim09,Zhang_bk_10} où les petites cellules sont à courte distance de communication, à faible puissance et à faible coût. 
Le déploiement massif de petites cellules peut fondamentalement améliorer le débit des données de communication à l'intérieur des locaux fermés et les performances de couverture du réseau cellulaire sans fil \cite{Andrews12,Le12}. 
En outre, certaines petites cellules, les femtocells, peuvent être déployées de façon aléatoire par les utilisateurs finaux et elles fonctionnent sur la même bande de fréquences que le macrocell existant afin d'améliorer l'utilisation du spectre.

En outre, le trafic intérieur effectué par femtocell peut être acheminé sur les connexions IP telles que xDSL (Symmetric/Asymmetric Digital Subscriber Line) pour réduire la charge de trafic des macrocells. Le  macrocell peut ainsi consacrer davantage de ressources radio pour assurer un service meilleur a l'extérieur des locaux. 
En outre, les femtocells, qui peuvent être déployés de manière qu’elles puissent s’ajouter facilement aux réseaux comme utilisateurs finaux ordinaires, nécessitent habituellement des dépenses en capital et à faible coût d'exploitation. Cependant, le déploiement dense des femtocells sur la même bande de fréquence avec les macrocells pose différents défis techniques. 
Tout d'abord, les fortes perturbations inter-niveaux entre macrocells et femtocells peuvent se produire et sérieusement affecter la performance du réseau si l’accès inter niveau n’est pas géré correctement \cite{Yavuz09}. 
Deuxièmement, les macro-UEs (MUEs) ont généralement une plus grande priorité à l'accès au spectre radio par rapport aux femto-UEs (FUEs). 
Ainsi, les exigences de qualité de service de MUEs doivent être protégées.
Par conséquent, l'interférence inter niveau introduite par FUEs au macrocell doit être contrôlée adéquatement. 
Enfin, les stratégies dynamiques et intelligentes de contrôle d'accès doivent être développées pour gérer efficacement l'interférence de réseau et l'équilibrage de charge trafic \cite{roche10}.

L'architecture C-RAN vise à exploiter l'infrastructure de l’informatique en nuage pour réaliser diverses fonctions et protocoles de réseau \cite{chinamobile2011,NGMN2013,maketresearch2013}. 
L'architecture générale C-RAN se compose de trois éléments principaux: (i) des processeurs centralisés ou une poule des unités de bande de base (BBU – Baseband Unit), (ii) d’un réseau de transport optique (lien Fronthaul), et (iii)  des RRH (Remote Radio Head) des unités d'accès avec des antennes situées sur des sites à distance. 
Le centre de traitement informatique comprenant un grand nombre de BBUs est le cœur de cette architecture où les BBUs fonctionnent comme des stations de base (BSs) virtuelles pour traiter les signaux en bande de base pour les UEs et optimiser l'allocation des ressources radio. 
En général, les RRHs peuvent être relativement simples, et peuvent donc être spatialement distribués sur le réseau pour plus pour réduire le coût et l'énergie du réseau.

En outre, le traitement centralisé permet de mettre en œuvre de manière sophistiquée la couche physique et les conceptions de gestion des ressources radio telles que la transmission à multipoint coordonné (CoMP – Coordinated Multipoint) et les techniques de réception proposées dans la norme sans fil LTE, la conception efficace en grappe pour RRH pour équilibrer l'amélioration de la capacité du réseau et la complexité de conception \cite{Costa-Perez_CM13, Liang_CST15}. 
D’autres avantages du C-RAN comprennent une liaison terrestre réduite et le trafic de cœur de réseau, des coûts réduits de déploiement et d'exploitation, une meilleure qualité de service et la différenciation pour soutenir efficacement les différentes applications sans fil \cite{mob_rep_13,chinamobile2011}. 
La réussite du déploiement de C-RAN, cependant, nous oblige à résoudre plusieurs problèmes techniques majeurs \cite{Costa-Perez_CM13,Liang_CST15}. 
En particulier, des techniques et des solutions pour le traitement avancé des signaux et la gestion des ressources radio efficaces doivent être développées pour atteindre l'amélioration de la performance du réseau potentiel tout en tenant compte des contraintes et de l'utilisation efficace de la capacité du fronthaul et des ressources infonuagiques.

\section{Contributions à la recherche}
L'objectif général de cette thèse doctorale consiste à développer les algorithmes d'allocation des ressources radio et de gestion d'interférences pour les réseaux cellulaires futures à haute vitesse de données.
En effet, nous avons développé (i) des techniques de gestion des ressources adaptatives capables à contrôler efficacement des interférences à la fois du même niveau (co-tier) ou d'inter niveau (cross-tier) sur des réseaux sans fil HetNets; et (ii) de nouvelles techniques d'accès conçues pour le C-RAN capables d'exploiter efficacement la radio, les ressources d'infonuagique et la capacité de fronthaul. Ces contributions de recherche seront décrites dans la section suivante. 

\subsection{Association de station de base et contrôle de l'alimentation pour monocanal HetNets sans fil}
Dans cette contribution, nous avons développé les techniques conjointes d'affectation des stations de base (BSA - Base Station Association) et de contrôle de puissance (PC – Power Control) pour les monocanal HetNets sans fil à multi-niveau. Dans le cadre de HetNets, la conception des techniques PC-et-BSA est difficile parce que les UEs de différents niveaux du réseau (par exemple, macrocell et femtocells) ont des priorités d'accès distinctes ainsi que les exigences différentes de qualité de service. 
Il existe plusieurs ouvrages proposant des solutions de gestion d'interférence avancées pour HetNets en utilisant PC et BSA dynamique \cite{jo09,chand09,madan10,yun11,Le12a,yanzan12}.
Cependant, la plupart de ces travaux ont examiné le mode d'accès fermé où MUEs ne sont pas autorisés à se connecter avec les stations de base (FBS - femto BS). En outre, l'examen conjoint des garanties de qualité de service et des utilisations efficaces des ressources du réseau n'a pas été traité à fond. Notre conception actuelle vise donc à répondre à ces limitations.
\begin{itemize}
\item Nous avons développé un algorithme généralisé de BSA et PC pour HetNets multi-niveau et prouvé sa convergence si la fonction de mise à jour de puissance satisfait la propriété de 2-sided-scalable (2.s.s).
\item Nous avons proposé un algorithme d'adaptation hybride qui ajuste les paramètres de conception clés pour soutenir les exigences de SINR différenciées de tous les UEs chaque fois que possible, tout en améliorant le débit du système.
\item Cet algorithme proposé est peut aboutir à des performances plus meilleures que d'autres algorithmes existants.
\item Nous décrivons la façon d'étendre la méthode proposée pour permettre la conception d'accès hybride dans les réseaux à deux niveaux macrocell-femtocell.
\end{itemize}

\subsubsection{Modèle de système}
\label{res_section2_r3}
Nous considérons les communications montantes dans un réseau sans fil hétérogène où $K$ BSs servent $M$ UEs sur le même spectre fréquentiel en utilisant CDMA. Supposons que chaque UE $i$ communique avec une seul BS à tout moment, qui est désigné comme $b_i$. Cependant, les UEs peuvent changer leur BSs associée au fil du temps. Le SINR de l’UE $i$ à la BS $b_i$ peut être écrit \cite{Alpcan08}:
\begin{equation}
\label{res_eq:sinr_r3}
\Gamma_i(\mathrm{p})=\dfrac{G h_{b_i i} p_i}{\sum_{j \neq i}{h_{b_ij}p_j}+\eta_{b_i}}
=\dfrac{p_i}{I_i \left( \mathrm{p}, b_i\right) }
\end{equation}
où $h_{k i}$ désigne le gain de canal entre BS $k$ et UE $i$, $p_i$ est la puissance d'émission d'UE $i$, $I_i \left( \mathrm{p},k\right) \triangleq  \dfrac{\sum_{j\neq i}{h_{kj}p_j}+\eta_{b_i}}{G g_{kii}}$ est l'interférence effective, et $\mathrm{p} = [p_1, ..., P_M]$. Nous allons parfois écrire $I_i \left( \mathrm{p}\right)$ au lieu de $I_i \left( \mathrm{p},b_i\right)$. 
Nous supposons que l’UE $i$ exige une qualité de service minimum en termes d'une SINR cible $\bar{\gamma}_i$:
\begin{equation}
\label{res_equ:SINR_cond_r3}
\Gamma_i(\mathrm{p}) \geq \bar{\gamma}_i, \:\: i \in \mathcal{M}.
\end{equation}
L'objectif de cette contribution consiste à développer les algorithmes distribués BSA et PC capable de maintenir autant que possible les exigences de SINR dans (\ref{res_equ:SINR_cond_r3}) tout en exploitant le gain de diversité multi-utilisateur pour augmenter le débit du système. Les algorithmes proposés visent donc à soutenir les applications de voix et de données à grande vitesse où les UEs de voix nécessitent généralement une cible fixée SINR $\bar{\gamma}_i$ tandis que les UEs de données cherchent à obtenir une SINR cible plus élevée que $\bar{\gamma}_i$ pour soutenir leurs applications à large bande.

\subsubsection{Algorithme généralisé de l'association de station de base et de contrôle de l'alimentation}
\label{res_sec:gnrl_PC_BAS}
Nous développons d'abord un algorithme généralisé BSA et PC. Plus précisément, nous allons nous concentrer sur un algorithme de PC itératif générale où chaque UE $i$ exécute la fonction de mise à jour de puissance suivante (puf – power update function) $p_i^{(n+1)} := J_i(\mathrm{p}^{(n)})=J'_i(I_i^{(\mathit{n})}(\mathrm{p}^{(n)}))$ où $n$ indique l'indice d'itération et $J_i (.) $, $J'_i (.) $ est le puf. En fait, ce genre d'algorithme de PC converge si nous pouvons prouver que sa puf correspondante est 2.s.s \cite{sung05, sung06}. 
En particulier, nous proposons un algorithme conjoint BSA et PC tel que résumé dans l'algorithm~\ref{res_alg:gms2_r3}. 
En vertu de cette conception, chaque UE choisit une BS qui se traduit par un minimum d'interférence effective et utilise la puf pour mettre à jour sa puissance. Cet algorithme garantit que chaque UE subit de faible interférence effective et donc un haut débit à la convergence. En général, les performances d'un algorithme de PC dépend de la façon dont nous concevons la fonction $\mathrm{J}(\mathrm{p})$.
\renewcommand{\baselinestretch}{0.9}
\small
\begin{algorithm}
\caption{\textsc{L'algorithme de minimum d'interférence effective BSA et PC}}
\label{res_alg:gms2_r3}
\begin{algorithmic}[1]
\STATE Initialisation:

- $p_i^{(0)}=0$ pour toutes UE $i$, $i \in \mathcal{M}$.

- $b_i^{(0)}$ est défini comme le BS le plus proche de IE $i$.

\STATE Itération $n$: Chaque UE $i$ ($ i \in \mathcal{M}$) effectue les opérations suivantes:

- Calculer l'interférence effective à BS $\mathit{k} \in D_i$ comme suit:

- Choisissez le BS $b_i^{(n)}$ avec le minimum $I_i^{(n)} ( \mathrm{p}^{(n-1)},\mathit{k})$, c-à-d, 
$ b_i^{(n)}=\mathrm{argmin}_{\mathit{k} \in D_i} I_i^{(n)} ( \mathrm{p}^{(n-1)},\mathit{k})$. Ensuite, nous avons $ I_i^{\sf o (\mathit{n})}(\mathrm{p}^{(n-1)}) = \mathrm{min}_{\mathit{k} \in D_i} I_i^{(n)} ( \mathrm{p}^{(n-1)},\mathit{k}) = I_i^{(n)} ( \mathrm{p}^{(n-1)},b_i^{(n)})$.

- Mise à jour la puissance d'émission pour la station de base choisie comme suit::
\begin{equation}
\label{res_p}
p_i^{(n)}=J'_i(I_i^{\sf o (\mathit{n})}(\mathrm{p}^{(n-1)})). 
\end{equation}
\vspace{-0.3cm}
\STATE Augmenter $n$ et revenir à l'étape 2 jusqu'à convergence.
\end{algorithmic}
\end{algorithm}
\renewcommand{\baselinestretch}{1.4}
\normalsize

\subsubsection{Algorithme HPC proposé}
L'algorithme hybride de contrôle de puissance (HPC - Hybrid Power Control) sera désormais développé pour tirer avantages des deux cibles TPC \cite{foschini93, yates95} et OPC \cite{sung05, sung06}. Plus précisément, l'algorithme HPC peut être décrit par la mise à jour itérative de puissance comme la suivante:
\begin{equation} \label{res_hpcrule}
p_i^{\left(n+1\right)} = J^{\mathsf{HPC}}_i \left( \mathrm{p}^{(n)}\right) = \mathrm{min}\left\lbrace P^{\mathsf{max}}_i, J_i\left( \mathrm{p}^{(n)}\right) \right\rbrace,
\end{equation}
où $J_i\left( \mathrm{p}\right) \triangleq \dfrac{\alpha_i\xi_i I_i\left( \mathrm{p}\right)^{-1}+ \bar{\gamma}_i I_i\left( \mathrm{p}\right)}{\alpha_i + 1}$, $P^{\mathsf{max}}_i$ est le budget de puissance de transmission de l’UE $i$ et $\xi_i=P_i^{\mathsf{max}2}/\bar{\gamma}_i$. 
Ici, l'algorithme de HPC est présenté sous une forme générale avec des paramètres $\alpha_i$. 
Plus précisément, pour $\alpha_i=0$, l'algorithme HPC devient l'algorithme TPC où chaque UE $i$ vise à atteindre son SINR ciblée. Lorsque $\alpha_i \rightarrow \infty$, chaque UE $i$ tentatives pour atteindre SINR plus (si elle est dans un état favorable) et HPC devient OPC. 
La convergence de l'algorithme HPC proposé peut être prouvée en montrant que sa puf est 2.s.s. comme indiqué dans le Théorème 2 de [J1].

\subsubsection{Algorithme adaptatif à double échelle de temps}
\label{res_sec:Gnrl_Time_Scal_Sepa}
Pour compléter notre objectif de conception en utilisant l'algorithme~\ref{res_alg:gms2_r3} et HPC, nous développons un algorithme adaptatif comme suit. 
Tout d'abord, nous dénotons les UEs dont les ratios SINR sont supérieurs ou égaux à son SINR ciblé comme UEs supportés, et les autres comme non supportés. Il est à noter également que UE $i$ est non-supporté si $I_i \left( \mathrm{p}\right) > I^{\mathsf{thr}}_i=P^{\mathsf{max}}_i/\bar{\gamma}_i$, et vice versa. En outre, $\Gamma_i(\mathrm{p}) \geq \bar{\gamma}_i$ si $\alpha_i > 0 $ et $I_i(\mathrm{p}^*) < I^{\mathsf{thr}}_i$. Un tel UE supporté peut réduire sa puissance en réglant son paramètre $\alpha_i$ pour aider les UEs non supportés à améliorer leurs SINR. Nous exploitons ce fait pour développer l'algorithme d'adaptation HPC, qui est décrit dans l'algorithme \ref{res_alg:gms3_r3}.
\renewcommand{\baselinestretch}{0.9}
\small
\begin{algorithm}
\caption{\textsc{L'algorithme d'Adaptation HPC}}
\label{res_alg:gms3_r3}
\begin{algorithmic}[1]
\STATE Initialisation: Mise à jour $\mathrm{p}^{(0)}=0$, $\Delta^{(0)}$ comme $\alpha_i^{(0)}=0$ pour les UEs de voix et $\alpha_i^{(0)}=\alpha_0$ ($\alpha_0 \gg 1$) pour les UEs de données, $\overline{N}^{\ast}=|\overline{\mathcal{U}}^{(0)}|$ et $\Delta^{\ast}=\Delta^{(0)}$.

\STATE Iteration $l$: 

- Exécuter l'algorithme HPC jusqu'à convergence avec $\Delta^{(l)}$.

- Si $|\overline{\mathcal{U}}^{(l)}|>\overline{N}^{\ast}$, mise à jour $\overline{N}^{\ast}=|\overline{\mathcal{U}}^{(l)}|$ et $\Delta^{\ast}=\Delta^{(l)}$.

- Si $\underline{\mathcal{U}}^{(l)} = \varnothing$ ou $\overline{\mathcal{U}}^{(l)} = \varnothing$, puis passez à l'étape 4.

- Si $\underline{\mathcal{U}}^{(l)} \neq \varnothing$ et $\overline{\mathcal{U}}^{(l)} \neq \varnothing$, puis exécuter le \textit{``processus de réactualisation''} comme suit:

$\;$ $\;$ a: Pour UE $i \in \underline{\mathcal{U}}^{(l)}$, mise à jour $\alpha_i^{(l+1)}=\alpha_i^{(l)}$.

$\;$ $\;$ b: Pour UE $i \in \overline{\mathcal{U}}^{(l)}$, mise à jour $\alpha_i^{(l+1)}$ pour que

$\;$ $\;$ $\;$ i) $\alpha_i^{(l)} > \alpha_i^{(l+1)} \geq 0$ si $\alpha_i^{(l)}>0$.

$\;$ $\;$ $\;$ ii) $\alpha_i^{(l+1)} = \alpha_i^{(l)}$ si $\alpha_i^{(l)}=0$.

\STATE Augmenter $l$ et revenir à l'étape 2 jusqu'à ce qu'il n'y a pas de demande de réactualisation pour $\Delta^{(l)}$.
\STATE Mise à jour $\Delta := \Delta^{\ast}$ et exécuter l'algorithme de HPC jusqu'à convergence.
\end{algorithmic}
\end{algorithm}
\renewcommand{\baselinestretch}{1.4}
\normalsize

Soit $\overline{\mathcal{U}}^{(l)}$ et $\underline{\mathcal{U}}^{(l)}$ respectivement les ensembles des UEs supportés et non supportés et à l'itération $l$. Nous utilisons $\overline{N}^{\ast}$ pour maintenir le nombre d’UEs supportés au cours de l'algorithme. Le \textit{``processus de réactualisation''} (\textit{``updating process''}) est conçu de manière  que tous les paramètres $\alpha_i$ des UEs supportés tendent vers à zéro si tous les UEs non supportés ne peuvent pas être aidés. L'algorithme d'adaptation HPC proposé atteint les performances souhaitables suivant es comme caractérisé dans le théorème suivant.
\begin{theorem}
\label{res_thm04}
Soit $\overline{N}_{\sf HPC}$ et $\overline{N}_{\sf TPC}$ le nombre d’UEs supportés dû à l'algorithme d'adaptation HPC proposé et l'algorithme de TPC, respectivement. Nous avons:

- $\overline{N}_{\sf HPC} \geq \overline{N}_{\sf TPC}$.

- Si toutes les exigences de SINR peuvent être remplies par l'algorithme de TPC alors tous les UEs peuvent atteindre leurs SINR cibles en utilisant l'algorithme HPC. En plus il existe des UEs qui peuvent atteindre SINRs plus élevés que les SINR cibles contrôlées par l'algorithme HPC.
\end{theorem}
Un exemple de conception illustrant le processus de mise à jour est présenté dans l'algorithme 4 de [J2] dans laque $\alpha_i$ est mis à jour par les deux processus local et global, c-à-d, la mise à jour localement $\alpha_i$ dans chaque cellule et mise à jour au niveau global $\alpha_i$ lorsqu’il y a UEs pas supportés après l'exécution du processus de mise à jour local.

\subsubsection{Résultats numériques}
\begin{figure}[!t]
        \centering
                \includegraphics[width=0.7 \textwidth]{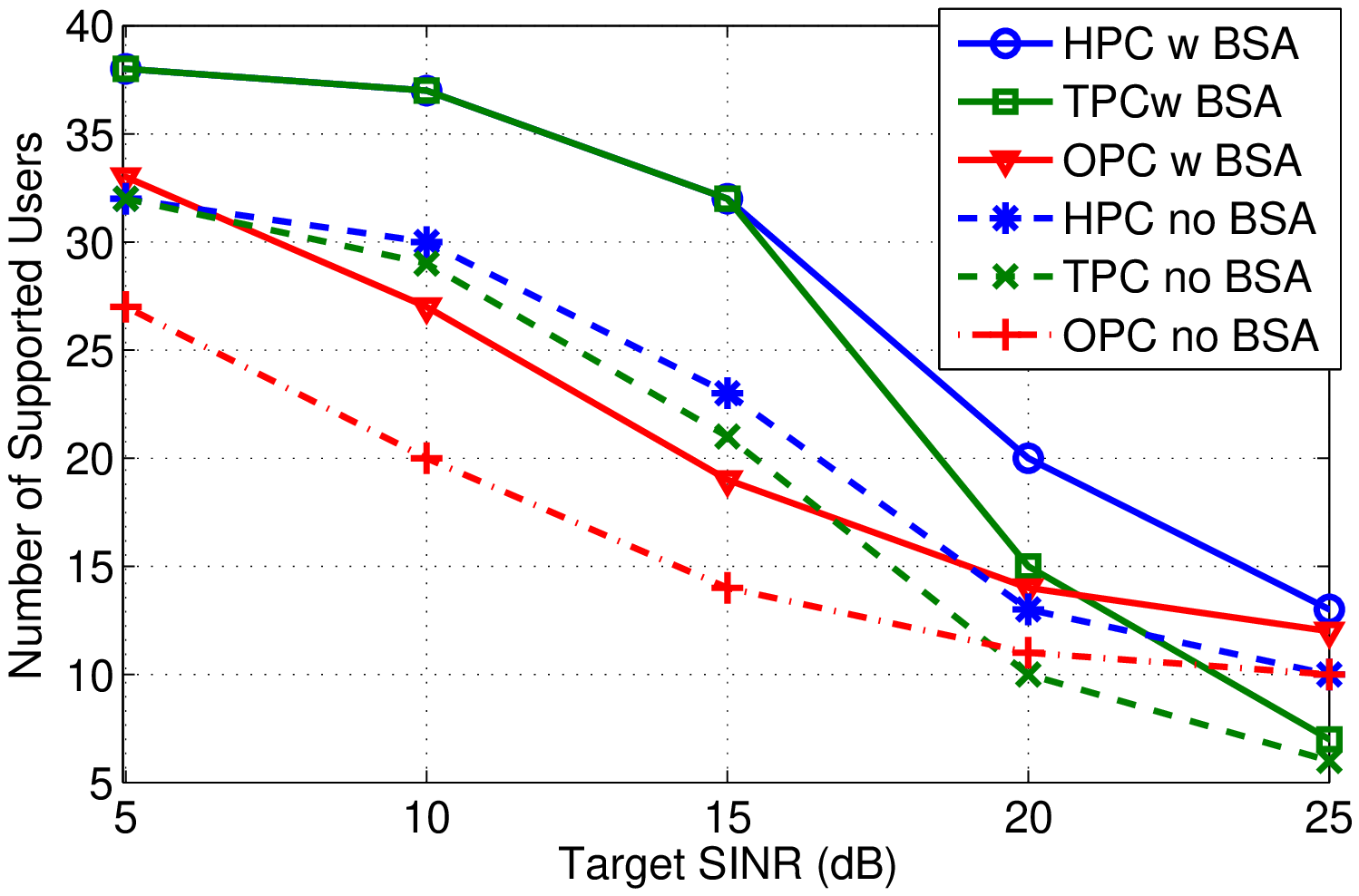}
                \caption{Le nombre d'UEs supportés pour les différents systèmes avec et sans l'aide de l'algorithme BSA.}
                \label{res_r3_fig1a}
\end{figure}
Nous présentons maintenant les résultats numériques montrant les performances des algorithmes proposés. 
Nous plaçons $10$ MUEs et FUEs (de $1$ à $3$ pour chaque cellule) à l'intérieur des cercles de rayons de $r_1=1000 \: m$ et $r_2=50 \: m$, respectivement.
Le gain de puissance du canal $h_ {ij}$ est choisi en fonction de la distance et le canal est supposé d’avoir un évanouissement de type Rayleigh. D'autres paramètres sont définis comme suit: $G=128$, $P_i^{\mathsf{max}}= 0.01 \: W $, $\eta_i = 10^{-13} \: W $, $ W_l = 12 \; dB $. La Fig.~\ref{res_r3_fig1a} montre le nombre d'UEs supportés pour les différents systèmes avec et sans l'aide de l'algorithme BSA. On peut voir que notre algorithme d'adaptation HPC  proposé peut maintenir les conditions de SINR pour le plus grand nombre d'équipements UE par rapport aux algorithmes de TPC et OPC. 
En addition, l'algorithme adaptatif de BSA permet d'augmenter le nombre d'équipements d'utilisateur supportés lorsqu'il est intégré conjointement avec ces algorithmes de PC.
\begin{figure}[!t]
        \centering        
                \includegraphics[width=0.7 \textwidth]{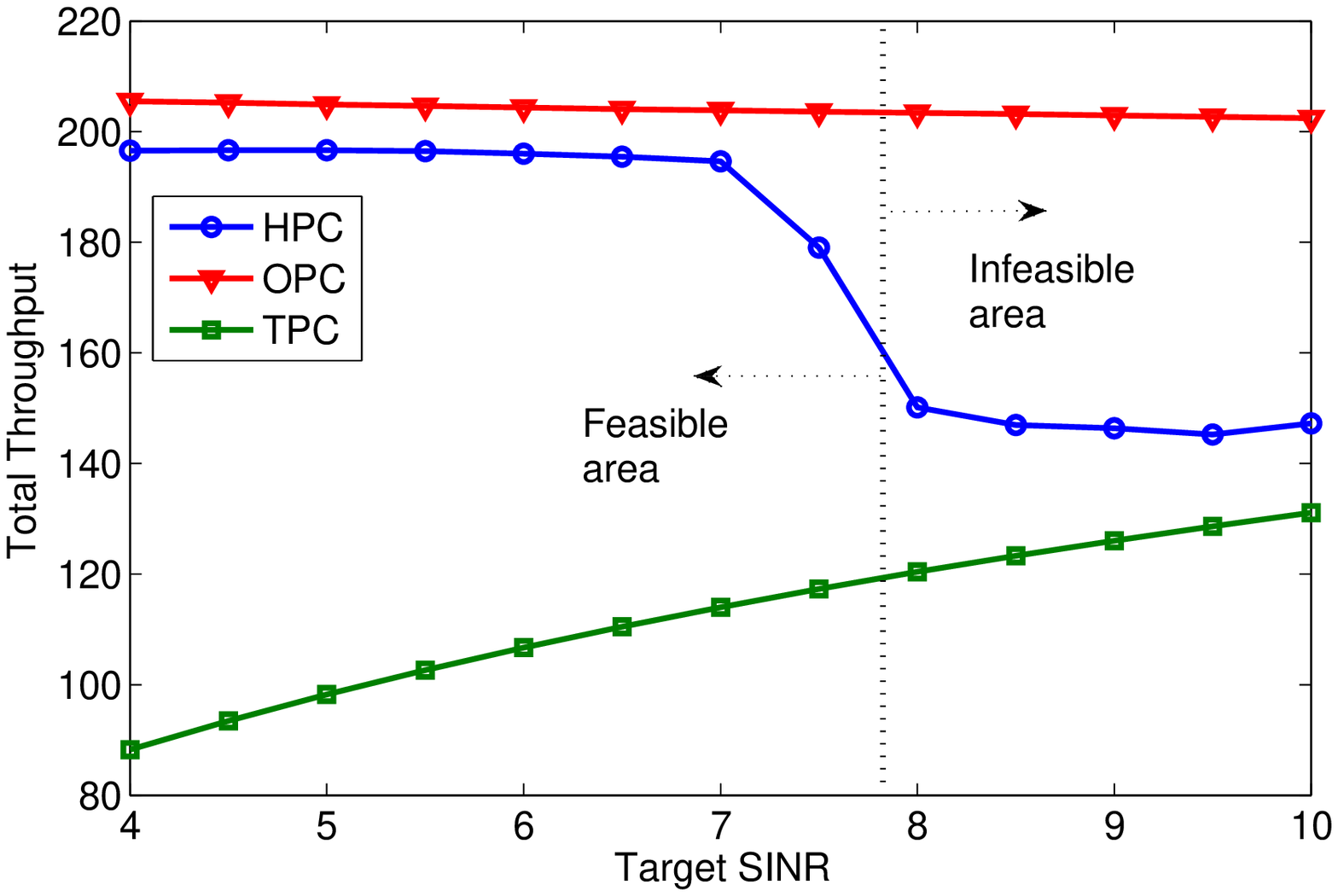}
                \caption{Le débit moyen obtenu par les différents régimes en fonction du rapport au SINR cible.}
 \label{res_r3_fig1b}
\end{figure}

La Figs.~\ref{res_r3_fig1b} illustre le débit moyen obtenu par les différents régimes par rapport à la SINR cible. Ici, le débit de UE $i$ est calculé comme $\log_2(1+ \Gamma_i)$ (b/s/Hz). Comme on peut le constater, le débit moyen obtenu par l'algorithme d'adaptation HPC est supérieur à celui de l'algorithme de TPC, mais inférieur à celui de l'algorithme OPC. En particulier, lorsque le réseau est faiblement chargé ($\bar{\gamma}_i <8$), notre algorithme atteint un débit beaucoup plus élevé que TPC. Cependant, l'écart devient plus petit lorsque la charge du réseau est plus élevée ($\bar{\gamma}_i \geq 8$). En effet, l'algorithme~\ref{res_alg:gms3_r3} tente de maintenir les exigences de SINR pour le plus grand nombre d’UEs.

\subsection{Allocation équitable de ressources pour HetNets à base d’OFDMA}

Dans cette partie, nous considérons de problème conjoint de l'attribution de sous-canaux et de contrôle puissance l’accès multiple à division de fréquence (OFDMA - Orthogonal Frequency Division Multiple Access) pour les réseaux femtocell. 
La conception d'une gestion efficace des ressources radio pour HetNets OFDMA multi-niveau est un important sujet de recherche \cite{perez09} et il y a eu quelques travaux existants dans ce domaine \cite{chuhan11,yanzan12,yushan12,Wangchi12,Long12,hoon11}.
Ceux-ci, cependant, ne considèrent pas PC dans leurs algorithmes d'allocation des ressources ou ne fournissent pas des garanties de qualité de service pour les UEs des deux niveaux du réseau. 
Notre conception vise à concevoir l’allocation conjointe des sous-canaux et de la puissance pour femtocells considérant l'équité pour les FUEs dans chaque femtocell, la protection des qualités de services (QoS) pour les MUEs, et les contraintes maximales de puissance. Au meilleur de notre connaissance, aucun des ouvrages existants considèrent conjointement tous ces problèmes de conception. En particulier, nous faisons les contributions suivantes.
\begin{itemize}
\item Nous présentons une formulation d'allocation des ressources de liaison montante équitable pour les HetNets à deux niveaux, qui vise à maximiser le debit minimal total de tous les femtocells soumis à différentes QoS et les contraintes du système.
\item Nous présentons deux algorithmes: (i) algorithme de recherche exhaustive, (ii) algorithme d'allocation des canaux et de contrôle de puissance distribué.
\item Nous décrivons ensuite comment étendre la conception proposée à différents scénarios y compris l'établissement de liaison descendante, la transmission adaptative à taux multiples et la conception de l'accès ouvert.
\end{itemize}
\subsubsection{Modèle de système}
Nous considérons un système en pleine réutilisation de fréquence dans laquelle il y a $M_{\sf f}$ FUEs desservi par $(K-1)$ FBS, qui sont sous-tendue par un macrocell servant $M_{\sf m}$ MUEs sur $N$ sous-canaux. Nous décrivons les allocations de sous-canal (SAs - subchannel assignments) par une matrice $\mathrm{\bf{A}} \in \mathfrak{R}^{M \times N}$ où
\begin{equation}
\label{res_eq:A}
\mathrm{\bf{A}}(i,n)=a^n_i=\left\lbrace \begin{array}{*{10}{l}}
1 & \text{si sous-canal $n$ est assigné à l'UE $i$}\\
0 & \text{si non}.
\end{array}
\right. 
\end{equation}
Nous supposons que la modulation M-QAM est adoptée pour les communications sur chaque sous-canal où la taille de la constellation est $s$. Selon \cite{proakis01}, le SINR cible pour le régime $s\text{-QAM}$ de modulation peut être calculée comme
\beq \label{res_eq:tSINR}
\bar{\gamma}(s) = \dfrac{\left[ \mathbb{Q}^{-1}( \overline{P}_e /x_s)\right]^2}{y_s}, \: s \in \mathfrak{M},
\eeq
où $\mathbb{Q}(.)$ représente la Q-fonction, $x_s=\frac{2(1-1/\sqrt{s})}{\log_2s}$, $y_s=\frac{3}{2(s-1)}$, et $\overline{P}_e$ est la valeur cible de BER. Ensuite, si schéma de modulation $s\text{-QAM}$ est utilisé, l'efficacité spectrale est ${\log_2s}/{N}$ (b/s/Hz).

Comme pour à la SA, nous définissons $\mathrm{\bf{P}}$ comme une matrice d'attribution de puissance (PA - power allocation) où $\mathrm{\bf{P}}(i,n)=p^n_i$. Ensuite, pour une solution SA et PA donnée, c-à-d, $\mathrm{\bf{A}}$ et $\mathrm{\bf{P}}$, le SINR réalisé à la BS $b_i$ en raison de la transmission d’UE $i$ sur le sous-canal $n$ peut être écrite comme
\begin{equation}
\label{res_sinr}
\Gamma_i^n(\mathrm{\bf{A}},\mathrm{\bf{P}})=\dfrac{a_i^n h_{b_i i}^n p_i^n}{\sum_{j \notin \mathcal{U}_{b_i}}{a_j^n h_{b_ij}^n p_j^n}+\eta_{b_i}^n}=\dfrac{a_i^np_i^n}{I_i^n(\mathrm{\bf{A}},\mathrm{\bf{P}})}.
\end{equation}

Nous supposons que $\mathrm{\bf{A}}_1$ (SA pour macrocell) est fixé alors que nous devons déterminer $\mathrm{\bf{A}}_k, \: 2 \leq k \leq K$ (SA pour femtocells) et l'AP correspondant. Pour protéger la qualité de service des MUEs, nous souhaitons maintenir une cible prédéterminée SINR $\bar {\gamma}_i^n$ pour chacun de ses sous-canaux attribués $n$. Pour FUEs, l'efficacité spectrale (bits/s/Hz) obtenue par FUE $i$ sur un sous-canal peut être écrite comme
\begin{equation}
\label{res_eq:rate_n}
r_i^n(\mathrm{\bf{A}},\mathrm{\bf{P}}) = \left\lbrace \begin{array}{*{5}{l}}
0, & \text{if } \Gamma_i^n(\mathrm{\bf{A}},\mathrm{\bf{P}}) < \bar{\gamma}_i^n,\\
r_{\sf f}, & \text{if } \Gamma_i^n(\mathrm{\bf{A}},\mathrm{\bf{P}}) \geq \bar{\gamma}_i^n,
\end{array} \right. 
\end{equation}
où $r_{\sf f}=(1/n)\log_2s^{\sf f}$ et $\bar{\gamma}_i^n = \bar{\gamma}(s^{\sf f})$.
Maintenant, nous pouvons exprimer l'efficacité spectrale totale obtenue par l'UE $i$ pour les matrices SA et PA donnés, $\mathrm{\bf{A}}$ et $\mathrm{\bf{P}}$, $R_i(\mathrm{\bf{A}},\mathrm{\bf{P}})=\sum_{n=1}^N r_i^n(\mathrm{\bf{A}},\mathrm{\bf{P}})$.
Pour imposer l'équité max-min pour tous les FUEs associés aux mêmes FBS, le problème d'allocation de ressource de liaison montante pour FUEs peut être formulé comme suit:
\begin{align}
\label{res_objfun_r4}
\mathop {\max} \limits_{(\mathrm{\bf{A}},\mathrm{\bf{P}}) } & \sum \limits_{2\leq k \leq K} \mathrm{R}^{(k)}(\mathrm{\bf{A}},\mathrm{\bf{P}}) =  \sum \limits_{2\leq k \leq K} \min_{i \in \mathcal{U}_k} R_i(\mathrm{\bf{A}},\mathrm{\bf{P}}) \\
\text{s. t.} & \sum_{i \in \mathcal{U}_k} a_i^n \leq 1, \:\:\: \forall k \in \mathcal{B} \text{ and } \forall n \in \mathcal{N}, \label{res_eq:c4} \\
{} & \sum_{n=1}^{N}p^n_i \leq P_i^{\texttt{max}}, \quad i \in \mathcal{U}, \label{res_powcon} \\
{} &  \Gamma_i^n(\mathrm{\bf{A}},\mathrm{\bf{P}}) \geq \bar{\gamma}_i^n,\:\: \text{if} \:\: a_i^n=1, \:\: \forall i \in \mathcal{U}_{\sf m}.\label{res_eq:rate_cond}
\end{align}
Ce problème d'allocation des ressources est un programme mixte en nombres entiers, ce qui est, par conséquent, NP-difficile qui est difficile pour résoudre.

\subsubsection{Algorithme de recherche exhaustive optimale}
Pour chaque sous-canal $n$, nous avons besoin de maintenir les contraintes de SINR $\Gamma^n_i(\mathrm{\bf{A}},\mathrm{\bf{P}}) \geq \bar{\gamma}_i^n$ pour tous les UEs qui sont attribués à ce sous-canal.
Par conséquent, la faisabilité d'une SA particulière peut être vérifiée en utilisant le théorème de Perron-Frobenius tel que présenté dans la Section~\ref{Ch2_sec_PC}. 
Comme le nombre de SA possibles est fini, l'algorithme de recherche exhaustive optimale peut être développé comme suit. 
Pour $\mathrm{\bf{A}}_1$ fixe et réalisable, soit $\Omega\{\mathrm{\bf{A}}\}$ la liste de toutes les solutions possibles SA qui satisfont aux contraintes SA (\ref{res_eq:c4}) et la condition de l'équité: $\sum_{n \in \mathcal{N}} a_i^n=\sum_{n \in \mathcal{N}} a_j^n=\tau_k$ pour tous les FUEs $i,j \in \mathcal{U}_k$.
Ensuite, nous trions la liste $\Omega\{\mathrm{\bf{A}}\}$ dans l'ordre décroissant de $\sum_{k=2}^K \tau_k$ et obtenons la liste triée $\Omega^{\ast}\{\mathrm{\bf{A}}\}$. Parmi toutes les solutions possibles SA, une possible réalisation de la valeur la plus élevée de la fonction d’objectif (\ref{res_objfun_r4}) et sa solution PA correspondante est la solution optimale.

\textit{Analyse de la complexité:} En calculant la cardinalité de $\Omega^{\ast}\{\mathrm{\bf{A}}\}$ et la complexité de la vérification de faisabilité pour chacun d'eux, la complexité de l'algorithme de recherche exhaustive peut être exprimée comme $O\left(K^3 \times N \times (N!)^{(K-1)}\right)$, qui est exponentielle du nombre dénombre de sous-canaux et de FBSs. Cet algorithme de recherche exhaustive optimale sera utilisé comme référence pour évaluer la performance de l'algorithme de faible complexité présenté ci-après.

\subsubsection{Algorithme sous-optimal et distribué}
Notre algorithme sous-optimal vise à attribuer le même nombre maximal de sous-canaux à FUEs dans chaque femtocell et effectuer l’optimisation Pareto PA pour FUEs et MUEs sur chaque sous-canal afin qu'ils répondent aux contraintes de SINR dans (\ref{res_eq:rate_n}) et (\ref{res_eq:rate_cond}).
Pour atteindre cet objectif de conception, nous proposons un nouvel algorithme d'allocation des ressources qui est décrit en détail dans l'algorithme~\ref{res_alg:gms1_r4}. 
L'opération clé dans cet algorithme est la SA en fonction du poids itératif qui est effectuée en parallèle à tous les femtocells. 
Le poids SA pour chaque paire de sous-canal et FUE est défini comme la multiplication de la puissance de transmission estimée et un facteur d'échelle capturant la qualité de l'allocation correspondante.
Plus précisément, chaque UE $i$ dans la cellule estimes la puissance de transmission $k$ sur le sous-canal $n$ à chaque itération $l$ de l'algorithme en utilisant l'algorithme TPC $p_i^{n,\texttt{min}}=  I_i^n(l) \bar{\gamma}_i^n$.
Ensuite, le poids d'affectation pour un FUE $i$ sur le sous-canal $n$ dans la cellule $k$ peut être défini comme $w_{i}^n = \chi_i^n p_i^{n,\texttt{min}}$ lorsque le facteur d'échelle $\chi_i^n$ est défini comme suit:
\begin{equation}
\label{res_eq:w2}
\chi_i^n  = \left\lbrace \begin{array}{*{5}{l}}
\alpha_i^{n}, & \text{if } p_i^{n,\texttt{min}} \leq \frac{P_i^{\texttt{max}}} {\tau_k}\\
\alpha_i^{n} \theta_i^{n}, & \text{if } \frac{P_i^{\texttt{max}}} {\tau_k} < p_i^{n,\texttt{min}} \leq P_i^{\texttt{max}}\\
\alpha_i^{n} \delta_i^{n}, & \text{if }  P_i^{\texttt{max}} < p_i^{n,\texttt{min}},
\end{array} \right. 
\end{equation}
où $\tau_k$ désigne le nombre de sous-canaux attribués pour chaque FUE en femtocell $k$; $\alpha_i^n \geq 1$ est un facteur aidant les MUEs à maintenir leur SINR cible (c-à-d, elle est augmentée si $\mathrm{\bf{A}}(i,n)=1$ si la contrainte SINR d’un MEU n’est pas respectée); $\theta_i^n,\delta_i^n \geq 1$ sont encore des facteurs qui sont augmentés si $\mathrm{\bf{A}}(i,n)=1$ tend à exiger une puissance de transmission plus grande que la puissance moyenne par sous-canal (ie, $\frac{P_i^{\texttt{max}}} {\tau_k}$) et le budget de puissance maximale (c.-à-, $P_i^{\texttt{max}}$).

Compte tenu des poids définis pour chaque FUE $i$, le femtocell $k$ trouve sa SA en utilisant l'algorithme hongrois standard (c-à-d, \textit{l'algorithme 14.2.3} donnée dans \cite{Jungnickel08}) pour résoudre le problème suivant
\beq
\min \limits_{\mathrm{\bf{A}}_k} \sum \limits_{i \in \mathcal{U}_k}\sum \limits_{n \in \mathcal{N}} a_i^n w_i^n \;\; \text{s.t.} \;\; \sum_{n \in \mathcal{N}} a_i^n = \tau_k \:\: \forall i \in \mathcal{U}_k. \label{res_eq:lcl_prob}
\eeq

\renewcommand{\baselinestretch}{0.9}
\small
\begin{algorithm}
\caption{\textsc{L'allocation des ressources de liaison montante distribuée}}
\label{res_alg:gms1_r4}
\begin{algorithmic}[1]
\STATE Initialisation: Mise à jour $p_i(0)=0$ $\forall i$, $\mathrm{\bf{A}}_1$ faisable, $\tau_k=\lfloor \frac{N}{\vert \mathcal{U}_k \vert} \rfloor$ et $\varrho_k=0$ pour tous FBSs, et $\alpha_i^{n}=\theta_i^{n}=1$, $\delta_i^{n}=N$ pour tous FUEs et sous-canaux.

\STATE \textbf{Pour le macrocell:}
\STATE MBS estime $\left\lbrace p_i^{n,\texttt{min}}\right\rbrace $ et vérifie la contrainte de puissance (PCON).

\IF{PCON satisfait} \STATE Mettre à jour puissance $p_i^{n,\texttt{min}}$.
\ELSIF{PCON non satisfait} \STATE Mise à jour puissance $p_i^{n,\texttt{min}}$ avec un facteur d'échelle vers le bas.
\STATE Trouver SC en utilisant la puissance maximale, et FUE générer la plus haute ingérence sur ce SC.
\STATE Augmenter $\alpha$ de cette FUE sur que SC.
\ENDIF
\STATE \textbf{Pour chaque femtocell $k \in \mathcal{B}_{\sf f}$:}
\STATE Chaque FBS $k$ estime $\left\lbrace p_i^{n,\texttt{min}}\right\rbrace $.
\IF{$\varrho_k=1$}  \STATE Fixer $\mathrm{\bf{A}}_k(l)=\mathrm{\bf{A}}_k(l-1)$ 
\ELSIF{$\varrho_k=0$}  \STATE Calculer le poids d'affectation des sous-canaux $\left\lbrace w_{i_u}^n\right\rbrace $ et exécuter l'algorithme Hungarian pour obtenir $W_k(l)$ et $\mathrm{\bf{A}}_k(l)$.
\IF{$W_k(l) > V \sum _{i \in \mathcal{U}_k}P_i^{\texttt{max}}$} \STATE Mise à jour $\tau_k:=\tau_k-1$. 
\ENDIF
\ENDIF
\STATE Vérifie la contrainte de puissance (PCON). 
\IF{PCON satisfait} \STATE Mise à jour puissance $p_i^{n,\texttt{min}}$ et $\varrho_{k,i}=1$. 
\ELSIF{PCON non satisfait} \STATE Mise à jour puissance comme $p_i^{n,\texttt{min}}$ avec un facteur d'échelle vers le bas.
\STATE Trouver SC et FUE, qui a passé la plus grande puissance.
\STATE Augmenter $\theta$ de ce FUE sur ce SC.
\ENDIF
\STATE Mise à jour $\varrho_{k}=\prod_{i \in \mathcal{U}_k}\varrho_{k,i}$.
\STATE Mise à jour $l:=l+1$, revenir à l'étape 2 jusqu'à ce que la convergence.
\end{algorithmic}
\end{algorithm}
\renewcommand{\baselinestretch}{1.4}
\normalsize

\textit{Analyse de la complexité:} La complexité de notre algorithme proposé est $O(K \times N^3)$ pour chaque itération en raison du processus hongrois SA à l'étape 16. Cependant, la SA local peut être réalisé en parallèle à tous $(K-1)$ femtocells. Par conséquent, la complexité d'exécution de notre algorithme est $O(N^3)$ multiplié par le nombre d'itérations nécessaires.

\subsubsection{Résultats numériques}

Pour effectuer les simulations, les MUEs et FUEs sont placés au hasard dans les cercles de rayons de $ r_1 = 1000m $ et $ r_2 = 30m $. Les gains d'alimentation de canal $h_{ij}^n$ sont générés en tenant compte à la fois la décoloration Rayleigh et les distances. Autres paramètres: $\eta_i = 10^{-13}\: W $, $W_l=12 \; dB $. 
Dans la Fig.~\ref{res_r4_fig1a}, nous montrons l'efficacité spectrale minimale totale de tous les femtocells (c-à-d la valeur objective optimale de (\ref{res_objfun_r4})) par rapport à la taille de la constellation de FUEs ($s^{\sf f}$) pour le petit réseau en raison de deux algorithmes l’un optimal et l’autre sous-optimal (l'algorithme~\ref{res_alg:gms1_r4}). Comme on le voit, pour les tailles faibles de la constellation (c-à-d, SINR cibles faibles), l'algorithme~\ref{res_alg:gms1_r4} peut atteindre presque la même efficacité spectrale que celle de l’algorithme optimal tandis que pour des valeurs plus élevées de $s^{\sf f}$, il en résulte l'efficacité spectrale légèrement inférieure à celle de l’algorithme optimal.
\begin{figure}[!t]
\centering
\includegraphics[width=0.7\textwidth]{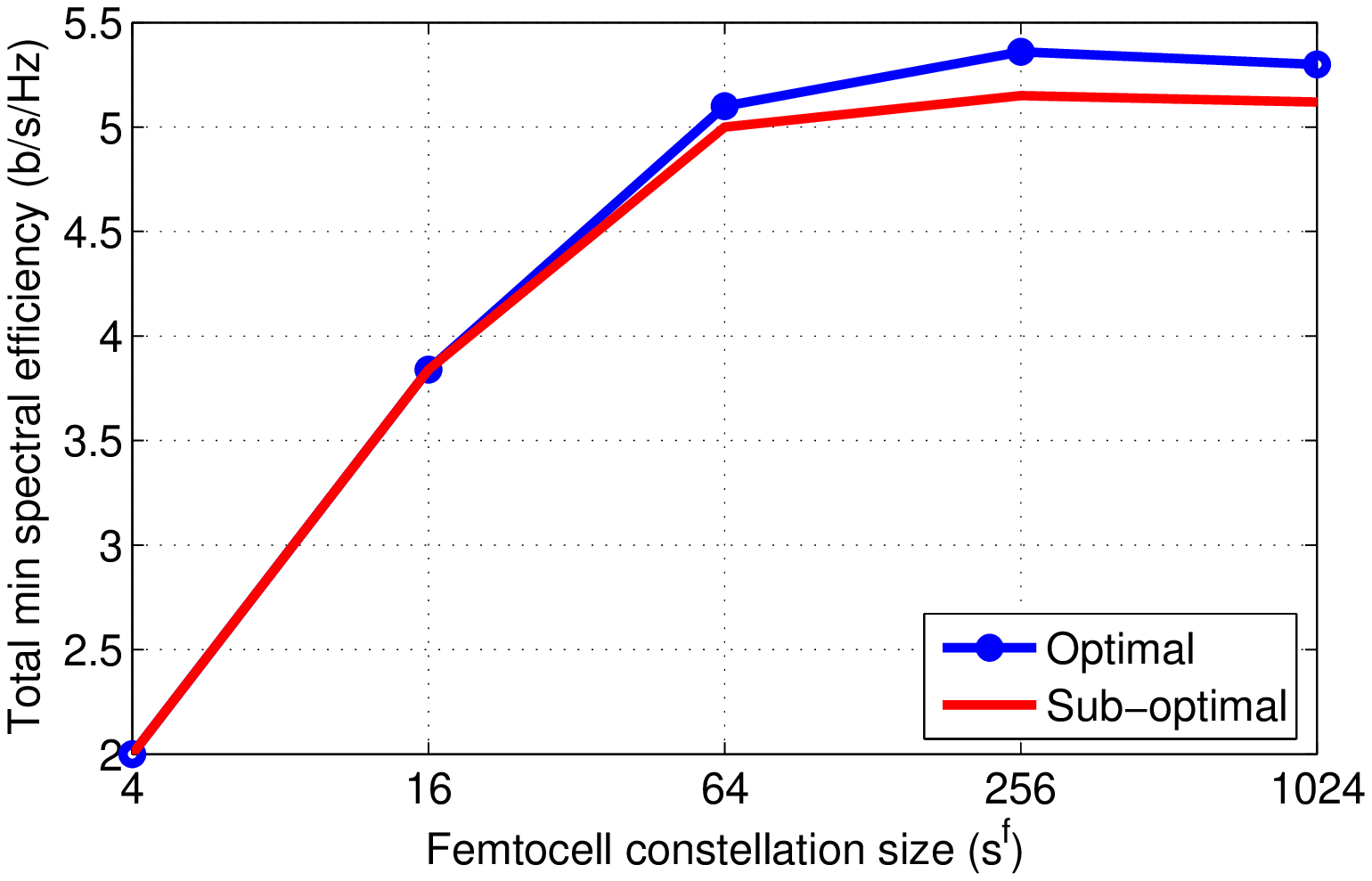}
\caption{L'efficacité spectrale minimale totale de tous les femtocells par rapport à la taille de la constellation de FUEs avec $P_i^{\mathsf{max}} = 0.01 \: W$.}
\label{res_r4_fig1a}
\end{figure}

Dans les figures. \ref{res_r4_fig:pf} et \ref{res_r4_fig:pm}, nous traçons l'efficacité femtocell minimum spectral totale par rapport à la puissance maximale de FUEs ($P_{\sf f}^{\texttt{max}}$) et MUEs ($P_{\sf m}^{\texttt{max}}$), respectivement, pour les différents niveaux de MUEs de modulation (le schéma de modulation des FUEs est $256$-QAM). Ces chiffres montrent que le total minimum spectrale efficacité augmente avec l'augmentation des budgets de puissance maximum, $P_{\sf f}^{\texttt{max}}$ ou $P_{\sf m}^{\texttt{max}}$. Toutefois, cette valeur est saturée quand les budgets de puissance maximum $P_{\sf f}^{\texttt{max}}$ ou $P_{\sf m}^{\texttt{max}}$ deviennent suffisamment grand. En outre, comme le nombre de MUEs augmente, l‘efficacité spectrale minimale totale des femtocell augmente grâce un meilleur gain de diversité offerte par le niveau macro.
\begin{figure}[!t]
\centering
\includegraphics[width=0.7\textwidth]{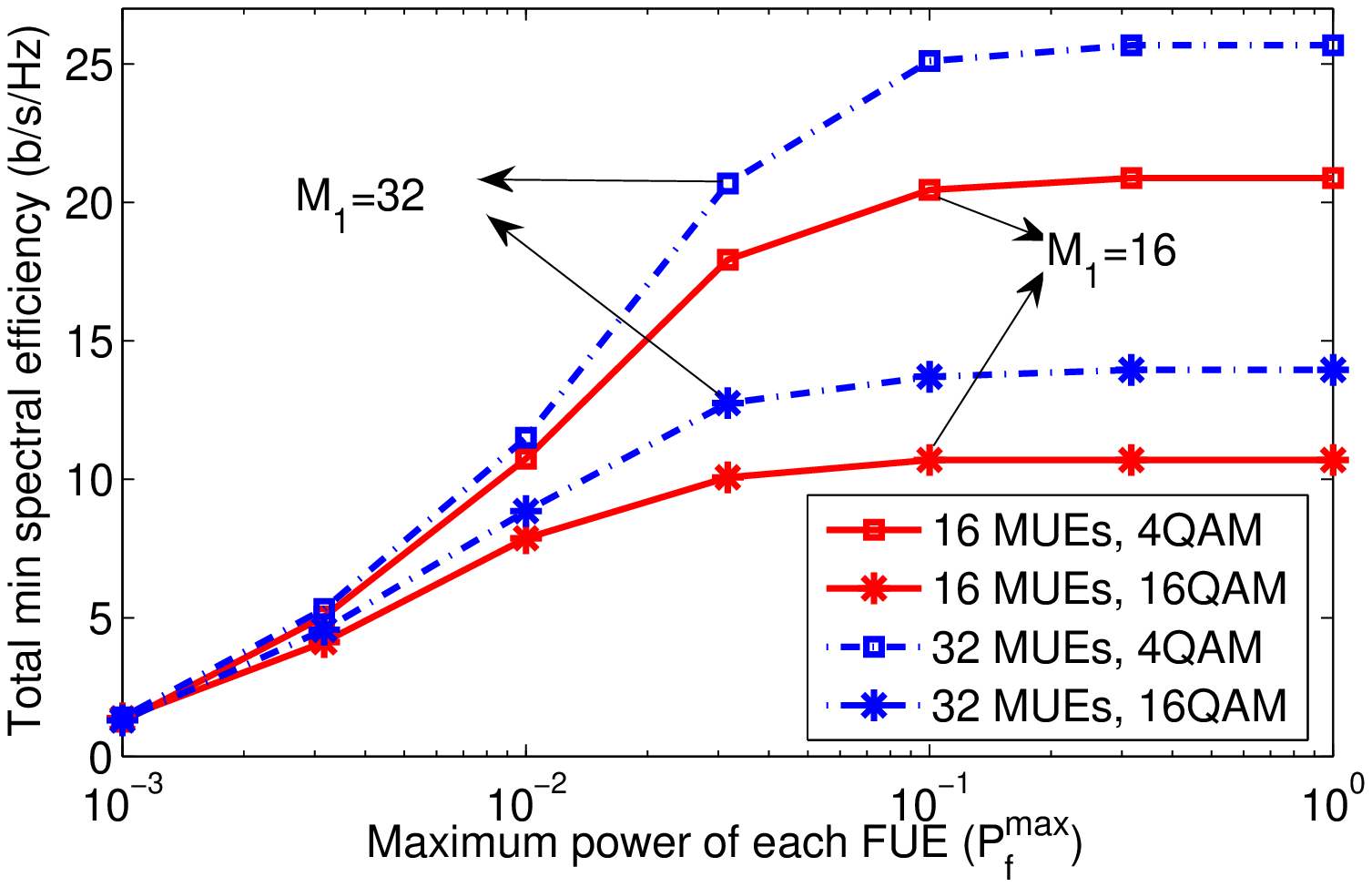}
                \caption{L'efficacité spectrale minimale totale de toutes en fonction du $P_{\sf f}^{\texttt{max}}$. }
                \label{res_r4_fig:pf}
\end{figure}

\begin{figure}[!t]
\centering
\includegraphics[width=0.7\textwidth]{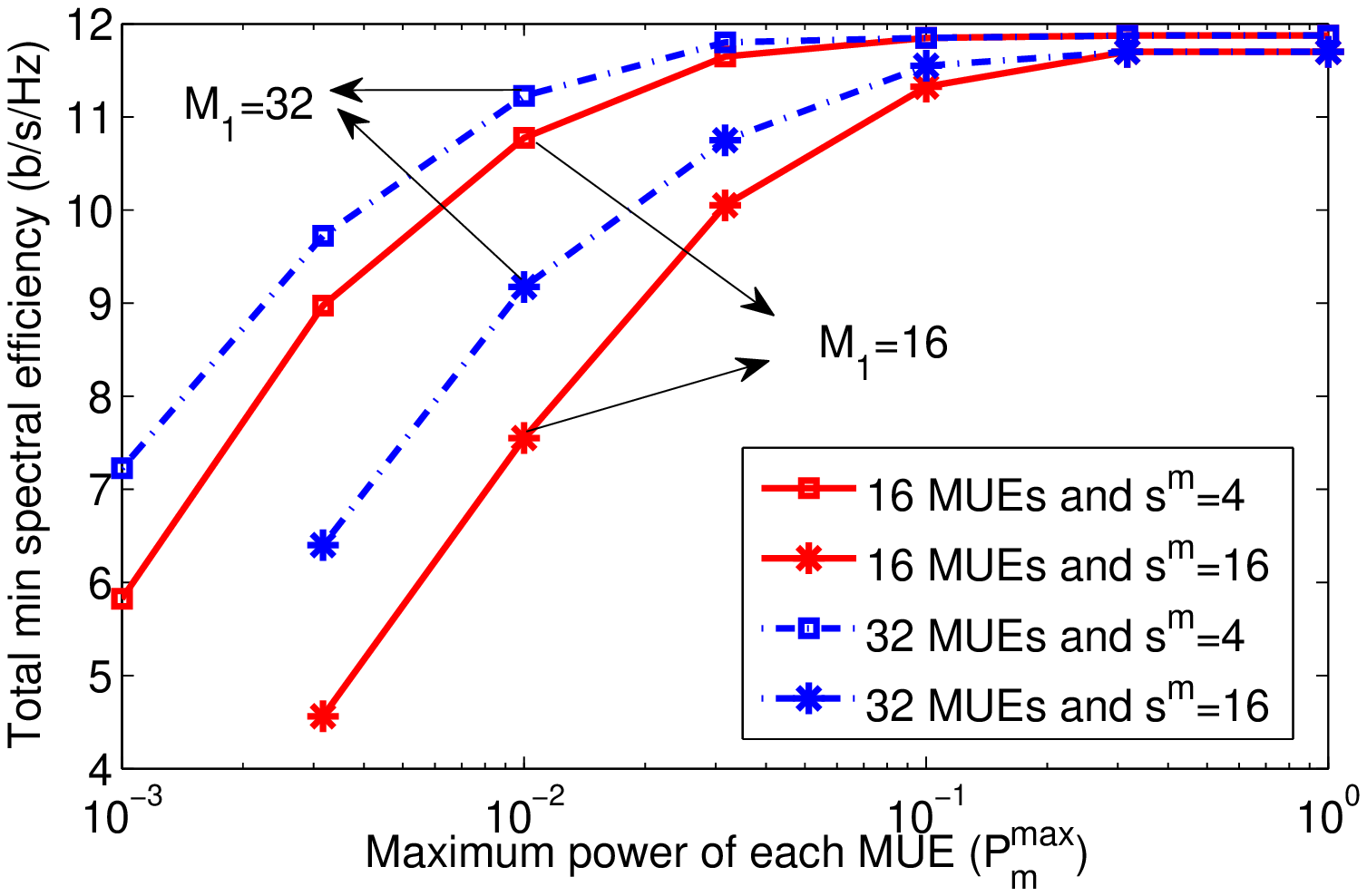}
                \caption{L'efficacité spectrale minimale totale de toutes en fonction du $P_{\sf m}^{\texttt{max}}$. }
                \label{res_r4_fig:pm}
\end{figure}

\subsection{Conception de transmission conjointe pour C-RANs avec la capacité limitée de fronthaul}

Dans cette contribution, nous considérons la conception de la transmission conjointe CoMP pour C-RAN qui considère explicitement la capacité du fronthaul ainsi que des contraintes de qualité de service d'UEs. En particulier, nous faisons les contributions suivantes.
\begin{itemize}
\item Nous formulons le problème conjoint de conception de transmission pour C-RAN qui optimise l'ensemble des RRH desservant chaque UE ainsi que leurs solutions de précodage et de répartition de puissance pour réduire au minimum la puissance totale de transmission compte tenu de la qualité de service des UES et contraintes de capacité du fronthaul.
\item Nous développons deux algorithmes différents de faible complexité pour résoudre le problème sous-jacent: l’algorithme fondé sur de la tarification et l’algorithme de relaxation itérative linéaire.
\item Nous étudions également les configurations étendues où il y a plusieurs contraintes individuelles de capacité du fronthaul et le scénario de transmission multi-flux.
\end{itemize}

\subsubsection{Modèle de système}
Nous considérons la conception de la transmission conjointe pour les communications CoMP de liaison descendante dans le C-RAN avec $K$ RRH et $M$ UEs. Soit $\mathcal{K}$ et $\mathcal{U}$ les ensembles des RRHs et UEs. 
Supposons que le RRH $k$ est équipé $N_k$ antennes et que chaque UE a une seule antenne. 
En outre, on note $\mathbf{v}_u^{k} \in \mathbb{C}^{N_k \times 1}$ comme vecteur de précodage à RRH $k$ correspondant à la transmission à l'UE $u$. 
Le SINR atteint par l'UE $u$ peut être décrit comme
\beq \label{res_eq:SINR}
\Gamma_u =\dfrac{ \left|  \sum \limits_{k \in \mathcal{K}} \mathbf{h}_u^{kH} \mathbf{v}_u^{k}\right| ^2 }
{\sum \limits_{i =1, \neq u}^{M} \left| \sum \limits_{l \in \mathcal{K}}\mathbf{h}_u^{lH} \mathbf{v}_i^{l}\right|^2 + \sigma^2 }.
\eeq
Soit $\mathbf{p}^k=[p^k_1 ... p^k_M]^T$ le vecteur de puissance de transmission de RRH $k$, et $\mathbf{p}=[\mathbf{p}^{1 T} ... \mathbf{p}^{K T}]^T$. 
Notez que $p^k_u=0$ implique que RRH $k$ ne sert pas l'UE $u$. 
En revanche, si $p^k_u>0$, le lien fronthaul du RRH $k$ à l'UE $u$ est activé pour transporter le signal de bande de base pour servir UE $u$ à sa SINR cible requise. 
Par conséquent, la capacité consommée totale des liens fronthaul peut être capturés par le vecteur $\mathbf{p}$ de puissance de transmission et le SINR cible de UEs, qui peut être écrit mathématiquement:
\beq \label{res_eq:Ck}
G(\mathbf{p})= \sum \limits_{k \in \mathcal{K}} \sum \limits_{u \in \mathcal{U}}  \delta( p_u^k ) R_u^{\sf{fh}}   
\eeq 
où $\delta(\cdot)$ désigne la fonction en escalier, et $R_u^{\sf{fh}}$ représente la capacité requise correspondante à UE $u$. 
En particulier, $R_u^{\sf{fh}}$ peut être exprimée comme \cite{dirk_14} $R_u^{\sf{fh}} = Q_{\sf{fh}} \log_2(1+\bar{\gamma}_u)$ où $Q_{\sf{fh}}$ représente un facteur multiplicateur. 
On peut poser le problème FCPM comme suit:
\begin{align}
(\mathcal{P}_{\sf{FCPM}}) \;\;\;\; \min \limits_{\lbrace \mathbf{v}_u^k \rbrace, \mathbf{p} } & \;\;\;\;\;\;\;\; \sum \limits_{k \in \mathcal{K}}\sum \limits_{u \in \mathcal{U}} \mathbf{v}_u^{kH}\mathbf{v}_u^{k}  \label{res_obj_1} \\
\;\;\;\;\;\;\;\; \text{s. t. } & \; \Gamma_u \geq \bar{\gamma}_u, \qquad \forall u \in \mathcal{U}, \label{res_eq:SINR_constraint} \\  
{} & \; \sum \limits_{u \in \mathcal{U}} p^k_u = \sum \limits_{u \in \mathcal{U}} \mathbf{v}_u^{kH}\mathbf{v}_u^{k} \leq P_k, \;\;\; \forall k \in \mathcal{K}, \label{res_eq:pwc} \\
{} & \; G(\mathbf{p})= \sum \limits_{k \in \mathcal{K}} \sum \limits_{u \in \mathcal{U}}  \delta( p_u^k ) R_u^{\sf{fh}}  \leq C, \label{res_eq:C_cons}
\end{align}
où $\bar{\gamma}_u$ désigne le SINR cible du UE $u$, $P_k$ ($k \in \mathcal{K}$) désigne la puissance maximale du RRH $k$, et $C$ est la limite de capacité du fronthaul.

\subsubsection{Algorithme fondé sur de la tarification}
Le premier algorithme de faible complexité a été développé en utilisant la méthode de pénalité pour faire face à la contrainte (\ref{res_eq:C_cons}). Plus précisément, nous considérons la soi-disant la capacité du fronthaul et le problème de minimisation de la puissance (PFCPM - Pricing-Based fronthaul Capacity and Power Minimization) fondée sur de la tarification, qui est défini comme l’équation suivante:
\beq
(\mathcal{P}_{\sf{PFCPM}}) \;\;\;\; \min \limits_{\lbrace \mathbf{v}_u^k \rbrace, \mathbf{p} } \;\; \Vert \mathbf{p} \Vert_{\mathbf{1}} + q G(\mathbf{p})  \;\; \text{s. t. } \;\; \text{ contraintes (\ref{res_eq:SINR_constraint}) et (\ref{res_eq:pwc}).} \label{res_obj_2}  
\eeq
Dans ce qui suit, nous établissons les résultats théoriques sur lesquelles nous développons le mécanisme pour mettre à jour le paramètre de prix $q$.
\begin{proposition} \label{res_R6_lm1} $G_{\sf{PFCPM}}(q)$ est une fonction décroissante de $q$ et sa borne inférieure est $G_{\sf{PFCPM}}(\bar{q}) $ où $\bar{q} = \sum \limits_{k \in \mathcal{K}} P_k/\sigma_{\sf{min}}$, et $\sigma_{\sf{min}}$ est la plus petite valeur non nulle de $\vert G(\mathfrak{a})-G(\mathfrak{a}^{\prime})\vert$ où ${\lbrace\mathfrak{a},\mathfrak{a}^{\prime}\rbrace \subset \mathcal{S}_\mathfrak{a}}$.
\end{proposition}
Ces résultats constituent la base à partir de laquelle nous pouvons développer un algorithme itératif présenté dans l'algorithme~\ref{res_R6_alg:gms2}. En fait, nous pouvons employer la méthode de recherche de bissection pour mettre à jour le paramètre de prix $q$ jusqu'à $G_{\sf{PFCPM}}(\bar{q}) = C$.
\renewcommand{\baselinestretch}{0.9}
\small
\begin{algorithm}[t]
\caption{\textsc{Algorithme Fondé sur de la Tarification pour le problème FCPM}}
\label{res_R6_alg:gms2}
\begin{algorithmic}[1]
\STATE Résoudre problème PFCPM utilisant Alg.~\ref{res_R6_alg:gms3} avec $q^{(0)}=\bar{q}$.
\IF{$G_{\sf{PFCPM}}(\bar{q}) > C$} \STATE Arrêt, le problème de FCPM est infaisable.
\ELSIF{$G_{\sf{PFCPM}}(\bar{q}) = C$} \STATE Arrêt, la solution est obtenue.
\ELSIF{$G_{\sf{PFCPM}}(\bar{q}) < C$} 
\STATE Mise à jour $l=0$, $q_{\sf{U}}^{(l)}=\bar{q}$ et $q_{\sf{L}}^{(l)}=0$.
\REPEAT 
\STATE Mise à jour $l=l+1$ et $q^{(l)}=\left(q_{\sf{U}}^{(l-1)} + q_{\sf{L}}^{(l-1)} \right)/2$.
\STATE Résoudre problème PFCPM utilisant Alg.\ref{res_R6_alg:gms3} avec $q^{(l)}$.
\IF{$G_{\sf{PFCPM}}(q^{(l)}) > C$} \STATE Mise à jour $q_{\sf{U}}^{(l)}=q_{\sf{U}}^{(l-1)}$ et $q_{\sf{L}}^{(l)}=q^{(l)}$.
\ELSIF{$G_{\sf{PFCPM}}(q^{(l)}) < C$} \STATE Mise à jour $q_{\sf{U}}^{(l)}=q^{(l)}$ et $q_{\sf{L}}^{(l)}=q_{\sf{L}}^{(l-1)}$.
\ENDIF
\UNTIL{$G_{\sf{PFCPM}}(q^{(l)}) = C$ or $q_{\sf{U}}^{(l)} - q_{\sf{L}}^{(l)}$ est trop petit}.
\ENDIF
\end{algorithmic}
\end{algorithm}
\renewcommand{\baselinestretch}{1.4}
\normalsize

Nous développons maintenant un algorithme efficace pour résoudre le problème $\mathcal{P}_{\sf{PFCPM}}$ basé sur l'approximation concave de la fonction en escalier. Plus précisément, la fonction $\delta(x)$ pour $x \geq 0$ peut être approcher par une fonction concave appropriée. Désignons $f_{\mathsf{apx}}^{(k,u)}(p^k_u)$ comme la fonction de pénalité concave qui se rapproche de la fonction $\delta(p^k_u)$ correspondant à un lien $(k,u)$. Ensuite, le problème $\mathcal{P}_{\sf{PFCPM}}$ peut alors être approché par le problème suivant:
\beq
 \min \limits_{\lbrace \mathbf{v}_u^k \rbrace,\mathbf{p}} \;\; \sum \limits_{k \in \mathcal{K}} \sum \limits_{u \in \mathcal{U}} p_u^k + q \sum \limits_{k \in \mathcal{K}} \sum \limits_{u \in \mathcal{U}} f_{\mathsf{apx}}^{(k,u)}\left( p_u^k \right) R_u^{\sf{fh}} \;\; 
 \text{s. t. } \;\; \text{ contraintes (\ref{res_eq:SINR_constraint}) et (\ref{res_eq:pwc}).} \label{res_objfun2}
\eeq
En appliquant la méthode du gradient, on peut résoudre le problème (1.22) en résolvant itérativement le problème suivant jusqu'à sa convergence.
\beq
 \min \limits_{\lbrace \mathbf{v}_u^k \rbrace} \;\;\; \sum \limits_{k \in \mathcal{K}} \sum \limits_{u \in \mathcal{U}} \alpha_u^{k(n)} \mathbf{v}_u^{kH}\mathbf{v}_u^{k} \;\; \text{s. t.} \;\; \text{ contraintes (\ref{res_eq:SINR_constraint}) and (\ref{res_eq:pwc})} \label{res_objfun3} 
\eeq
où
\beq \label{res_eq:alp}
\alpha_u^{k(n)}=  1 + q \nabla f_{\mathsf{apx}}^{(k,u)}\left( p_u^k \right)R_u^{\sf{fh}}.
\eeq
Le problème (\ref{res_objfun3}) est un problème de minimisation des sommes pondérées de puissance, qui peut être transformé dans le programme semi-défini convexe (SDP) tel que présenté dans la Section~\ref{Ch2_sec_PMP}. Ensuite, l'algorithme pour résoudre le problème $\mathcal{P}_{\sf{PFCPM}}$ est présenté dans l'algorithme \ref{res_R6_alg:gms3}. L’algorithme~\ref{res_R6_alg:gms2}, qui est proposé pour résoudre le problème $\mathcal{P}_{\sf{FCPM}}$, est basé sur la solution du problème de $\mathcal{P}_{\sf{PFCPM}}$, qui peut être obtenu en utilisant l'algorithme~\ref{res_R6_alg:gms3}. En outre, les résultats indiqués dans la proposition~\ref{res_R6_lm1} et les propriétés standards de la méthode du gradient garantissent la convergence de cet algorithme.
\renewcommand{\baselinestretch}{0.9}
\small
\begin{algorithm}[!t]
\caption{\textsc{Algorithme Fondé sur de la SDP pour le problème PFCPM}}
\label{res_R6_alg:gms3}
\begin{algorithmic}[1]
\STATE Initialisation: Mise à jour $n=0$, et $\alpha_u^{k(0)}=1$ pour tout liens RRH-UE $(k,u)$.

\STATE Itération $n$:
\begin{description}
\item[a.] Résoudre le problème (\ref{res_objfun3}) avec $\left\lbrace \alpha_u^{k(n-1)} \right\rbrace $ pour obtenir $(\mathbf{p}^{(n)},\lbrace \mathbf{v}_u^k \rbrace^{(n)})$.
\item[b.] Réactualiser $\left\lbrace \alpha_u^{k(n)} \right\rbrace $ comme dans (\ref{res_eq:alp}).
\end{description} 
\STATE Définir $n:=n+1$ et revenir à l'étape 2 jusqu'à ce que la convergence.
\end{algorithmic}
\end{algorithm}
\renewcommand{\baselinestretch}{1.4}
\normalsize

\subsubsection{Algorithme de relaxation itérative linéaire}
Nous présentons maintenant l'algorithme de ralaxation itérative linéaire pour traiter directement avec la fonction en escalier dans la contrainte (\ref{res_eq:C_cons}). Tout d'abord, le problème de $\mathcal{P}_{\sf{FCPM}}$ peut être approché par le problème suivant:
\begin{align}
 \min \limits_{\lbrace \mathbf{v}_u^k \rbrace, \mathbf{p}} & \;\;\; \Vert \mathbf{p} \Vert_{\mathbf{1}}  \label{res_obj_4} \\
 \text{s. t. } & \; \text{ contraintes (\ref{res_eq:SINR_constraint})}  \nonumber  \\
 {} & \; \sum \limits_{k \in \mathcal{K}} \sum \limits_{u \in \mathcal{U}} f_{\mathsf{apx}}^{(k,u)}( p_u^k ) R_u^{\sf{fh}} \leq C. \label{res_const4}
\end{align}
Nous approchons la fonction $f_{\mathsf{apx}}^{(k,u)}(p_u^k)$ par une fonction linéaire, qui est basée sur les propriétés de la dualité de conjugué des fonctions convexes \cite{rockafellar70} comme suit:
\beq \label{res_eq:g2}
f_{\mathsf{apx}}^{(k,u)}( p_u^k )= \inf \limits_{z_u^k} \left[ z_u^k p_u^k - f_{\mathsf{apx}}^{(k,u)\ast}(z_u^k) \right]  
\eeq
où $f_{\mathsf{apx}}^{(k,u)\ast}(z)= \inf \limits_{w} \left[ z w -f_{\mathsf{apx}}^{(k,u)}(w) \right]$. On peut vérifier que le problème d'optimisation sur le côté droit de (\ref{res_eq:g2}) atteint son minimum à
\beq \label{res_eq:z}
\hat{z}_u^k=\nabla f_{\mathsf{apx}}^{(k,u)}( w )\vert_{w=p_u^k}.
\eeq
Par conséquent, pour une valeur donnée de $\hat {z} $, le problème (1.25) - (1.26) peut être reformulé par:
\begin{align}
 {} & \min \limits_{\lbrace \mathbf{v}_u^k \rbrace} \:\:\: \sum \limits_{k \in \mathcal{K}} \sum \limits_{u \in \mathcal{U}} \mathbf{v}_u^{kH}\mathbf{v}_u^{k}  \label{res_obj_1sdp} \\
 \text{s. t. } & \; \text{ contraintes (\ref{res_eq:SINR_constraint})}  \nonumber  \\
  {} & \!\!\!\!\! \sum \limits_{k \in \mathcal{K}} \sum \limits_{u \in \mathcal{U}} \! \hat{z}_u^k  R_u^{\sf{fh}}  \mathbf{v}_u^{kH}\mathbf{v}_u^{k} \leq   C \!  + \! \sum \limits_{k \in \mathcal{K}} \sum \limits_{u \in \mathcal{U}} \! R_u^{\sf{fh}} f_{\mathsf{apx}}^{(k,u)\ast}(\hat{z}_u^k).  \label{res_const5}
\end{align}
Le problème décrit par les équations (\ref{res_obj_1sdp})-(\ref{res_const5}) est en effet le problème bien connu de minimisation des sommes de puissance, qui peut être résolu en le transformant en le SDP comme décrit dans la Section~\ref{Ch2_sec_PMP}. En résumé, nous pouvons atteindre nos objectifs de conception en mettant à jour $\left\lbrace \hat{z}_u^k \right\rbrace$ itérativement. En faisant cette mise à jour, nous pouvons résoudre à plusieurs reprises le problème d'optimisation de précodage (\ref{res_obj_1sdp})-(\ref{res_const5}). L’algorithme est décrit dans l'algorithme~\ref{res_R6_alg:gms4}. En outre, les propriétés de la fonction conjuguée garantissent la convergence de cet algorithme.
\renewcommand{\baselinestretch}{0.9}
\small
\begin{algorithm}[!t]
\caption{\textsc{Algorithme de Relaxation Itérative Linéaire}}
\label{res_R6_alg:gms4}
\begin{algorithmic}[1]
\STATE Commencer par une solution et définir $l=0$.
\REPEAT 
\STATE Calculer $\left\lbrace \hat{z}_u^{k,(l)} \right\rbrace $ comme dans (\ref{res_eq:z}) pour tout $(k,u)$.
\STATE Résoudre le problème (\ref{res_obj_1sdp})-(\ref{res_const5}) avec $\left\lbrace \hat{z}_u^{k,(l)}\right\rbrace $.
\STATE Définir $l=l+1$.
\UNTIL Convergence.
\end{algorithmic}
\end{algorithm}
\renewcommand{\baselinestretch}{1.4}
\normalsize

\subsubsection{Résultats numériques}
Nous considérons trois RRH dans le cadre de simulation où les distances entre leurs centres sont égales à $500 \; m$. De plus, UEs sont placés au hasard dans un cercle dont le centre est colocalisé avec un de ces trois RRH et le rayon de chaque cercle est de $125 \; m$. Les gains de canal sont générés en considérant à la fois l’évanouissement Rayleigh et la perte de puissance de transmission dans le canal. D'autres paramètres sont définis comme suit: $\sigma^2=10^{-13} \; W$, $\sigma^2=10^{-13} \; W$, $\tau=10^{-8}$, $P_k=3 \; W$, et $N_k=4$, pour tout $k \in \mathcal{K}$.
\begin{figure}[!t]
\begin{center}
\includegraphics[width=0.7\textwidth]{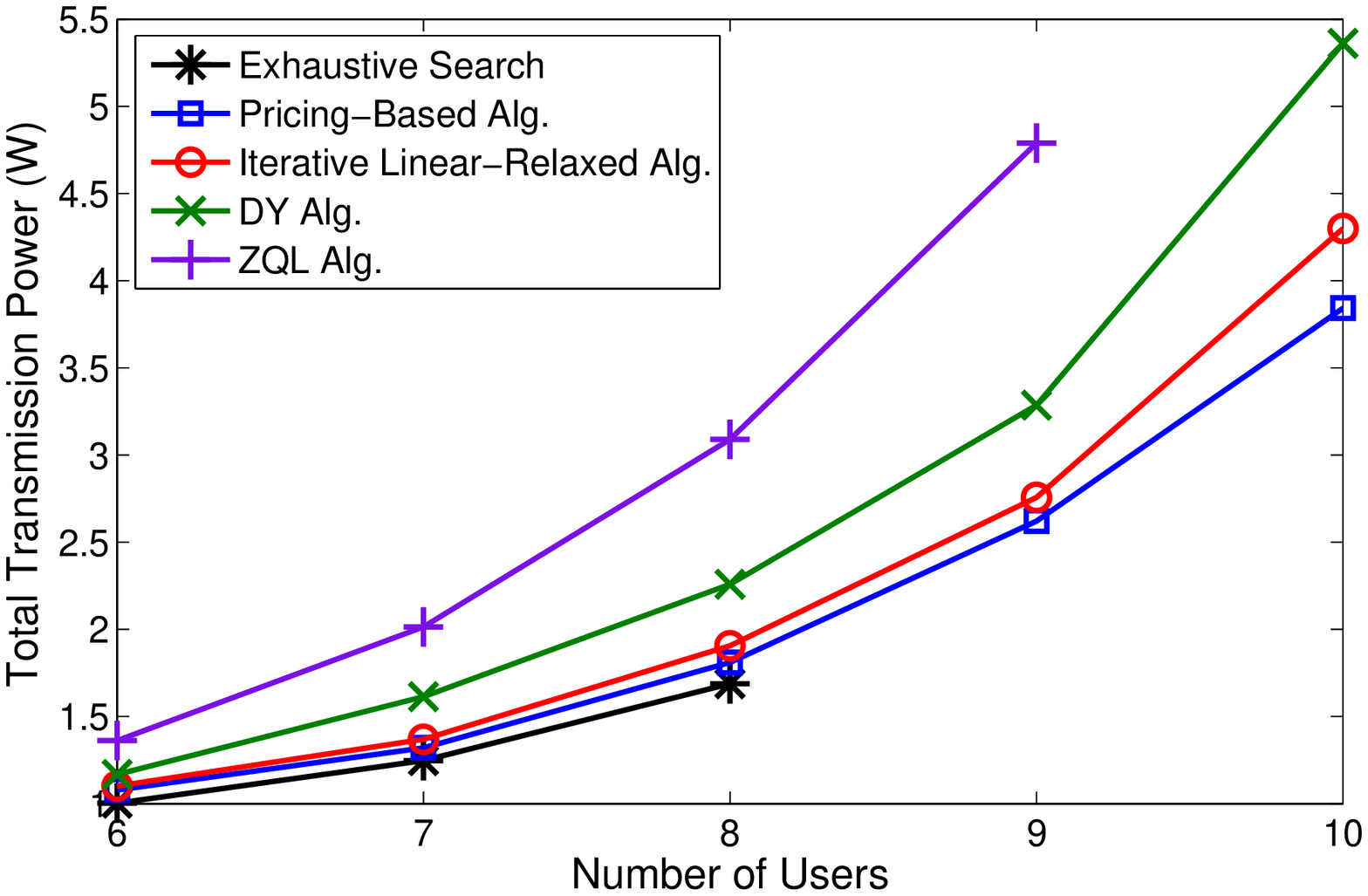}
\end{center}
\caption{Les puissances totales de transmission en fonction du nombre d'UEs.}
\label{res_SMR_3BS_PvsM}
\end{figure}

\begin{figure}[!t]
\begin{center}
\includegraphics[width=0.7\textwidth]{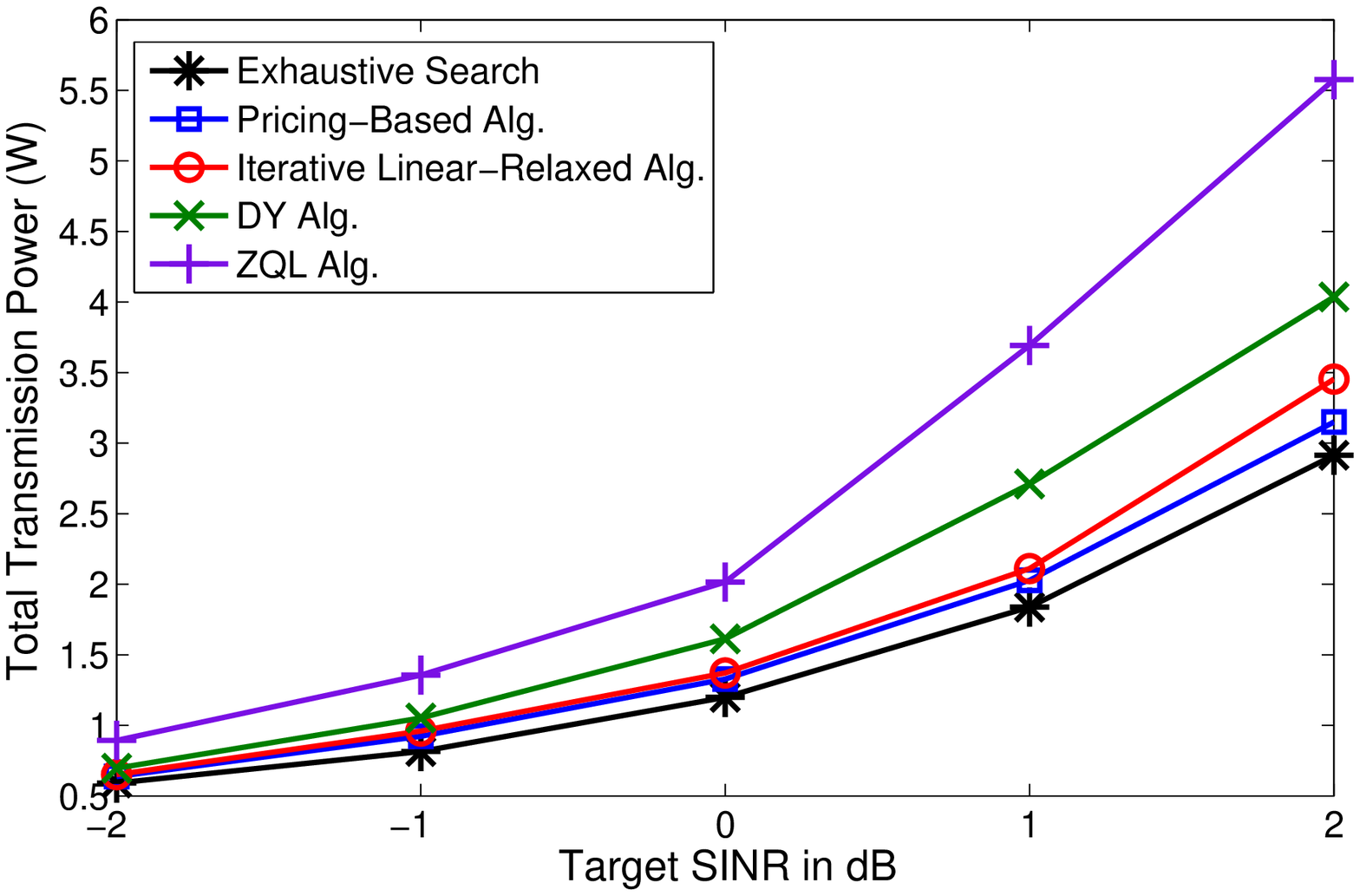}
\end{center}
\caption{Les puissances totales de transmission en fonction des SINR ciblées des UEs.}
\label{res_SMR_3BS_PvsSINR}
\end{figure}

Dans les figures \ref{res_SMR_3BS_PvsM} et \ref{res_SMR_3BS_PvsSINR} , nous montrons les puissances totales de transmission de tous RRH obtenues par la méthode de recherche exhaustive, les algorithmes proposés, ainsi que deux algorithmes de référence, en fonction du nombre d'UEs et des SINR ciblées des UEs, respectivement. 
Comme on le voit, les algorithmes proposés conduisent à la puissance de transmission totale plus faible par rapport aux algorithmes Dy et ZQL. De plus, l’algorithme fondé sur de la tarification est légèrement meilleur que l'algorithme de relaxation itérative linéaire, et les deux algorithmes proposés exigent la puissance d'émission totale légèrement supérieure à celle de l'algorithme de recherche exhaustive optimale.
En outre, l'algorithme de DY surpasse l'algorithme de ZQL mais ces algorithmes exigent des puissances considérablement plus élevées par rapport à nos algorithmes. 
Ainsi, ces figures montrent que la puissance de transmission total augmente lorsque le nombre des UEs augmente dans le cas de tout algorithm.

\subsection{Allocation de ressources pour la virtualisation sans fil OFDMA-basé C-RAN}
Cette contribution considère la conception de l'allocation des ressources pour la liaison montante OFDMA C-RAN soutenant plusieurs opérateurs (OPs). À notre connaissance, la conception de liaison montante C-RAN compte tenu des contraintes sur la capacité du fronthaul et des ressources limitées de l’infoniagoque n'a pas été étudiée auparavant. Cette ligne de travail vise à combler cette lacune dans la littérature C-RAN existante où nous faisons les contributions suivantes.
\begin{itemize}
\item Nous considérons l'allocation des ressources pour la liaison montante OFDMA C-RAN virtualisée où plusieurs OPs partagent l'infrastructure C-RAN et des ressources appartenant à un fournisseur d'infrastructure. Le fournisseur de l'infrastructure et OPs sont intéressés à maximiser leurs profits, qui sont modélisés par deux problèmes d'allocation de ressources différent, de niveau supérieur et inférieur. Le problème de niveau supérieur se concentre sur la tranchement de capacité de trancher fronthaul et des ressources d’informatique en nuage pour tous les OPs. Le niveau inférieur maximise la somme de taux de l'OP en optimisant les taux de transmission des utilisateurs et de l'allocation de bits de quantification pour les signaux en bande de base I/Q comprimés, qui doivent être transférés de RRH vers le nuage. Notre conception vise à maximiser le bénéfice de la somme pondérée à la fois des fournisseurs d'infrastructure et les OPs tenant compte des contraintes sur la limite de calcul de la capacité et de nuages fronthaul.
\item Nous développons un cadre algorithmique en deux étapes. Dans la première étape, nous transformons à la fois au niveau supérieur et niveau inférieur des problèmes dans les problèmes relaxés d'optimisation correspondants en relaxant les variables discrètes sous-jacentes en variables continues. Nous montrons que ces problèmes sont convexes et donc ils peuvent être résolus de façon optimale. Dans la deuxième étape, nous proposons deux méthodes pour utiliser la solution obtenue en résolvant les problèmes relaxés pour parvenir à une solution réalisable finale pour le problème d'origine. 
\end{itemize}

\subsubsection{Modèle de système}
Nous considérons le système OFDMA de liaison montante C-RAN avec $K$ RRHs ($K$ cellules) supportant $O$ différents OPs. Nous supposons que chaque cellule utilise l'ensemble du spectre avec $S$ blocs de ressources physiques (PRBs). Désignons $\mathcal{S}$ et $\mathcal{S}_k^o$ comme l'ensemble de tous les PRBs et l'ensemble des PRBs attribués pour à l'OP $o$ dans la cellule $k$. 
Soit $x_{k}^{(s)}$ le signal de bande de base transmis sur PRB $s$ dans la cellule $k$. 
Le signal reçu à RRH $k$ sur PRB $s$ peut être écrit comme:
\beq
y^{(s)}_{k} = \sum_{j \in \mathcal{K}} h_{k,j}^{(s)} \sqrt{p_{j}^{(s)}} x_{j}^{(s)} + \eta_k^{(s)},
\eeq
où $p_{j}^{(s)}$ désigne la puissance de transmission correspondante à $x_{j}^{(s)}$, $h_{k,j}^{(s)}$ est le gain de canal du UE affecté au PRB $s$ dans la cellule $j$ à RRH $k$, et $\eta_k$ désigne le bruit gaussien.
Supposons que RRH $k$ utilise $b_{k}^{(s)}$ bits pour quantifier chacune des parties I/Q de $y^{(s)}_{k}$ avant de les transmettre au nuage. 
La puissance de bruit de quantification et le SINR du signal quantifié peut être formulée par:
\beqn 
q^{(s)}_{k}(b^{(s)}_{k}) \simeq   \dfrac{\sqrt{3} \pi}{2^{2 b_{k}^{(s)} +1}}  Y^{(s)}_{k}, \label{res_eq:q_kmn} \\
\gamma^{(s)}_{k}(b^{(s)}_{k}) = \dfrac{D_k^{(s)}}{I_k^{(s)}+ q^{(s)}_{k}(b^{(s)}_{k})} \simeq \dfrac{D_k^{(s)}}{I_k^{(s)} + \frac{\sqrt{3} 
\pi Y^{(s)}_{k}}{2^{2 b_{k}^{(s)} +1}} }, \label{res_eq:SINR}
\eeqn
où $Y^{(s)}_{k}$ représente la puissance du signal reçu $y^{(s)}_{k}$, soit $ Y^{(s)}_{k} = \sum \limits_{j \in \mathcal{K}} \vert h_{k,j}^{(s)} \vert ^2 p_{j}^{(s)} + \sigma_k^{(s)2} = D_k^{(s)} + I_k^{(s)}$. 
De plus, $I_k^{(s)}$ représente l'interférence multi-cellules sur le PRB $s$. 
Nous supposons que le processus de quantification satisfait $q^{(s)}_{k}(b^{(s)}_{k}) \leq \sqrt{Y^{(s)}_{k} I^{(s)}_{k}}$, ce qui peut garantir la qualité du signal quantifiée que $\gamma^{(s)}_{k}(b^{(s)}_{k}) \geq \underline{\gamma}^{(s)}_{k} = \sqrt{\bar{\gamma}^{(s)}_{k} + 1} - 1$.
Cette exigence de conception est équivalente à $b^{(s)}_{k} \geq \lceil \underline{b}^{(s)}_{k} \rceil = \dfrac{1}{2}\left( \log_2 \left( \sqrt{3}\pi \sqrt{\dfrac{ Y^{(s)}_{k}}{I^{(s)}_{k}}} \right) -1\right)$.
Le nombre total de bits de quantification du OP $o$ dans la cellule $k$ peut être exprimé comme:
\beq
B_k^o = (2N_{\sf{RE}}) \sum_{s \in \mathcal{S}_k^o} b_{k}^{(s)},
\eeq 
où $N_{\sf{RE}}$ est le nombre d'éléments de ressources (RE – Ressource Elements) en une seconde où RE correspond à un symbole de données sur chaque sous-porteuse. Nous supposons que le débit de données $r^{(s)}_{k}$ (en \textit{`` bits par utilisation de canal (bits pcu)''}, ou \textit{``bits par RE''}) pour la transmission sur PRB $s$ est choisi parmi un ensemble prédéterminé de taux $\mathcal{R} = \lbrace R_1, R_2, ..., R_{M_R} \rbrace$.
La complexité des calculs nécessaires pour décoder correctement les bits d'information peut être exprimée par:
\beq \label{res_eq:cpt}
C^{(s)}_{k}  = \chi^{(s)}_{k}(r^{(s)}_{k},b^{(s)}_{k})  =  A r^{(s)}_{k} \!\! \left[B - 2 \log_2 \! \left( \log_2 \! 
\left( \! 1 +\gamma^{(s)}_{k}(b^{(s)}_{k}) \! \right) \! - r^{(s)}_{k} \! \right) \! \right].
\eeq
Soit $\mb{r}^o$ et $\mb{b}^o$ les vecteurs représentant les taux et le nombre de bits de quantification correspondant à OP $o$. La complexité de calcul totale requise par OP $o$ (en \textit{``bit itération par seconde (bips)''}) peut être écrite comme
\beq
C_{\sf{tt}}^o(\mb{r}^o,\mb{b}^o) = N_{\sf{RE}} \sum_{k \in \mathcal{K}} \sum_{s \in \mathcal{S}^o_k} C^{(s)}_{k}.
\eeq

Nous sommes maintenant prêts à présenter la formulation du problème. 
Soit $C^o$ et $B^o_k$ la complexité de calcul et la capacité du fronthaul correspondante à OP $o$ dans la cellule $k$, et leurs prix sont appelés $\psi^o$ (\textcent$/bips$) et $\beta^o_k$ (\textcent$/bps$), respectivement. L'OP $o$ doit payer une somme de $G^{\sf{InP}}_o=\psi^o C^o + \sum \limits_{k \in \mathcal{K}} \beta^o_k B^o_k$ pour utiliser la tranche de ressource allouée. 
De plus, soit $G^{\sf{InP}} = \sum \limits_{o \in \mathcal{O}} G^{\sf{InP}}_o$ les revenus obtenus par le fournisseur d'infrastructure. Ignorant le coût de fonctionnement du fournisseur d'infrastructure, ce revenu $G^{\sf{InP}}$ peut être considéré comme le bénéfice de ce fournisseur d'infrastructure. Notez que nous ne considérons pas l'optimisation pour les revenus ou les coûts liés au spectre fréquentiel, qui sont supposéd fixes.

Pour plus de commodité de notation, nous allons introduire le vecteur $\mb{B}^o = [B^o_1,...,B^o_K]$. De plus, on note $R_o(C^o,\mb{B}^o)$ la somme de taux maximale atteinte par tous les UEs d’OP $o$. $R_o(C^o,\mb{B}^o)$ représente aussi le résultat des problème de niveau inférieur $(\mathcal{P}^o)$ correspond à OP $o$:
\begin{subequations} \label{res_obj_11}
\begin{align}
 R_o(C^o,\mb{B}^o) = \max \limits_{\mb{r}^o, \mb{b}^o} & \;\;\;\;\;  \sum_{k \in \mathcal{K}} \sum_{s \in \mathcal{S}^o_k} r^{(s)}_{k}  \label{res_obj_1a} \\
\text{s. t. } & \; \sum_{k \in \mathcal{K}} \sum_{s \in \mathcal{S}^o_k} C^{(s)}_{k} \leq C^o/N_{\sf{RE}}, \label{res_obj_1b} \\
 {} & \;  r^{(s)}_{k} \leq \log_2 \left(1 + \gamma^{(s)}_{k}(b^{(s)}_{k}) \right), \forall k \in \mathcal{K}, \forall s \in \mathcal{S}^o_k, \label{res_obj_1c}\\
 {} & \; \sum_{s \in \mathcal{S}^o_k} b_{k}^{(s)} \leq B^o_k/(2N_{\sf{RE}}), \forall k \in \mathcal{K}, \label{res_obj_1d} \\
  {} & \; b^{(s)}_{k} \geq \lceil \underline{b}^{(s)}_{k} \rceil, \forall k \in \mathcal{K}, \forall s \in \mathcal{S}^o_k, \label{res_obj_1e} \\
 {} & \; \text{$b^{(s)}_{k}$ est un entier et } r^{(s)}_{k} \in \mathcal{M}_R, \forall k \in \mathcal{K}, \forall s \in \mathcal{S}^o_k. \label{res_obj_1f} 
\end{align}
\end{subequations}
Supposons que le prix par unité de taux d’OP $o$ est $\rho^o$ (\textcent$/bps$). Ensuite, le bénéfice atteint OP $o$ peut être exprimé en $G^{\sf{OP}}_o= \rho^o N_{\sf{RE}} R_o(C^o,\mb{B}^o) - G^{\sf{InP}}_o$. Le problème de niveau supérieur vise à maximiser la somme pondérée des bénéfices du fournisseur d'infrastructure et les OPs, qui peuvent être déclarés comme
\begin{subequations} \label{res_obj_U}
\begin{align}
 \max \limits_{\lbrace C^o, \mb{B}^o \rbrace} & \;\;\;\;\; \upsilon^{\sf{InP}} G^{\sf{InP}} + \sum \limits_{o \in \mathcal{O}} \upsilon^o G^{\sf{OP}}_o, \label{res_obj_Ua} \\
\text{s. t. } & \;\;\; \sum_{o \in \mathcal{O}} C^{o} \leq \bar{C}_{\sf{cloud}}, \label{res_obj_Ub} \\
{} & \;\;\; \sum_{o \in \mathcal{O}} B^{o}_k \leq \bar{B}_k, \forall k \in \mathcal{K}. \label{res_obj_Uc} 
\end{align}
\end{subequations}

\subsubsection{Problème de relaxation et de convexité}
Nous approchons d'abord les problèmes de niveau inférieur en relaxant leurs variables de discrètes on obtient alors le problème de relaxation au niveau inférieur (RLL – Relaxed Lower Level):
\begin{align}
 \max \limits_{\mb{r}^o, \mb{b}^o} \;\;\;  \sum_{k \in \mathcal{K}} \sum_{s \in \mathcal{S}} r^{(s)}_{k} \;\;\; \text{s. t. } & \text{ (\ref{res_obj_1b})-(\ref{res_obj_1e}),} \label{res_ctb4_obj_2} \\
 {} & R_{\sf{min}} \leq r^{(s)}_{k} \leq R_{\sf{max}}, \forall k \in \mathcal{K}, \forall s \in \mathcal{S}, \label{res_obj_1f2}
\end{align}
où $R_{\sf{min}}$ et $R_{\sf{max}}$ sont le taux plus bas, le plus élevé et le taux fixé $\mathcal{M}_R$. Soit $\bar{R}_o(C^o,\mb{B}^o)$ la solution optimale de ce problème de RLL. Le problème de niveau supérieur peut être relaxé en un problème de relaxation au niveau supérieur (RUL – Relaxed Upper Level):
\begin{align}
 \max \limits_{\lbrace C^o, \mb{B}^o \rbrace} \;\;  \sum \limits_{o \in \mathcal{O}} \Psi^o(C^o, \mb{B}^o) \;\; \text{s. t.} \;\;\; \text{(\ref{res_obj_Ub}), (\ref{res_obj_Uc})}, \label{res_obj_U2}
\end{align}
où $\Psi^o(C^o, \mb{B}^o)=(\upsilon^{\sf{InP}}-\upsilon^o) G^{\sf{InP}}_o + \upsilon^o \rho^o N_{\sf{RE}} \bar{R}_o(C^o,\mb{B}^o)$.  Nous indiquons quelques résultats importants sur lesquels nous pouvons développer des algorithmes pour résoudre notre problème de conception.
\begin{theorem} \label{res_R8_thr1} $\chi^{(s)}_{k,u}(r^{(s)}_{k},b^{(s)}_{k})$ est une fonction conjointement convexe par rapport aux variables $(r^{(s)}_{k},b^{(s)}_{k})$ si $q^{(s)}_{k}(b^{(s)}_{k}) \leq \sqrt{Y^{(s)}_{k} I^{(s)}_{k}}$.
\end{theorem}

\begin{proposition} \label{res_R8_prt2}
Le problème d'optimisation RLL est convexe.
\end{proposition}

\begin{theorem} \label{res_R8_thr2}
$\bar{R}_o(C^o,\mb{B}^o)$ est une fonction concave par rapport à $C^o$ et $\mb{B}^o$.
\end{theorem}

\begin{proposition} \label{res_R8_prt5}
Le problème d'optimisation RUL est convexe.
\end{proposition}

\subsubsection{Algorithmes proposés}
En se basant sur les propositions~\ref{res_R8_prt2}-\ref{res_R8_prt5}, nous pouvons développer un algorithme pour résoudre le problème RLL de manière optimale. Nous exprimons d'abord la double fonction $g^o(\lambda)$ du problème RLL correspondant à OP $o$ en assouplissant les contraintes du calcul infonuagique comme suit:
\begin{align}
g^o(\lambda^o) \! = \! \max \limits_{\mb{r}^o, \mb{b}^o} \Phi^o \! ( \! \lambda^o \! ,\mb{r}^o \!, \mb{b}^o \! ) \;\;\; \text{ s. t. (\ref{res_obj_1c})-(\ref{res_obj_1e}) et (\ref{res_obj_1f2}),} \label{res_obj_3}
\end{align}
où $\lambda^o$ représente le multiplicateur de Lagrange correspondant à la contrainte de calcul de nuage, et
\begin{align}
\Phi^o(\lambda^o,\mb{r}^o, \mb{b}^o) \! = \! \sum_{k \in \mathcal{K}} \! \sum_{s \in \mathcal{S}^o_k} \! r^{(s)}_{k} \! - \! \lambda^o \!\! \left( \! \sum_{k \in \mathcal{K}} \! \sum_{s \in \mathcal{S}^o_k} \! C^{(s)}_{k} \! - \! \dfrac{C^o}{N_{\sf{RE}}} \! \right).
\end{align}
Ensuite, le double problème peut être décrit comme $\min \limits_{\lambda^o \geq 0} g^o(\lambda^o)$. Étant donné que le double problème est convexe, il peut être résolu en utilisant l'algorithme de sous-gradient standard comme suit:
\beq \label{res_eq:updld}
\lambda^o_{(l+1)} = \left[\lambda^o_{(l)} + \delta^o_{(l)}\left( \sum_{k \in \mathcal{K}} \sum_{s \in \mathcal{S}^o_k} C^{(s)}_{k} - \dfrac{C^o}{N_{\sf{RE}}} \right)  \right]^{+},
\eeq
où $l$ désigne l'indice d'itération et $\delta^o_{(l)}$ représente le pas de mise à jour.

L'étape suivante consiste à déterminer les solutions optimales de $\mb{r}^o$ et $\mb{b}^o$ pour une valeur de $\lambda^o$ donnée. Selon les résultats de la proposition~\ref{res_R8_prt2}, le problème \ref{res_obj_3} est également convexe. En fait, il est possible de déterminer la solution optimale pour l'une des deux variables $\mb{r}^o_k$ et $\mb{b}^o_k$ tout en maintenant l'autre fixe. Plus précisément, la solution optimale de $\mb{r}^o_k$ pour un donné $\mb{b}^o_k$ peut être exprimée comme
\beq \label{res_eq:r_opt}
r^{(s)\star}_{k} \!\! = \max \! \left[ \! R_{\sf{min}},  \min \! \left( \! t(b^{(s)}_{k}),R_{\sf{max}}, r\vert_{\frac{\partial w(r)}{\partial r} 
= - E^{(s)}_{k}} \! \right) \! \right]
\eeq
où $w(r)= 2\lambda^o A r \log_2 \left( 1 - r/{t(b^{(s)}_{k})} \right)$. De même, la solution optimale de $\mb{b}^o_k $ pour une donnée $\mb{r}^o_k$ peut être exprimée comme
\beq \label{res_eq:b_opt}
b^{(s)\star}_k =  \max \left(\lceil \underline{b}^{(s)}_{k} \rceil, t^{-1}\left( r^{(s)}_{k}\right), b\vert_{\frac{\partial z(b)}{\partial b} = \mu} \right), 
\eeq
où $\mu$ doit être déterminée pour satisfaire $\sum_{s \in \mathcal{S}^o_k} b_{k}^{(s)\star} = B^o_k/(2N_{\sf{RE}})$. L'algorithme pour résoudre le problème de RLL est résumé dans l'algorithme~\ref{res_R8_alg:gms1}.
\renewcommand{\baselinestretch}{0.9}
\small
\begin{algorithm}[!t]
\caption{\textsc{Algorithme pour Résoudre le Problème RLL}}
\label{res_R8_alg:gms1}
\begin{algorithmic}[1]
\STATE Initialisation: Mise à jour $r_k^{(s)} = R_{\sf{min}}$ pour tous $(k,s) \in \mathcal{K} \times \mathcal{S}$, $\lambda^o_{(0)}=0$ and $l=0$. Choisissez un paramètre de tolérance $\varepsilon$ pour la convergence.
\REPEAT
\FOR{$k \in \mathcal{K}$} 
\REPEAT
\STATE Fixer $\mb{r}^o_k$ et mettre à jour $\mb{b}^o_k$ comme dans (\ref{res_eq:b_opt}) avec $\lambda^o_{(l)}$.
\STATE Fixer $\mb{b}^o_k$ et mettre à jour $\mb{r}^o_k$ comme dans (\ref{res_eq:r_opt}) avec $\lambda^o_{(l)}$.
\UNTIL $\mb{r}^o_k$ et $\mb{b}^o_k$ soient constantes.
\ENDFOR 
\STATE Calculer tous les $C_k^{(s)}$avec les valeurs actuelles de  $\mb{r}^o_k$ et $\mb{b}^o_k$.
\STATE Mettre à jour $\lambda^o_{(l+1)}$ comme dans (\ref{res_eq:updld}).
\STATE Mise à jour $l=l+1$.
\UNTIL $\vert \lambda^o_{(l)} - \lambda^o_{(l-1)} \vert < \varepsilon$.
\end{algorithmic}
\end{algorithm}
\renewcommand{\baselinestretch}{1.4}
\normalsize

Nous présentons maintenant un algorithme basé sur le sous-gradient pour résoudre le problème de la RUL. Plus précisément, l'algorithme de mise à jour itérative de sous-gradient du $C^o,\mb{B}^o$ peut être décrit comme suit:
\beq \label{res_eq:updcb}
[C^o,\mb{B}^o]_{(n+1)} = \mathcal{P} \left[[C^o,\mb{B}^o]_{(n)} + \tau^o_{(n)} \nabla \Psi^o(C^o,\mb{B}^o)  \right]
\eeq
où $\tau^o_{(n)}$ est le pas de mise à jour, $\mathcal{P}\left[*\right]$ désigne l'opération de projection dans la région faisable du problème de la RUL, et
\beq
\nabla \Psi^o(C^o,\mb{B}^o)=\left[ \dfrac{\partial\Psi^o(C^o,\mb{B}^o)}{\partial C^o}, \dfrac{\partial \Psi^o(C^o,\mb{B}^o)}{\partial B^o_1}, \ldots, 
 \dfrac{\partial \Psi^o(C^o,\mb{B}^o)}{\partial B^o_K} \right] ^T.
\eeq

On peut en effet déterminer $\nabla \Psi^o(C^o,\mb{B}^o)$ en utilisant l'algorithme~\ref{res_R8_alg:gms1}. Nous résumons ensuite la procédure de mise à jour de $[C^o,\mb{B}^o]$'s dans l'algorithme~\ref{res_R8_alg:gms2}, qui résout le problème RUL.
\renewcommand{\baselinestretch}{0.9}
\small
\begin{algorithm}
\caption{\textsc{Algorithme pour Résoudre le Problème RUL}}
\label{res_R8_alg:gms2}
\begin{algorithmic}[1]
\STATE Initialisation: Mise à jour $C^o_{(0)} = \bar{C}_{\sf{cloud}}/O$, et $B^o_{k,(0)} = \bar{B}_k/O$ pour tous $(o,k) \in \mathcal{O} \times \mathcal{K}$, $\nu^o_{(0)}=0$ pour tous $o \in \mathcal{O}$, et $n=0$. 
\REPEAT
\STATE Exécuter Algorithm~\ref{res_R8_alg:gms1} pour obtenir $\lbrace \bar{R}_o(C^o,\mb{B}^o)\rbrace $ pour certaines valeurs des $\lbrace [C^o,\mb{B}^o]_{(n)}\rbrace$ $\forall o \in \mathcal{O}$
obtenues à partir de l'itération précédente, qui sont utilisées pour calculer $\nabla \Psi^o(C^o,\mb{B}^o)$.
\STATE Mettre à jour $[C^o,\mb{B}^o]_{(n+1)}$ pour tous $o \in \mathcal{O}$ d'après (\ref{res_eq:updcb}).
\STATE Mise à jour $n=n+1$.
\UNTIL Convergence.
\end{algorithmic}
\end{algorithm}
\renewcommand{\baselinestretch}{1.4}
\normalsize

L'algorithme itératif proposé converge si $\delta$ et $\tau$ sont choisis de manière appropriée, par exemple, $\delta^o_{(l)} = 1 /\sqrt{l}$ et $\tau^o_{(n)} = 1 /\sqrt{n}$ \cite{Bertsekas99}. Après l'exécution de l'algorithme \ref{res_R8_alg:gms2}, on peut obtenir une solution réalisable ${C^o,B^o}$ et les valeurs optimales $\left\lbrace r^{(s) \bigstar}_k \right\rbrace $ et $\left\lbrace b^{(s) \bigstar}_k \right\rbrace $ des problèmes relaxés.

Notez que $\left\lbrace r^{(s) \bigstar}_k \right\rbrace $ et $\left\lbrace b^{(s) \bigstar}_k \right\rbrace $ sont des nombres réels, tandis que les variables d'origine ne prennent que des valeurs discrètes. Pour résoudre ce problème, nous proposons deux méthodes d'arrondies, c-à-d, la méthode d'arrondissement itérative (\textit{IR – Iterative Rounding}) et la méthode d’arrondi et d’ajustement (\textit{RA – one-time Rounding and Adjustment}), qui complètent les variables continues réalisées à partir de l'algorithme \ref{res_R8_alg:gms2} pour atteindre les solutions discrètes et réalisables correspondantes pour le problème initial.

\subsubsection{Résultats numériques}
Nous considérons un réseau avec sept cellules et trois OPs ($O=3$). Les gains de canal sont générés en considérant à la fois l’évanouissement Rayleigh et la diminution de puissance transmise dans le canal. 
D'autres paramètres sont définis comme suit: $\sigma^2=10^{-13} \; W$ and $p^{(s)}_k=0.1 \;W$, $T'=0.2$, $\zeta=6$, et $\epsilon_{\sf{ch}}=10\%$.
Pour démontrer les impacts de la limite du calcul infonuagique et la capacité du fronthaul sur la performance du système, nous montrons la somme de taux du système réalisable, qui est le résultat de l'algorithme~\ref{res_R8_alg:gms2} fixant le problème RLL, par rapport à la limite du calcul infonuagique ($\bar{C}_{\sf{cloud}}$) et la fronthaul capacité $\bar{B}_k$ dans la Fig.~\ref{res_Fig03}. Il est évident que plus la limite de l’infonuagique est grande, plus le résultat de la capacité du fronthaul dans la somme de taux du système est élevé comme prévu. En outre, la somme de taux devient saturé quand la limite de l’infonuagique ou de capacité du fronthaul devient suffisamment grande.
\begin{figure}[!t]
\begin{center}
\includegraphics[width=0.7\textwidth]{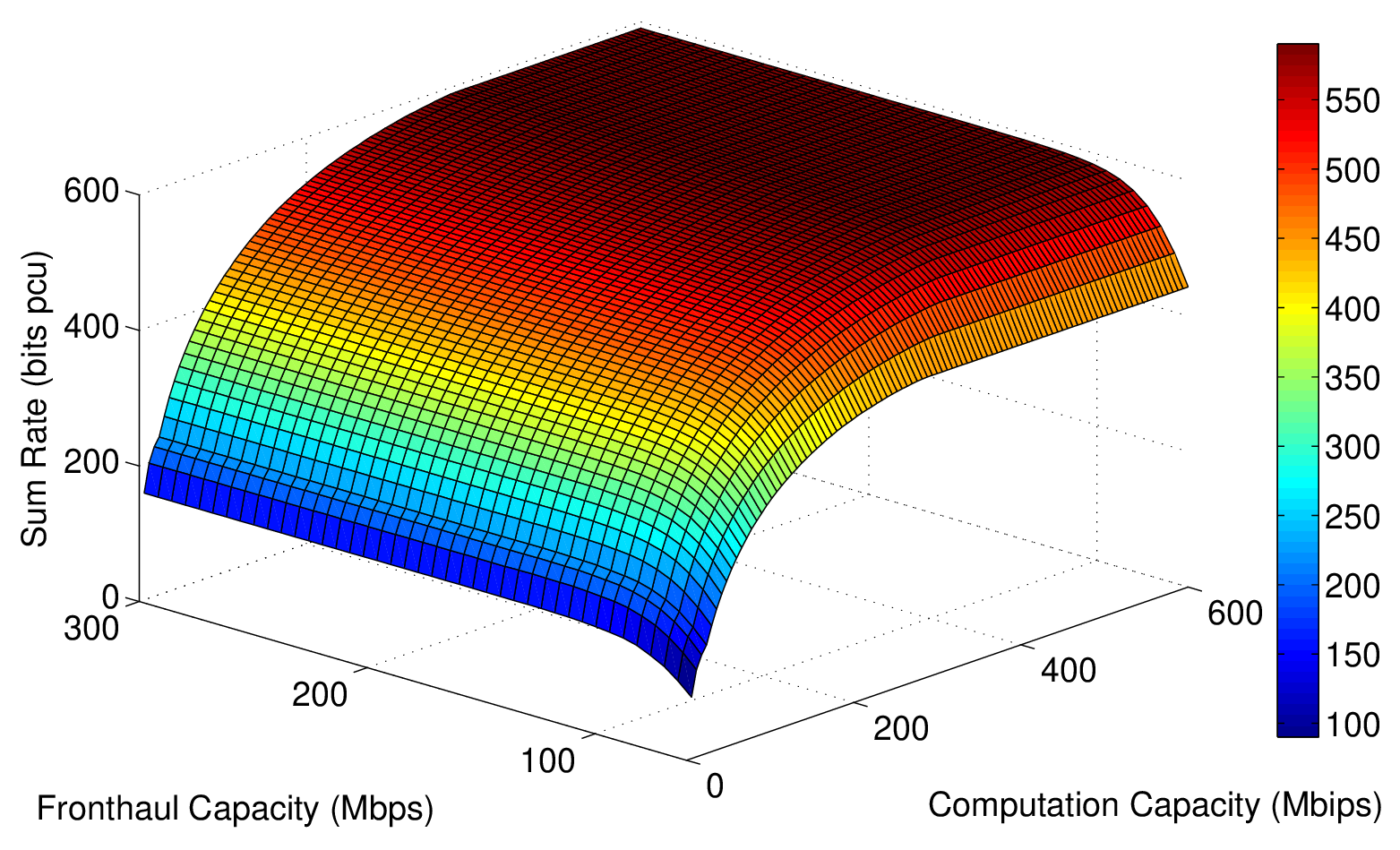}
\end{center}
\caption{La somme de taux du système en fonction des $\bar{C}_{\sf{cloud}}$ et $\bar{B}_k$.}
\label{res_Fig03}
\end{figure}

Dans la figure. \ref{res_Fig04}, nous illustrons les variations de taux de somme du système obtenu par
 nos algorithmes proposés sans arrondir (Prop. Relaxed. Alg.), méthode d'arrondi IR (Prop. Alg. with IR), RA- méthode d'arrondi RA (Prop. Alg. with RA), et l'algorithme glouton rapide (Greedy Alg.), par rapport au nombre de PRBs dans chaque cellule. Pour obtenir ces résultats numériques, on ajoute séquentiellement un plus PRB pour chaque OP dans chaque cellule pour obtenir différents points sur chaque courbe. Comme on peut le remarquer, nos algorithmes proposés surpassent considérablement l'algorithme glouton dans les scénarios étudiés. En outre, la méthode IR conduit à un meilleur rendement que celui obtenu par le procédé RA. Fait intéressant, le débit total du système augmente d'abord, puis diminue à mesure que le nombre de PRBs disponible dans chaque cellule devient plus grand.
 
\begin{figure}[!t]
\begin{center}
\includegraphics[width=0.7\textwidth]{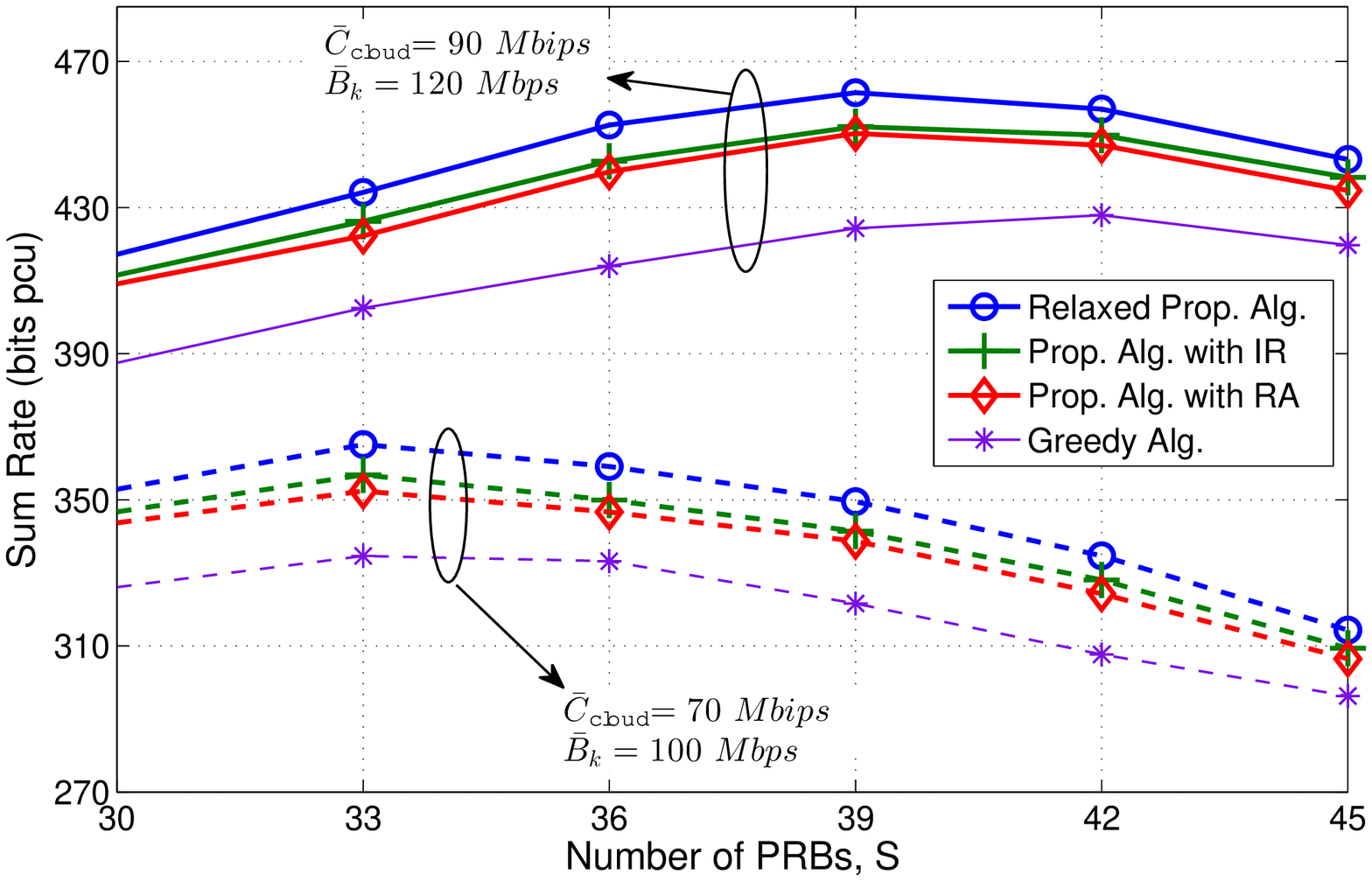}
\end{center}
\caption{La somme de taux du système en fonction du nombre des PRBs ($S$).}
\label{res_Fig04}
\end{figure}

\section{Remarques finales}

Dans cette thèse, différents techniques et algorithmes ont été proposés pour la gestion des nouvelles ressources pour les HetNets sans fil et C-RAN. Ces derniers sont deux candidats d'architecture importants pour les réseaux sans fil à grande vitesse de la 5G. En particulier, ce travail de recherche a permis quatre contributions importantes.
Premièrement, un modèle conjoint PC et BAS a été proposé pour les HetNets sans fil à multi-niveau. La solution HPC proposée et développée est capable à équilibrer efficacement l'amélioration du débit du réseau  et la réalisation des QoS requises pour le nombre maximum des UEs.
Deuxièmement, les algorithmes distribués d'allocation des ressources ont été proposés afin de maximiser l'efficacité spectrale minimale totale au niveau des femtocell tout en assurant l'équité entre les FEUs et la protection QoS pour tous les MUEs dans les réseaux sans fil HetNets à deux niveaux femtocell et macrocell. 

Troisièmement, nous avons développé les algorithmes efficaces à faible complexité pour répondre aux besoins des C-RAN en liaison descendante visant à réduire au minimum la puissance d'émission totale tout en respectant des contraintes de la transmission, de la capacité du fronthaul, et de la QoS des UEs. 
Enfin, nous avons mis au point un cadre algorithmique  qui est capable de capturer des intérêts et des interactions pertinents entre le fournisseur d'infrastructure et les OPs pour la liaison montante dans un C-RAN virtualisée supportant plusieurs OP. En particulier, la conception proposée tient compte des contraintes pratiques sur la limite de capacité du fronthaul et d'infonuagique a pour but de maximiser le profit des fournisseurs d'infrastructure ainsi que des OPs. Ceci est possible par l'optimisation de la répartition des tarifs pour les UEs et le nombre de bits de quantification des signaux en bande de base en liaison montante des utilisateurs. Ce travail de recherche a produit cinq articles de revues (quatre publications \cite{VuHa_TVT14_BSA_PC,VuHa_TVT14_PC_SA,VuHa_TVT16,Tri_Access15} et un en soumission \cite{VuHa_sTWC16}) et neuf articles publiés dans les conférences prestigieuses en communications sans fil \cite{VuHa_VTC12,VuHa_WCNC13,VuHa_GBCWS13,vuha_ciss_2014,VuHa_WCNC14,vuha_globecom_2014,VuHa_WCNC15,VuHa_WCNC16,VuHa_ICC16}. 

%% file: chap3/Ha_chap3.tex

\chapter{Introduction}
\label{chapter1}

The fifth-generation (5G) wireless cellular system, which would be deployed by 2020, is expected to deliver significantly higher
capacity and better network performance compared to those of the current fourth-generation (4G) system. This is required to meet various practical technical
and service challenges. 
Specifically, it is predicted that tens of billions of wireless devices will be connected to the wireless network over next few years.
Together with the increasing number of connections, the amount of mobile data traffic has been exploding at an exponential rate. 
Therefore, the 5G wireless mobile network should be able to support
the data traffic volume of an order of magnitude larger than that in the current wireless network \cite{Zander13,Boccardi14,Bhushan14,Le_EU15}.

Therefore, more advanced wireless architecture, as well as radical and innovative access technologies, must be proposed to 
meet the exponential growth of mobile data and connectivity requirements in the coming years \cite{Dottling09, 
Lopez-Perez11, 3GPP,cisco13,cisco16}. 
Toward this end, two important wireless cellular architectures, namely wireless heterogeneous networks (HetNets) based on
the dense deployment of small cells and the cloud radio access networks (Cloud-RAN or C-RAN) have been proposed and actively 
studied by both academic and industry communities. These two terms Cloud-RAN and C-RAN
 will be used changeably in the sequel. This doctoral dissertation focuses on the radio resource management 
designs for wireless systems exploiting these two potential network architectures.


\section{5G Wireless System and Enabling Technologies}

\begin{figure}[!t]
\begin{center}
\includegraphics[width=0.7\textwidth]{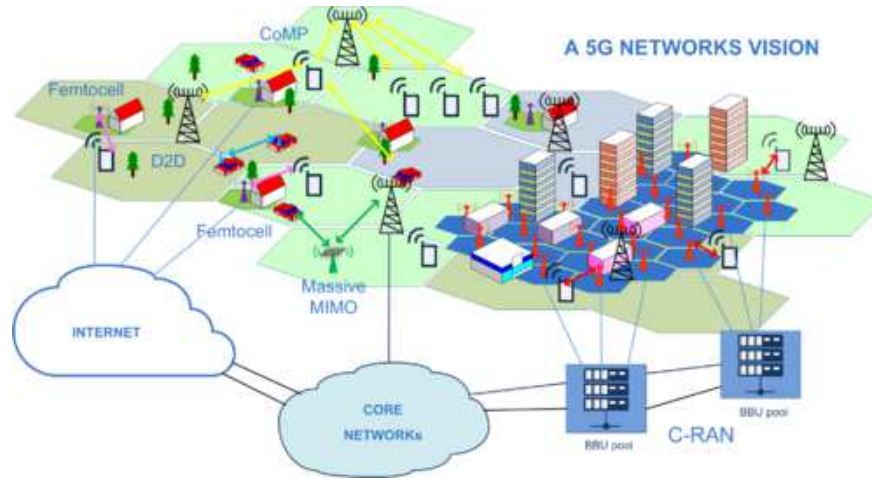}
\end{center}
\caption{Next-generation Wireless Cellular Networks}
\label{Ch1_fig:5G}
\end{figure}

An illustration of the next-generation wireless cellular network is given in Fig.~\ref{Ch1_fig:5G}, which integrates various enabling technologies and
novel network architectures for scalable and efficient support of various wireless applications
including existing human-type and emerging machine-type ones with diverse quality-of-service (QoS) \cite{Bhushan14,Osseiran14,Chin14}.
This is to realize the vision of the hyper-connected world with billions of wireless
connections supporting the so-called internet of things \cite{Le_EU15, Hellemans15}. 
Toward this end, different key 5G enabling technologies have been actively developed
by both academia and wireless industry in recent years \cite{Costa-Perez_CM13,Liang_CST15},
some of which are illustrated in Fig.~\ref{Ch1_fig:5G_sol}.

\begin{figure}[!t]
\begin{center}
\includegraphics[width=0.7\textwidth]{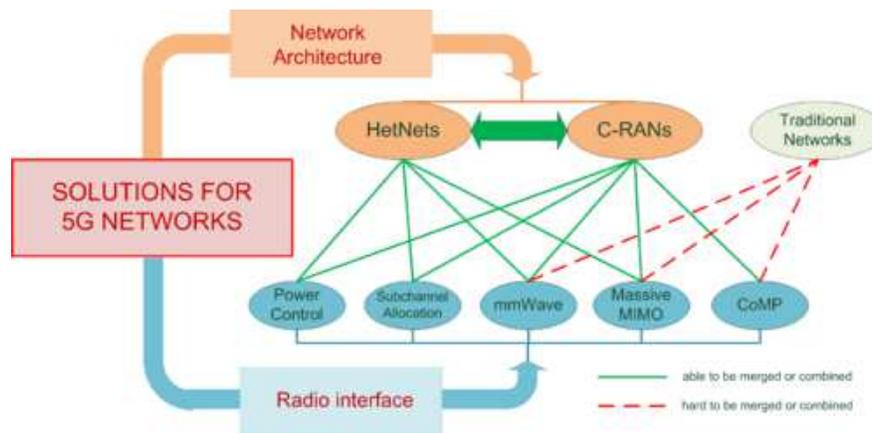}
\end{center}
\caption{Enabling 5G Wireless Technologies}
\label{Ch1_fig:5G_sol}
\end{figure}

Various advanced 5G wireless technologies have been researched actively over the past years.
Specifically, there is growing interest in exploiting the millimeter wave (mmWave) for communications
to meet the aggressive network capacity requirements from emerging Gigabit broadband applications \cite{Pi11,Rappaport13}.
This is especially significant because most microwave frequency bands below $6$ GHz have been very crowded \cite{Pi11,Rappaport13}. 
Another important technology, namely massive MIMO, which employs the large number of antennas 
to fundamentally increase the spectrum and energy efficiency has also been considered by both academic and industry communities.
Moreover, massive MIMO can enable us to achieve significant performance enhancements while avoiding
complicated and costly network coordination strategies such as the multipoint transmission or reception (CoMP) techniques \cite{Le_EU15}.

Furthermore, wireless virtualization and software defined networking technologies allow more efficient and flexible sharing
of wireless infrastructure and radio resources for better quality of service and improved cost efficiency.
Finally, the 5G wireless system should also integrate existing wireless solutions and techniques proposed for 3G/4G and 
other wireless access technologies such as dynamic power control (PC), adaptive subcarrier allocation (SA), multiple
input multiple output (MIMO), and coordinated multipoint transmission or reception (CoMP) techniques  \cite{Dottling09, 
Lopez-Perez11, 3GPP}, \cite{Valenzuela06, Gesbert10, Marsch11,Americas12,cheng13,Yang-WC13}.
\nomenclature{MIMO}{multiple input multiple output}

Design of innovative network architectures plays a very important role in meeting the stringent requirements of the next-generation wireless system \cite{cisco16}.
Two such potential wireless cellular architectures have been proposed for 5G networks, namely wireless HetNets with dense deployment of small cells
\cite{Chandrasekhar08,Claussen08,Kim09,Zhang_bk_10} and the Cloud Radio Access Networks (C-RAN) \cite{chinamobile2011,NGMN2013,maketresearch2013}.
Massive deployment of small cells in coexistence with existing macrocells in the HetNets can fundamentally enhance
 network capacity, energy efficiency, and coverage performance \cite{Andrews12,Le12}, \cite{Damnjanovic11}.
Moreover, thanks to the powerful cloud processing capability, C-RAN enables us to realize computationally extensive communication schemes 
such as CoMP in a centralized manner \cite{tanno2010}.
Moreover, C-RAN and HetNets can be jointly deployed by exploiting residential internet broadband links for backhauling \cite{Liu_infocom_13}, to provide
 short-range communication for mmWave and massive MIMO \cite{Boccardi14,Bhushan14,Le_EU15,Pi11,Rappaport13}.
More discussions of the benefits and technical challenges of these network architectures are described in more details in the following sections.
\nomenclature{RF}{Radio Frequency}  
 
\section{Emerging Wireless Network Architectures}

\subsection{Heterogeneous Networks with Small Cells}

The wireless HetNet is typically based on the dense deployment of small cells such as femtocells and picocells in coexistence
with existing macrocells \cite{Chandrasekhar08,Claussen08,Kim09,Zhang_bk_10,zhang12} where small cells
have short communications range, low power, and low cost as illustrated in Fig.~\ref{Ch1_fig:HetNets}. Moreover,
certain small cells such as femtocells can be randomly deployed by end user equipments (UEs) and they operate on the same
frequency band with the existing macrocell to enhance the spectrum utilization, network capacity, and energy efficiency \cite{Andrews12, Calin10,Gambini13}.
Furthermore, indoor traffic supported by femtocells can be backhauled through the IP connections
such as DSL to reduce the traffic load of the macrocells through adaptive load balancing \cite{Gambini13,Aijaz13}. 
The macrocell tier can thus dedicate more radio resources to better service outdoor UEs.
Femtocells, which can be deployed by end users in a plug-and-play manner, usually require low capital expenditure and operating cost.

\begin{figure}[!t]
\begin{center}
\includegraphics[width=0.5\textwidth]{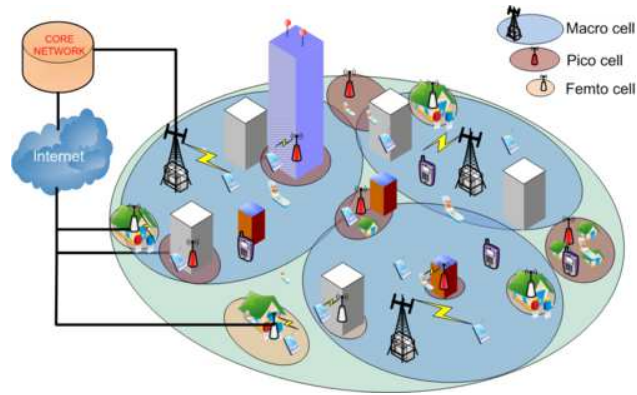}
\end{center}
\caption{Heterogeneous wireless networks with multiple tiers of cells}
\label{Ch1_fig:HetNets}
\end{figure}

However, dense deployment of femtocells on the same frequency band with the macrocells poses various technical challenges.
First, the strong cross-tier interferences between macrocells and femtocells may occur, which may severely impact the network performance
if not managed properly \cite{Andrews12,Lopez-Perez11,Yavuz09}. Second, macro UEs (MUEs) typically have higher priority
in accessing the radio spectrum compared to the femto UEs (FUEs); therefore, the QoS requirements of MUEs must be protected and maintained. 
Therefore, the cross-tier interference induced by FUEs to the macrocell tier considering fairness and access priority must be appropriately controlled. 
Lastly, dynamic and intelligent access control strategies must be developed to efficiently manage the 
network interference and traffic load balancing \cite{Claussen08, Lopez-Perez11, Lu13, roche10}.

\subsection{Cloud Radio Access Networks (C-RAN)}

\begin{figure}[!t]
\begin{center}
\includegraphics[width=0.6\textwidth]{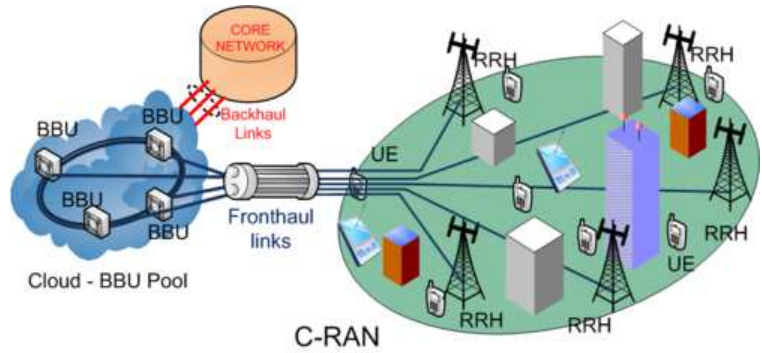}
\end{center}
\caption{C-RAN architecture}
\label{Ch1_fig:CRANs}
\end{figure}

The C-RAN architecture aims at exploiting the cloud computing infrastructure to realize various network functions and protocols
\cite{chinamobile2011,NGMN2013,maketresearch2013}. The general C-RAN architecture consists of three main 
components, namely  (i) centralized processors or baseband unit (BBU) pool, (ii) the optical 
transport network (i.e., fronthaul (FH) links), and  (iii) remote radio head (RRH) access units with antennas located at remote sites
as illustrated in Fig.~\ref{Ch1_fig:CRANs}  \cite{chinamobile2011,Demestichas13,Checko15}. 
The cloud processing center comprising a large number of BBUs is the heart of this architecture where
BBUs operate as virtual base stations (BSs) to perform different tasks such as processing of
 baseband signals for UEs and optimizing the radio resource allocation. 
RRHs can be relatively simple, which can be spatially distributed over the network for more energy and 
cost-efficient network deployment. In general, different levels of centralization can be realized in C-RAN
where the full centralization design requires heavy signal processing in the cloud while the partial centralization
design requires certain processing functions to be performed at the RRHs \cite{chinamobile2011,Rost_com_mag_14,Wubben14}.

Additionally, the centralized processing enables to implement sophisticated physical-layer and radio resource management designs
such as the coordinated multi-point (CoMP) transmission and reception techniques proposed in the LTE wireless standard,
efficient clustering design for RRHs to balance between the network capacity enhancement and design complexity 
\cite{mob_rep_13,chinamobile2011, Costa-Perez_CM13, Liang_CST15}.
Other benefits of C-RAN include reduced backhaul and network core traffic, deployment and operation costs,
enhanced network capacity, service quality and differentiation to effectively support different wireless applications \cite{mob_rep_13,chinamobile2011}. 
Centralized C-RAN is considered as an eco-friendly infrastructure where the network resources (radio spectrum, computation power, fronthaul links) 
can be sliced for multiple operators to realize the wireless network virtualization (WNV), which results in enhanced
QoS for mobile users and better supports the deployment of new wireless technologies and services \cite{Costa-Perez_CM13,Liang_CST15}. 

Successful deployment of C-RAN, however, requires us to resolve several major technical challenges \cite{Costa-Perez_CM13,Liang_CST15}. 
In particular, efficient techniques and solutions for advanced signal processing and radio resource management must be developed
to achieve the potential network performance enhancement while accounting for constraints and efficient utilization of fronthaul
capacity and cloud computing resources. This is especially important for wireless systems with multiple antennas
where heavy MIMO processing complexity is usually required to leverage the multiplexing and diversity gains.
Moreover, C-RAN design must account for practical network constraints such as those
on limited cloud computation power and finite capacity of the fronthaul transport network. 

\section{Research Challenges and Motivations}

To realize these benefits of HetNets and C-RAN in next-generation wireless systems, one has to address many technical challenges 
spanning different layers of the protocol stack \cite{Andrews12,Lopez-Perez11,Checko15,Liang_CST15}. This doctoral dissertation
limits its research scope on resource allocation for these two potential network architectures. In particular, we consider
both single-carrier and multi-carrier resource allocation design issues for the wireless HetNet while we study the MIMO-based
precoding techniques and wireless virtualization design for the C-RAN setting.


In single-carrier HetNets, small-cells usually operate on the same licensed frequency band with the existing macrocells, which can create 
severe cross-tier interference for MUEs \cite{Le12}. Thus, adaptive power control (PC) is required to mitigate the co-channel
interference \cite{foschini93, yates95, Yates95a, hanly95, farrokhi97, teng10, sung05, sung06}.
There have been several existing works on PC and interference management for macrocell-femtocell HetNets
\cite{jo09,chand09,yun11,Le12a,yanzan12,madan10}. However, most of 
these works did not consider the scenarios where users of different network tiers have
differentiated QoS requirements (e.g., voice versus data users).

Other important issues are joint dynamic BS association (BSA) and interference management for
efficient load balancing and access mode control (i.e., the closed, open, and hybrid access).
Most existing works on these topics, however, only consider the closed access mode \cite{jo09,chand09,yun11,Le12a,yanzan12};
therefore, fixed BSA solutions of different users to their BSs are assumed. 
The open and hybrid access, however, allows flexible BSA to mitigate cross-tier interference than the closed access
 \cite{roche10}. Hence, an efficient hybrid access scheme with dynamic BSA, is an important research issue,
which is addressed in this doctoral dissertation.


Multi-carrier wireless HetNets based on orthogonal frequency division multiple access (OFDMA) provide
more flexibility in radio resource management and robustness against the multi-path fading \cite{guoqing06,Pischella10}.
Radio resource management of OFDMA-based HetNets for appropriate protection 
of existing macrocells and interference mitigation through  subchannel allocation (SA) and PC is a challenging 
research issue \cite{perez09, Le12,Zahir13}, which have received great research interests \cite{chuhan11,yanzan12,
yushan12,Wangchi12,Long12,hoon11,renchao12,lu12}.

These works in \cite{chuhan11,yanzan12,yushan12,Wangchi12,Long12}, however, only optimize the SA without considering the   
PC while the authors in \cite{hoon11} only study PC for the given SA of all users.
The joint SA and PC design have been studied in \cite{renchao12,lu12}; however, 
these works do not provide QoS guarantees for users of both network tiers.
Moreover, the fairness among FUEs has not been addressed well in these existing works. Therefore, an 
efficient joint SA and PC design which provides QoS guarantees for users of different tiers and
fairness among FUEs is a critical research topic in OFDMA-based HetNets which is tackled in our dissertation. 


In the C-RAN, cloud centralized processing facility can be leveraged
to realize the CoMP techniques for efficient interference management and network performance enhancement.
However, C-RANs typically has limited cloud processing resource and fronthaul capacity. 
For joint transmission design in the CoMP C-RAN, optimization of the set of remote radio heads (RRHs) serving
each user as well as the precoding and transmission powers considering these practical constraints
is a difficult problem. 

There have been some research efforts in the CoMP design for efficient exploitation of the fronthaul capacity
in C-RAN \cite{letaief14,luo_arxiv,Quek2013,cheng13,BBDai_14}.
In \cite{letaief14,luo_arxiv,Quek2013,cheng13}, it is shown how such problems can be formulated as
a sparse beamforming  problem, which can be solved by applying the compressive sensing (CS) techniques \cite{donoho06,tao06,bach12}.
However, these existing designs have not considered the limited fronthaul capacity constraint.
Our dissertation makes important contributions in this direction where we address the CoMP design for C-RAN
considering such fronthaul capacity constraint.


In general, appropriate combination of C-RAN and WNV techniques can result in significant benefits 
including reduction of the CAPEX and OPEX, enhancement of the QoS and resource utilization, and 
fast deployment of the new products and services\cite{Costa-Perez_CM13,Liang_CST15,Le_EU15}.
Resource allocation and slicing design \cite{Checko15, Liang_CST15}, which decides how to share
the centralized computation resource \cite{Wubben14, Suryaprakash15}, fronthaul capacity \cite{Peng14}, radio resources in C-RAN \cite{Checko15}
among different operators while efficient support all mobile users is a challenging research topic \cite{Liang_CST15}.
However, there have been very few research works considering the resource allocation to realize the  WNV in the C-RAN \cite{Wang_JSAC,Kim2016,Yang_OE2016}. 
Moreover, such design considering both limited fronthaul capacity and cloud computation resource has not been studied. 
Our dissertation has made some significant contributions along this line.

\section{Research Contributions and Organization of the Dissertation}

The overall objective of this Ph.D. research is to
develop radio resource allocation and interference management algorithms that directly resolve
important technical issues of the future high-speed wireless cellular networks. Particularly, 
we have made the following main contributions.

\begin{enumerate}
\item \textit{Joint Base Station Association and Power Control Design for HetNets \cite{VuHa_TVT14_BSA_PC,VuHa_VTC12}:} 
We propose a universal joint Base Station Association and Power Control (PC) algorithm 
for HetNets. Specifically, the proposed algorithm iteratively updates
the BSA solution then the transmit power of each user. We prove the convergence of this algorithm when the power update function of the PC strategy satisfies the so-called ``two-sided scalable function''
 property. Then, we develop a novel hybrid PC scheme based on which an adaptation
mechanism for the HPC algorithm is further proposed to support the signal-to-interference-plus-noise ratio (SINR)
 requirements of all users whenever possible while exploiting the multiuser diversity
to improve the system throughput. We show that the proposed HPC adaptation algorithm outperforms the
well-known PC algorithm in both feasible and infeasible systems.

\item \textit{Fair Subchannel Allocation and Power Control Design for OFDMA HetNets \cite{VuHa_TVT14_PC_SA,VuHa_WCNC13}:}
We are interested in the fair resource sharing solution for users in each femtocell that maximizes the
total minimum spectral efficiency of all femtocells subject to protection constraints for the prioritized macro users.
First, we formulate the optimization problem for the fairness resource allocation design, and propose an optimal
exhaustive search algorithm. Then, we develop a distributed and low-complexity algorithm to find an efficient solution for the problem.  We prove that the proposed algorithm converges and analyze its complexity. 
Then, we extend the proposed algorithm in three different directions,
namely downlink context; resource allocation with rate adaptation for femto users; and consideration of a hybrid access strategy 
where some macro users are allowed to connect with nearby femto base stations to improve the performance of the femto tier.

\item \textit{Coordinated Transmission Design for C-RANs with Limited Fronthaul Capacity \cite{VuHa_TVT16,VuHa_WCNC14,vuha_ciss_2014,vuha_globecom_2014,VuHa_WCNC15}:} 
Our design aims to optimize the set of remote radio heads (RRHs) serving each user as well as the precoding and transmission powers to
minimize the total transmission power while maintaining the fronthaul capacity and users' QoS constraints. The 
 fronthaul capacity constraint involves a non-convex and discontinuous function which renders the optimal exhaustive
search method unaffordable for large networks. To address this challenge, we propose two low-complexity algorithms.
The first pricing-based algorithm solves the underlying problem through iteratively tackling a related pricing
problem while appropriately updating the pricing parameter. In the second iterative linear-relaxed algorithm, we directly address the fronthaul constraint function by iteratively approximating it with a suitable linear form using a conjugate function and solving the corresponding convex problem.

\item \textit{Resource Allocation for Wireless Virtualization of OFDMA-Based Cloud Radio Access Networks \cite{VuHa_sTWC16,VuHa_ICC16}:}
We consider the resource allocation for the virtualized OFDMA uplink C-RAN where multiple OPs share the C-RAN
infrastructure and resources owned by an infrastructure provider. 
In particular, our design aims at
maximizing the weighted sum profit of both infrastructure provider and OPs considering constraints
on the fronthaul capacity and cloud computation limit.
This design is modelled by two different resource allocation
problems, namely upper-level and lower-level problems, respectively. 
The upper-level problem focuses on optimizing the slices of the fronthaul capacity and cloud computing resources for all OPs to maximize the weighted sum profits. 
While the lower-level maximizes the OP's sum rate by optimizing users' transmission rates and quantization bit allocation 
for the compressed I/Q baseband signals. 
We then develop a two-stage algorithmic framework to fulfil the design. 
In the first stage, by relaxing the underlying discrete variables 
to the continuous variables, we transform both upper-level and lower-level problems into
the corresponding optimization problems which are convex and can be solved optimally.
In the second stage, we propose two methods to round the solution obtained by solving the relaxed problems
to achieve a final feasible solution for the original problem.
\end{enumerate}

The remaining of this dissertation is organized as follows. 
Chapter \ref{Ch2} reviews some fundamental resource allocation techniques and the literature on resource allocation
 and interference management for HetNets and C-RANs, respectively.
In Chapter \ref{Ch3}, we present the joint BSA and PC design for CDMA-based wireless HetNets.
Then, we discuss the fair SA and PC design for OFDMA-based HetNets in the Chapter \ref{Ch4}.
The advanced CoMP transmission designs for C-RANs with limited fronthaul capacity are proposed in Chapter \ref{Ch5}.
In Chapter \ref{Ch8}, we describe the resource allocation design for wireless virtualization of the uplink C-RAN considering limited computation
 and fronthaul capacity. Chapter \ref{Ch9} summarizes the main contributions of the dissertation and makes some recommendations for future research directions.

%% file: chap4/Ha_chap4.tex

\chapter{Background and Literature Review}
\label{Ch2}

This chapter provides the necessary background and survey the related works in resource allocation for wireless cellular networks.
Specifically, some fundamental resource allocation techniques in cellular networks will be briefly discussed. 
Then, review of existing works for resource allocation in wireless HetNets and C-RANs, will be presented.

\section{Fundamental Resource Allocation Techniques}

\subsection{Adaptive Power Control}

Consider single-carrier multicell wireless cellular network as illustrated in Fig.~\ref{Ch2_fig:multicell_networks}.
This can be applied to the downlink or uplink communications scenarios.
We refer to the (downlink or uplink) transmission between UE $i$ and its BS as link $i$.
Let $p_i \geq 0$ be the transmit power corresponding to link $i$ and $\eta_i$ be the power of the additive white Gaussian noise (AWGN) at the receiver. 
Denote the channel gain from the transmitter of link $i$ to its receiver as $h_{i,i}$, 
and the channel gain from the transmitter of link $j$ to the receiver of link $i$ as $h_{i,j}$. 
Note that for downlink communications the transmitter of link $i$ is BS $i$ serving UE $i$ while
for the uplink case the transmitter is UE $i$. For convenience, 
we arrange the transmit powers of all users in a power vector denoted as $\mathrm{p} = \left( p_1,p_2,...,p_M \right)$ where $M$ 
represents the number of UEs in the network. The received signal to interference plus noise ratio (SINR) of UE $i$ can be 
expressed as
\beq \label{Ch2_eq:SINRi}
\Gamma_i\left( \mathrm{p} \right)=\dfrac{h_{i,i}p_i}{\sum_{j \neq i} h_{i,j}p_j+\eta_i}=\dfrac{p_i}{R_i\left( \mathrm{p} \right)},
\eeq
where $R_i \left( \mathrm{p} \right)$ is called the effective interference at the receiver of link $i$, which is defined as
\beq \label{Ch2_eq:Rp}
 R_i \left( \mathrm{p} \right) \triangleq  \dfrac{\sum_{j\neq i}{h_{i,j}p_j}+\eta_{b_i}}{h_{i,i}}.
\eeq
\begin{figure}[!t] 
        \centering
        \begin{subfigure}[c]{0.4\textwidth}
                \begin{center}
\includegraphics[height=50mm]{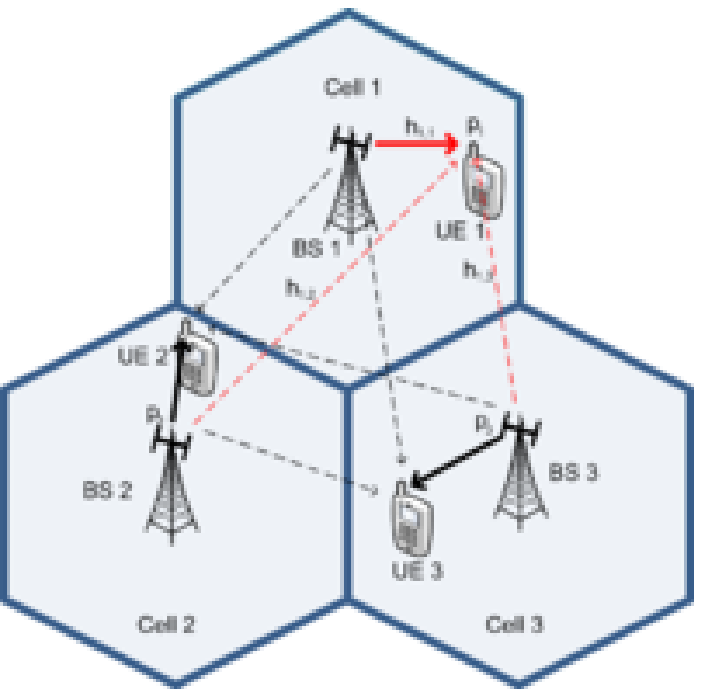}
\end{center}
\caption{Downlink} \label{Ch2_fig:MC_DL}
        \end{subfigure}%
        \qquad  
        ~ 
        \begin{subfigure}[c]{0.4\textwidth}
\begin{center}
\includegraphics[height=50mm]{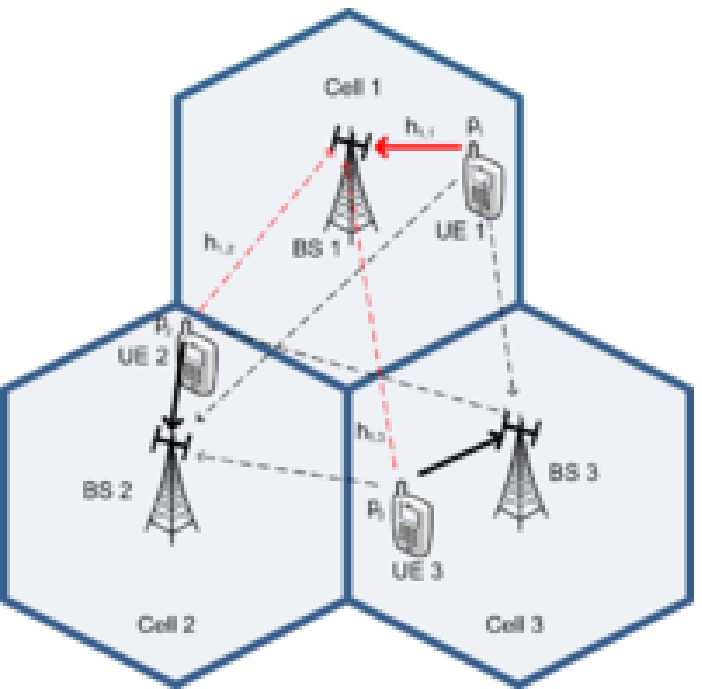}
\end{center}
\caption{Uplink} \label{Ch2_fig:MC_UL}
        \end{subfigure}
         \caption{Co-channel Interference in Wireless Cellular Networks}
         \label{Ch2_fig:multicell_networks}
\end{figure}

It can be observed from (\ref{Ch2_eq:SINRi}) that the co-channel interference $\sum_{j \neq i} h_{i,j}p_j$ may significantly decrease
the SINR of UE $i$, which negatively impacts the reliability of the underlying communication. 
Fig.~\ref{Ch2_fig:multicell_networks} indeed illustrates typical interference scenarios in multicell uplink and downlink communications. 
In the downlink, UE 1 in cell 1 receives not only the desired signal from its serving BS 1 but also the interfering signals from BSs 2 and 3. 
Similarly, the signal transmitted by UE 1 to its BS 1 in the uplink scenario is interfered by those from UEs 2 and 3 in the adjacent
cells.

To mitigate the negative impacts of co-channel interference, adaptive power control (PC) is typically required.
In practice, the transmit power is updated iteratively to cope with the time-varying nature of the wireless environment.
The transmit power update rules of many popular iterative PC algorithms such as that proposed by Foschini and Miljanic in \cite{foschini93} 
can be expressed in the general form as follows:
\beq \label{Ch2_eq:PCgnr}
p_i^{(n+1)} = J_i \left( \mathrm{p}^{(n)} \right),
\eeq
where $n$ denotes the iteration index and $J_i(\cdot)$ is called the power update function (p.u.f.). 
In general, the p.u.f. can be designed so that the iterative power update converges
to the desired equilibrium that fulfills certain design goal. Moreover, the p.u.f.
must be appropriately designed to guarantee the convergence. 
One important result along this line was derived by Sung and Leung in \cite{sung05, sung06} 
where it was proved that this kind of iterative PC algorithm converges to an equilibrium if the 
corresponding p.u.f. is a \textit{two-sided scalable function} (2.s.s.) whose definition is given as follows. 
\begin{definition}
\label{Ch2_def_2s_scal}
A p.u.f. $J( \mathrm{p})$ is a 2.s.s. function with respect to $\mathrm{p}$ if for any parameter $a >1$ and two
 power vectors $\mathrm{p}$ and $\mathrm{p}^{\prime}$ that satisfy $\frac{1}{a}\mathrm{p} \leq \mathrm{p}^{\prime} \leq a \mathrm{p}$, we 
must have $(1/a) J \left( \mathrm{p}\right) < J(\mathrm{p} ^{\prime} ) < a J(\mathrm{p})$.
\end{definition}
This result can be employed to design a suitable p.u.f. for a desired PC algorithm. We will describe
several such popular PC algorithms in the following.

\subsubsection{Important Power Control Algorithms} \label{Ch2_sec_PC}

We now describe two popular PC algorithms, namely opportunistic PC (OPC) algorithm \cite{sung05, sung06} and target-SINR-tracking PC (TPC) algorithm
 \cite{foschini93}. Specifically,
the OPC algorithm proposed by Sung and Leung in \cite{sung05, sung06} aims to exploit the multi-user 
diversity to improve the system throughput where the transmit power of link $i$ is updated iteratively as 
\beq \label{Ch2_eq:OPC}
p_i^{(n+1)} = J^{\sf{OPC}}_i \left( \mathrm{p}^{(n)} \right) = \dfrac{\xi_i}{R_i \left( \mathrm{p}^{(n)} \right)},
\eeq
where $\xi_i$ is a predetermined coefficient corresponding to link $i$.
This algorithm indeed allocates more power to stronger wireless link $i$, which observes the smaller effective 
interference $R_i \left( \mathrm{p}^{(n)} \right)$. Therefore, this PC algorithm converges to an equilibrium
which efficiently exploits the different channel conditions of different links to enhance the network throughput. 

The distributed TPC algorithm, which was proposed by Foschini and Miljianic \cite{foschini93}, enables users to 
reach their desired target SINRs with minimum powers. The TPC algorithm iteratively updates the
 transmit power of link $i$ as 
\beq \label{Ch2_eq:TPC}
p_i^{(n+1)} = J^{\sf{TPC}}_i \left( \mathrm{p}^{(n)} \right) = \dfrac{\bar{\gamma}_i}{\Gamma_i\left( \mathrm{p}^{(n)} \right)} p_i^{(n)}=\bar{\gamma}_i R_i \left( \mathrm{p}^{(n)} \right),
\eeq
where $\bar{\gamma}_i$ denotes the target SINR of UE $i$ while $p_i^{(n)}$, $\Gamma_i\left( \mathrm{p}^{(n)} \right)$, and $R_i \left( \mathrm{p}^{(n)} \right)$ 
represent the transmit power, measured SINR, and effective interference at the receiver of link $i$ in the iteration $n$, respectively. 

It can be observed that the TPC algorithm attempts to attain the target SINR by updating its power based on the ratio between
the target SINR and current SINR in each iteration.
It was shown in \cite{sung05, sung06, foschini93} that both p.u.fs. $J^{\sf{OPC}}_i$ and $J^{\sf{TPC}}_i$
of the OPC and TPC PC algorithms satisfy the properties of a 2.s.s. function given in Definition~\ref{Ch2_def_2s_scal}.
Therefore, both PC algorithms are guaranteed to converge.

Recall, however, that the design goal of the TPC algorithm is to ensure all UEs to achieve their required SINR, i.e., 
\beq \label{Ch2_eq:SINRcnt}
\Gamma_i(\mathrm{p}) \geq \bar{\gamma}_i, \: i = 1,...,M,
\eeq
which is clearly not always feasible for any wireless systems. 
In particular, the system is called feasible if it is possible for all UEs to achieve their target SINRs.
To further study the feasibility, we first rewrite the constraints (\ref{Ch2_eq:SINRcnt}) for all UEs in matrix form as follows:
\begin{equation}
\label{Ch2_eq:cond_mtx1}
(\mathrm{\bf{I}}-\mathrm{\bf{G}}\mathrm{\bf{H}})\mathrm{p}^n \geq \mathrm{g},
\end{equation}
where $\mathrm{g}=\left[ g_{1},...,g_{M} \right]^T$ with $g_i=\frac{\eta_{b_i}^n \bar{\gamma}_i^n}{h_{i,i}}$, $\mathrm{\bf{I}}^n$ is $M \times M$ identity matrix, $\mathrm{\bf{G}}=\text{diag}\{\bar{\gamma}_{1},...,\bar{\gamma}_{M}\}$, and $\mathrm{\bf{H}}^n$ is a $M \times M$ matrix defined as
\begin{equation}
\label{Ch2_eq:H}
\left[ \mathrm{\bf{H}}_{i,j}\right]=\left\lbrace \begin{array}{*{5}{l}}
0, & \text{if } j=i,\\
\frac{h_{i,j}}{h_{i,j}}, & \text{if } j \neq i.
\end{array} \right.
\end{equation}
Then, the feasibility of constraints (\ref{Ch2_eq:SINRcnt}) can be studied through investigating the inequality (\ref{Ch2_eq:cond_mtx1}) 
by applying the Perron-Frobenius theorem \cite{Horn-85,seneta73,gant71}.
Toward this end, let $\lambda_i$ denote the $i^{\sf th}$ eigenvalue of matrix $\mathrm{\bf{G}}\mathrm{\bf{H}}$ and 
$\rho(\mathrm{\bf{G}}\mathrm{\bf{H}})={\sf max}_{ i} | \lambda_i |$ represent the maximum value of the
modulus of all eigenvalues (i.e., the spectral radius of $\mathrm{\bf{G}}\mathrm{\bf{H}}$). 
The Perron-Frobenius theorem \cite{Horn-85,seneta73,gant71}, which can be applied for non-negative matrix $\mathrm{\bf{G}}\mathrm{\bf{H}}$,
is relevant for the model under investigation where this theorem implies the following fact \cite{bambos00}. 

\noindent
Fact: If $\mathrm{\bf{G}}\mathrm{\bf{H}}$ is a matrix with non-negative elements, the following statements are equivalent:
\begin{enumerate}
\item There exists a non-negative power vector $\mathrm{p}$ such that $(\mathrm{\bf{I}}-\mathrm{\bf{G}}\mathrm{\bf{H}})\mathrm{p} \geq \mathrm{g}$.
\item $\rho(\mathrm{\bf{G}}\mathrm{\bf{H}}) \leq 1$.
\item $\left( \mathrm{\bf{I}}-\mathrm{\bf{G}}\mathrm{\bf{H}}\right)^{-1}=\sum_{k=0}^{\infty} (\mathrm{\bf{G}}\mathrm{\bf{H}})^k$ exists and is positive element-wise.
\end{enumerate}

According to these results, if $\rho(\mathrm{\bf{G}}\mathrm{\bf{H}}) \leq 1$ then we can determine
 a solution of $(\mathrm{\bf{I}}-\mathrm{\bf{G}}\mathrm{\bf{H}})\mathrm{p} = \mathrm{g}$ (i.e., the underlying inequality is met with equality) as
\beq \label{Ch2_eq:pw_chk}
\mathrm{p}^{\star}=\left( \mathrm{\bf{I}}- \mathrm{\bf{G}}\mathrm{\bf{H}} \right)^{-1}\mathrm{g},
\eeq
which is also the Pareto-optimal solution of $(\mathrm{\bf{I}}-\mathrm{\bf{G}}\mathrm{\bf{H}})\mathrm{p} \geq \mathrm{g}$ where Pareto-optimality means that
any feasible solution $\mathrm{p}$ for the inequality  $(\mathrm{\bf{I}}-\mathrm{\bf{G}}\mathrm{\bf{H}})\mathrm{p} \geq \mathrm{g}$ is not smaller
than $\mathrm{p}^{\star}$ element-wise (i.e., there not exist a feasible solution $\mathrm{p}$ for $(\mathrm{\bf{I}}-\mathrm{\bf{G}}\mathrm{\bf{H}})\mathrm{p} \geq \mathrm{g}$ so that t
$p_i \leq p_i^{\star}, \forall i$ and $p_j < p_j^{\star}$ for some $j$).
Interestingly, this Pareto-optimal power vector can be achieved at the equilibrium by the distributed TPC algorithm \cite{bambos00}.

\subsection{Resource Allocation in Multi-Carrier Wireless Networks}

We now present some fundamentals of radio resource allocation for the multi-carrier-based wireless network, which is based on the 
orthogonal frequency division multiple access (OFDMA). Resource allocation for OFDMA-based wireless networks involves the
allocation of spectral resources (subchannels or finer spectrum resource units such as LTE physical resource blocks (PRBs))
and transmit powers, which is a difficult mixed integer program. Moreover, resource allocation can be performed
in the single-cell or multi-cell settings. Also, appropriate inter-cell interference management must be addressed appropriately
 if aggressive spectrum reuse is employed. In the following, we discuss some typical resource allocation problems
for the OFDMA-based wireless system in the multi-cell setting.

Consider a multicell and multiuser wireless cellular network with $K$ cells and $M$ UEs served by $K$ BSs.
Moreover, the BS $k$ in cell $k$ serves $M_k$ UEs. Let $\mathcal{U}_k$ denote the set of UEs in the $k$-th cell
and let $\mathcal{U} \triangleq \cup_{k=1}^{K} \mathcal{U}_k= \left\lbrace 1,...,M \right\rbrace$ denote the set of all UEs,
and $\mathcal{B}$ represent the set of all BSs. We denote $\mathcal{N}=\left\lbrace {1,2,...,N} \right\rbrace$ as
 the set of available orthogonal subchannels, which will be allocated for UEs in all cells. 
We assume that there is no interference among transmissions on different subchannels.

To represent the subchannel assignment (SA) decisions, we introduce SA variables $a^n_i$
and arrange them into a SA matrix for all $M$ UEs over $N$ subchannels denoted as  $\mathrm{\bf{A}} \in \mathfrak{R}^{M \times N}$,
which are defined as follows:
\begin{equation}
\label{Ch2_eq:A}
\mathrm{\bf{A}}(i,n)=a^n_i=\left\lbrace \begin{array}{*{10}{l}}
1 & \text{if  subchannel $n$ is assigned for UE $i$}\\
0 & \text{otherwise}.
\end{array}
\right. 
\end{equation}

There are specific constraints on the SA matrix depending on the SA and frequency reuse strategy. 
In the following, we study a wireless system with full frequency reuse where all $N$ subchannels are 
reused over the cells. We further assume that each subchannel can be allocated to at most one UE in any cell 
 to avoid the strong intra-cell co-channel interference. These SA allocation constraints can be expressed as
\begin{equation}
\label{Ch2_eq:c4}
\sum_{i \in \mathcal{U}_k} a_i^n \leq 1, \:\:\: \forall k \in \mathcal{B} \text{ and } \forall n \in \mathcal{N}.
\end{equation}

For power allocation, let $p^n_i$ represent the transmission power for UE $i$ on subchannel $n$ where $p^n_i \geq 0$.
The following constraints on total transmission powers must be imposed
\beqn \label{Ch2_powcon}
\sum_{n=1}^{N} p^n_i \leq P_i^{\texttt{max}}, \quad i \in \mathcal{U}, \text{ for UL,}\\
\sum_{i \in \mathcal{U}_k} \sum_{n=1}^{N} p^n_i \leq P_{\sf{BS},k}^{\texttt{max}}, \quad i \in \mathcal{U}, \text{ for DL,}
\eeqn
where $P_i^{\texttt{max}}$ denotes the maximum transmission power of UE $i$ and $P_{\sf{BS},k}^{\texttt{max}}$ represents the maximum 
transmission power of BS $k$. Similar to the SAs, we define the power allocation (PA) matrix $\mathrm{\bf{P}}$ as an $M \times N$  
whose elements are denoted as $\mathrm{\bf{P}}(i,n)=p^n_i$. 

For convenience, we also introduce $\mathrm{\bf{A}}_k, \mathrm{\bf{P}}_k \in \mathfrak{R}^{\vert \mathcal{U}_k \vert \times N}$
to represent the SA and PA matrices for UEs in cell $k$ over $N$ subchannels, respectively. 
Let $h_{i,j}^n$ denote the channel gain from the transmitter of link $j$ to the receiver of link $i$ and $\eta_i^n$ be the noise power at the 
receiver of link $i$ on subchannel $n$. For a given SA and PA solution (i.e., given $\mathrm{\bf{A}}$ and $\mathrm{\bf{P}}$), the SINR achieved
by the transmission of link $i$ on subchannel $n$ can be written as
\begin{equation}
\label{Ch4_sinr}
\Gamma_i^n(\mathrm{\bf{A}},\mathrm{\bf{P}})=\dfrac{a_i^n h_{i,i}^n p_i^n}{\sum_{j \neq i}{a_j^n h_{i,j}^n p_j^n}+\eta_{i}^n}.
\end{equation}

In general, resource allocation requires to jointly optimize some SA and PC optimization problems, which
are specified by specific design objectives. In the following, we present some popular resource allocation problems. 
\begin{itemize}
\item Weighted sum rate maximization problem:

Resource allocation decisions in this case can be obtained by solving the following optimization problem
\begin{eqnarray}
\label{Ch2_eq:max_SR}
\mathop {\max} \limits_{(\mathrm{\bf{A}},\mathrm{\bf{P}}) } \sum \limits_{1\leq i \leq M} \sum \limits_{1\leq n \leq M} \mathrm{R}_i^{(n)}(\mathrm{\bf{A}},\mathrm{\bf{P}}) \nonumber \\
\text{s. t. constraints} \quad (\ref{Ch2_eq:c4}), (\ref{Ch2_powcon}),  
\end{eqnarray}
where $\mathrm{R}_i^{(n)}(\mathrm{\bf{A}},\mathrm{\bf{P}})=\log(1+\Gamma_i^n(\mathrm{\bf{A}},\mathrm{\bf{P}}))$ represents
the rate  achieved by the link $i$ over subchannel $n$.

\item Maximization of the sum of minimum of weighted min rates

In this case, we have solve the following optimization problem
\begin{eqnarray}
\label{Ch2_eq:max_min}
\mathop {\max} \limits_{(\mathrm{\bf{A}},\mathrm{\bf{P}}) } \sum \limits_{1 \leq k \leq K} \omega_k \min_{i \in \mathcal{U}_k} \mathrm{R}_i(\mathrm{\bf{A}},\mathrm{\bf{P}}) \nonumber \\
\text{s. t. constraints} \quad (\ref{Ch2_eq:c4}), (\ref{Ch2_powcon}),
\end{eqnarray}
where $\mathrm{R}_i(\mathrm{\bf{A}},\mathrm{\bf{P}})=\sum \limits_{1\leq n \leq M} a_i^n \mathrm{R}_i^{(n)}(\mathrm{\bf{A}},\mathrm{\bf{P}})$
describes the sum rate achieved by UE $i$ and $\omega_k$ denotes the weight of cell $k$.
\end{itemize}

These resource allocation problems have been studied in several works \cite{ZhuHan07,Kwon_ICC_06,guoqing06,HZhang11,K_Yang_09,Koutsopoulos_06,K_Yang11,Ksairi10,
Pischella10,Venturino09} where mostly heuristic and suboptimal algorithms have been proposed to tackle them.

\subsection{MIMO Precoding Designs in Multicell Wireless Systems} \label{C2-CoMP-sect}

This section presents some recent advances in MIMO signal processing and precoding design for the MIMO multicell system.
In particular, we will briefly discuss the coordinated multipoint (CoMP) transmission and reception techniques
and the precoding design for power minimization.

\subsubsection{Coordinated Multipoint (CoMP) Transmission and Reception Techniques}

There are different CoMP schemes, which can be different in the levels of data/information sharing among cells, 
including uplink interference-aware detection, joint multicell scheduling, multicell link adaptation, joint multicell signal processing
 and downlink coordinated scheduling/beamforming, joint processing.
In this dissertation, we only consider the downlink joint  transmission CoMP design for a multicell system which consists of $K$ BSs serving $M$ UEs where 
data signals related of each UE can be simultaneously transmitted from multiple BSs as illustrated in Fig.~\ref{Ch2_fig:CoMP}. 

\begin{figure}[!t]
\begin{center}
\includegraphics[height=50mm]{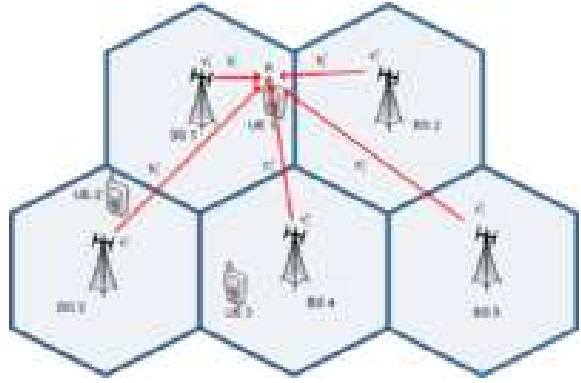}
\end{center}
\caption{CoMP Based MIMO Signal Processing in Multicell Wireless Systems }
\label{Ch2_fig:CoMP}
\end{figure}

Let $\mathcal{K}$ and $\mathcal{U}$ denote the sets of BSs and UEs in the system, respectively.
Suppose that BS $k$ is equipped with $N_k$ antennas ($k \in \mathcal{K}$) and each UE is equipped with a single antenna. 
In the CoMP joint transmission strategy, each UE is served by a specific set of BSs, which simultaneously transmit data signals to the underlying UE
in the downlink. Specifically, suppose that the data intended for UE $u$ can be made available at the set of BSs serving UE $u$
and we denote the set of serving BSs for UE $u$ as $\mathcal{R}_u$. Then, we have $\mathcal{R}_u \subset \mathcal{K}$.
Note that each BS can serve several UEs and we denote the set of UEs served by BS $k$ as $\mathcal{U}_k$.

Denote the precoding vector applied to the signal transmitted to UE $u$ from BS $k$ as $\mathbf{v}_u^{k} \in \mathbb{C}^{N_k \times 1}$.
The transmission power used by RRH $k$ for transmission to user $u$ can be expressed as
\beq \label{Ch5_eq:pv}
p^k_u=\mathbf{v}_u^{kH}\mathbf{v}_u^{k}.
\eeq 
Suppose that $x_u \in \mathbb{C}$ denotes the signal symbols of unit power, which are transmitted by BSs in the set $\mathcal{R}_u$ to UE $u$.
Then, the signal $y_u$ received at UE $u$ can be written as
\beq
\label{Ch5_eq:rx_sig}
y_u= \underbrace{\sum \limits_{k \in \mathcal{R}_u} \mathbf{h}_u^{kH} \mathbf{v}_u^{k} x_u}_\text{desired signal} 
 + \underbrace{ \sum \limits_{i =1, \neq u}^{M} \sum \limits_{l \in \mathcal{R}_i} \mathbf{h}_u^{lH} \mathbf{v}_i^{l} x_{i}}_\text{interference} + \eta_u,
\eeq
where $\mathbf{h}_u^k \in \mathbb{C}^{N_k \times 1}$ denotes the channel vector between BS $k$ and UE $u$, and $\eta_u$ represents
 the noise at UE $u$. 
Then, the SINR achieved by UE $u$ can be expressed as
\beq \label{Ch2_eq:SINR}
\Gamma_u =\dfrac{ \left|  \sum \limits_{k \in \mathcal{R}_u} \mathbf{h}_u^{kH} \mathbf{v}_u^{k}\right| ^2 }
{\sum \limits_{i =1, \neq u}^{M} \left| \sum \limits_{l \in \mathcal{R}_i}\mathbf{h}_u^{lH} \mathbf{v}_i^{l}\right|^2 + \sigma^2 },
\eeq
where $\sigma^2$ represents the noise power.

\subsubsection{Precoding Design for Power Minimization} \label{Ch2_sec_PMP}

We now describe an interesting precoding design for the CoMP-based MIMO multicell system, which 
aims to minimize the power consumption. In particular, we assume that the precoding design 
must satisfy UEs' QoS in terms of minimum SINR. This QoS constraint for each UE $u$ can be expressed as
\beq \label{Ch2_eq:SINR_ct}
\Gamma_u \geq \bar{\gamma}_u, \qquad \forall u \in \mathcal{U},
\eeq
where $\bar{\gamma}_u$ denotes the target SINR of UE $u$.
In addition to the SINR constraints in (\ref{Ch2_eq:SINR_ct}), we also impose the  total transmission power constraints for individual BSs as
\beq \label{Ch2_eq:pwc}
\sum \limits_{u \in \mathcal{U}} p^k_u = \sum \limits_{u \in \mathcal{U}_k} \mathbf{v}_u^{kH}\mathbf{v}_u^{k} \leq P_k, \qquad \forall k \in \mathcal{K},
\eeq 
where $P_k$ ($k \in \mathcal{K}$) denotes the maximum power of BS $k$. The power minimization problem (PMP), which optimizes
the precoding vectors subject to QoS and power constraints, can be stated as follows:
\begin{align}
\min \limits_{\lbrace \mathbf{v}_u^k \rbrace} & \;\;\; \sum_{k \in \mathcal{K}} \sum_{u \in \mathcal{U}} p^k_u \nonumber \\
\text{s. t. } & \; \text{ constraints (\ref{Ch2_eq:SINR_ct}) and (\ref{Ch2_eq:pwc})}.  \label{Ch2_PMP}
\end{align}

We now describe how to transform this problem into an appropriate form based on which we can determine its optimal solution.
Let $\mathbf{v}_u$ denote the precoding solution over all RRHs in $\mathcal{R}_u$,
which is defined as $\mathbf{v}_u=[\mathbf{v}_u^{u_1T}, \mathbf{v}_u^{u_2T} ... \mathbf{v}_u^{u_{a_u}T}]^T$ where $\{u_1,...,u_{a_u}\} = \mathcal{R}_u$ and $a_u=|\mathcal{R}_u|$.
So, we have $\mathbf{v}_u \in \mathbb{C}^{N_u \times 1}$ where $N_u=\sum_{k \in \mathcal{R}_u} N_k$.
Let us define $\mathbf{W}_u =\mathbf{v}_u\mathbf{v}_u^H$ then we have $\mathbf{W}_u \in \mathbb{C}^{N_u \times N_u}$. 
Note that $\mathbf{W}_u$ is positive semi-definite ($\mathbf{W}_u \succeq \mathbf{0}$) and has rank one because it is generated from vector $\mathbf{v}_u$.

We also define the channel vector $\mathbf{h}_{u,i}=[\mathbf{h}_u^{i_1T}, \mathbf{h}_u^{i_2T}, ... \mathbf{h}_u^{i_{a_i}T}]^T$ and $\mathbf{H}_{u,i}=\mathbf{h}_{u,i}\mathbf{h}_{u,i}^{H}$ for all ${u,i} \in \mathcal{U}$. Then, the SINR constraint for UE $u$ in (\ref{Ch2_eq:SINR_ct}), and the power constraints (\ref{Ch2_eq:pwc})
 can be rewritten as 
\beqn
\mathsf{Tr}(\mathbf{H}_{u,u}\mathbf{W}_u) - \bar{\gamma}_u \sum \limits_{i \in \mathcal{U}/u} \mathsf{Tr}(\mathbf{H}_{u,i}\mathbf{W}_i) \geq \bar{\gamma}_u \sigma^2, \;\; \forall u \in \mathcal{U}, \label{Ch2_eq:SINR_mtx} \\
\sum \limits_{u \in \mathcal{U}} \mathsf{Tr}(\mathbf{E}_u^k \mathbf{W}_u) \leq P_k, \;\; \forall k \in \mathcal{K}, \label{Ch2_eq:pw_mtx}
\eeqn
where $\mathbf{E}_u^k=\mathsf{diag}(\mathbf{0}_{N_{u_1} \times 1},..., \mathbf{1}_{N_{u_i} \times 1},...,\mathbf{0}_{N_{u_{a_u}} \times 1})$ if $u_i=k$.
Therefore, PMP can be transformed into the following semidefinite program (SDP)
\begin{align}
 \min \limits_{\left\lbrace \mathbf{W}_u\right\rbrace_{u=1}^{M} } & \sum \limits_{u=1}^M \mathsf{Tr}(\mathbf{W}_u) \label{Ch2_SDP} \\
 \text{s.t. } & \text{ constraints (\ref{Ch2_eq:SINR_mtx}), (\ref{Ch2_eq:pw_mtx}),} \nonumber\\
 {} & \mathbf{W}_u \succeq \mathbf{0}, \mathsf{rank}(\mathbf{W}_u)=1, \; \forall u. \label{Ch2_W_rk1}
\end{align}
This transformation reveals a special structure of the precoding design problem. Specifically, if we remove the rank-one constraints in (\ref{Ch2_W_rk1}) 
from this problem then the resulting optimization problem is convex. In fact, this relaxed problem is the convex SDP, which can be, therefore, solved 
 by using  standard optimization tools such as CVX solver \cite{Boyd2009}.
As stated in \textit{Theorem 3.1} of \cite{Bengtsson99} and \textit{Lemma 2} of \cite{Ottersten01}, if this transformed problem is feasible, then it has at 
least one solution with $\mathsf{rank}(\mathbf{W}_u)=1$, for all $u \in \mathcal{U}$.

Because the relaxed SDP problem is convex, $\mathbf{v}_u$ can be calculated as the eigenvector of $\mathbf{W}_u$ 
if the optimum solution is unique. Such unique optimum solution satisfies the rank-one constraint. 
As discussed in \cite{Bengtsson99,Ottersten01}, the relaxed SDP may have more than one optimum solution, which
means that the CVX solver cannot ensure to return the rank-one solution in general.
However, this situation almost surely never happens in practice, except for cases where the channels from two groups are exactly symmetric. 
Nevertheless, if the algorithm does produce one of such solutions where $\mathbf{W}_u$ does not have rank one, we can still obtain a 
rank-one optimum solution from that solution by using the method described in \textit{Lemma 5} of \cite{Ottersten01}, or by 
 solving the rank constrained problem presented in \textit{Algorithm~2} of
 \cite{Huang10}, or by utilizing the best rank-one approximation solution
 based on the largest eigenvalue and the corresponding eigenvector as discussed in Section~II of \cite{lou10}.

\subsection{C-RAN Fronthaul Utilization} \label{C2-FH-sect}

C-RAN can be considered as a two-hop relay architecture consisting of a wireless access hop between UEs and RRHs, and a wireless/fiber fronthaul
 hop between RRHs and BBU pool in the cloud. Transmissions over the wireless access hop are impacted by the physical 
wireless fading channels between RRHs and UEs and the co-channel interference. Moreover, transmissions over the wireless/fiber fronthaul
hop between RRHs and the cloud is affected by the finite-capacity fronthaul links and the employed signal compression/quantization.
Hence, resource allocation design in C-RAN must account for new practical constraints and focus on efficient utilization of capacity-limited fronthaul links.

\subsubsection{C-RAN Fronthaul Utilization for Uplink Communications }

In uplink transmission, the baseband symbols received by RRHs must be transferred via the fronthaul links
to the cloud for further processing. To accomplish this, ``quantize-and-forward'' pre-processing is required
 where each RRH first quantizes its received baseband signals and then sends the corresponding quantized codewords 
to the cloud through the capacity-limited fronthaul links.
These compressed signals will be further processed by the central BBU pool which
 performs joint decoding for all UEs. In this uplink scenario, the actual fronthaul capacity utilization
 depends on the quality of the received signal and the quantization techniques employed by the RRHs.

\subsubsection{C-RAN  Fronthaul Utilization for Downlink Communications}

In the downlink transmission, the BBU pool first processes the baseband signals and optimizes the precoding vectors, and then compresses the 
pre-coded signals before transmitting them to the UEs through the RRHs. As described in \cite{ParkICC13,shamai15}, there 
are two potential strategies to realize such design, namely ``Compression-After-Precoding'' (CAP) and ``Compression-Before-Precoding'' (CBP).
Moreover, these two strategies can decide which information will be transferred via the capacity-limited fronthaul links.
For the CAP strategy, which is considered in \cite{shamai13a,shamai13b}, the cloud precodes the data streams of users 
using the corresponding precoding vectors, then compresses the resulting signals and forwards them to the corresponding RRHs over 
the fronthaul transport network. 
Specifically, each baseband signal corresponding to each antenna is compressed and forwarded according to certain ``quantize-and-forward'' scheme. 
Hence, the quantization technique employed at cloud will decide the consumed fronthaul capacity in this case.

For the CBP strategy, the cloud directly compresses the precoding vectors and forwards the compressed precoding vectors as well as users' data streams 
to the corresponding RRHs \cite{ParkICC13}.
Then, the RRHs precode the signals, up-convert them to form RF signals for transmissions to users.
In this scenario, information streams and precoding vectors are compressed and transmitted from the cloud to RRHs separately. 
Therefore, the required fronthaul capacity to support the communication link between user $u$ to RRH $k$ can be calculated as
\beq \label{Ch2:R_fh_DL}
R_u^{k,\sf{fh}} = R_u^{k,\sf{pr}} + R_u^{\sf{dt}}(\bar{\gamma}_u),
\eeq
where $R_u^{\sf{dt}}(\bar{\gamma}_u)$ and $R_u^{k,\sf{pr}}$
correspond to the fronthaul capacity consumptions for transferring the information stream and precoding vector $\mathbf{v}_u^k$ for user $u$ as given in 
equation (3) of \cite{ParkICC13}, respectively. 
In fact, $R_u^{\sf{dt}}(\bar{\gamma}_u)$ is a function of the target SINR and it can be calculated as 
\beq
R_u^{\sf{dt}}(\bar{\gamma}_u) = \log_2(1+\bar{\gamma}_u).
\eeq
Moreover, $R_u^{k,\sf{pr}}$ in (\ref{Ch2:R_fh_DL}) can be predetermined based on the desirable quantization quality of the precoding vectors.
When the precoding vectors are quantized, and then employed at RHHs, there are the quantization errors.
In general, if the corresponding quantization noise is sufficiently small compared to the interference and noise at receivers, the quantization 
errors can be omitted. 

\subsubsection{Quantize-and-Forward Processing} \label{Ch2_QaF}

We assume that all baseband symbols $y^k_u$, which corresponds to the transmission link between RRH $k$ and UEs $u$, must be transferred 
via the fronthaul links to the cloud from RRH for the uplink communication or from the cloud to the RRHs in the downlink communications. 
To accomplish this, certain ``quantize-and-forward'' pre-processing strategy must be employed where $y^k_u$ is first quantized, then the 
corresponding quantized codewords will be transmitted over the fronthaul link.
Assume that the same $b^k_u$ bits are utilized to quantize the I/Q parts of the received symbol $y^k_u$.
Then, according to the results in \cite{Bucklew79}, the quantization noise power can be approximated as
\beq
q^k_u(b^k_u) \simeq  {2 Q(y^k_u)}/{2^{b^k_u}},
\eeq
where $ Q(y) = {\left( \int_{-\infty}^{\infty} f(y)^{1/3} dy \right)^3 }/{12}$, and $f(y^k_u)$ is the probability density function of
 the I/Q parts of $y^k_u$. 
Assuming a Gaussian distribution of the signal to be quantized, we have \cite{Baracca14}
\beq \label{eq:q_kmn}
q^k_u(b^k_u) \simeq   \dfrac{\sqrt{3} \pi}{2^{2 b^k_u +1}}  Y^k_u,
\eeq
where $Y^k_u$ is the power of baseband symbol $y^k_u$.
Then, the total number of quantization bits which are transferred via the fronthaul link between RRH $k$ and cloud can be expressed as
\beq
B_k = \sum_{u \in \mathcal{U}_k} 2b^k_u,
\eeq 
where $\mathcal{U}_k$ denotes the set of UEs which are served by RRH $k$.
Let $\mb{b}$ denote the vector that represents the numbers of quantization bits for all users in the network.
For a given $\mb{b}$, the quantized symbol of $y^{(s)}_{k}$ can be then written as 
\beq
\tilde{y}^k_u = y^k_u + e^k_u,
\eeq
where $e^k_u$ represents the quantization error for $y^k_u$, which has zero mean and variance $q^k_u(b^k_u)$. 

\subsection{C-RAN Cloud Computation Complexity} \label{Ch2_CCE}

Another important issue in C-RAN design is to provision and optimize the cloud computation resources.
In fact, the BBU pool in cloud performs many critical computation and processing tasks in the C-RAN.
In particular, the BBU pool performs various operations of
 MAC protocols, decodes or encodes the users' information, processes modulation/mapping to generate the baseband signals.
It may also optimize the precoding/decoding vectors for MIMO transmissions and the resource allocation for interference mitigation.

These processes will require ample computational resources which are dependent directly on the parametrization
and operating regime of C-RAN \cite{Suryaprakash15}. 
If the cloud is built with two limited computational resources, computation overload or outage may occur and users' signals
 may not be decoded successfully even though the received signals' quality is satisfactory.

In general, the processing requirement for the UE's data, or computation complexity (in Giga Operations Per Second - GOPS), required by the
 physical layer can be determined based on the \textit{frequency domain processing} (FDP) and \textit{coding processing} (CP) \cite{Werthmann13}.
While the CC of FDP scales with the number of RF channels corresponding to the transmission of a particular UE, e.g., number of antennas serving the UE, the 
computation complexity of CP depends on the coding rate, the type of coding, and encoding/decoding sides.
Hence, this CP requirements can be separated into those for the uplink and downlink communications.
In the uplink, the BS has to perform decoding and if the turbo code is employed, then the computation complexity of CP corresponding to UE $u$'s signal can be calculated as \cite{Rost_TWC15}
\beq \label{Ch2_eq:cpt}
C_u = A r_u \left[B - 2 \log_2 \left( \log_2 \left( 1 +\Gamma_u \right) - r_u \right) \right],
\eeq
where $r_u$ is the data rate chosen for UE $u$, $\Gamma_u$ is the SINR value corresponding to that user's transmission, $A = 1/\log_2(\zeta - 1)$, $B = \log_2\left( {(\zeta -2)}/{(\zeta T(\epsilon_{\sf{ch}}))}\right)$, $\zeta$ is a parameter 
related to the connectivity of the decoder, $T(\epsilon_{\sf{ch}})  =  -{T'}/{\log_{10}(\epsilon_{\sf{ch}})}$, $T'$ is another model parameter,
 and $\epsilon_{\sf{ch}}$ is the target channel outage probability. 
The set $\left\lbrace T', \zeta \right\rbrace $ can be selected by calibrating (\ref{Ch2_eq:cpt}) with an actual turbo-decoder 
implementation or a message-passing decoder.

In the downlink, the processing requirement for CP scales with the number of bits at a coding rate $C$ \cite{Werthmann13}.  
Generally, the computation complexity required for a specific UE is a function of number of utilized antennas, the modulation bits, coding rate,
 number of data streams,  and number of allocated resource blocks \cite{sabella14}. 
Specifically, the computation complexity required for each UE $u$ can be expressed in the following general form 
\beq
\label{eq:xu}
X_u = f(N_u,b_u),
\eeq
where $f(\cdot,\cdot)$ is a non-linear increasing function with its variables, $N_u$ denotes the total number of antennas serving UE $u$, and $b_u$ is the modulation bits per symbol in the data stream $x_u$.
Here, the number of antennas corresponding to transmission of UE $u$ is the total antennas of all RRHs serving that UE, which can be written as
\beq
N_u = \sum_{k \in \mathcal{R}_u} N_k,
\eeq
where $\mathcal{R}_u$ is the set of RRHs serving UE $u$, and $N_k$ is the number of antennas equipped at RRH $k$.
One particular model for $f(\cdot,\cdot)$ proposed in \cite{Werthmann13} is given as 
\beq \label{CE_func_ex}
f_{\sf{ex}}(N_u,b_u) = \left( 27 N_u + 9 N^2_u + 3 C b_u \right)/5,
\eeq
where $C$ is the coding rate and $b_u$ is the modulation bits for each transmission symbol.

\section{Literature Review}
\subsection{Power Control Design for Single-Carrier Wireless HetNets}

As introduced above, PC is an important technique to mitigate the interference among UEs, which has been investigated extensively in conventional wireless networks \cite{foschini93, yates95, Yates95a, hanly95, farrokhi97, teng10, sung05, sung06}. 
Recently, various PC and interference management algorithms have been proposed in some existing works for femtocell networks 
\cite{jo09,chand09,yun11,Le12a}. In particular, to protect the existing
Macrocell User Equipments (MUEs) while enabling a scalable femtocell deployment, Jo \textit{et al.} \cite{jo09} proposed two PC strategies, which aim to set the transmit powers of 
FUEs to maintain the cross-tier interference below a fixed threshold at MBSs. 
However, the devised solution is neither distributed nor Pareto-optimal.
By employing the game theory approaching methods, \cite{chand09,yun11,Le12a} have proposed various PC strategies for both macrocells and femtocells. 

Interestingly, PC has shown more advantages in providing efficient mechanisms to mitigate interference or enhancing the networks' performance when it is integrated with other allocation techniques such as admission control (users removal) and BS association (BSA) \cite{rasti11a, rasti11b, alpcan02, han05, 
lee05, Altman06, Saraydar01, Xiao03}.
In HetNets, the BSA algorithm design depends on the underlying access mode implemented at femtocells.
Hence, the joint design of BSA and PC in HetNets has also attracted much attention.
Specifically, femtocells allow MUEs to connect with their FBSs in the open access while they do not MUEs to do so in the closed access \cite{roche10, Le12, golaup09}.
However, most existing works on interference management and PC for femtocell networks assume
the closed access \cite{jo09,chand09,yun11,Le12a,yanzan12}, which are not flexible enough to support both voice and data UEs in different network tiers. 
In contrast, the open access is more effective in mitigating cross-tier interference compared to the closed access;
however, it may result in uncontrollable performance degradation of FUEs \cite{roche10,madan10}. 
Design of an efficient hybrid access scheme that can balance between advantages and 
disadvantages of the other two access modes is, therefore, an interesting and important research topic.

\subsection{Resource Allocation for Multi-Carrier Wireless HetNets}

There have been also several works studying the resource allocation for multi-carrier-based HetNets \cite{chuhan11,yanzan12,yushan12,Wangchi12,Long12,hoon11,renchao12,lu12}. In particular, \cite{chuhan11} has proposed a spectrum sharing and access control strategy for OFDMA femtocell networks by using the water-filling algorithm and game theory technique. 
Moreover, the authors in \cite{yanzan12} have proposed two methods to mitigate the uplink interference for OFDMA femtocell networks. 
In the first method, Femtocell User Equipments (FUEs) that produce strong interference to MUEs are only allowed to use dedicated subchannels
 while the remaining FUEs can utilize all subchannels assigned for the femto tier. In the second method, an auction-based algorithm is devised to optimize the channel assignment for both tiers to mitigate the co-tier interference. 
In addition, by employing a heuristic approach, \cite{yushan12} developed a greedy algorithm to enable self-organization for OFDMA femtocells to avoid the cross-interference. 

In \cite{Wangchi12}, Cheung  et al. study the  network performance where the macrocells and femtocells utilize separate sets of subchannels or share the whole
 spectrum under both open and closed access strategies. The QoS-aware admission control design for OFDMA femtocells
 is conducted in \cite{Long12}. 
All these mentioned works, however, do not consider PC 
 in their proposed resource allocation algorithms.
In \cite{hoon11}, the authors have developed an adaptive femtocell interference management algorithm comprising three control loops that run 
continuously and separately at MBS and FBSs to determine initial femto maximum power, target SINRs for FUEs, and to control the transmission power, respectively. However, an efficient SA design
is not investigated and the maximum power constraint for each UE is not considered in this work. 
The joint SA and PC optimization has been investigated in \cite{renchao12,lu12}. 
In particular, \cite{renchao12} aims to enhance the energy 
efficiency and while the main design objective of \cite{lu12} is to achieve user-level fairness for cognitive femtocells. 
However, the algorithms developed in these works have not provided QoS guarantees for users of both network tiers.

\subsection{CoMP Transmission Design for C-RANs}

There are some recent works that address the coordinated transmission issues in C-RANs. 
In \cite{fan_arxiv}, authors have proposed an efficient clustering algorithm to reduce the number of computations in the 
centralized cloud and ease the transmission design when the number of RRHs is very large. 
Liu and Lau in \cite{vincent_lau_14} attempt to maximize the average weighted sum-rate for uplink C-RANs by studying the joint optimization of antenna selection, regularization, and power allocation. 
Then, they employ the so-called random matrix theory method to decouple the corresponding non-convex problem into some sub-problems which can be tackled more easily.
In addition, the works \cite{letaief14,luo_arxiv} minimize the total network power consumption by optimizing the precoding vectors for all RRHs where
 \cite{letaief14} investigates the downlink case while \cite{luo_arxiv} addresses both the downlink and uplink scenarios. 

In these works, total power consumption is calculated by accounting for the powers due to radio transmission and fronthaul-link operations. 
The underlying problem is then formulated as the sparse beamforming problem which could be solved by employing the compressed 
sensing techniques \cite{donoho06,tao06,bach12}. 
The work in \cite{YShi15} proposes a generic stochastic coordinated transmission design for downlink C-RANs where the CSI is uncertain. A novel 
stochastic DC (difference-of-convex) programming algorithm is then proposed, which can allow to obtain the optimal solution.
However, the coordinated transmission design in C-RANs considering practical constraints such as limited cloud computation power
 and fronthaul capacity deserves further research.

\subsection{Fronthaul Utilization and Cloud Computation in C-RAN}

We now briefly discuss the literature on fronthaul utilization and cloud computation complexity in C-RANs. 
Early works in these directions are presented in \cite{shamai13a,shamai13b} where they propose
 compression techniques to minimize the amount of data transmitted over the fronthaul transport network.
The compression and quantization techniques for efficient transmissions over fronthaul links have recently been 
considered for uplink C-RANs in \cite{YZhou14,LLiu15} where \cite{YZhou14} aims to maximize the network weighted sum rate
 while \cite{LLiu15} focuses on maximizing the minimum SINR of all UEs.
Furthermore, the work in \cite{XRao15} considers a distributed fronthaul compression scheme at the 
RRHs to reduce the fronthaul loading in C-RAN. This work also attempts to develop a joint recovery strategy at the BBUs by 
employing the compressive sensing techniques for the end-to-end recovery of the transmitted signals from the users. 

The works in \cite{Werthmann13,sabella14,Rost_TWC15,Rost_GBC15} have studied on the optimization of computational resources of C-RANs.
Specifically, \cite{Werthmann13} has modeled the computation complexity  for downlink communication based on the operations 
of electrical circuits for processing the users' baseband signals. 
The work in \cite{sabella14} then applies this computation complexity model to estimate the energy efficiency benefits of C-RAN.
In addition, the authors in \cite{Rost_TWC15} have modeled the complexity model in uplink C-RAN assuming the turbo coding.
These authors have also attempted to optimize the rate allocation for each uplink transmission to maximize the system sum rate
considering the limitation on the cloud computation capacity \cite{Rost_GBC15}.
This work, however, assumes unlimited fronthaul capacity and does not optimize the resource allocation.
Interestingly, combining the wireless virtualization and C-RANs has been considered in \cite{XTran15} where
sharing the computing-resource in C-RAN for multiple virtual BSs and designing a joint
dynamic radio clustering and cooperative beamforming scheme that maximizes the downlink weighted sum-rate system utility have been studied. 
However, the computation complexity model in this paper is very simple.

\section{Concluding Remarks}

In this chapter, we have discussed fundamental resource allocation techniques and the literature on radio resource
allocation and interference management for wireless HetNets and C-RANs. 
Particularly, we have presented the PC and SA techniques, the CoMP as well the fronthaul utilization and cloud computation complexity in C-RANs.
Then, the literature survey has summarized the existing state-of-the-art approaches and pointed our their limitations in resolving
 technical challenges pertaining to the design of next-generation high-speed wireless cellular networks. 
These limitations of the existing literature motivate us to develop more efficient and practical solutions, which
will be presented in Chapters~\ref{Ch3}–\ref{Ch8} of this dissertation.

%% file: chap5/Ha_chap5.tex
\chapter{Distributed Base Station Association and Power Control for Heterogeneous Cellular Networks}
\renewcommand{\rightmark}{Chapter 5.  Distributed BSA and PC for HetNets}
\label{Ch3}
The content of this chapter was published in IEEE Transactions on Vehicular Technology in the following paper:

Vu N. Ha and Long B. Le, ``Distributed Base Station Association and Power Control for Heterogeneous Cellular Networks,'' {\em IEEE Trans. Veh. Tech.,} vol. 63, no. 1, pp. 282--296, Jan. 2014. 

\section{Abstract}
\label{Ch3_Abs}
In this paper, we propose a universal joint base station (BS) association (BSA) and power control (PC) algorithm for
heterogeneous cellular networks. Specifically, the proposed algorithm iteratively updates the BSA solution and the transmit power
of each user. Here, the new transmit power level is expressed as a
function of the power in the previous iteration, and this function is
called the power update function (puf).
We prove the convergence
of this algorithm when the puf of the PC strategy satisfies the
so-called ``two-sided scalable (2.s.s.) function'' property.
Then, we develop a novel hybrid PC (HPC) scheme by using noncooperative game theory and prove that its corresponding puf is 2.s.s.
Therefore, this HPC scheme can be employed in the proposed joint
BSA and PC algorithm.
We then devise an adaptation mechanism for the HPC algorithm so that it can support the signal-to-interference-plus-noise ratio (SINR) requirements of all users
whenever possible while exploiting multiuser diversity to improve
the system throughput.
We show that the proposed HPC adaptation algorithm outperforms the well-known Foschini–Miljianic PC
algorithm in both feasible and infeasible systems. 
In addition, we present the application of the developed framework to design a
hybrid access scheme for two-tier macrocell–femtocell networks.
Numerical results are then presented to illustrate the convergence
of the proposed algorithms and their superior performance, compared with existing algorithms in the literature.

\section{Introduction}
\label{Ch3_Introduction}
Femtocells have recently emerged as a potential solution to enhance indoor capacity and coverage \cite{Yeh08}.
Efficient deployment and operation of femtocells, however, demand to resolve various technical challenges beyond those existing in
the traditional cellular network, including network architecture design, interference management, and synchronization issues.
In fact, interference management for the macrocell–femtocell network is one of the most critical issues, which needs to be resolved to achieve enhanced network capacity and support users' quality of service (QoS) in different network tiers \cite{Yeh08,Claussen08,roche10,Le12,Panti12}.
Since the emerging femtocells utilize the same licensed spectrum with existing macrocells, both cotier interference within
each network tier (i.e., macro and femto tiers) and cross-tier interference among different network tiers need to be
properly managed. For code-division multiple-access (CDMA) networks, power control (PC) and base station association
(BSA) algorithms provide efficient mechanisms to mitigate interference and maximize the system throughput
\cite{Yates95a,hanly95,farrokhi97,teng10}.

A good BSA algorithm typically provides an efficient mechanism to associate mobile users with one or several serving
base stations (BSs) so that certain performance metrics of interest are optimized.
In general, BSA can be decided based on different factors and metrics, such as achievable rates, transmit
powers, geographical locations, and cell load \cite{kyuho09}.
In addition, the BSA can be jointly designed with PC \cite{Yates95a,hanly95,farrokhi97,teng10}.
In the two-tier macrocell–femtocell network, the design of a BSA
algorithm also depends on the underlying access mode implemented at femtocells.
Specifically, femtocells allow macro user equipment (MUE) devices to connect with their femto base
stations (FBSs) in the open access, whereas they do not allow MUE devices to do so in the closed access \cite{roche10,Le12,golaup09}.
In general, the open access is more effective in mitigating cross-tier interference compared with the closed access; however, it
may result in uncontrollable performance degradation of femto users \cite{roche10}.
Design of an efficient hybrid access scheme that can balance between advantages and disadvantages of the other
two access modes is, therefore, an interesting and important research topic.

PC is an important research topic that has been extensively investigated in the literature 
\cite{foschini93,yates95,Yates95a,hanly95,farrokhi97,teng10,sung05,sung06,rasti11a,rasti11b,alpcan02,han05,lee05,Altman06,Xiao03}.
Among existing PC algorithms, there are two popular PC algorithms, namely, the
target-signal-to-interference-plus-noise ratio (SINR)-tracking PC (TPC) algorithm \cite{foschini93}
and the opportunistic PC (OPC) algorithm \cite{sung05,sung06}.
Specifically, Foschini and Miljianic were the first to propose the distributed TPC algorithm that can
support predetermined target SINRs for all users with minimum powers (i.e., achieving Pareto optimality) \cite{foschini93}.
This PC algorithm was extended to consider maximum power constraints and distributed BS association in \cite{yates95,Yates95a,hanly95,farrokhi97}.
However, the TPC algorithm aims to support fixed target SINRs, which is, therefore, more applicable to the voice application. 
In contrast, the OPC algorithm, which was proposed by Sung and Leung in \cite{sung05,sung06}, 
aims at exploiting multiuser diversity to improve the system throughput.
In these papers, Sung and Leung introduced the so-called ``two-sided scalable (2.s.s.) function'', which was
the extended version of the “standard function” notion proposed by Yates \cite{yates95}, to prove the convergence of the OPC algorithm.
Unfortunately, the OPC algorithm cannot provide any QoS support for users in the network.
There have been also some more recent works that investigated the joint user removal and
PC algorithms \cite{rasti11a,rasti11b}.
In addition, various PC algorithms have been developed by employing the game theory approach \cite{alpcan02,han05,lee05,Altman06,Xiao03}.
In most cases, proposed PC algorithms are proved to converge to the Nash equilibrium (NE) of the underlying
game. In particular, a pricing-based approach has been taken to design PC algorithms in \cite{alpcan02}
that strike to achieve an efficient throughput and a power tradeoff.
In addition, sufficient conditions for convergence and having the largest number of
users attain their target SINR have been presented.
\nomenclature{NE}{Nash Equilibrium}

There are some existing works that consider PC and interference management problems for femtocell networks \cite{jo09,chand09,yun11,Le12a,yanzan12,madan10}.
In particular, Jo \textit{et al.} \cite{jo09} proposed two PC strategies, which aim to set the transmit powers of 
femto user equipment (FUE) devices to maintain the cross-tier interference below a fixed threshold at macrocell BSs (MBSs). 
Various PC strategies for both macrocells and femtocells were also developed by using the game theory approach in \cite{chand09,yun11,Le12a}. In particular, 
we previously proposed distributed PC algorithms for closed-accessed macro-femto networks in \cite{Le12a}.

Interference management methods with dynamic subcarrier allocation and cell association were proposed in \cite{yanzan12} and \cite{madan10}, respectively.
Most existing works on interference management and PC for femtocell networks, however, assume
the closed access \cite{jo09,chand09,yun11,Le12a,yanzan12}, which are not flexible enough to support both voice and data users in different network tiers.
This paper aims to resolve these limitations.
Specifically, we make the following contributions.

\begin{itemize}
\item We develop a generalized BSA and PC algorithm and prove its convergence for any 2.s.s. power update function (puf). 
We then propose a hybrid PC (HPC) scheme that can be used in this general algorithm.
In addition, we describe two alternative designs, namely, decomposed and joint BSA and PC design.
For the decomposed design, we develop a simple load-aware BSA algorithm.
These designs focus on the heterogeneous cellular network with different
types of cells (e.g., macrocells, microcells, and femtocells) and users with different QoS
requirements (e.g., voice and data users). 

\item We propose an HPC adaptation algorithm that adjusts parameters of the HPC scheme
to support the differentiated SINR requirements of all users whenever possible while enhancing the system throughput. 
We then present the application of the proposed framework to the two-tier macrocell-femtocell networks.

\item We present extensive numerical results to validate the developed theoretical
results and to demonstrate the efficacy of our proposed algorithms.
\end{itemize}

The remainder of this paper is organized as follows: We describe the system model in Section~\ref{Ch3_section2}. In Section~\ref{Ch3_section3}, we 
present the generalized BSA and PC algorithm and the HPC scheme. 
We propose an HPC adaptation algorithm and present its application to macrocell-femtocell networks in Section~\ref{Ch3_sec:Time_Scal_Sepa}. 
Numerical results are presented in Section~\ref{Ch3_result}, followed by conclusion in Section~\ref{Ch3_ccls}.
 
\section{System Model}
\label{Ch3_section2}

We consider uplink communications in a heterogeneous wireless cellular network. 
We assume that there are $M$ users, which are labeled $1, 2, ..., M$, transmitting information to $K$ BSs, which are labeled ${1, 2, ..., K}$, on the same spectrum by using CDMA. 
Each of these BSs can belong to one of the available cell types (e.g., macrocells, microcells, picocells, and femtocells) in a heterogeneous cellular network.
Let $\mathcal{M}$ and $\mathcal{K}$ be the sets of all users and BSs, respectively, i.e., $\mathcal{M}=\left\lbrace 1, 2, ..., M \right\rbrace $ and $\mathcal{K}=\left\lbrace {1,2,...,K} \right\rbrace $. Assume that each user $i$ communicates with only one BS at any time, which is denoted as $b_i \in \mathcal{K}$. 
However, users can change their associated BSs over time. 
Note that if $b_i \equiv b_j$, then users $i$ and $j$ are associated with same BS.
Let $D_i$ be the set of BSs that user $i$ can be associated with, which are nearby BSs of user $i$ in practice.
Note that we have $D_i\subseteq \mathcal{K}$ and $b_i \in D_i$.
We are interested in developing a BSA strategy which determines how each user $i$ will choose one BS $b_i \in D_i$ to communicate with based on its observed channel state information and interference. 

Let the transmit power of user $i$ be $p_i$, whose maximum value is $\bar{p}_i$, i.e., $0 \leq p_i \leq \bar{p}_i$.
We arrange transmit powers of all users in a vector, which is denoted by $\mathrm{p} = \left( p_1,p_2,...,p_M \right)$.
Let $h_{ij}$ be the channel power gain from user $j$ to BS $\mathit{i}$, and $\eta_i$ be the noise power at BS $\mathit{i}$.
Then, the SINR of user $i$ at BS $b_i$ can be written as \cite{Alpcan08}
\begin{equation}
\label{Ch3_sinr}
\Gamma_i(\mathrm{p})=\dfrac{G h_{b_i i} p_i}{\sum_{j \neq i}{h_{b_ij}p_j}+\eta_{b_i}}=\dfrac{p_i}{R_i \left( \mathrm{p}, b_i\right) }
\end{equation}
where $G$ is the processing gain, which is defined as the ratio of spreading bandwidth to the symbol rate; $R_i \left( \mathrm{p}, b_i \right) $ is the effective interference to user $i$, which is defined as
\begin{equation}
\label{Ch3_Rp}
 R_i \left( \mathrm{p}, b_i \right) \triangleq  \dfrac{\sum_{j\neq i}{h_{b_ij}p_j}+\eta_{b_i}}{G h_{b_ii}}
\end{equation}
where $R_i \left( \mathrm{p},k\right)$ is the effective interference experienced by user $i$ at BS $\mathit{k}$.
We will sometimes write $R_i \left( \mathrm{p}\right)$ instead of  $R_i \left( \mathrm{p},b_i\right)$ when there is no confusion.
We assume that each user $i$ requires the minimum QoS in terms of a target SINR $\hat{\gamma}_i$ $\forall i \in \mathcal{M}$. 
In particular, these QoS requirements can be written as follows:
\begin{equation}
\label{Ch3_equ:SINR_cond}
\Gamma_i(\mathrm{p}) \geq \hat{\gamma}_i, \:\: i \in \mathcal{M}.
\end{equation}

The objective of this paper is to develop distributed BSA and PC algorithms that can maintain 
the SINR requirements in (\ref{Ch3_equ:SINR_cond}) (whenever possible) while exploiting the multiuser diversity gain to increase the system throughput. 
The proposed algorithms, therefore, aim to support both voice and high-speed data applications.
Moreover, we aim to achieve these design objectives for a heterogeneous wireless environment where there are different 
kinds of users with differentiated QoS targets (e.g., voice and data users) and potentially different cell types (e.g., macrocells, microcells, picocells, and femtocells).
In particular, voice users would typically require some fixed target SINR $\hat{\gamma}_i$, whereas data users would
seek to achieve higher target SINR $\hat{\gamma}_i$ to support their broadband applications (e.g., video and Internet browsing).

\section{Base Station Association and Power Control}
\label{Ch3_section3}
\subsection{Generalized BSA and PC Algorithm}
\label{Ch3_sec:gnrl_PC_BAS}
We develop a general minimum effective interference BSA and PC algorithm and prove its convergence.
Specifically, we will focus on a general iterative PC algorithm where
each user $i$ in the network performs the following power update $p_i^{(n+1)} := J_i(\mathrm{p}^{(n)})$, where
$n$ denotes the iteration index, and $J_i(.)$ is the puf\footnote{It can be verified that the pufs of the TPC
and OPC schemes are 2.s.s.}. 
In fact, this kind of PC algorithm converges if we can prove 
that its corresponding puf is a \textit{2.s.s. function} according to
Sung and Leung \cite{sung05, sung06}.
The challenges involved in designing such a PC algorithm are that we have to ensure its puf is \textit{2.s.s.}, 
it fulfills our design objectives for the heterogeneous cellular network, and it can be
implemented in a distributed manner.
In addition, we seek to design a PC algorithm
that can be integrated with an efficient BSA mechanism. 
Toward this end, we give the definition of \textit{2.s.s. function}
in the following.

\begin{definition}
\label{Ch3_def_2s_scal}
A puf $\mathrm{J( p ) }=\left[ J_1(\mathrm{p}), \right. $ $ \left. ...,J_M(\mathrm{p})\right]^{\mathrm{T}} $ is \textit{2.s.s.} with
 respect to $\mathrm{p}$ if for all $a >1$ and any power vector $\mathrm{p}^{\prime}$ satisfying $\frac{1}{a}\mathrm{p} \leq \mathrm{p}^{\prime} \leq a \mathrm{p}$, we have
\begin{equation}
\label{Ch3_eq_def_2s_scal}
(1/a) J_i \left( \mathrm{p}\right) < J_i(\mathrm{p} ^{\prime} ) < a J_i(\mathrm{p}) \; \forall i \in \mathcal{M}.
\end{equation}
\end{definition}

We will consider a joint BSA and PC algorithm where the puf of the PC scheme satisfies the \textit{2.s.s.} property stated in
Definition~\ref{Ch3_def_2s_scal}. Under this design, each user chooses its ``best'' BS and updates its transmit power in the distributed manner.
Specifically, the proposed algorithm is described in  Algorithm \ref{Ch3_alg:gms2} where each user chooses a BS that results in minimum effective interference
and updates its power by using any \textit{2.s.s.} puf $\mathrm{J}(\mathrm{p})=\mathrm{J}'(\mathrm{R}(\mathrm{p}))=[J'_1(R_1(\mathrm{p})),...,J'_M(R_M(\mathrm{p}))]$, where $\mathrm{R}(\mathrm{p})=[R_1(\mathrm{p}),...,R_M(\mathrm{p})]$. The proposed algorithm can be implemented distributively with any distributed PC algorithm.
To realize the BSA, each user needs to estimate the effective interference levels for different nearby BSs of interest.
Then, each user $i$ can choose one BS $k$ in $D_i$ with a minimum value of $R_i( \mathrm{p},k)$.
This algorithm ensures that each user experiences low effective interference and, therefore, high throughput at convergence.

User $i$ can estimate $R_i( \mathrm{p},k)$ if it has information about the total interference and noise power [i.e., the denominator of (\ref{Ch3_sinr})]
and the channel power gain $h_{ki}$ according to SINR expression in (\ref{Ch3_sinr}) (since the current transmit power level $p_i$ is readily available).
In addition, the channel power gains $h_{ki}$ can be estimated by BS $k$ and sent back to each user by using the pilot signal and any standard channel estimation technique \cite{steiner94,Amico03,Shakya09}.
Moreover, the total interference and noise power for each user can estimated as follows.
Each BS estimates the total received power then broadcasts this value to its connected users.
Each user $i$ can calculate the total interference and noise power by subtracting its received
signal power (i.e., $h_{ki}p_i$) from the total receiving power broadcast by the BS.
Therefore, calculation of the effective
interference only requires the standard channel estimation of $h_{ki}$ and estimation of the total receiving power at the BS.
In addition, signaling is only involved in sending these values from the BS to its connected users, which is relatively mild and can be conducted over the air. 
More importantly, estimation of the effective interference $R_i\left( \mathrm{p}\right)$ and, therefore, the proposed PC algorithm given in (\ref{Ch3_hpcrule}) can be implemented in a distributed fashion. 

\begin{algorithm}[!t]
\caption{\textsc{Minimum Effective Interference BSA and PC Algorithm}}
\label{Ch3_alg:gms2}
\begin{algorithmic}[1]
\STATE Initialization:
\begin{itemize}
\item $p_i^{(0)}=0$ for all user $i$, $i \in \mathcal{M}$.
\item $b_i^{(0)}$ is set as the nearest BS of user $i$.
\end{itemize} 
\STATE Iteration $n$: Each user $i$ ($ i \in \mathcal{M}$) performs the following:
\begin{itemize}
\item Calculate the effective interference at BS $\mathit{k} \in D_i$ as follows:
\begin{equation}
\label{Ch3_Rp2}
R_i^{(n)}( \mathrm{p}^{(n-1)},\mathit{k})=\dfrac{\sum_{j\neq i}{h_{\mathit{k}j}p_j^{(n-1)}}+\eta_k}{G h_{\mathit{k}i}}.
\end{equation}
\item Choose the BS $b_i^{(n)}$ with the minimum effective interference, i.e., 
\begin{align}
b_i^{(n)}&=&\mathrm{argmin}_{\mathit{k} \in D_i} R_i^{(n)} ( \mathrm{p}^{(n-1)},\mathit{k}). \label{Ch3_bi} \\
R_i^{\sf o (\mathit{n})}(\mathrm{p}^{(n-1)})& =&\mathrm{min}_{\mathit{k} \in D_i} R_i^{(n)} ( \mathrm{p}^{(n-1)},\mathit{k}) \nonumber \\
& = & R_i^{(n)} ( \mathrm{p}^{(n-1)},b_i^{(n)}). \label{Ch3_Rmin}
\end{align}
\item Update the transmit power for the chosen BS as follows:
\begin{equation}
\label{Ch3_p}
p_i^{(n)}=J'_i(R_i^{\sf o (\mathit{n})}(\mathrm{p}^{(n-1)})). 
\end{equation}
\end{itemize} 
where $J'_i(R_i(\mathrm{p}))$ is the puf with respect to $R_i(\mathrm{p})$.
\STATE Increase $n$ and go back to step 2 until convergence.
\end{algorithmic}
\end{algorithm}

Note that we have expressed the pufs with respect to both $\mathrm{p}$ and $\mathrm{R}(\mathrm{p})$.
For the expression with respect to $R(\mathrm{p})$, the $i$th element of the puf is denoted by $J'_i(R_i(\mathrm{p}))$, where $R_i(\mathrm{p})$ is the effective interference experienced by user $i$ given in (\ref{Ch3_Rp}).
Therefore, we have $\mathrm{J}(\mathrm{p}) = \mathrm{J}'(\mathrm{R}(\mathrm{p}))$,
which depends on both the transmit powers of all users and the BS with which each user $i$ is associated with in general.
We establish the convergence of Algorithm~\ref{Ch3_alg:gms2} in the following by utilizing the \textit{2.s.s. function} approach.
Toward this end, we recall the convergence result for any PC algorithm that employs a bounded \textit{2.s.s.} puf in the following lemma \cite{sung05}.

\begin{lemma}
\label{Ch3_lemma3.1}
Assume that $\mathrm{J}(\mathrm{p})$ is a \textit{2.s.s.} function, whose elements ${J}_i(\mathrm{p})$ are bounded by zero
and $\bar{p}_i$, i.e., $0 \leq {J}_i(\mathrm{p}) \leq \bar{p}_i$.
Consider the corresponding power update $p_i^{(n+1)} := J_i(\mathrm{p}^{(n)}), \forall i$,
where $n$ denotes the iteration index. Then, we have the following results.
\begin{enumerate} 
\item The puf $\mathrm{J}(\mathrm{p})$ has a unique fixed point corresponding to a transmit power vector $\mathrm{p}^{\ast}$ that satisfies $\mathrm{p}^{\ast}=\mathrm{J}(\mathrm{p}^{\ast})$.
\item Given an arbitrary initial power vector $\mathrm{p}^{(0)}$, the PC algorithm based on puf $\mathrm{J}(\mathrm{p})$  converges 
to that unique fixed point $\mathrm{p}^{\ast}$.
\end{enumerate}
\end{lemma}
\begin{proof}
The results stated in this lemma have been established for the \textit{2.s.s.} puf in \cite{sung05}. 
\end{proof}

We are now ready to state one important result for the joint BSA and PC operations described in Algorithm~\ref{Ch3_alg:gms2} in the following theorem.

\begin{theorem}
\label{Ch3_thm3.1}
Assume $\mathrm{J}(\mathrm{p})$ and $\mathrm{J}'(\mathrm{R}(\mathrm{p}))$ are arbitrary 2.s.s. pufs with respect to $\mathrm{p}$ and $\mathrm{R}(\mathrm{p})$, respectively. Then, Algorithm~\ref{Ch3_alg:gms2} converges to an equilibrium. 
\end{theorem}

\begin{proof} The proof is given in Appendix \ref{Ch3_apdx.thm3.1}.
\end{proof}

To prove that Algorithm~\ref{Ch3_alg:gms2} converges, we first show that the puf $\mathrm{J}^{\sf o}(\mathrm{p})=[J_1^{\sf o}(\mathrm{p}),...,J_M^{\sf o}(\mathrm{p})]$ is a \textit{2.s.s.} function with respect to $\mathrm{p}$, where $J_i^{\sf o}(\mathrm{p})=J'_i(R_i^{\sf o}(\mathrm{p}))$, and $R_i^{\sf o}(\mathrm{p}) =\mathrm{min}_{\mathit{k} \in D_i} R_i \left( \mathrm{p},\mathit{k}\right), \forall i \in \mathcal{M}$. Then, the convergence can be established by applying the results in Lemma~\ref{Ch3_lemma3.1}.
The result in this theorem implies if pufs $\mathrm{J}(\mathrm{p})$ and $\mathrm{J}'(\mathrm{R}(\mathrm{p}))$ are \textit{2.s.s.}, then the BSA strategy described 
in (\ref{Ch3_bi}) and (\ref{Ch3_Rmin}) results in a composite 2.s.s. puf, which corresponds to the joint BSA and PC operation.
In other words, the proposed BSA scheme proposed in Algorithm~\ref{Ch3_alg:gms2} preserves the \textit{2.s.s.} property of the employed puf $\mathrm{J}(\mathrm{p})$.

\subsection{HPC Algorithm}
\label{Ch3_hpcscheme}
To complete the design presented in Algorithm~\ref{Ch3_alg:gms2}, we need to develop a distributed PC strategy, which is employed in (\ref{Ch3_p}).
In general, the performance of a PC algorithm depends on how we design the corresponding \textit{2.s.s.} puf $\mathrm{J}(\mathrm{p})$. 
We will propose a new puf, which is denoted by $\mathrm{I}^H(\mathrm{p})$, and the corresponding PC algorithm in the following. 

\subsubsection{Game-Theoretic Formulation}
\label{Ch3_sec:PC_game}
We develop the distributed PC algorithm by using the noncooperative game theory approach.
In particular, we define a PC game as follows.
\begin{itemize}
\item Players: This is the set of mobile users $\mathcal{M}$.
\item Strategies: Each user $i$ chooses transmit power in set $[0,\bar{p}_i]$.
\item Payoffs: User $i$ is interested in maximizing the following payoff function:
\begin{equation} 
\label{Ch3_payoff}
U_i(\mathrm{p}) \triangleq - \alpha_i ( p_i  - \xi_i R_i(\mathrm{p})^{\frac{x}{x-1}})^2 - (p_i-\widehat{\gamma}_i R_i(\mathrm{p}))^2
\end{equation}
where $\widehat{\gamma}_i$ denotes the target SINR for user $i$; $x$ is a special parameter whose desirable value will
be revealed in Theorem~\ref{Ch3_thm3.2}, and $\alpha_i$ and $\xi_i$ are nonnegative control
parameters, i.e., $\alpha_i, \xi_i \geq 0$, which will be adaptively adjusted to achieve our design objectives.
\end{itemize}

This game-theoretic formulation arises quite naturally in autonomous spectrum access scenarios, such
as heterogeneous wireless networks where mobile users tend to be selfish and are only interested in maximizing
their own benefits. Using this formulation, we will develop an iterative PC algorithm in which
each user maximizes its own payoff in each iteration given the chosen power levels from other users in the
previous iteration (i.e., each user plays the \textit{best response} strategy).
To devise such an algorithm,
each user $i$ chooses the power level, which is obtained by setting the first derivative of the underlying
user's payoff function to zero. 

In fact, by maximizing the payoff function given in (\ref{Ch3_payoff}) each user $i$ strikes to balance between
achieving the SINR target $\widehat{\gamma}_i$ and exploiting its potential favorable channel condition to increase
its SINR. While it is quite intuitive that maximizing $- (p_i-\widehat{\gamma}_i R_i(\mathrm{p}))^2$ enables user $i$ to reach its target 
SINR $\widehat{\gamma}_i$, the design intuition in optimizing the first term $-\alpha_i \left( p_i  - \xi_i R_i(\mathrm{p})^{\frac{x}{x-1}}\right)^2$
may not be very straightforward.
We state one result, which reveals the engineering intuition behind this in the following lemma.

\begin{lemma}
\label{Ch3_lemma3.2} 
Consider the game formulation described above with infinite power budget (i.e., $\bar{p}_i = \infty \forall i$). Then, best responses due to the following two payoff functions are the same:
\begin{align} 
U_i^{(1)}(\mathrm{p}) \triangleq \Gamma_i^x - \lambda_i p_i \label{Ch3_opc1} \\ 
U_i^{(2)}(\mathrm{p}) \triangleq - \left( p_i  - \xi_i R_i(\mathrm{p})^{\frac{x}{x-1}}\right)^2 \label{Ch3_opc2}
\end{align}
if $\xi_i=\left( \lambda_i/x\right)^{\frac{1}{x-1}}$ and $0<x<1$, where $\lambda_i$ represents the pricing coefficient of user $i$.
\end{lemma}
\begin{proof} The proof is given in Appendix~\ref{Ch3_apdx.lemma3.2}
\end{proof}

\begin{remark}
\label{Ch3_rk:HPC_OPC_TPC}
It has been shown that OPC algorithm can be achieved if users iteratively play 
their corresponding best-response strategies with payoff function $U_i^{(1)}(\mathrm{p})$ for $x=1/2$ \cite{sung05,sung06}. 
Moreover, if each user $i$ plays the best response strategies using $U_i(\mathrm{p})$ with $\alpha_i=0$ then we can obtain the well-known TPC algorithm. Therefore, our chosen payoff function in (\ref{Ch3_payoff}) can be used to design an HPC strategy that exploits the advantages of both OPC and TPC algorithms. 
\end{remark}

\subsubsection{Proposed HPC Algorithm}
We are now ready to develop an HPC algorithm corresponding to the payoff function in (\ref{Ch3_payoff}).
Specifically, we can derive the power update rule for the HPC algorithm according to the best response strategy
of the underlying payoff function. After some manipulations, we can obtain the following best response under the chosen payoff function (\ref{Ch3_payoff}):
\begin{equation} \label{Ch3_bpupd}
p_i = I_i\left( \mathrm{p}\right) \triangleq \dfrac{\alpha_i\xi_i R_i\left( \mathrm{p}\right)^{\frac{x}{x-1}}+ \hat{\gamma}_i R_i\left( \mathrm{p}\right)}{\alpha_i + 1}.
\end{equation}
Considering the maximum power constraints, the HPC algorithm employs the following iterative power update:
\begin{equation} \label{Ch3_hpcrule}
p_i^{\left(n+1\right)} = I_i^H \left( \mathrm{p}^{(n)}\right) = \mathrm{min}\left\lbrace \bar{p}_i, I_i\left( \mathrm{p}^{(n)}\right) \right\rbrace 
\end{equation}
where $n$ denotes the iteration index, and $I_i\left( \mathrm{p}\right)$ is given in (\ref{Ch3_bpupd}).
Here, parameters $\alpha_i$ can be used to control the desirable performance of the proposed HPC algorithm.
Specifically, by setting $\alpha_i=0$, user $i$ actually employs the standard Foschini-Milijanic TPC algorithm to 
achieve its target SINR $\hat{\gamma}_i$, whereas if $\alpha_i \rightarrow \infty$, user $i$ attempts to achieve higher SINR (if it is
in favorable condition). 
It can be observed that each user $i$ only needs to calculate or estimate the effective interference $R_i\left( \mathrm{p}\right)$ to update its power in the proposed HPC algorithm.  

\subsubsection{Convergence of HPC Algorithm}
\label{Ch3_convergence_HPC}
Here, we establish the convergence condition for the proposed HPC algorithm
by using the \textit{2.s.s function} approach given in Definition~\ref{Ch3_def_2s_scal}.
We state a sufficient condition under which puf $\mathrm{I}^H\left( \mathrm{p}\right)$ in (\ref{Ch3_hpcrule})
is \textit{2.s.s.} as in Theorem~\ref{Ch3_thm3.2}.

\begin{theorem}
\label{Ch3_thm3.2}
If the parameter $x$ of function $I_i\left( \mathrm{p}\right)$ given in (\ref{Ch3_bpupd}) satisfies $0 < x \leq 1/2$, then
puf $I_i^H \left( \mathrm{p}\right)$ given in (\ref{Ch3_hpcrule}) is \textit{2.s.s.}
In addition, the proposed HPC algorithm in (\ref{Ch3_hpcrule}) converges to the NE of the underlying PC game. 
\end{theorem}

\begin{proof}
The convergence of the proposed HPC algorithm immediately follows from the results of Lemma~\ref{Ch3_lemma3.1} if we can prove that its puf $\mathrm{I}^H\left( \mathrm{p}\right)=\left[ I_1^H(\mathrm{p}),...,I_M^H(\mathrm{p})\right]$ 
 in (\ref{Ch3_hpcrule}) is \textit{2.s.s.} with respect to $\mathrm{p}$. 
In addition, the resulting equilibrium (i.e., the power vector at convergence) is the NE of the PC game defined
 in Section \ref{Ch3_sec:PC_game} since users play the best response strategy.
We will prove that $\mathrm{I}^H\left( \mathrm{p}\right)$ is \textit{2.s.s.} in Appendix~\ref{Ch3_apdx.thm3.2}.
\end{proof}

The proposed HPC scheme will be employed in Algorithms~\ref{Ch3_alg:gms2} and \ref{Ch3_alg:gms3} presented in this paper. Note, however, that this proposed HPC scheme
can be used as a standalone PC algorithm in general.

\subsection{Alternative BSA and PC Designs}
We describe two alternative designs for the BSA and PC operations in the following.

\subsubsection{Decomposed BSA and PC Design}
We can design the BSA and PC algorithms separately and implement them over different time scales.
In particular, a BSA solution
can be obtained and fixed based on the average long-term channel state information.
Then, PC is applied to the
 BSA solution to achieve the design objectives.
The advantage of this decomposed design is that BSA is only updated as the long-term channel state information significantly changes. Therefore, this design requires infrequent updates of the BSA solution. This would lead to reduced computation and signaling complexity.

\begin{algorithm}[!t]
\caption{\textsc{Load-Aware BSA Algorithm}}
\label{Ch3_alg:gms1}
\begin{algorithmic}[1]
\STATE Each BS estimates the average channel power gains from nearby users to itself.
\STATE By assuming all users transmit with their maximum powers, each BS estimates/calculates SINRs achieved by nearby users;
it transmits these estimated SINRs to them.
\STATE Upon receiving estimated SINRs from all potential BSs, each user will associate with the BS achieving
the largest estimated SINR. 
\end{algorithmic}
\end{algorithm}

We propose a load-aware BSA algorithm for this design in Algorithm~\ref{Ch3_alg:gms1}. In this algorithm, each user is associated with the BS that results
in the highest SINR by using the average channel power gains and assuming maximum transmit powers. To implement this algorithm, each BS needs to estimate
the SINRs for nearby users and broadcasts these estimated SINRs to them, and based on this, each user chooses the associated BS.  
Given the BSA solution, users run the HPC strategy to settle the power levels. This decomposed design would be more applicable to the fast-fading environment where the joint dynamic BSA and power control algorithm may not work efficiently.

\subsubsection{Joint BSA and PC Design}
\label{Ch3_sec:joint_BAS_PC}
In general, the network throughput can be improved by performing joint BSA and PC algorithm presented in Algorithm~\ref{Ch3_alg:gms2}.
In particular, the BSA solution is dynamically updated jointly with PC, which employs the proposed HPC scheme presented
Section~\ref{Ch3_hpcscheme}. Here, each user $i$ chooses one BS that achieves the minimum effective interference $R_i(\mathrm{p},k)$
and updates its power level accordingly. When a particular user is in the common neighborhood of several BSs, the fluctuation of
its channel power gains toward these BSs may result in varying BSA decisions even for fixed users. This means
that each user may transmit data to different BSs  over a short time interval. In addition,
 to facilitate the dynamic BSA, each user needs to estimate the effective interference frequently. This requires 
more signaling overhead in the control channel compared with the decomposed BSA and PC design. 

\section{Hybrid Power Control Adaptation Algorithm}
\label{Ch3_sec:Time_Scal_Sepa}

The equilibrium achieved by the proposed HPC scheme at convergence for either decomposed or joint
BSA and PC design depends on its parameters, namely $\alpha_i$ and $\xi_i$ for $i \in \mathcal{M}$.
Here, we develop decentralized mechanisms to adjust these parameters so that target SINRs of all users can be achieved whenever possible
while enhancing the system throughput. The proposed adaptive mechanisms comprise two updating operations in two different time scales, i.e., running HPC algorithm (and the appropriate BSA strategy) to achieve the NE point in the small time scale and updating $\alpha_i$ and $\xi_i$ for all users to achieve desirable NE in the large time scale. Moreover, we discuss the application of the proposed framework to the two-tier macrocell-femtocell networks. Let $\Delta= \left\lbrace \alpha_i | i \in \mathcal{M}\right\rbrace $ and $\Xi= \left\lbrace \xi_i | i \in \mathcal{M}\right\rbrace $ be the set of $\alpha$ and $\xi$ parameters of all users in the puf of the HPC scheme, respectively.

\subsection{Two Time-Scale Adaptive Algorithm}
\label{Ch3_sec:Gnrl_Time_Scal_Sepa}

\begin{figure}[!t]
        \centering
\includegraphics[width=0.9 \textwidth]{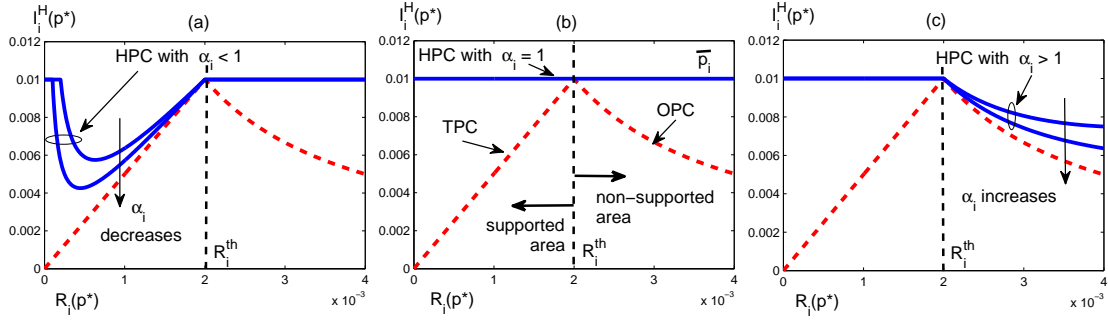}
\caption{Relationship between HPC puf $I_i^{'H}(R_i(\mathrm{p}))$ and  $R_i(\mathrm{p})$ for different values of $\alpha_i$ ($\bar{p}_i=0.01\:W$), where $R_i(\mathrm{p}^{\ast}) \leq R_i^{th}$ for supported users and $R_i(\mathrm{p}^{\ast}) > R_i^{th}$ for nonsupported users .}
\label{Ch3_fig:HPCalpha}
\end{figure}

As discussed in Remark \ref{Ch3_rk:HPC_OPC_TPC}, the TPC scheme is a special case of the proposed HPC scheme with $\alpha_i=0$ for $i \in \mathcal{M}$.
It is known that the TPC scheme is able to support all users' target SINRs expressed in (\ref{Ch3_equ:SINR_cond}) as long as the system is feasible [i.e., the
SINR requirements in (\ref{Ch3_equ:SINR_cond}) can be supported] \cite{rasti11a}. However, the TPC scheme fails to achieve high system throughput when the system is feasible and lowly-loaded. In addition,
the TPC scheme may not be able to support the largest possible number of users when the system is infeasible. Our objective is to develop
an adaptive strategy for the proposed HPC algorithm to overcome these limitations of the TPC scheme.

Toward this end, if user $i$ is a voice user who is only interested in maintaining its target SINR $\hat{\gamma}_i$, then
we can simply set $\alpha_i=0$ in (\ref{Ch3_payoff}).
For each data user $i$, we will fix $\xi_i$ while adaptively updating power and $\alpha_i$ in two different time scales to achieve the 
design objectives. Specifically, each data user $i$ will run the HPC and BSA algorithm for a given  $\alpha_i$ until convergence
(i.e., in the small time scale); then it updates  $\alpha_i$ accordingly (i.e., in the large time scale).
To set the value for $\xi_i$, suppose that data user $i$ would need to use its maximum power $\bar{p}_i$ to reach the target
SINR $\hat{\gamma}_i$. Then, the value of $\xi_i$ can be found using the result in (\ref{Ch3_power_opc}), shown in Appendix~\ref{Ch3_apdx.lemma3.2}, as 
\begin{equation}
\label{Ch3_xi2}
\xi_i=\left(\bar{p}_i/\hat{\gamma}_i^x \right) ^{\frac{1}{1-x}}
\end{equation}
where we have substituted $\Gamma_i = \hat{\gamma}_i$ to (\ref{Ch3_power_opc}). 
We now describe how to update $\alpha_i$ in $\Delta$.
In the HPC scheme, puf $I_i^{\prime H}(R_i(\mathrm{p}))=I_i^H(\mathrm{p})$ depends on the effective interference $R_i(\mathrm{p})$ and
the value of $\alpha_i$. Let $\mathrm{p}^*$ denote the power vector at convergence, which is obtained by running the proposed HPC algorithm for a particular vector $\alpha$. 
We illustrate the relationship between $I_i^{\prime H}(R_i(\mathrm{p}^*))$ and $R_i(\mathrm{p}^*)$ in Fig.~\ref{Ch3_fig:HPCalpha} 
where $R^{\sf th}_i = \bar{p}_i/\hat{\gamma}_i$ is a threshold effective interference of user $i$. We reveal the relationship
between $\Gamma_i(\mathrm{p}^*)$, $R_i(\mathrm{p}^*)$, and $\alpha_i$ in the following lemma.

\begin{lemma}
\label{Ch3_lemma3.3} 
Assume that $\xi_i=\left(\bar{p}_i/\hat{\gamma}_i^x \right) ^{\frac{1}{1-x}}$ and $0< x \leq 1/2$. Then, we have the following.
\begin{enumerate}
\item If $R_i(\mathrm{p}^*) > R_i^{\sf th}$, then $\Gamma_i(\mathrm{p}^*) < \hat{\gamma}_i$, $\forall \: \alpha_i \geq 0$.
\item If $R_i(\mathrm{p}^*) \leq R_i^{\sf th}$, then $\Gamma_i(\mathrm{p}^*) \geq \hat{\gamma}_i$, $\forall \: \alpha_i \geq 0$.
\item If $R_i(\mathrm{p}^*) < R_i^{\sf th}$, then we have:
\begin{itemize}
\item $\Gamma_i(\mathrm{p}^*) = \hat{\gamma}_i$ iff $\alpha_i = 0$.
\item $I_i(\mathrm{p}^*)$ decreases if $\alpha_i$ decreases.
\end{itemize} 
\end{enumerate}
\end{lemma}

\begin{proof} The proof is given in Appendix~\ref{Ch3_apdx.lemma3.3}. 
\end{proof}

It is shown in Fig. \ref{Ch3_fig:HPCalpha} that if $\alpha_i<1$, HPC curves become closer to the TPC curve as $\alpha_i$ decreases,
whereas if $\alpha_i>1$, HPC curves become closer to the OPC curve as $\alpha_i$ increases.
This is because $I_i(\mathrm{p})$ in the puf of the proposed HPC scheme is a weighted sum of those in the TPC and OPC schemes.
Recall that our design objectives are to maintain the SINR requirements for all users expressed in (\ref{Ch3_equ:SINR_cond}) (whenever possible) while
enhancing the system throughput. Let user $i$ be a supported user (non-supported user) if its SINR greater (less) than its target SINR,
which occurs as its effective interference $R_i(\mathrm{p}^*)$ is less (greater) than the threshold $R_i^{\sf th}$, respectively. 
Note that a supported user $i$ will have its SINR greater than the target SINR if $\alpha_i >0$. We refer to such a supported user as
a potential user in the following. In fact, any potential user $i$ can reduce its transmit power to enhance the SINRs of other users (since this reduces
the effective interference experienced by other users) as being implied by the results in Lemma~\ref{Ch3_lemma3.3}. 
Therefore, by adjusting $\alpha_i$, each user $i$ can vary its SINR and assist other users in improving their SINRs.

\begin{algorithm}[!t]
\caption{\textsc{HPC Adaptation Algorithm}}
\label{Ch3_alg:gms3}
\begin{algorithmic}[1]
\STATE Initialization:
 \begin{itemize}
\item Set $\mathrm{p}^{(0)}=0$, i.e., $p_i^{(0)}=0, \: \forall i \in \mathcal{M}$.
\item Set $\Delta^{(0)}$ as $\alpha_i^{(0)}=0$ for voice user and $\alpha_i^{(0)}=\alpha_0$ ($\alpha_0 \gg 1$) for data user.
\item Set $\overline{N}^{\ast}=|\overline{\mathcal{M}}^{(0)}|$ and $\Delta^{\ast}=\Delta^{(0)}$.
\end{itemize} 
\STATE Iteration $l$: 
\begin{itemize}
\item Run HPC algorithm until convergence with $\Delta^{(l)}$.
\item If $|\overline{\mathcal{M}}^{(l)}|>\overline{N}^{\ast}$, set $\overline{N}^{\ast}=|\overline{\mathcal{M}}^{(l)}|$ and $\Delta^{\ast}=\Delta^{(l)}$.
\item If $\underline{\mathcal{M}}^{(l)} = \varnothing$ or $\overline{\mathcal{M}}^{(l)} = \varnothing$, then go to step 4.
\item If $\underline{\mathcal{M}}^{(l)} \neq \varnothing$ and $\overline{\mathcal{M}}^{(l)} \neq \varnothing$, then run the \textit{``updating process''} as follows:\\[-0.15cm]
\hrulefill \\[-0.1cm]
\textit{$\%$ Start of updating process $\%$}
\begin{description}
\item[a:] For user $i \in \underline{\mathcal{M}}^{(l)}$, set $\alpha_i^{(l+1)}=\alpha_i^{(l)}$.
\item[b:] For user $i \in \overline{\mathcal{M}}^{(l)}$, set $\alpha_i^{(l+1)}$ so that 

i) $\alpha_i^{(l)} > \alpha_i^{(l+1)} \geq 0$ if $\alpha_i^{(l)}>0$.

ii) $\alpha_i^{(l+1)} = \alpha_i^{(l)}$ if $\alpha_i^{(l)}=0$.
\end{description}
\textit{$\%$ End of updating process $\%$}\\[-0.25cm]
\hrulefill
\end{itemize} 
\STATE Increase $l$ and go back to step 2 until there is no update request for $\Delta^{(l)}$. 
\STATE Set $\Delta := \Delta^{\ast}$ and run the HPC algorithm until convergence.
\end{algorithmic}
\end{algorithm}

We exploit this fact to develop the HPC adaptation algorithm, which is described in Algorithm~\ref{Ch3_alg:gms3}.
Let $\overline{\mathcal{M}}^{(l)}$ and $\underline{\mathcal{M}}^{(l)}$ be the sets of supported and nonsupported users in iteration $l$, respectively.
We use $\overline{N}^{\ast}$ to keep the number of supported users during the course of the algorithm.
In each iteration, 
each user $i$ will run the proposed HPC algorithm and slowly update $\alpha_i$ based on the achieved equilibrium. 
All data users $i$ initially set $\alpha_i$ to be a sufficiently large value so that we reach the first equilibrium 
that favors strong users.
For each nonsupported user $i$, we maintain its parameter $\alpha_i$ to keep its power updating process stable. 
In addition, user $i$ can save its parameter $\alpha_i$ at the instant when the number of nonsupported users decreases and reloads this value later as the global update process terminates (in step 4).
This helps to prevent potential users from reducing their transmit powers unnecessarily. 
Discussion about the general \textit{``updating process''} is given in Algorithm~\ref{Ch3_alg:gms3} whose specific design will be 
presented in Algorithm~\ref{Ch3_alg:gms4}. In fact, detailed design of the \textit{``updating process''} can be done to meet specific design objectives. However, if 
the \textit{``updating process''} is designed in such a way that all parameters $\alpha_i$ of supported users tend to zero
if all nonsupported users cannot be saved, then the proposed HPC adaptation algorithm can achieve the following desirable performance. 

\begin{theorem}
\label{Ch3_thm3.3}
Let $\overline{N}_{\sf HPC}$ and $\overline{N}_{\sf TPC}$ be the numbers of supported users due to the proposed HPC adaptation algorithm (i.e., Algorithm~\ref{Ch3_alg:gms3}) and the TPC algorithm, respectively. Then, we have the following.
\begin{enumerate}
\item $\overline{N}_{\sf HPC} \geq \overline{N}_{\sf TPC}$.
\item If the network is feasible (i.e., all SINR requirements can be fulfilled by the TPC algorithm), then all users achieve their target SINRs by 
using HPC adaptation algorithm, and there exist feasible users who achieve SINRs higher than the target values under the HPC adaptation algorithm.
\end{enumerate} 
\end{theorem}

\begin{proof}
The proof is given in Appendix~\ref{Ch3_apdx.thm3.3}.
\end{proof}

This theorem implies that the proposed HPC adaptation algorithm achieves better performance than the traditional TPC algorithm for both
infeasible and feasible systems. Specifically, the HPC adaptation algorithm can support at least the same
number of users and achieve higher SINRs and, therefore, higher total throughput compared with those due to the TPC
algorithm.

\subsection{Practical Algorithm Design}
\label{Ch3_sec:Prac_Agrm}
We propose a practical design of the \textit{``updating process''} for HPC adaptation presented in Algorithm \ref{Ch3_alg:gms3}. In fact, this process attempts to convert nonsupported users into as many supported
users as possible by decreasing effective interference levels of
nonsupported users below threshold value $R_i^{\sf th}$.
Note that the interference received by a user is contributed by all other users
in the network where intracell interference forms a significant
portion of the total interference.
Exploiting this fact, we propose to update $\alpha_i$ in $\Delta$ in a process consisting of both local and global updates, i.e., locally varying $\alpha_i$ within each cell and globally varying $\alpha_i$ when there are cells whose users are not supported after performing the local updating process. 

Let $\mathcal{M}_k$ be the set of users associated with BS $\mathit{k}$ at the equilibrium, i.e., $\mathcal{M}_k= \left\lbrace j | b_j \equiv k \right\rbrace$. 
In addition, let $\overline{\mathcal{M}}_k$ and $\underline{\mathcal{M}}_k$ be the sets of supported and nonsupported users in cell $k$, respectively. Therefore, we have $\overline{\mathcal{M}}_k \cup \underline{\mathcal{M}}_k = \mathcal{M}_k$.
Let $\Delta_k= \left\lbrace \alpha_i | i \in \mathcal{M}_k \right\rbrace $. For given values  of $\alpha_i$ in $\Delta_k$, if the SINRs of users in $\mathcal{M}_k$ are all satisfied or unsatisfied, then we do not change $\alpha_j$ ($j \in \mathcal{M}_k$) further in the local updating process. If there are both the supported and nonsupported users in a particular cell, we propose to reduce the power of the potential users to improve the SINRs of nonsupported users. In particular, we can calculate the reduction ratio of effective interference required for the weak user $i$ ($i \in \underline{\mathcal{M}}_k$) as follows:
\begin{equation}
\label{Ch3_equ:reduce_ratio}
\beta_i=\dfrac{R_i(\mathrm{p})-R_i^{\sf th}}{R_i(\mathrm{p})}, \; i \in \underline{\mathcal{M}}_k.
\end{equation}

Consultation of $R_i(\mathrm{p})$ in (\ref{Ch3_Rp}) suggests that to assist user $i$ ($i \in \underline{S}_k$) in reducing $R_i(\mathrm{p})$ by a factor $\beta_i$, all potential users must reduce their transmit power by a factor at least $\beta_i$. In addition, we need to assist the weakest user who has the highest effective interference in achieving its target SINR. The reduction ratio of effective interference for this weakest user can be expressed as
\begin{equation}
\label{Ch3_equ:betamax}
\beta_{\sf max}^k=\mathrm{max}_{i \in \underline{\mathcal{M}}_k} \beta_i.
\end{equation}
On the other hand, each potential user $j$ must maintain its SINR requirement, i.e., its new SINR must not be less than its target SINR. Hence, its
 new transmit power level must be not less than $R_j(\mathrm{p}) \hat{\gamma}_j$. Therefore, the expected transmit power of potential user $j$ 
 in cell $k$ can be written as
\begin{equation}
\label{Ch3_equ:exp_p}
p_j^{\sf exp}=\mathrm{max} \left\lbrace p_j(1-\beta_{\sf max}^k), R_j(\mathrm{p}) \hat{\gamma}_j \right\rbrace ,\; j \in \overline{\mathcal{M}}_k.
\end{equation}
Using this result, we can calculate the parameter $\alpha_j$ for potential user $j$ after using (\ref{Ch3_hpcrule}) and performing some manipulations as 
\begin{equation}
\label{Ch3_equ:alpha_strg}
\alpha_j = g(p_j^{\sf exp})= \dfrac{p_j^{\sf exp}-R_j(\mathrm{p}) \hat{\gamma}_j}{\xi_j R_j(\mathrm{p})^{\frac{x}{x-1}}-p_j^{\sf exp}}, \: j \in \overline{\mathcal{M}}_k.
\end{equation}
This local updating process for each cell k is run until all users meet their SINR requirements or there is no potential
user in this cell. A cell is called unsatisfied if it still contains nonsupported users (and there is no potential user). If
there exist unsatisfied cells, the BS of each unsatisfied cell will send a ``warning message'' through the backhaul to seek
global assistance. Then, the global updating process will be performed as soon as a ``warning message'' is broadcast. In this
global updating process, all potential users $i$ will 
update their parameters $\alpha_i$ to assist the unsatisfied cells.
The global updating process can be run in parallel with the local updating process, i.e., 
a potential user $i$ may reduce $\alpha_i$ twice in one updating step. 

We propose a practical \textit{``updating process,''} which is described in Algorithm \ref{Ch3_alg:gms4}.
Let $\overline{\mathcal{M}}_k^{(l)}$ and $\underline{\mathcal{M}}_k^{(l)}$ be the sets of supported and non-supported users in cell $k$ and 
iteration $l$, respectively. In this proposed algorithm, each user attempts to assist other users in achieving their SINR targets while achieving high SINR for itself. After obtaining the equilibrium in each iteration, each user $i$ whose SINR is lower than
the target $\hat{\gamma}_i$ at the equilibrium will maintain its parameter $\alpha_i$. 
The supported or potential users who have their $\alpha_i$ greater than zero will update their $\alpha_i$ according to (\ref{Ch3_equ:alpha_strg}) or scale down their $\alpha_i$ further by a scaling factor of $\varsigma$ ($\varsigma > 1$) if they receive ``warning messages.'' A ``warning message'' from a particular cell includes the number of nonsupported users of the cell. 
Hence, each user knows whether its reduction of $\alpha_i$ is helpful or not.
Each user stores its parameter $\alpha_i$ at the instant when the number of nonsupported users decreases and reloads this value after the global update process terminates. The global process will terminate as all users attain their target SINRs or all the supported users reduce their parameters $\alpha_i$ to zero. 

\begin{algorithm}[!t]
\caption{\textsc{Exemplified Design of Updating Process}}
\label{Ch3_alg:gms4}
\begin{algorithmic}[1]
\STATE Run the following \textit{``updating process''} at each cell $\mathit{k}$ (${k=1,2,...,K}$) in iteration $l$ (only for data users):
\begin{itemize}
\item If $\underline{\mathcal{M}}_k^{(l)} \neq \varnothing$, then consider the following cases: 
\begin{description}
\item[a:] For $i \in \underline{\mathcal{M}}_k^{(l)}$, set $\alpha_i^{(l+1)}=\alpha_i^{(l)}$.
\item[b:] For $i \in \overline{\mathcal{M}}_k^{(l)}$, perform the following:
\begin{list}{}{}
\item[b1:] Set $\alpha_i^{(l+1)}=g(p_j^{\sf exp})$ if BS $k$ does not receive any ``warning massage.''
\item[b2:] Set $\alpha_i^{(l+1)}=g(p_j^{\sf exp})/\varsigma$ if BS $k$ receives a ``warning massage.''
\end{list}
\item[c:] If $\alpha_i^{(l)}=0,\:\: \forall i \in \overline{\mathcal{M}}_k^{(l)}$ then BS $k$ sends a ``warning message'' to other BSs in the network.
\end{description}
\item If $\underline{\mathcal{M}}_k^{(l)} = \varnothing$, consider the following cases: 
\begin{description}
\item[a:] If BS $k$ receives a ``warning message,'' then set $\alpha_i^{(l+1)}=\alpha_i^{(l)}/\varsigma$ for $i \in \overline{\mathcal{M}}_k^{(l)}$.
\item[b:] If BS $k$ does not receive any ``warning message,'' then set $\alpha_i^{(l+1)}=\alpha_i^{(l)}$ for $i \in \overline{\mathcal{M}}_k^{(l)}$.
\end{description}
\end{itemize} 
\end{algorithmic}
\end{algorithm}
\subsection{Application to Two-Tier Macrocell-Femtocell Networks}
\label{Ch3_sec:Aplc_FC}
\subsubsection{Two-tier Macrocell-Femtocell Networks}
\label{Ch3_subsec:femto model}
We consider a two-tier network where $M_f$ FUE devices served by $K$ FBSs
are underlaid with one macrocell serving $M_m$ MUE devices. 
We denote the sets of MUE and FUE devices by $\mathcal{M}_m \triangleq \left\lbrace 1,...,M_m\right\rbrace $ and $\mathcal{M}_f \triangleq \left\lbrace M_m+1,...,M_m+M_f \right\rbrace $, respectively, and the set of all users is $\mathcal{M} = \mathcal{M}_m \cup \mathcal{M}_f$. 
We assume that FUE devices are only served by the BS of its cell.
However, MUE devices are allowed to connect to any BS in the network.
All users in this heterogeneous network are assumed to employ the proposed HPC scheme.
In addition, by allowing users to adjust appropriate PC parameters, we can
provide useful mechanisms for FBSs to control the spectrum access from associated MUE devices, which enables the hybrid
access design for the underlying two-tier network. 

\subsubsection{Control MUE devices' Access at Femtocells}
\label{Ch3_subsec:ctrlacc}
Recall that
MUE devices are allowed to connect to any FBS to enhance
their performance and to reduce the interference to other users
in the network. However, MUE devices' connections to a particular femtocell without appropriate control may significantly
degrade the throughput of FUE devices in the underlying femtocell, which is not desirable from the FUE devices' viewpoint.
To resolve this issue, we propose a hybrid access strategy to
control the access of MUE devices. Toward this end, we assume
that each FUE $i$ has two target SINRs, e.g., $\hat{\gamma}_i^{[1]}$ and $\hat{\gamma}_i^{[2]}$ where $\hat{\gamma}_i^{[1]}<\hat{\gamma}_i^{[2]}$. Then, FUE devices can attempt to reach the higher target SINRs as long as the network
condition allows. Otherwise, FUE devices seek to maintain at least the lower target SINRs when it is possible.

Each FUE $i$ can adaptively track the higher target SINR by taking the following procedure. It employs 
the proposed HPC adaptation algorithm to achieve the lower target SINR $\hat{\gamma}_i^{[1]}$ in the first phase.
If this is successfully fulfilled,
it attempts to achieve the higher target SINR $\hat{\gamma}_i^{[2]}$ in the second phase.
Toward this end,  
FUE $i$ simply sets parameters $\xi_i$ and $R_i^{\sf th}$ by using the higher target SINR.
Then, it runs the proposed HPC adaptation algorithm again. 
As FUE devices at a particular femtocell attempt to achieve higher target SINR, MUE devices associated with the underlying femtocell have to decrease their transmit
powers and, therefore, achieve lower SINRs.
Hence, the proposed HPC adaptation algorithm enables us to attain flexible spectrum sharing between FUE and MUE devices
at each femtocell. 

\section{Numerical Results}
\label{Ch3_result}

\begin{figure}[!t]
        \centering
\includegraphics[width=0.5 \textwidth]{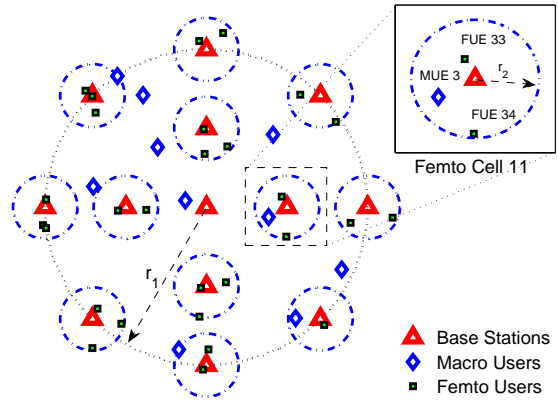}
\caption{Simulated two-tier macrocell-femtocell network.}
\label{Ch3_fig:system}
\end{figure}

We present illustrative numerical results to demonstrate the performance of the proposed algorithms. The network setting and user placement for our simulations are illustrated in Fig. \ref{Ch3_fig:system}, where MUE and FUE devices are randomly located inside circles of radii of $r_1 = 1000\:m$ and $r_2 =50\:m$, respectively. 
Assume that there are heavy walls at the boundaries of the femtocells.
We fix $M_m=10$ except for the results given in Fig.~\ref{Ch3_fig:BAS} and randomly choose the number of FUE devices in each femtocell from 1 to 3. Then, eight users of either tier are set as voice users randomly. The channel power gain $h_{ij}$ is chosen according to the path loss $L_{ij}=A_i \mathrm{log}_{10}(d_{ij})+B_i+C\mathrm{log}_{10}(\frac{f_c}{5})+ W_l n_{ij}$, where $d_{ij}$ is the distance from user $j$ to BS $i$; $(A_i,B_i)$ are set as $(36,40)$ and $(35,35)$ for MBS and FBSs, respectively; $C=20$, $f_c=2.5$ GHz; $W_l$ is the wall-loss value; $n_{ij}$ is the number of walls between BS $i$ and user $j$. Other parameters are set as follows (unless stated otherwise): $x=0.5$,  processing gain $G=128$, maximum power $\bar{p}_j = 0.01 \: W, \: \forall j \in \mathcal{M},$ noise power $\eta_i=10^{-13} \: W$, $\forall \: i \in \mathcal{K}$. We set $n_{ij}$ equal to the number of cell boundaries that the corresponding signal traverses and the
wall-loss value $W_l=12$dB (except for Fig.~\ref{Ch3_fig:ASE_WL}). The SINR presented in each figure is either in linear or decibel scale, which is
stated below each relevant figure in this section.

In Algorithm~\ref{Ch3_alg:gms3}, we have to run the HPC power updates given in (\ref{Ch3_xi2}) in the inner loop before updating parameters $\alpha_i^{(l)}$ in each iteration of the outer loop 
if there exists any users who cannot achieve their target SINRs (i.e., the set $\underline{\mathcal{M}}^{(l)} \neq \emptyset$). 
We can limit the number of HPC power updates (iterations) to smaller than a predetermined limit instead
of waiting for HPC power updates to fully converge to improve the convergence time of Algorithm~\ref{Ch3_alg:gms3}. 
Moreover, we can also improve the convergence time of Algorithm~\ref{Ch3_alg:gms3} by appropriately setting scaling factor $\varsigma$ in Algorithm~\ref{Ch3_alg:gms4}.
In fact, using larger scaling factor $\varsigma$ enables us to relieve the network congestion more quickly but
may degrade the throughput performance. To study the impacts of these possible parameter settings on the convergence time
 of the proposed algorithm, we show the number of required
iterations versus different values of $\varsigma$ in Fig.~\ref{Ch3_fig:num_irt_scf}, where the number of HPC power updates (iterations)
are limited to three, four, five, and ``no limit'', which are indicated as ``limit inner loop'' in this figure. The average total throughput
remains almost unchanged for the considered values of $\varsigma$ and number of HPC power updates.
We do not present the throughput variation here due to
the space constraint. Fig.~\ref{Ch3_fig:num_irt_scf} confirms that we can significantly improve
the convergence time of Algorithm~\ref{Ch3_alg:gms3} by choosing $\varsigma$ sufficiently large (e.g., $\varsigma>10$) and a small number of
HPC power updates (e.g., three or five HPC power updates are sufficient).
Learning from these studies, we set $\varsigma=16$ and the number of 
 HPC power updates equal to five to obtain all other results presented in the following.

\begin{figure}[!t]
        \centering
\includegraphics[width=0.7 \textwidth]{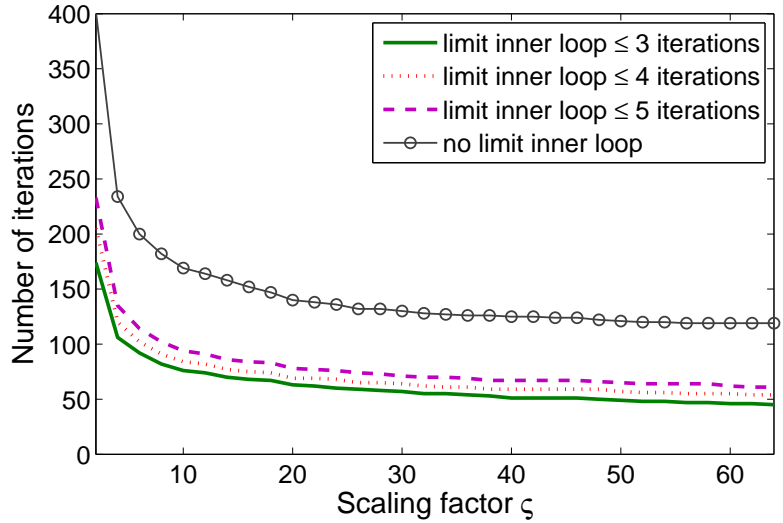}
\caption{Number of iterations versus scaling factor $\varsigma$.}
\label{Ch3_fig:num_irt_scf}
\end{figure}

\begin{figure}[!t]
        \centering
\includegraphics[width=0.7 \textwidth]{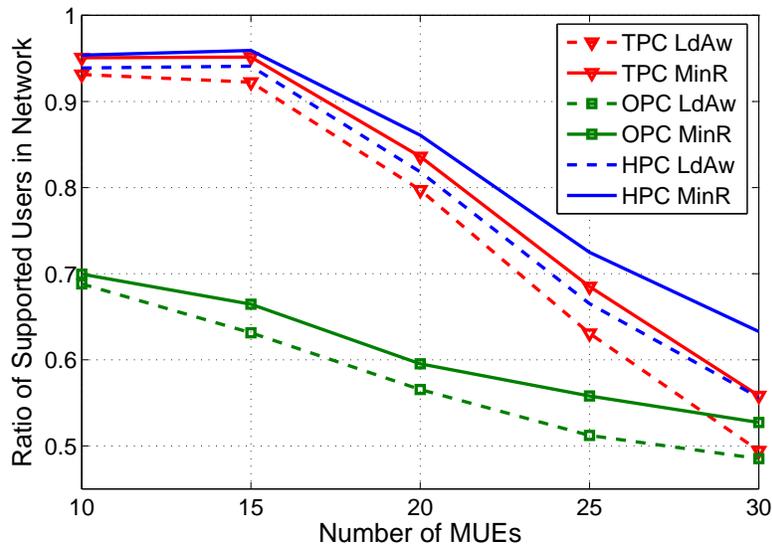}
\caption{Ratio between the number of supported users and the number of all users under different BSA schemes for target SINR equal $8$ (in linear scale).}
\label{Ch3_fig:BAS}
\end{figure}

In Fig.~\ref{Ch3_fig:BAS}, we show ratios between the number of supported users and the total number of users for the following schemes: dynamic BSA and fixed load-aware BSA strategies (i.e., BSA in Algorithm \ref{Ch3_alg:gms2} and Algorithm \ref{Ch3_alg:gms1}, respectively) with different PC strategies, namely, TPC, OPC, and 
 HPC adaptation algorithm (i.e., Algorithms \ref{Ch3_alg:gms3} with the updating process given in Algorithm \ref{Ch3_alg:gms4}). In this figure, we increase the number of MUE devices from 10 to 30.
 For the fixed BSA (i.e., in the decomposed BSA and PC design), we use the average channel power gains based on the path loss only. For dynamic BSA 
(i.e., in the joint BSA and PC design), we determine BSA results based on instantaneous channel power gains comprising both path loss and Rayleigh fading, which is 
represented by an exponentially distributed random variable with the mean value of 1.
Results corresponding to fixed
and dynamic BSAs are indicated by ``LdAw'' and ``MinR'' together with the corresponding PC schemes (i.e., TPC, OPC, and HPC) in this figure, respectively.
As shown in this figure, the proposed HPC scheme outperforms existing TPC and OPC schemes under both joint and decomposed BSA and PC designs.
Moreover, the dynamic BSA proposed in Algorithm \ref{Ch3_alg:gms2} achieves better performance that the fixed BSA for all three PC strategies. 
In addition, the performance gap between joint and decomposed BSA and PC designs becomes quite significant when the number of MUE devices is sufficiently large.
Therefore, the joint BSA and PC proposed in Algorithm \ref{Ch3_alg:gms2} will be employed to obtain other results in the section.

\begin{figure}[!t]
        \centering
\includegraphics[width=0.7 \textwidth]{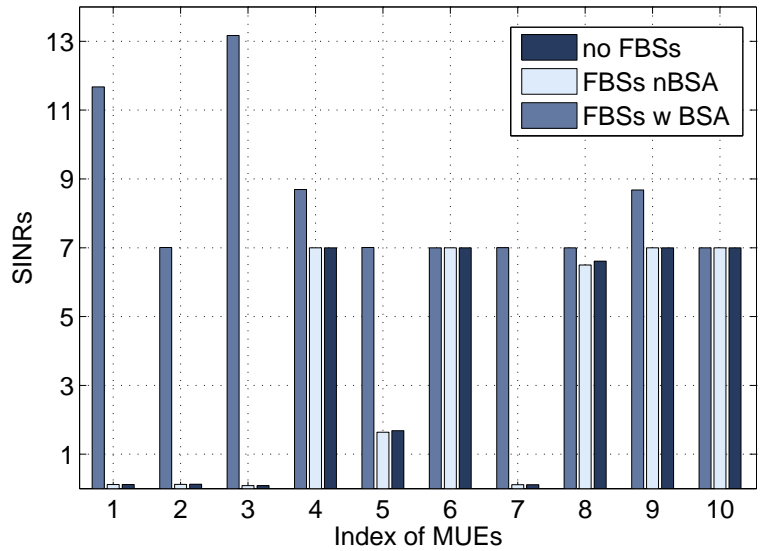}
\caption{SINRs of MUE devices with and without FBS association (closed versus hybrid access) for target SINRs of 7 (in linear scale).
The ``index of MUE devices'' in the horizontal axis indicates individual MUE devices in the simulation.}
\label{Ch3_fig:MSINR}
\end{figure}

In Fig. \ref{Ch3_fig:MSINR}, we present the SINRs of MUE devices for three cases: 1) There is no femtocell; 2) there are femtocells, but MUE devices are not allowed to connect with FBSs (i.e., closed access); 3) there are femtocells, and MUE devices are allowed to connect with FBSs (i.e., hybrid access). Results for these three cases are indicated as ``no FBSs'', ``FBSs nBSA'', and ``FBSs wBSA'' in this figure, respectively. The target SINRs of all users are set to 7. It can be observed that the SINRs of MUE devices do not decrease much with the deployment of femtocells. This is due to the fact that the large wall loss prevents femtocells from creating large cross-tier interference to MUE devices. Moreover, by allowing MUE devices to connect with nearby FBSs, we can significantly increase their SINRs.  Moreover, as shown in Fig.~\ref{Ch3_fig:MSINR}, the average SINRs per
one FUE for the cases with and without FBS association are approximately 1267 and 39, respectively. These results confirm
the great benefits of the hybrid access over the closed access in the two-tier macrocell–femtocell network in improving the
SINRs of both MUE and FUE devices.

The operating point achieved by Algorithm~\ref{Ch3_alg:gms3} may not lie on the Pareto-optimal boundary of the capacity region.
Therefore, there can exist a different point that strictly dominates this operating point. To understand the throughput gap between 
our solution and a ``throughput optimal'' solution, we have developed an exhaustive search algorithm that aims to reach a more 
throughput-efficient point on the Pareto-optimal boundary of the capacity region.  
This exhaustive search algorithm performs two key steps in each iteration, which are described in the following. \textit{In the first step}, we consider all
possible pairs of users $(i,j)$ and scale their SINRs as $\Gamma_i^{(k)}=\Gamma_i^{(k)}/\theta$ and $\Gamma_j^{(k)}=\theta\Gamma_j^{(k)}$,
where scaling factor $\theta>0$ is searched so that the maximum throughput (i.e., spectral efficiency) can be achieved and 
all supported users are still supported after the SINR updates for the users $(i,j)$.
The operation in this step aims to increase the total throughput while maintaining the set of supported users. 
\textit{In the second step}, we find exactly one user where the total throughput can be increased the most as we
increase the SINR of this user while keeping infeasibility of the resulting set of SINRs. Then, the SINR of the chosen user
is set to the maximum value. These two steps are performed in each iteration of the search algorithm until convergence.

In Fig.~\ref{Ch3_fig:max_utility}, we compare the throughput achieved by Algorithm~\ref{Ch3_alg:gms3} and the exhaustive search algorithm.
As shown in this figure, the exhaustive search algorithm achieves higher throughput than that achieved by Algorithm~\ref{Ch3_alg:gms3}.
Note, however, that the exhaustive search algorithm relies on the operating point
achieved by Algorithm~\ref{Ch3_alg:gms3} for initialization.
In addition, it is the centralized algorithm and is based on exhaustive search.
Fig.~\ref{Ch3_fig:max_utility} also reveals that the throughput gap between the two algorithms is pretty small for low target SINR values where the
QoS constraints of all users can be supported. This confirms the efficacy of our proposed low-complexity algorithm.

\begin{figure}[!t]
        \centering
\includegraphics[width=0.7 \textwidth]{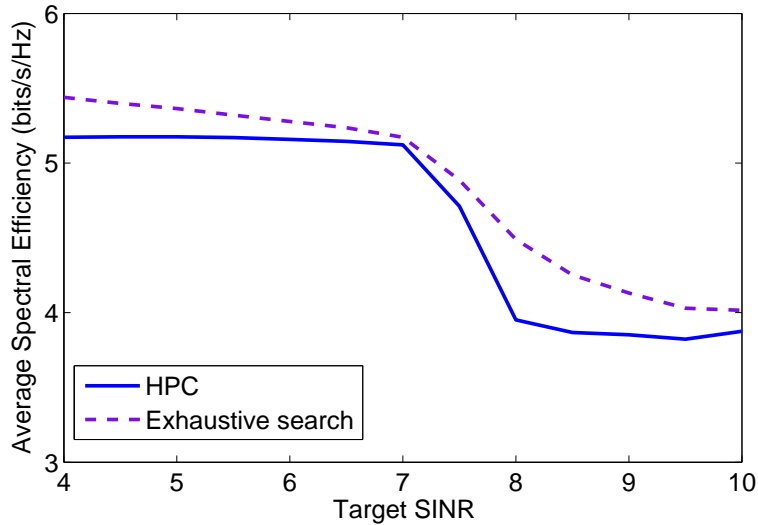}
\caption{Comparison of the throughput achieved by Algorithm~\ref{Ch3_alg:gms3} and the exhaustive search algorithm. (The target SINR is in linear scale.)}
\label{Ch3_fig:max_utility}
\end{figure}

\begin{figure}[!t]
        \centering
\includegraphics[width=0.7 \textwidth]{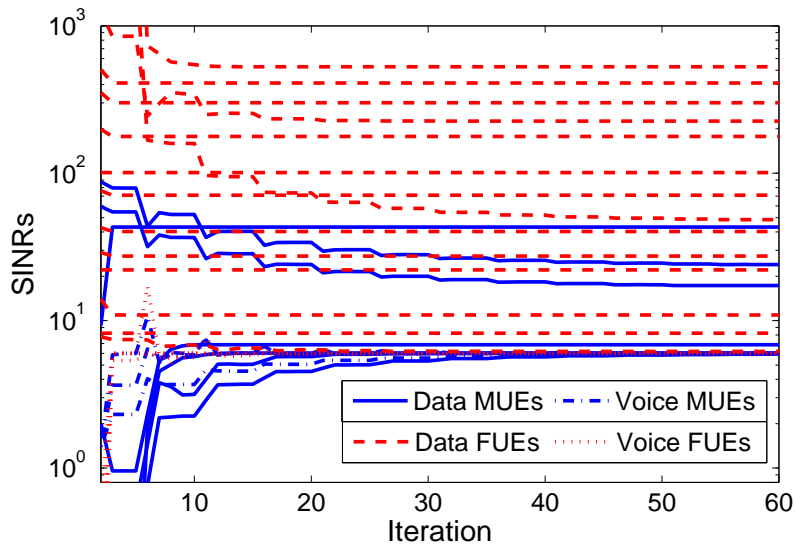}
\caption{SINRs of MUE and FUE devices for target SINRs of 6 (in linear scale).}
\label{Ch3_fig:SINR_fsbl}
\end{figure}

\begin{figure}[!t]
        \centering
\includegraphics[width=0.7 \textwidth]{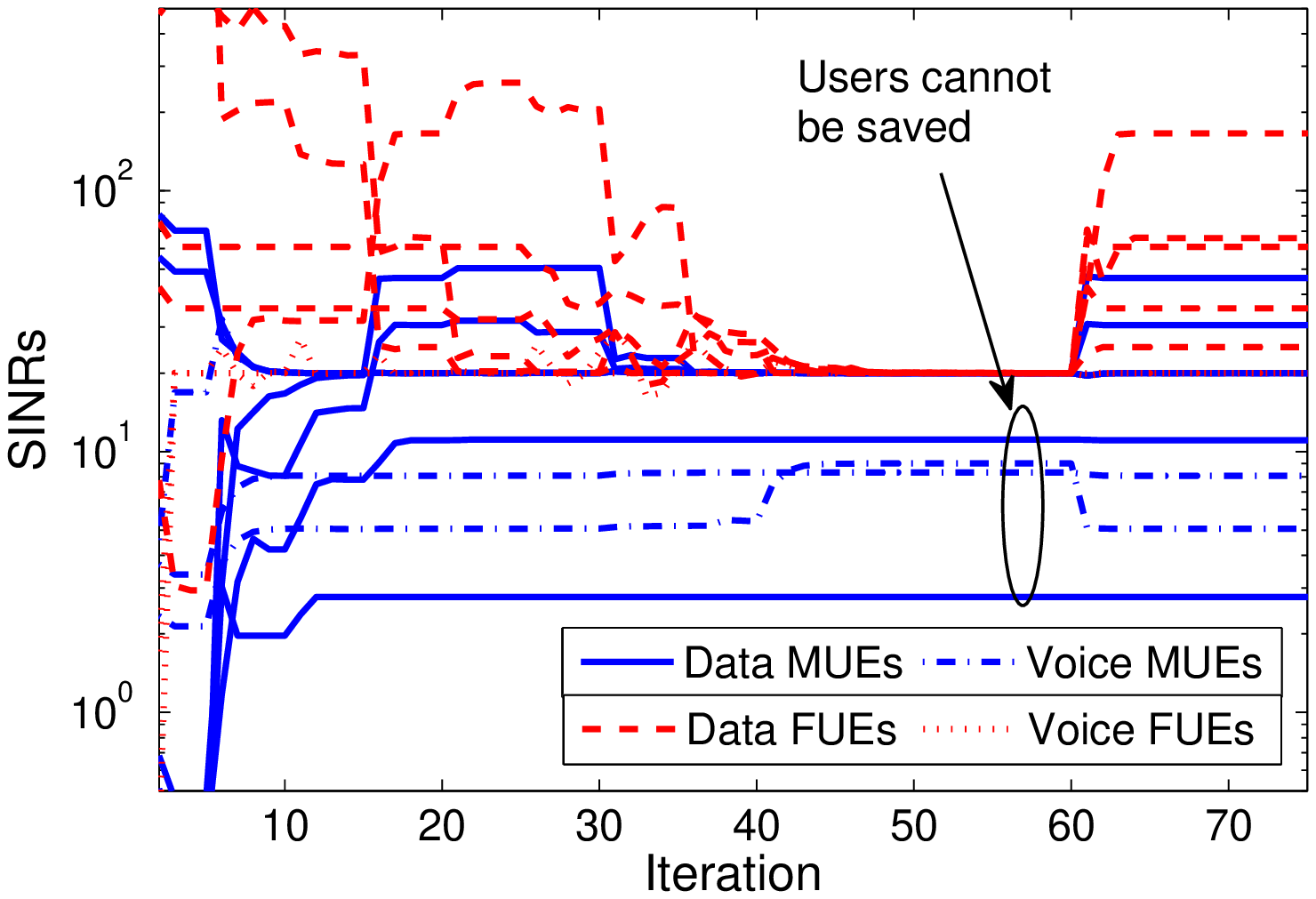}
\caption{SINRs of all users for target SINRs of 20 (in linear scale).}
\label{Ch3_fig:SINR_infsbl}
\end{figure}

Figs.~\ref{Ch3_fig:SINR_fsbl} and \ref{Ch3_fig:SINR_infsbl} show the convergence of the HPC adaptation algorithm when we set $\hat{\gamma}_i=6$ and $20$ for all users, respectively. Fig.~\ref{Ch3_fig:SINR_fsbl} confirms that the HPC adaptation algorithm converges to an equilibrium for which some users exactly achieve their target
SINRs while others attain SINRs higher than their target values.
This set of results corresponds to the feasible system where all SINR requirements can be supported. This figure shows
that the proposed HPC adaptation algorithm not only attains all SINR requirements but enables strong users to achieve
high throughput performance as well.
The results in Fig.~\ref{Ch3_fig:SINR_infsbl} correspond to an infeasible system where we cannot attain all SINR requirements.
The dynamics of the HPC adaptation algorithm demonstrated in this figure can be interpreted as follows.
Some potential users $i$ attempt to reduce their parameters $\alpha_i$ to zero to assist nonsupported users to achieve their target SINRs initially. 
As soon as potential users recognize that there are still nonsupported users, they reload the last $\alpha^{\ast}$ values (i.e., in step 4 of Algorithm~\ref{Ch3_alg:gms3}),
 which result in improvements for their SINRs.

\begin{figure}[!t]
        \centering
\includegraphics[width=0.7 \textwidth]{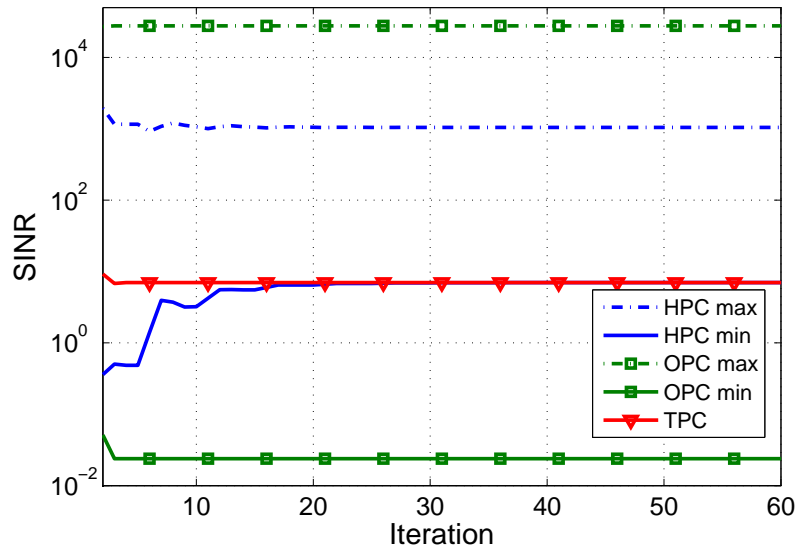}
\caption{Maximum and minimum SINRs achieved by OPC, TPC and HPC adaptation schemes for target SINRs of 7 (in linear scale).}
\label{Ch3_fig:SINRcomp}
\end{figure}

In Fig.~\ref{Ch3_fig:SINRcomp}, we show the fairness achieved by the OPC, TPC and HPC adaptation algorithms by showing the highest and lowest SINRs of all users for different strategies (indicated as X min and X max in this figure, where X represents the corresponding PC strategy). 
In the TPC scheme, all users reach the same target SINR.
In contrast, the OPC scheme results in very different SINRs at the equilibrium where the strongest
user attains a very high SINR, whereas the weak users achieve SINRs that are very close to zero.
Therefore, the OPC scheme tends to allocate more resources to strong users to achieve
high overall throughput at the cost of significantly penalizing weak users.
This figure shows that our proposed HPC adaptation algorithm allows the weak users to achieve their target
SINRs while enabling strong users to settle at higher SINRs.
Therefore, our proposed algorithm strikes to balance between achieving required QoS while exploiting multiuser diversity to
enhance the system throughput, which would be very desirable for data-driven wireless systems.

\begin{figure}[!t]
        \centering
\includegraphics[width=0.7 \textwidth]{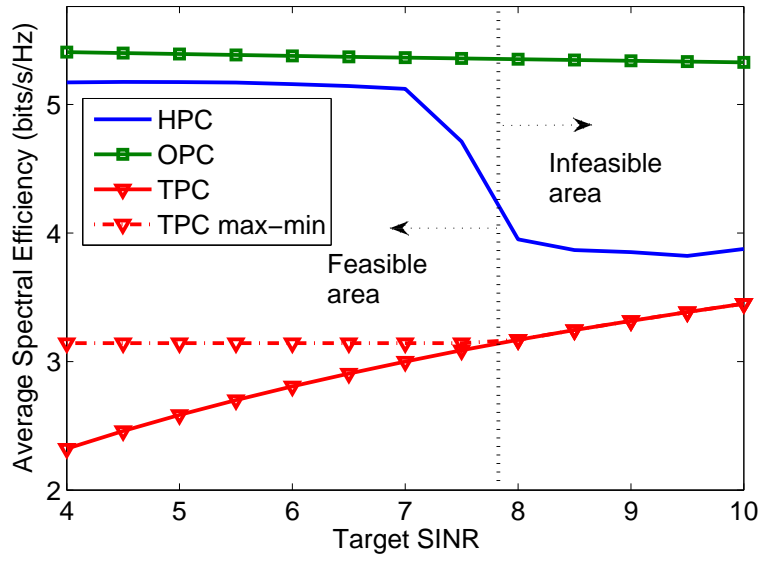}
\caption{Average spectral efficiency versus target SINR (in linear scale).}
\label{Ch3_fig:thrvSINR}
\end{figure}

\begin{figure}[!t]
        \centering
\includegraphics[width=0.7\textwidth]{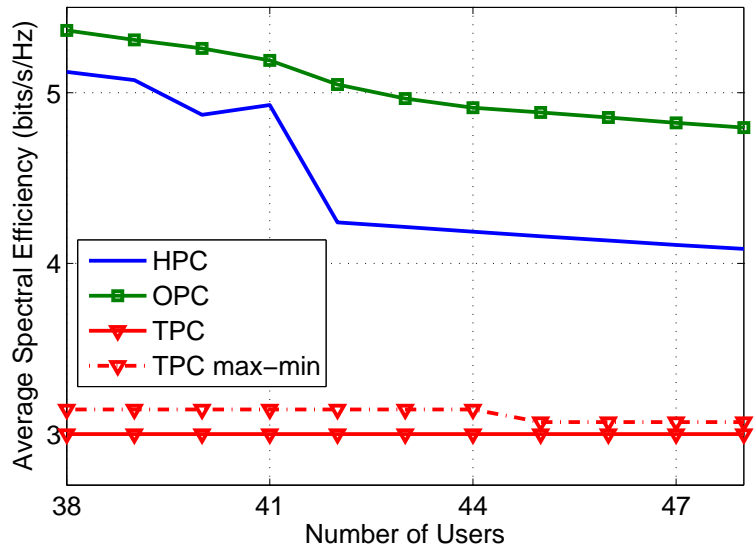}
\caption{Average spectral efficiency versus number of all users for target SINRs of 7 (in linear scale).}
\label{Ch3_fig:thrvFUEs}
\end{figure}

Figs.~\ref{Ch3_fig:thrvSINR} and \ref{Ch3_fig:thrvFUEs} show the average spectral efficiency (i.e.,
throughput) achieved by different schemes versus the target SINR and the number of all users, respectively. 
Here, the spectral efficiency of user $i$ is calculated as $\log_2(1+\Gamma_i)$ (in bits per second per hertz).
To obtain the results in these two figures, we randomly generate user locations and obtain results for this fixed topology.
For the results in Fig.~\ref{Ch3_fig:thrvFUEs}, we sequentially add more FUE devices to obtain different points on each curve. For
reference, we also present the average spectral efficiency of a TPC max–min scheme in which we slowly increase the target
SINR for all users as long as the system is still feasible under the TPC scheme. (This scheme is denoted as TPC max–min
in these two figures.) As shown in these figures, our HPC adaptation algorithm attains higher average spectral efficiency
than both traditional TPC and TPC max–min schemes but lower average spectral efficiency than the OPC algorithm. 
In particular, when the network is lowly loaded ($\hat{\gamma}_i < 8$ in Fig.~\ref{Ch3_fig:thrvSINR}), our proposed algorithm attains much higher 
spectral efficiency than the traditional TPC and TPC max-min schemes. Moreover, when the network load is higher ($\hat{\gamma}_i \geq 8$), the gaps between our proposed algorithm and the two TPC schemes become smaller. This is because Algorithm \ref{Ch3_alg:gms3} attempts to maintain the target SINRs for all users as the TPC schemes. 

\begin{figure}[!t]
        \centering
\includegraphics[width=0.7 \textwidth]{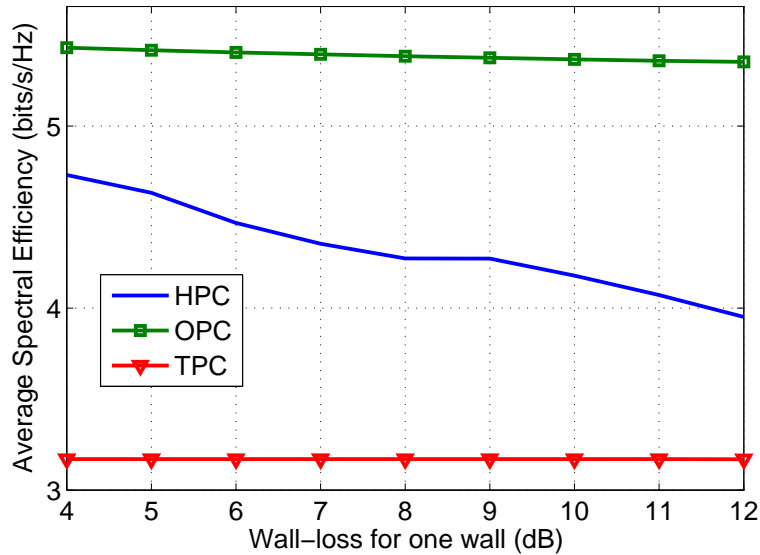}
\caption{Average spectral efficiency versus the wall-loss value for target SINRs of 8 (in linear scale).}
\label{Ch3_fig:ASE_WL}
\end{figure}

\begin{figure}[!t]
        \centering
\includegraphics[width=0.7 \textwidth]{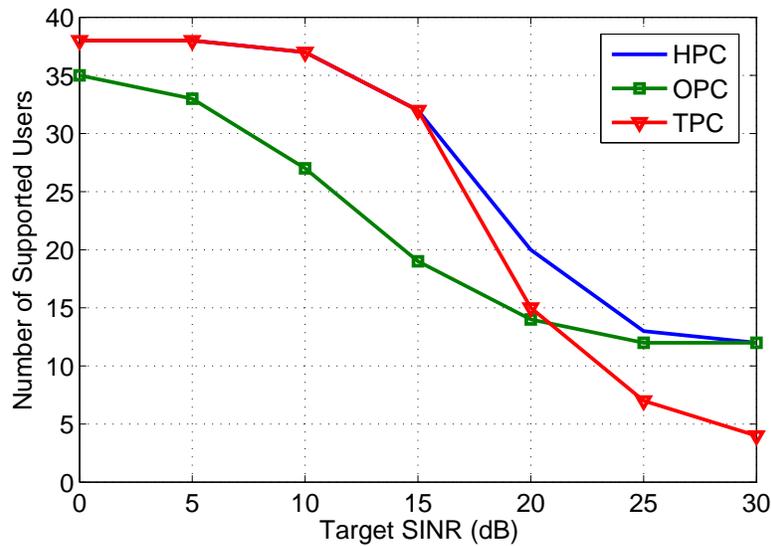}
\caption{Number of supported users in HPC, TPC and OPC schemes versus target SINR (in decibel scale).}
\label{Ch3_fig:struvSINR}
\end{figure}

In Fig.~\ref{Ch3_fig:ASE_WL}, we study the impact of wall-loss value $W_l$ (therefore, the path loss $L_{ij}$) on the average spectral efficiency of different schemes. This figure shows that the average spectral efficiency of the TPC scheme remains unchanged, whereas
that of the OPC scheme slightly decreases with the increasing wall-loss value $W_l$. This is because the system load under consideration is light; hence, the number of users that can be supported by the TPC scheme remains the same for different values of $W_l$. Moreover, the OPC scheme allocates most capacity to a few strong users who achieve very high SINRs and
contribute a significant fraction of the total throughput. As a result, the throughput contribution from these strong users only
slightly decreases with increasing wall-loss value $W_l$ due to the concavity characteristic of the throughput formula. Finally, the spectral efficiency
of proposed HPC scheme decreases with $W_l$; however, it consistently achieves better performance than the TPC scheme.

Fig.~\ref{Ch3_fig:struvSINR} shows the number of supported users for different
schemes. It is shown that our proposed HPC adaptation algorithm can maintain the SINR requirements for the larger
number of users compared with the TPC and OPC algorithms.
In particular, our proposed algorithm can support roughly the same number of users as the TPC algorithm for low values of target SINRs. However, our proposed algorithm performs better than the TPC scheme for higher target SINRs.
In fact, almost all users utilize the maximum power under the TPC scheme at high target SINRs, which prevent all of
them from achieving their target SINRs. On the other hand, in our proposed scheme, nonsupported users tend to use as
small transmit power as possible (i.e., high values of $\alpha_i$).
This enables more users to attain their target SINRs. The figure also shows that the proposed HPC adaptation algorithm performs better than the OPC scheme for all values of the target SINR.

\begin{figure}[!t]
        \centering
\includegraphics[width=0.7 \textwidth]{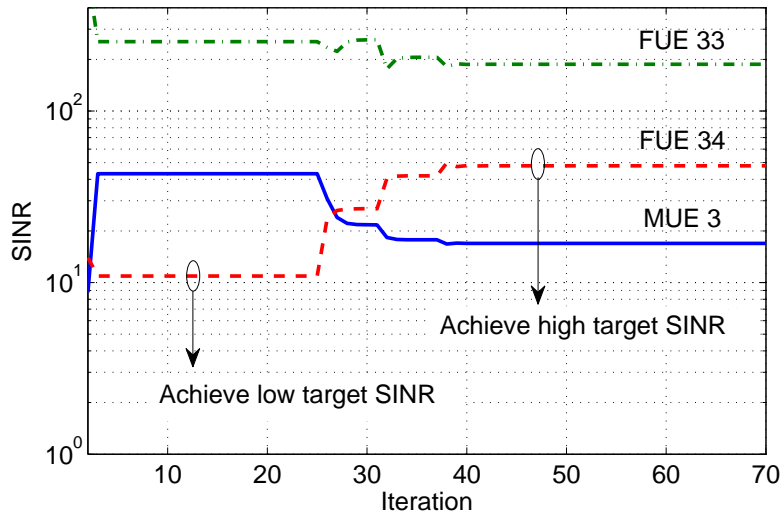}
\caption{SINR of users in femtocell 11 as FUE 34 attempts to achieve the higher target SINR (in linear scale).}
\label{Ch3_fig:hybrid}
\end{figure}

In Fig.~\ref{Ch3_fig:hybrid}, we illustrate the SINR dynamics in femtocell 11 under the association of one MUE (MUE 3).
This femtocell and the corresponding users are indicated in Fig.~\ref{Ch3_fig:system}.
Specifically, FUE devices attempt to achieve the low target SINR initially, which is equal to 5. As this SINR requirement is fulfilled, FUE 34 is not ``happy'' and attempts to achieve a higher target SINR of 50.
This figure shows that the proposed HPC adaptation algorithm allows FUE 34 to achieve this QoS goal where it gradually
increases its power to reach the new target SINR. 
This, in turn, results in the decrease in SINR of MUE 3 and the other FUE (FUE 33).
This figure, therefore, confirms that our proposed algorithm allows FUE devices to achieve flexible spectrum
sharing with associated MUE devices, which is a desirable feature to realize the hybrid access mode at femtocells.
FUE devices 3, 33, and 34 are indicated in Fig.~\ref{Ch3_fig:system}.

\section{Conclusion}
\label{Ch3_ccls}
We have developed a general distributed BSA and PC algorithm and demonstrated how it can be utilized to design an efficient hybrid spectrum access strategy for two-tier macrocell–femtocell networks.
The BSA mechanism integrated
in this algorithm aims to minimize user effective interference.
In addition, a novel HPC scheme has been proposed, which can
be used in this algorithm. We have proved the convergence of
the proposed algorithm by using the 2.s.s. function approach.
We have also developed an efficient adaptation mechanism for
the proposed algorithm so that users can achieve their SINR
requirements whenever possible while we can exploit multiuser
diversity to increase the system throughput. Numerical results
were presented to demonstrate the desirable performance of the
proposed algorithms.

\section{Appendices}

\subsection{Proof of Theorem~\ref{Ch3_thm3.1}}
\label{Ch3_apdx.thm3.1}
According to Lemma~\ref{Ch3_lemma3.1}, a PC algorithm converges if the corresponding puf is 2.s.s.
Hence, we can prove Theorem~\ref{Ch3_thm3.1} by 
showing that the puf $J^{\sf o}(\mathrm{p})=[J_1^{\sf o}(\mathrm{p}),...,J_M^{\sf o}(\mathrm{p})]$ is a 2.s.s. function with respect to $\mathrm{p}$, where $J_i^{\sf o}(\mathrm{p})=J'_i(R_i^{\sf o}(\mathrm{p}))$, and $R_i^{\sf o}(\mathrm{p}) =\mathrm{min}_{\mathit{k} \in D_i} R_i \left( \mathrm{p},\mathit{k}\right), \forall i \in \mathcal{M}$. We will prove $J_i^{'}(R_i^{\sf o}(\mathrm{p}))$ in the following.

Recall that we have assumed $\mathrm{J}'(\mathrm{R}(\mathrm{p}))$ is 2.s.s. with respect to $\mathrm{R}(\mathrm{p})$, where $\mathrm{J}(\mathrm{p})=\mathrm{J}'(\mathrm{R}(\mathrm{p}))$. Therefore, $\mathrm{J}'(\mathrm{R}^{\sf o }(\mathrm{p}))$ is a 2.s.s. function with respect to $\mathrm{R}^{\sf o }(\mathrm{p})$. Consequently, the proof is equivalent to showing the following statement: If $\frac{1}{a}\mathrm{p} \leq \mathrm{p}^{\prime} \leq a \mathrm{p}$, then we have
\begin{equation}
\label{Ch3_eqn:Ri2ss}
\frac{1}{a}R_i^{\sf o }(\mathrm{p}) \leq R_i^{\sf o }(\mathrm{p}^{\prime}) \leq a R_i^{\sf o }(\mathrm{p})\; \forall i \in \mathcal{M}.
\end{equation} 
We will prove (\ref{Ch3_eqn:Ri2ss}) by using contradiction. Let $b_i$ and $b_i^{\prime}$ be the base stations chosen by user $i$ corresponding to $R_i^{\sf o }(\mathrm{p})$ and $R_i^{\sf o }(\mathrm{p}^{\prime})$, respectively. From our assumption $\frac{1}{a}\mathrm{p} \leq \mathrm{p}^{\prime} \leq a \mathrm{p}$ for $a >1$, we have $\frac{1}{a}p_j \leq p_j^{\prime} \leq a p_j, \: \forall j \in \mathcal{M}$. Performing some simple manipulations, we can obtain
\begin{equation}
\label{Ch3_eqn:Ribi}
\begin{array}{c}
(1/a)R_i \left( \mathrm{p},b_i^{\prime}\right) \leq R_i^{\sf o}(\mathrm{p} ^{\prime} ) \leq a R_i(\mathrm{p},b_i^{\prime}), \\
(1/a)R_i^{\sf o} \left( \mathrm{p}\right) \leq R_i(\mathrm{p} ^{\prime},b_i ) \leq a R_i^{\sf o}(\mathrm{p}).
\end{array}
\end{equation}

Now suppose that (\ref{Ch3_eqn:Ri2ss}) is not satisfied if $\frac{1}{a}\mathrm{p} \leq \mathrm{p}^{\prime} \leq a \mathrm{p}$. 
Then, we must have $\frac{1}{a}R_i^{\sf o }(\mathrm{p}) > R_i^{\sf o }(\mathrm{p}^{\prime})$ or $R_i^{\sf o }(\mathrm{p}^{\prime}) > a R_i^{\sf o }(\mathrm{p})$. Let us consider these
possible cases in the following.
\begin{itemize}
\item If $\frac{1}{a}R_i^{\sf o }(\mathrm{p}) > R_i^{\sf o }(\mathrm{p}^{\prime})$, then using the results in (\ref{Ch3_eqn:Ribi}) yields
\begin{equation}
\label{Ch3_eqn:cntr01}
R_i^{\sf o }(\mathrm{p}) > R_i \left( \mathrm{p},b_i^{\prime}\right).
\end{equation}
\item Similarly, if $R_i^{\sf o }(\mathrm{p}^{\prime}) > a R_i^{\sf o }(\mathrm{p})$, then using the results in (\ref{Ch3_eqn:Ribi}) yields 
\begin{equation}
\label{Ch3_eqn:cntr02}
R_i^{\sf o }(\mathrm{p}^{\prime}) >R_i(\mathrm{p}^{\prime},b_i).
\end{equation}
\end{itemize}

Both (\ref{Ch3_eqn:cntr01}) and (\ref{Ch3_eqn:cntr02}) indeed result in contradiction to the definition of $R_i^{\sf o }(\mathrm{p})$ and $R_i^{\sf o }(\mathrm{p}^{\prime})$, which is given in (\ref{Ch3_Rmin}). Hence, we have $\frac{1}{a}R_i^{\sf o }(\mathrm{p}) \leq R_i^{\sf o }(\mathrm{p}^{\prime}) \leq a R_i^{\sf o }(\mathrm{p})$ if $\frac{1}{a}\mathrm{p} \leq \mathrm{p}^{\prime} \leq a \mathrm{p}$, which completes the proof of the theorem.

\subsection{Proof of Lemma~\ref{Ch3_lemma3.2}}
\label{Ch3_apdx.lemma3.2}

We can rewrite the payoff function in (\ref{Ch3_opc1}) as
\begin{equation}
\label{Ch3_mopc1}
U_i^{(1)}(\mathrm{p})=\Gamma_i^x-\lambda_i p_i=\left( \frac{p_i}{R_i\left( \mathrm{p}\right) } \right)^x -\lambda_i p_i.
\end{equation}
Taking the first and second derivatives of this payoff function with respect to $p_i$ yields
\begin{eqnarray}
\frac{\partial U_i^{(1)}}{\partial p_i}&=&\frac{x}{R_i\left( \mathrm{p}\right)}\left( \frac{p_i}{R_i\left( \mathrm{p}\right) }\right)^{x-1}-\lambda_i \label{Ch3_fdev}\\
\frac{\partial^2 U_i^{(1)}}{\partial p_i^2} &=&\frac{x(x-1)}{R_i\left( \mathrm{p}\right)^2}\left( \frac{p_i}{R_i\left( \mathrm{p}\right) }\right)^{x-2}. \label{Ch3_sdev}
\end{eqnarray}
From (\ref{Ch3_sdev}), it can be easily verified that $\frac{\partial^2 U_i^{(1)}}{\partial p_i^2}<0$ for $0<x<1$, which implies that
$U_i^{(1)}(\mathrm{p})$ is a concave function. Therefore, the best response can be obtained by setting the first derivative in (\ref{Ch3_fdev}) to zero.
After some manipulations, we can obtain the best response corresponding to the payoff function $U_i^{(1)}(\mathrm{p})$ as follows:
\begin{equation}
\label{Ch3_power_opc}
p_i = R_i\left( \mathrm{p}\right)^{\frac{x}{x-1}}\left( \lambda_i / x\right)^{\frac{1}{x-1}}=\xi_i R_i\left( \mathrm{p}\right)^{\frac{x}{x-1}} 
\end{equation}
where the second relationship in (\ref{Ch3_power_opc}) holds for $\xi_i=\left( \lambda_i/x\right)^{\frac{1}{x-1}}$.
Moreover, it can be also verified that the best response obtained in (\ref{Ch3_power_opc}) is exactly that corresponding to payoff function $U_i^{(2)}(\mathrm{p})$. Therefore, we have proved the lemma.

\subsection{Proof of Theorem~\ref{Ch3_thm3.2}}
\label{Ch3_apdx.thm3.2}
Here, we will show that puf $\mathrm{I}^H\left( \mathrm{p}\right)=\left[ I_1^H(\mathrm{p}),...,\right.$ $\left. I_M^H(\mathrm{p})\right] $ 
 in (\ref{Ch3_hpcrule}) is 2.s.s. with respect to $\mathrm{p}$. Note that if $\frac{1}{a}\mathrm{p} \leq \mathrm{p}^{\prime} \leq a \mathrm{p}$ for a given $a >1$, then from the definition of $R_i\left( \mathrm{p}\right) $ in (\ref{Ch3_Rp}), we have 
\begin{equation}
\label{Ch3_ineq_Rb}
(1/a)R_i \left( \mathrm{p}\right) \leq R_i(\mathrm{p} ^{\prime} ) \leq a R_i(\mathrm{p}), \; \forall i \in \mathcal{M}.
\end{equation}
Let $I_i^{\prime H} (R_i( \mathrm{p})) = \mathrm{min} \left\lbrace \bar{p}_i, I'_i(R_i(\mathrm{p})) \right\rbrace$, which is a function of $R_i(\mathrm{p})$.
Note that we have $I_i^H\left( \mathrm{p}\right)=I_i^{\prime H} (R_i( \mathrm{p}))$.
To prove that $\mathrm{I}^H(\mathrm{p})$ is 2.s.s. with respect to $\mathrm{p}$, we can equivalently show that $I_i^{\prime H} (R_i( \mathrm{p}))$ is 2.s.s. with respect to $R_i(\mathrm{p})$. 

Since $\frac{x}{x-1}<0$ for $0<x\leq 1/2$, it can be verified that $I'_i(R_i(\mathrm{p}))$ is convex with respect to $R_i(\mathrm{p})>0$ and $\lim_{R_i(\mathrm{p})\rightarrow \left\lbrace 0,\infty\right\rbrace }I'_i(R_i(\mathrm{p}))=\infty$. Hence, there are, at most, two values of $R_i(\mathrm{p})$ that satisfy $I'_i(R_i(\mathrm{p}))=\bar{p}_i$.

If there is no or only one such intersection point, we have $I'_i(R_i(\mathrm{p}))\geq \bar{p}_i$ or $I_i^{\prime H}(R_i(\mathrm{p}))= \bar{p}_i$. For both cases, we have $(1/a) I_i^{\prime H}(R_i(\mathrm{p})) < I_i^{\prime H}(R_i(\mathrm{p}^{\prime} )) < a I_i^{\prime H}(R_i(\mathrm{p}))$ since $a>1$.

If there are two intersection points, let $R_i^l$ and $R_i^u$ be the values of $R_i(\mathrm{p})$ at these two points, where $R_i^l<R_i^u$. 
Moreover, the puf of HPC algorithm in these cases must satisfy
\begin{equation}
\label{Ch3_proof_thm01_2}
I_i^{\prime H}(R_i(\mathrm{p})) = \left\lbrace {\begin{array}{*{20}{ll}}
I'_i(R_i(\mathrm{p})), & \text{if} \: R_i^l \leq R_i( \mathrm{p} ) \leq R_i^u\\
\bar{p}_i, & \text{otherwise}
\end{array}} \right.
\end{equation} 
Let us consider all possible cases in the following.
\begin{itemize}
\item If $\left\lbrace R_i( \mathrm{p} ),R_i( \mathrm{p}^{\prime} )\right\rbrace \not\subset \left[ R_i^l,R_i^u\right] $ then we have $I_i^{\prime H}(R_i(\mathrm{p}))=I_i^{\prime H}(R_i(\mathrm{p}^{\prime}))=\bar{p}_i$. Therefore, it is easy to obtain that $(1/a) I_i^{\prime H}(R_i(\mathrm{p})) < I_i^{\prime H}(R_i(\mathrm{p}^{\prime} )) < a I_i^{\prime H}(R_i(\mathrm{p}))$.
\item If $\left\lbrace R_i( \mathrm{p} ),R_i( \mathrm{p}^{\prime} )\right\rbrace \subset \left[ R_i^l,R_i^u\right]$, then we have $I_i^{\prime H}(R_i(\mathrm{p}))=I'_i(R_i(\mathrm{p}))$ and $I_i^{\prime H}(R_i(\mathrm{p}^{\prime}))=I'_i(R_i(\mathrm{p}^{\prime}))$. Let $r=\frac{R_i( \mathrm{p} )}{R_i( \mathrm{p} ^{\prime})}$, we have 
\begin{equation}
\label{Ch3_r02}
I'_i(R_i(\mathrm{p}))=\dfrac{(1/r)^{\frac{x}{1-x}} \alpha_i \xi_i R_i( \mathrm{p}^{\prime})^{\frac{x}{x-1}}+ r\hat{\gamma}_i R_i( \mathrm{p}^{\prime})}{\alpha_i + 1}.
\end{equation}
In addition, using the results in (\ref{Ch3_ineq_Rb}) yields
$\left\lbrace 1/r,r,(1/r)^{\frac{x}{1-x}}\right\rbrace \subset \left[ 1/a,a\right]$ 
since $0<x \leq 1/2$. Therefore, we have
$(1/a)I'_i(R_i(\mathrm{p}^{\prime})) \leq I'_i(R_i(\mathrm{p})) \leq a I'_i(R_i(\mathrm{p}^{\prime})).$
Moreover, these inequalities hold if $x=1/2$ and both $1/r$ and $r$ are equal to $a$ or $1/a$. Evidently, this cannot be satisfied since $a>1$. Thus, if $0<x \leq 1/2$, we must have
\begin{equation}
\label{Ch3_r05}
(1/a)I_i^{\prime H}(R_i(\mathrm{p})) < I_i^{\prime H}(R_i(\mathrm{p}^{\prime} )) < a I_i^{\prime H}(R_i(\mathrm{p})).
\end{equation}
\item Finally, we have remaining cases where $R_i( \mathrm{p})$ or $R_i( \mathrm{p}^{\prime})$ $\in \left[ R_i^l,R_i^u\right] $. We will only consider the case $R_i( \mathrm{p}) \leq R_i^l \leq R_i( \mathrm{p}^{\prime})\leq R_i^u$ $(\ast)$ since the proofs for other cases can be similarly done. 
Using the results in $(\ast)$ yields $I_i^{\prime H}(R_i(\mathrm{p}))=\bar{p}_i$, $I_i^{\prime H}(R_i^l)=\bar{p}_i=I'_i(R_i^l)$, and $I_i^{\prime H}(R_i(\mathrm{p}^{\prime}))=I'_i(R_i(\mathrm{p}^{\prime}))$. After performing some simple manipulations, we obtain
\begin{equation}
\label{Ch3_Ril}
(1/a)R_i^l \leq R_i(\mathrm{p} ^{\prime} ) \leq a R_i^l.
\end{equation}
By applying the results in the previous case, we have $\frac{1}{a}I_i^{\prime H}(R_i^l) < I_i^{\prime H}(R_i(\mathrm{p}^{\prime})) < a I_i^{\prime H}(R_i^l) $. Since $I_i^{\prime H}(R_i^l)=I_i^{\prime H}(R_i(\mathrm{p}))=\bar{p}_i$, the relationship (\ref{Ch3_r05}) holds, which completes the proof for this case.
\end{itemize}

In summary, we have proved that our proposed puf is 2.s.s. for $0<x \leq 1/2$. In fact, when $x>1/2$, $(1/r)^{\frac{x}{1-x}}$ may not be in $\left[ 1/a,a\right]$. Hence, the p.u.f. might not be 2.s.s.

\subsection{Proof of Lemma~\ref{Ch3_lemma3.3}}
\label{Ch3_apdx.lemma3.3}
According to the power update rule of the HPC algorithm given in (\ref{Ch3_hpcrule}), the SINR of user $i$ at convergence
can be expressed as
\begin{equation}
\label{Ch3_eqn:gamma_HPC}
\Gamma_i(\mathrm{p}^*) =\dfrac{p_i^*}{R_i(\mathrm{p}^*)}=\dfrac{\mathrm{min}\left\lbrace \bar{p}_i,I_i(\mathrm{p}^*)\right\rbrace }{R_i(\mathrm{p}^*)}.
\end{equation}
Recall that $R^{\sf th}_i = \bar{p}_i/\hat{\gamma}_i$. Using the relationship in (\ref{Ch3_xi2}) yields
\begin{eqnarray}
\label{Ch3_equ:gam_appC1}
\Gamma_i(\mathrm{p}^*) =\left\lbrace \begin{array}{*{10}{l}}
\hat{\gamma}_i \dfrac{R_i^{\sf th}}{R_i(\mathrm{p}^*)}, & p_i^*=\bar{p}_i,\\
\hat{\gamma}_i\dfrac{\alpha_i [R_i^{\sf th}/R_i( \mathrm{p}^*)]^{\frac{1}{1-x}}+1 }{\alpha_i + 1}, &p_i^* < \bar{p}_i.
\end{array}
\right. 
\end{eqnarray}

We can now utilize this result to prove the lemma. Specifically, let us consider the following cases. 
\begin{enumerate}
\item If $R_i(\mathrm{p}^*) > R_i^{\sf th}$, then we have $R_i^{\sf th}/R_i( \mathrm{p}^*)$ and $[R_i^{\sf th}/R_i( \mathrm{p}^*)]^{\frac{1}{1-x}}$ strictly less than $1$ 
since $0<x \leq 1/2$. Hence, we have $\Gamma_i(\mathrm{p}^*) < \hat{\gamma}_i$ if $p_i^*=\bar{p}_i$ in (\ref{Ch3_equ:gam_appC1}). We now prove 
statement 1 of the lemma for the remaining case.
For $\alpha_i>0$, we have $\Gamma_i(\mathrm{p}^*) < \hat{\gamma}_i$ according to (\ref{Ch3_equ:gam_appC1}). In addition, for $\alpha_i=0$, we have $I_i(\mathrm{p}^*)=\hat{\gamma}_iR_i(\mathrm{p}^*)>\hat{\gamma}_iR_i^{\sf th}=\bar{p}_i$. Hence, the HPC power update in (\ref{Ch3_hpcrule}) results in
 $p_i^*=\bar{p}_i$, which leads to $\Gamma_i(\mathrm{p}^*)=\hat{\gamma}_i \dfrac{R_i^{\sf th}}{R_i(\mathrm{p}^*)} < \hat{\gamma}_i$. Therefore, we have completed
 the proof for the first statement of the lemma. 

\item If $R_i(\mathrm{p}^*) \leq R_i^{\sf th}$, then we have $R_i^{\sf th}/R_i(\mathrm{p}^*) \geq 1$. Hence, using the results in (\ref{Ch3_equ:gam_appC1}) yields $\Gamma_i(\mathrm{p}^*) \geq \hat{\gamma}_i$, $\forall \: \alpha_i \geq 0$, which completes proof for the second statement of the lemma.
 
\item If $R_i(\mathrm{p}^*)<R_i^{\sf th}$ and $\alpha_i=0$, then we have $I_i(\mathrm{p}^*)=\hat{\gamma}_iR_i(\mathrm{p}^*)<\hat{\gamma}_iR_i^{\sf th}=\bar{p}_i$. As a result, we have $p_i^*=I_i(\mathrm{p}^*)$ due to the HPC power update in (\ref{Ch3_hpcrule}), which results in $\Gamma_i(\mathrm{p}^*)=\hat{\gamma}_i$. 
Inversely, if $\Gamma_i(\mathrm{p}^*) = \hat{\gamma}_i$ and $R_i(\mathrm{p}^*)<R_i^{\sf th}$, then we must have $p_i^* < \bar{p}_i$ [i.e., the second
case in (\ref{Ch3_equ:gam_appC1})]. Therefore, we have $p_i^*=I_i(\mathrm{p}^*)$ and $\alpha_i=0$, 
which completes the proof for the following statement of the lemma: 
 If $R_i(\mathrm{p}^*) < R_i^{\sf th}$, then $\Gamma_i(\mathrm{p}^*) = \hat{\gamma}_i$ iff $\alpha_i = 0$. 
Note that we have $\Gamma_i(\mathrm{p}^*) > \hat{\gamma}_i$ if and only if $\alpha_i>0$ and $R_i(\mathrm{p}^*)<R_i^{\sf th}$. 
We now prove the last statement of the lemma. Let us take the derivative of $I_i(\mathrm{p}^*)$ with respect to $\alpha_i$, then
 we have
\begin{equation} \label{Ch3_difGam}
\dfrac{\partial I_i(\mathrm{p}^*)}{\partial \alpha_i}=R_i(\mathrm{p}^*)\hat{\gamma}_i\dfrac{[R_i^{\sf th}/R_i(\mathrm{p}^*)]^{\frac{1}{1-x}}-1}{(1+\alpha_i)^2}.
\end{equation}
This implies that if we have $R_i(\mathrm{p}^*) < R_i^{\sf th}(\mathrm{p})$, then $I_i(\mathrm{p}^*)$ decreases as $\alpha_i$ decreases. 
\end{enumerate}

\subsection{Proof of Theorem~\ref{Ch3_thm3.3}}
\label{Ch3_apdx.thm3.3}

\subsubsection{Proof of Statement 1 in Theorem~\ref{Ch3_thm3.3}}
If $\overline{N}_{\sf HPC}=|\mathcal{M}|$, we obviously have $\overline{N}_{\sf HPC} \geq \overline{N}_{\sf TPC}$.
Therefore, we only need to consider the scenarios where there are one or several non-supported users after running Algorithm \ref{Ch3_alg:gms3}. In these cases, as
Algorithm~\ref{Ch3_alg:gms3} terminates (i.e., the end of step 3 in Algorithm~\ref{Ch3_alg:gms3}), we have $\alpha_i^{(s)}=0, \: \forall i \in \overline{\mathcal{M}}^{(s)}$, where $s$ denotes the index of the last iteration.
Let $\Delta^{s}$ is the set storing the values of $\alpha_i^{(s)}, \: i \in \mathcal{M}$, i.e., $\alpha_i^{(s)}$ is the value of $\alpha_i$ for user $i$ in the last iteration (before we set them equal to $\alpha_i^{\ast} \in \Delta^{\ast}$ in step 4 of Algorithm~\ref{Ch3_alg:gms3}). Since each potential user $i$ reduce $\alpha_i$
until it reaches zero in Algorithm \ref{Ch3_alg:gms3}, we have $\alpha_i^{(s)}=0, \: \forall i \in \overline{\mathcal{M}}^{(s)}$ and $\overline{N}_{\sf HPC} \geq \overline{N}^{s}$, where $\overline{N}^{s}$ is the number of supported users after running HPC puf using $\Delta^{s}$. Hence, $\overline{N}_{\sf HPC} \geq \overline{N}_{\sf TPC}$ holds if we can prove that 
\begin{equation}
\label{Ch3_equ:Ns_Ntpc}
\overline{N}^{s} \geq \overline{N}_{\sf TPC}.
\end{equation}
We will prove this in the following.
Now, let $\mathrm{p}_s^{\ast}$ and $\mathrm{p}_{\sf TPC}^{\ast}$ denote the power vectors at convergence, which are obtained by running the HPC algorithm with $\Delta^{s}$ and the TPC algorithm, respectively. Note that $\mathrm{p}_{\sf TPC}^{\ast}$ is also the power vector at convergence if we run TPC algorithm with initial powers set as $\mathrm{p}_{\sf TPC}^{(0)}=\mathrm{p}_s^{\ast}$. Let us consider this initial setting for the TPC scheme and investigate the following
 possible scenarios for a particular user $i$. 

\begin{itemize}
\item If user $i$ is supported, then we have $R_i(\mathrm{p}_{\sf TPC}^{(0)})=R_i(\mathrm{p}_s^{\ast}) \leq R_i^{\sf th}$, and $p_{i, \sf TPC}^{(0)}=p_{i,s}^{\ast}=\hat{\gamma}_i R_i(\mathrm{p}_{\sf TPC}^{(0)})$ since $\alpha_i^{(s)}=0$.
\item If user $i$ is nonsupported, then we have $R_i(\mathrm{p}_{\sf TPC}^{(0)})=R_i(\mathrm{p}_s^{\ast}) > R_i^{\sf th}$ and $p_{i,\sf TPC}^{(0)} \leq \bar{p}_i$. 
\end{itemize}

Let us analyze the dynamics of the power updates due to the TPC algorithm.
In particular, the power of each user is updated in the next iteration as follows.
\begin{itemize}
\item $p_{i, \sf TPC}^{(1)}=\hat{\gamma}_i R_i(\mathrm{p}_{\sf TPC}^{(0)})=p_{i, \sf TPC}^{(0)}$ if $R_i(\mathrm{p}_{\sf TPC}^{(0)}) \leq R_i^{\sf th}$.
\item $p_{i, \sf TPC}^{(1)}=\bar{p}_i \geq p_{i,\sf  TPC}^{(0)}$ if $R_i(\mathrm{p}_{\sf TPC}^{(0)}) > R_i^{\sf th}$.
\end{itemize}

Therefore, we always have $\mathrm{p}_{\sf TPC}^{(1)} \geq \mathrm{p}_{\sf TPC}^{(0)}$. Moreover, it can be verified that $R_i(\mathrm{p}_{\sf TPC}^{(n)}) \geq R_i(\mathrm{p}_{\sf TPC}^{(n-1)}), \: \forall i \in \mathcal{M}$ if $\mathrm{p}_{\sf TPC}^{(n)} \geq \mathrm{p}_{\sf TPC}^{(n-1)}$. Using this result and consulting the
TPC puf, we can show that $\mathrm{p}_{\sf TPC}^{(n+1)} \geq \mathrm{p}_{\sf TPC}^{(n)}$ if $\mathrm{p}_{\sf TPC}^{(n)} \geq \mathrm{p}_{\sf TPC}^{(n-1)}$. Hence, under the TPC algorithm, $\mathrm{p}_{\sf TPC}^{(n)}$ will increase and converge to $\mathrm{p}_{\sf TPC}^{\ast}$ if $\mathrm{p}_{\sf TPC}^{(1)} \geq \mathrm{p}_{\sf TPC}^{(0)}$. Therefore, we have
\begin{equation}
\label{Ch3_equ:ps_ptpc}
\mathrm{p}_{\sf TPC}^{\ast} \geq \mathrm{p}_{\sf TPC}^{(0)} = \mathrm{p}_s^{\ast}.
\end{equation}

From this, we can achieve $R_i(\mathrm{p}_{\sf TPC}^{\ast}) \geq R_i(\mathrm{p}_s^{\ast}), \: \forall i \in \mathcal{M}$. Hence, if user $i$ is a supported user under TPC algorithm, $R_i(\mathrm{p}_{\sf TPC}^{\ast})$ must not be greater than $R_i^{\sf th}$. As a result, $R_i(\mathrm{p}_{s}^{\ast})$ is also not greater than $R_i^{\sf th}$. This implies that user $i$ is also a supported user under HPC algorithm with $\Delta^{s}$. Therefore, any users that can be supported by TPC algorithm must also be supported by the HPC algorithm. Hence, (\ref{Ch3_equ:Ns_Ntpc}) holds. This completes the proof that $\overline{N}_{\sf HPC} \geq \overline{N}_{\sf TPC}$.

\subsubsection{Proof of Statement 2 in Theorem~\ref{Ch3_thm3.3}}
For a feasible system, all users can achieve their target SINRs under the TPC algorithm, i.e., $\overline{N}_{\sf TPC}=|\mathcal{M}|$. Using the argument in the proof of the first statement of Theorem~\ref{Ch3_thm3.3}, we conclude that $\overline{N}_{\sf HPC}$ is also equal to $|\mathcal{M}|$. In addition, under the HPC adaptation algorithm, we have $\Gamma_i(\mathrm{p}_{\sf HPC}^{\ast}) \geq \hat{\gamma}_i$, $\forall i \in \mathcal{M}$. Therefore, supposing all potential users $i$ decrease their $\alpha_i$ slowly, some of them will achieve SINR greater than the target values if $\alpha_i>0$ at convergence. This completes the proof for the second 
statement of Theorem~\ref{Ch3_thm3.3}.

%% file: chap6/Ha_chap6.tex
\chapter{Fair Resource Allocation for OFDMA Femtocell Networks with Macrocell Protection}
\renewcommand{\rightmark}{Chapter 6.  Fair Resource Allocation for OFDMA Femtocells with Macrocell Protection}
\label{Ch4}
The content of this chapter was published in IEEE Transactions on Vehicular Technology in the following paper:

Vu N. Ha and Long B. Le, ``Fair resource allocation for OFDMA femtocell networks with macrocell protection,'' {\em IEEE Trans. on Veh. Tech.,} vol. 63, no. 3, pp. 1388--1401, Mar. 2014.

\section{Abstract}
\label{Ch4_Abs}
We consider the joint subchannel allocation and power control problem for orthogonal frequency-division multiple-access (OFDMA) femtocell networks in this paper. 
Specifically, we are interested in the fair resource-sharing solution for users
in each femtocell that maximizes the total minimum spectral efficiency of all femtocells subject to protection constraints for
the prioritized macro users. Toward this end, we present the
mathematical formulation for the uplink resource-allocation
problem and propose an optimal exhaustive search algorithm.
Given the exponential complexity of the optimal algorithm, we
develop a distributed and low-complexity algorithm to find an
efficient solution for the problem. We prove that the proposed
algorithm converges and we analyze its complexity. Then, we
extend the proposed algorithm in three different directions,
namely, downlink context, resource allocation with rate adaptation
for femto users, and consideration of a hybrid access strategy
where some macro users are allowed to connect with nearby femto
base stations (FBSs) to improve the performance of the femto
tier. Finally, numerical results are presented to demonstrate the
desirable performance of the proposed algorithms.

\section{Introduction}
\label{Ch4_Introduction}
Massive deployment of small cells such as femtocells provides an important solution to fundamentally improve the indoor throughput and coverage performance of wireless cellular networks \cite{Andrews12}. 
Moreover, fourth-generation
(4G) and beyond cellular networks are based on orthogonal
frequency-division multiple access (OFDMA) that provides
flexibility in radio resource management and robustness against
adverse effects of multipath fading \cite{perez09}. Hence, the design of an
OFDMA femtocell tier so that it can coexist efficiently with
the existing macrocell tier is an important research topic. One
particular challenge in realizing this objective stems from the
fact that femtocells are randomly deployed and they operate on
the same frequency band with the macrocell tier, which can create strong cross-tier interference to the macrocells. Therefore,
to successfully deploy the femtocell network, it is required to
resolve many technical challenges, which range from resource
allocation and synchronization \cite{guoqing06} to protection of existing
macrocells against cross-tier interference from femtocells \cite{Le12}.

There have been some existing works studying the resource allocation for OFDMA-based femtocell networks \cite{chuhan11,yanzan12,yushan12,Wangchi12,Long12,hoon11,renchao12,lu12}. 
In \cite{chuhan11}, spectrum sharing and access control strategies are proposed for femtocell networks by using the water-filling algorithm and game theory technique. 
The authors in \cite{yanzan12} propose two
methods to mitigate the uplink interference for OFDMA femtocell networks. In the first method, femtocell user equipments
(FUEs) that produce strong interference to macrocell user
equipments (MUEs) are only allowed to use dedicated subchannels, whereas the remaining FUEs can utilize all subchannels
assigned for the femto tier. In the second method, an auction-based algorithm is devised to optimize the channel assignment
for both tiers to mitigate the co-tier interference. Interference avoidance and resource allocation issues for OFDMA femtocell networks are also studied in \cite{yushan12} 
where a greedy algorithm is developed to enable self-organization of femtocells. 
In \cite{Wangchi12}, Cheung  \textit{et al.} study the network performance where the macrocells and femtocells utilize separate sets of subchannels or share the whole spectrum under both open- and closed-access strategies.
The quality-of-service (QoS)-aware admission control design for OFDMA femtocells is conducted in \cite{Long12}. All these
existing works, however, do not consider power control in their
proposed resource-allocation algorithms.

In \cite{hoon11}, an adaptive femtocell interference management algorithm comprising three control loops that run continuously and
separately at macro and femto base stations (MBS and FBS) to determine initial femto maximum power, to decide target
signal-to-interference-plus-noise ratios (SINRs) for FUEs, and to control the transmission power, respectively, is developed.
However, subchannel assignment (SA) is not studied, and the maximum power constraint for each user is not considered in
this paper. 
Spectrum sharing and power allocation (PA) are investigated in \cite{renchao12,lu12}.
In particular, \cite{renchao12} aims to enhance the energy efficiency, whereas the main design objective of \cite{lu12} is to achieve user-level fairness for cognitive femtocells. However, the algorithms
developed in \cite{renchao12} and \cite{lu12} do not provide QoS guarantees for users of both network tiers.
\nomenclature{MBS}{Macro Base Station}

In this paper, we study the joint subchannel and PA problem for femtocell networks considering fairness for FUEs in
each femtocell, protection of MUEs, and maximum power constraints.
To the best of our knowledge, none of the existing works on the OFDMA-based femtocell network jointly
consider all these design issues.
In particular, we make the following contributions.
\begin{itemize}
\item
We formulate the uplink fair subchannel and PA problem,
study its optimal structure, and present an algorithm
to find its optimal solution. To overcome the exponential complexity of the centralized optimal exhaustive
search algorithm, we then develop a distributed and low-complexity resource-allocation algorithm. The proposed
algorithm iteratively updates the SAs at each femtocell
based on carefully designed assignment weights and
transmission power values for the subchannels accordingly. We prove the convergence of the proposed algorithm and analyze its complexity.

\item 
We extend the proposed solution in three important directions. Specifically, we show how the proposed resource-allocation algorithm can be adapted for the downlink context. Moreover, we develop a distributed algorithm
for the scenario where FUEs in different femtocells are allowed to choose different target rates per subchannel
to better adapt to the network interference.
Finally, we extend the proposed algorithm to implement the hybrid
spectrum access where some MUEs are permitted to connect to nearby FBSs.

\item 
Numerical results are presented to demonstrate the efficacy of the proposed algorithms and their relative
performance compared with the optimal algorithm. In particular, we demonstrate the impacts of target rates
per subchannel (i.e., modulation schemes), maximum power values, rate adaptation, and hybrid access on the
performance of the femtocell network.
\end{itemize}

The remainder of this paper is organized as follows. 
We describe the system model and the uplink problem formulation in Section~\ref{Ch4_section2}.
In Section~\ref{Ch4_sec:OSOA}, we present both optimal exhaustive
search and suboptimal resource-allocation algorithms. Extensions for the downlink scenario, adaptive-rate transmission, and
hybrid access control are described in Section~\ref{Ch4_sec:extend}. 
Numerical results are presented in Section~\ref{Ch4_result} followed by conclusion in Section~\ref{Ch4_ccls}. Key notations used in this paper are summarized in Table~\ref{Ch4_keynotation1}.
 
\section{System Model}
\label{Ch4_section2}
\subsection{System Model}

We consider the OFDMA-based two-tier macrocell-femtocell network where users of both tiers share the spectrum comprising $N$ subchannels.
We assume that there are $M_{\sf f}$ FUEs served by $(K-1)$ FBSs, which 
are underlaid by one macrocell serving $M_{\sf m}$ MUEs.
Let $\mathcal{U}_k$ be the set of user equipments (UEs) in the $k$-th cell, i.e., they are served by the base station (BS) $k$ of the corresponding tier.
For convenience, let $\mathcal{U}_{\sf m} \triangleq \mathcal{U}_1$ represent the set of MUEs and $\mathcal{U}_{\sf f} \triangleq \cup_{k=2}^{K} \mathcal{U}_k= \left\lbrace M_{\sf m}+1,...,M_{\sf m}+M_{\sf f} \right\rbrace$ denote the set of all FUEs. In addition, let $\mathcal{U}$ and $\mathcal{B}$ be the sets of all UEs and BSs, respectively.
Then, we have $\mathcal{U}= \mathcal{U}_{\sf m} \cup \mathcal{U}_{\sf f} = \left\lbrace 1, 2, ..., M \right\rbrace $ and $\mathcal{B}=\left\lbrace {1,2,...,K} \right\rbrace $ where BS $1$ is assumed to be the MBS and $\mathcal{B}_{\sf f}=\left\lbrace {2,...,K} \right\rbrace$ denotes the set of all FBSs.

We assume the fixed BS association for all UEs of both tiers in the network (i.e., each UE is served by a fixed BS in the
corresponding tier) to present the problem formulation here.
Relaxation of this assumption under the hybrid access design will be considered in Section~\ref{Ch4_sec:extend}.
Now, let $b_i \in \mathcal{B}$ denote the BS serving UE $i$, and
let $\mathcal{N}=\left\lbrace {1,2,...,N} \right\rbrace$ be the set of all orthogonal subchannels.
We assume there is no interference among transmissions on
 different subchannels. We consider a system with full frequency reuse where all $N$ subchannels are allocated for
UEs in all cells of either tier. To describe the SAs, 
let $\mathrm{\bf{A}} \in \mathfrak{R}^{M \times N}$ be the SA matrix for all $M$ UEs over $N$ subchannels where
\begin{equation}
\label{Ch4_eq:A}
\mathrm{\bf{A}}(i,n)=a^n_i=\left\lbrace \begin{array}{*{10}{l}}
1 & \text{if  subchannel $n$ is assigned for UE $i$}\\
0 & \text{otherwise}.
\end{array}
\right. 
\end{equation}
We assume that a subchannel can be allocated to, at most, one
UE in any cell. Then, we have
\begin{equation}
\label{Ch4_eq:c4}
\sum_{i \in \mathcal{U}_k} a_i^n \leq 1, \:\:\: \forall k \in \mathcal{B} \text{ and } \forall n \in \mathcal{N}.
\end{equation}
These constraints will be applied to all resource-allocation problems studied in this paper.

\subsection{Physical-layer Model}
We assume that $M$-ary quadrature amplitude modulation ($M\text{-QAM}$) is adopted for communications on each subchannel where the constellation size $s$ for $s\text{-QAM}$ is chosen from a predetermined set $\mathfrak{M}$.
To guarantee acceptable performance when a constellation size $s$
is adopted on a particular subchannel, we need to ensure that the SINR on that subchannel is not smaller than a corresponding
target SINR value $\bar{\gamma}(s)$. To determine $\bar{\gamma}(s)$, let $f\left( \gamma(s) \right)$ be the expression of 
bit error rate (BER) for SINR $\gamma(s)$ when the constellation size $s$ is adopted. Suppose we want to maintain
the BER to be not greater than a target value $\overline{P}_e$. 
Then, the SINR must satisfy the following
for the constellation size $s$:
\beqn \label{Ch4_tsnr}
\gamma(s) \geq  \bar{\gamma}(s) = f^{-1}\left( \overline{P}_e \right),
\eeqn
where $f^{-1}(.)$ denotes the inverse function of the BER function $f(.)$.
To demonstrate this calculation, suppose that $M\text{-QAM}$ modulation with square signal constellations (i.e., $s\text{-QAM}$ where $s \in \mathfrak{M} \triangleq \{4,16,64,256, \ldots, s_{\sf max}\} = \{2^{2\sigma}\vert$ $\sigma \text{ is a positive integer and } \sigma \leq 1/2 \log_2 s_{\sf max} \}$) is employed. According to \cite{proakis01}, the BER of the $s\text{-QAM}$ scheme with Gray encoding can be approximated as
\beq \label{Ch4_eq:ber}
f\left( \gamma(s) \right) \approx x_s\mathbb{Q}(\sqrt{y_s\gamma(s)}), \: s \in \mathfrak{M},
\eeq
where $\mathbb{Q}(.)$ stands for the \textit{Q}-function, $x_s=\frac{2(1-1/\sqrt{s})}{\log_2s}$, $y_s=\frac{3}{2(s-1)}$, and $\gamma(s)$ is the SINR. Using the
result in (\ref{Ch4_tsnr}), the target SINR for the $s\text{-QAM}$ modulation scheme can be calculated as
\beq \label{Ch4_eq:tSINR}
\bar{\gamma}(s) = \dfrac{\left[ \mathbb{Q}^{-1}( \overline{P}_e /x_s)\right]^2}{y_s}, \: s \in \mathfrak{M},
\eeq
where $\overline{P}_e$ is the target BER value. We assume that ideal Nyquist
data pulses are used where the symbol rate on each subchannel
is equal to $BW/N$ where $BW$ is the total bandwidth of $N$ subchannels. Then, if $s\text{-QAM}$ modulation scheme is employed, the
spectral efficiency per 1 Hz of system bandwidth is $\frac{\log_2s}{N}$ (in bits per second per hertz).

\subsection{Uplink Resource Allocation Problem}
Let $p^n_i$ represent the transmission power of UE $i$ over subchannel $n$ in the uplink where $p^n_i \geq 0$.
We impose the following constraints on the total transmission power values:
\begin{equation} \label{Ch4_powcon}
\sum_{n=1}^{N}p^n_i \leq P_i^{\texttt{max}}, \quad i \in \mathcal{U},
\end{equation}
where $P_i^{\texttt{max}}$ is the maximum transmission power of UE $i$.
Similar to the SAs, we define $\mathrm{\bf{P}}$ as an $M \times N$ PA matrix where $\mathrm{\bf{P}}(i,n)=p^n_i$. 
For convenience, we also define partitions of SA and PA matrices $\mathrm{\bf{A}}$ and $\mathrm{\bf{P}}$ as follows.
In particular, let $\mathrm{\bf{A}}_k, \mathrm{\bf{P}}_k \in \mathfrak{R}^{\vert \mathcal{U}_k \vert \times N}$ represent the SA and PA matrices 
for UEs in cell $k$ over $N$ subchannels, respectively. 

Let $h_{ij}^n$ and $\eta_i^n$ be the channel power gain from UE $j$ to BS $\mathit{i}$ and the noise power at BS $\mathit{i}$ over subchannel $n$, respectively. 
Then, for a given SA and PA solution, i.e., given $\mathrm{\bf{A}}$ and $\mathrm{\bf{P}}$, the SINR achieved at
 BS $b_i$ due to the transmission of UE $i$ over subchannel $n$  can be written as
\begin{equation}
\label{Ch4_sinr}
\Gamma_i^n(\mathrm{\bf{A}},\mathrm{\bf{P}})=\dfrac{a_i^n h_{b_i i}^n p_i^n}{\sum_{j \notin \mathcal{U}_{b_i}}{a_j^n h_{b_ij}^n p_j^n}+\eta_{b_i}^n}=\dfrac{a_i^np_i^n}{I_i^n(\mathrm{\bf{A}},\mathrm{\bf{P}})},
\end{equation}
where $I_i^n(\mathrm{\bf{A}},\mathrm{\bf{P}})$ is the effective interference corresponding to UE $i$ on subchannel $n$, which is defined as
\begin{equation}
\label{Ch4_eq:I}
 I_i^n (\mathrm{\bf{A}},\mathrm{\bf{P}}) \triangleq  \dfrac{\sum_{j \notin \mathcal{U}_{b_i}}{a_j^n h_{b_ij}^n p_j^n}+\eta_{b_i}^n}{h_{b_ii}^n}.
\end{equation}
We assume that SAs for MUEs have been determined by a
certain mechanism and fixed, whereas PAs over the corresponding subchannels for MUEs are updated to cope with the cross-tier interference due to transmissions of FUEs. This means that $\mathrm{\bf{A}}_1$ is
fixed while we need to determine $\mathrm{\bf{A}}_k, \: 2 \leq k \leq K$ and the corresponding PAs.
This assumption is quite natural
since the macro tier has been deployed before the introduction
of the femto tier; therefore, it would be more reasonable to keep
the existing subchannel allocation mechanism and solution in
the macro tier unchanged.

To protect the QoS of the licensed MUEs, we wish to maintain a predetermined target SINR $\bar{\gamma}_i^n$ for each of its 
assigned subchannel $n$. Specifically, we have the following constraints for the MUEs:
\begin{equation}
\label{Ch4_eq:rate_cond}
\Gamma_i^n(\mathrm{\bf{A}},\mathrm{\bf{P}}) \geq \bar{\gamma}_i^n,\:\: \text{if} \:\: a_i^n=1, \:\: \forall i \in \mathcal{U}_{\sf m},
\end{equation}
where $\bar{\gamma}_i^n = \bar{\gamma}(s^{\sf m})$, which is calculated as in (\ref{Ch4_tsnr}) if MUE $i$ employs $s^{\sf m}\text{-QAM}$
modulation scheme.

For FUEs, we assume that they all employ $s^{\sf f}\text{-QAM}$ for a predetermined $s^{\sf f}$.
Then, the spectral efficiency (in bits per second per hertz) achieved by FUE $i$ on one subchannel can be written as 
\begin{equation}
\label{Ch4_eq:rate_n}
r_i^n(\mathrm{\bf{A}},\mathrm{\bf{P}}) = \left\lbrace \begin{array}{*{5}{l}}
0, & \text{if } \Gamma_i^n(\mathrm{\bf{A}},\mathrm{\bf{P}}) < \bar{\gamma}_i^n,\\
r_{\sf f}, & \text{if } \Gamma_i^n(\mathrm{\bf{A}},\mathrm{\bf{P}}) \geq \bar{\gamma}_i^n,
\end{array} \right. 
\end{equation}
where $r_{\sf f}=(1/n)\log_2s^{\sf f}$, and $\bar{\gamma}_i^n = \bar{\gamma}(s^{\sf f})$, which is calculated as in (\ref{Ch4_tsnr}).
Note that we have assumed that FUEs in all femtocells employ the same constellation size $s^{\sf f}$ for ease of exposition. We will 
discuss the more general case in Section~\ref{Ch4_sec:extend}. 
In practice, we typically choose $s^{\sf f} \geq s^{\sf m}$ since FUEs can achieve higher
transmission rates on each subchannel than MUEs due to their
short distance to the associated FBS. This is indeed equivalent to $\bar{\gamma}(s^{\sf f}) \geq \bar{\gamma}(s^{\sf m})$. In the analysis, we mostly use the same notation
 $\bar{\gamma}_i^n$ to refer to the required target SINR for FUEs or MUEs where $\bar{\gamma}_i^n=\bar{\gamma}(s^{\sf f})$ for $\forall i \in \mathcal{U}_{\sf f}$ and $\bar{\gamma}_i^n=\bar{\gamma}(s^{\sf m})$ for $\forall i \in \mathcal{U}_{\sf m}$. 
Now, we can express the total spectral efficiency achieved by UE $i$ for given SA and PA matrices $\mathrm{\bf{A}}$ and $\mathrm{\bf{P}}$ as
\begin{eqnarray}
\label{Ch4_eq:rate}
R_i(\mathrm{\bf{A}},\mathrm{\bf{P}})=\sum_{n=1}^N r_i^n(\mathrm{\bf{A}},\mathrm{\bf{P}}).
\end{eqnarray}
To impose the max–min fairness for all FUEs associated with
the same FBS, we define the minimum spectral efficiency of
femtocell $k$ as the smallest value of spectral efficiency values
of all FUEs in that femtocell, e.g.,
\begin{equation}
\label{Ch4_eq:rate_fcell}
\mathrm{R}^{(k)}(\mathrm{\bf{A}},\mathrm{\bf{P}})=\min_{i \in \mathcal{U}_k} R_i(\mathrm{\bf{A}},\mathrm{\bf{P}}).
\end{equation}
The uplink resource-allocation problem for FUEs can be formulated as follows:
\begin{eqnarray}
\label{Ch4_objfun}
\mathop {\max} \limits_{(\mathrm{\bf{A}},\mathrm{\bf{P}}) } \sum \limits_{2\leq k \leq K} \mathrm{R}^{(k)}(\mathrm{\bf{A}},\mathrm{\bf{P}}) =  \sum \limits_{2\leq k \leq K} \min_{i \in \mathcal{U}_k} R_i(\mathrm{\bf{A}},\mathrm{\bf{P}}) \\
\text{s. t. constraints} \quad (\ref{Ch4_eq:c4}), (\ref{Ch4_powcon}),  (\ref{Ch4_eq:rate_cond}), \label{confun}
\end{eqnarray}

Therefore, the objective of this uplink resource-allocation problem aims to balance between attaining an equal rate for
FUEs in each femtocell (max-min fairness \cite{luo08}) and high total femtocell spectral efficiency
subject to SA constraints, FUE power constraints, and protection constraints for MUEs. Note that the SINR constraints for FUEs
are embedded in the calculation of FUE spectral efficiency in (\ref{Ch4_eq:rate_n}). The resource-allocation problem (\ref{Ch4_objfun})-(\ref{confun}) 
can be transformed into a standard mixed integer program, which is, therefore, NP-hard.

\section{Optimal and Supoptimal Algorithms}
\label{Ch4_sec:OSOA}

Here, we study the feasibility of a particular SA solution,
which reveals the optimal structure of the optimization problem (\ref{Ch4_objfun})-(\ref{confun})
based on which we develop optimal and suboptimal algorithms.

\subsection{Feasibility of an SA Solution}
\label{Ch4_sec:fsb_chk}

The resource-allocation problem (\ref{Ch4_objfun}) and (\ref{confun}) involves finding a joint SA and PA solution.
For a certain SA solution represented by matrix $\mathrm{\bf{A}}$ satisfying the SA constraints (\ref{Ch4_eq:c4}), we can indeed find its ``best'' PA and verify
its feasibility with respect to the constraints in (\ref{confun}). In particular, we wish to maintain the SINR constraints
for FUEs and MUEs in (\ref{Ch4_eq:rate_n}) and (\ref{Ch4_eq:rate_cond}), respectively.
Specifically, for each subchannel $n$ we need to maintain SINRs constraints $\Gamma^n_i(\mathrm{\bf{A}},\mathrm{\bf{P}}) \geq \bar{\gamma}_i^n$ 
for both MUEs and FUEs that are allocated with this subchannel.

Now, let $\mathcal{U}^n=\{i \in \mathcal{U}|a_i^n \neq 0\}=\{n_1,...,n_{c^n}\}$ denote the set of UEs of both tiers that are assigned
subchannel $n$ and $c^n=|\mathcal{U}^n|$ be the number of elements in set $\mathcal{U}^n$. Then, the SINR constraints for UEs in set $\mathcal{U}^n$, i.e.,
$\Gamma^n_i(\mathrm{\bf{A}},\mathrm{\bf{P}}) \geq \bar{\gamma}_i^n, \: i \in \mathcal{U}^n$, can be
rewritten in matrix form as 
\begin{equation}
\label{Ch4_eq:cond_mtx1}
(\mathrm{\bf{I}}^n-\mathrm{\bf{G}}^n\mathrm{\bf{H}}^n)\mathrm{\bf{p}}^n \geq \mathrm{\bf{g}}^n,
\end{equation}
where $\mathrm{\bf{g}}^n=\left[ g_{n_1}^n,...,g_{n_{c^n}}^n \right]^T$ with $g_i^n=\frac{\eta_{b_i}^n \bar{\gamma}_i^n}{h_{b_ii}^n}$, $\mathrm{\bf{I}}^n$ is $c^n \times c^n$ identity matrix, $\mathrm{\bf{G}}^n=\text{diag}\{\bar{\gamma}_{n_1}^n,...,\bar{\gamma}_{n_{c^n}}^n\}$, $\mathrm{\bf{p}}^n=[p_{n_1}^n,...,p_{n_{c^n}}^n]^T$, and $\mathrm{\bf{H}}^n$ is a $c^n \times c^n$ matrix defined as
\begin{equation}
\label{Ch4_eq:H}
\left[ \mathrm{\bf{H}}^n_{i,j}\right]=\left\lbrace \begin{array}{*{5}{l}}
0, & \text{if } j=i,\\
\frac{h_{b_{n_i}n_j}^n}{h_{b_{n_i}n_i}^n}, & \text{if } j \neq i.
\end{array} \right.
\end{equation}

To study feasible solutions for inequality (\ref{Ch4_eq:cond_mtx1}), we recall some standard results due to the Perron-Frobenius theorem \cite{Horn-85,seneta73,gant71} and their application to the power control problem \cite{bambos00}.
Let $\lambda_i$ be the $i^{\sf th}$ eigenvalue of matrix $\mathrm{\bf{G}}^n\mathrm{\bf{H}}^n$ and $\rho(\mathrm{\bf{G}}^n\mathrm{\bf{H}}^n)={\sf max}_{ i} | \lambda_i |$ be the maximum value of the
modulus of all eigenvalues (i.e., the spectral radius of $\mathrm{\bf{G}}^n\mathrm{\bf{H}}^n$). The results in the Perron-Frobenius theorem can be applied for nonnegative matrix $\mathrm{\bf{G}}^n\mathrm{\bf{H}}^n$,
which is the case for our model since all power channel gains, noise power values, and target SINRs are real and nonnegative numbers. Specifically, the Perron-Frobenius theorem implies the following fact, which has been clearly stated
in \cite{bambos00}. 

Fact: If $\mathrm{\bf{G}}^n\mathrm{\bf{H}}^n$ has non-negative elements, the following statements are equivalent:
\begin{enumerate}
\item There exists a nonnegative power vector $\mathrm{\bf{p}}$ such that $(\mathrm{\bf{I}}^n-\mathrm{\bf{G}}^n\mathrm{\bf{H}}^n)\mathrm{\bf{p}}^n \geq \mathrm{\bf{g}}^n$.
\item $\rho(\mathrm{\bf{G}}^n\mathrm{\bf{H}}^n) \leq 1$.
\item $\left( \mathrm{\bf{I}}-\mathrm{\bf{G}}^n\mathrm{\bf{H}}^n\right)^{-1}=\sum_{k=0}^{\infty} (\mathrm{\bf{G}}^n\mathrm{\bf{H}}^n)^k$ exists and is positive elementwise.
\end{enumerate}
From these results, if $\rho(\mathrm{\bf{G}}^n\mathrm{\bf{H}}^n) \leq 1$, then we can find a solution of $(\mathrm{\bf{I}}^n-\mathrm{\bf{G}}^n\mathrm{\bf{H}}^n)\mathrm{\bf{p}}^n = \mathrm{\bf{g}}^n$ (i.e., the
equality case of the considered inequality) as
\beq \label{Ch4_eq:pw_chk}
\mathrm{\bf{p}}^{n\star}=\left( \mathrm{\bf{I}}- \mathrm{\bf{G}}^n\mathrm{\bf{H}}^n \right)^{-1}\mathrm{\bf{g}}^n,
\eeq
which is also the Pareto-optimal solution of $(\mathrm{\bf{I}}^n-\mathrm{\bf{G}}^n\mathrm{\bf{H}}^n)\mathrm{\bf{p}}^n \geq \mathrm{\bf{g}}^n$ where Pareto-optimality means that
any feasible solution $\mathrm{\bf{p}}^n$ for the inequality  $(\mathrm{\bf{I}}^n-\mathrm{\bf{G}}^n\mathrm{\bf{H}}^n)\mathrm{\bf{p}}^n \geq \mathrm{\bf{g}}^n$ is not smaller
than $\mathrm{\bf{p}}^{n\star}$ elementwise (i.e., there exists no feasible solution $\mathrm{\bf{p}}^n$ for $(\mathrm{\bf{I}}^n-\mathrm{\bf{G}}^n\mathrm{\bf{H}}^n)\mathrm{\bf{p}}^n \geq \mathrm{\bf{g}}^n$ so that 
$p_i^n \leq p_i^{n\star}, \forall i$ and $p_j^n < p_j^{n\star}$ for some $j$). These results has been used in \cite{bambos00} and several references therein.

In addition, this Pareto-optimal power vector can be achieved at the equilibrium by employing the following well-known distributed Foschini-Miljanic power updates \cite{foschini93,bambos00}:
\beq \label{Ch4_dispow}
p_i^n(l+1) := p_i^n(l) \frac{\bar{\gamma}_i^n}{\Gamma^n_i(l)} = I_i^n(l)\bar{\gamma}_i^n,
\eeq
where $\Gamma^n_i(l)$ and $I_i^n(l)$ are the SINR and effective interference achieved by UE $i$ in iteration $l$, respectively. 
If there is no feasible solution for (\ref{Ch4_eq:cond_mtx1}) (when $\rho(\mathrm{\bf{G}}^n\mathrm{\bf{H}}^n) > 1$), the transmission power values due to (\ref{Ch4_dispow}) will increase to infinity \cite{foschini93,bambos00}.
Now, we are ready to state one important result for a given SA solution $\mathrm{\bf{A}}$ that satisfies the SA constraints (\ref{Ch4_eq:c4}) in the
 following proposition.

\begin{prop}
\label{Ch4_thm02}
Suppose that we can find a finite Pareto-optimal PA solution $\mathrm{\bf{p}}^n$ on the subchannel $n$ for the considered SA $\mathrm{\bf{A}}$,
which is given in (\ref{Ch4_eq:pw_chk}) (i.e., the spectrum radius $\rho(\mathrm{\bf{G}}^n\mathrm{\bf{H}}^n)<1, \: \forall n$).
Then, the underlying SA $\mathrm{\bf{A}}$ is feasible if these Pareto-optimal PA vectors $\mathrm{\bf{p}}^n$ satisfy the
power constraints in (\ref{Ch4_powcon}).
\end{prop}
\begin{proof}
The Pareto-optimal PA solution $\mathrm{\bf{p}}^n$ on each subchannel $n$ requires the minimum power values in the elementwise sense for all
UEs, which are assigned subchannel $n$, to meet their SINR requirements. Therefore, the considered SA solution $\mathrm{\bf{A}}$ is feasible
if and only if the corresponding Pareto-optimal PA solutions $\mathrm{\bf{p}}^n$ on all $N$ subchannels satisfy
the power constraints in (\ref{Ch4_powcon}). Therefore, we have completed the proof of the proposition.
\end{proof}

Since the number of possible SAs is finite, the result in this proposition pays the way to develop optimal exhaustive
search algorithm, which is presented in the following.

\subsection{Optimal Algorithm}
\label{Ch4_sec:OA}
\subsubsection{Exhaustive Search Algorithm}
Based on the results in Section~\ref{Ch4_sec:fsb_chk}, we can find the optimal solution for the 
resource-allocation problem (\ref{Ch4_objfun}) and (\ref{confun}) by performing exhaustive
search as follows. For a fixed and feasible $\mathrm{\bf{A}}_1$,\footnote{The SA solution for MUEs corresponding to $\mathrm{\bf{A}}_1$ is feasible if there exists a PA solution that satisfies the constraints in (\ref{Ch4_eq:rate_cond}) when there is no femto tier.} let $\Omega\{\mathrm{\bf{A}}\}$ be the list of all potential SA solutions
 that satisfy the SA constraints (\ref{Ch4_eq:c4}) and the fairness condition $\sum_{n \in \mathcal{N}} a_i^n=\sum_{n \in \mathcal{N}} a_j^n=\tau_k$ for all FUEs $i,j \in \mathcal{U}_k$. Then, we sort the list $\Omega\{\mathrm{\bf{A}}\}$ in the decreasing order of $\sum_{k=2}^K \tau_k$ to obtain the sorted list $\Omega^{\ast}\{\mathrm{\bf{A}}\}$. Then, the feasibility of each SA solution in the list $\Omega^{\ast}\{\mathrm{\bf{A}}\}$ can be verified, as presented in Section~\ref{Ch4_sec:fsb_chk}. 
Among all feasible SA solutions, the feasible one achieving the highest value of the objective function (\ref{Ch4_objfun}) and its corresponding
PA solution given in (\ref{Ch4_eq:pw_chk}) is the optimal solution of the optimization problem (\ref{Ch4_objfun}) and (\ref{confun}).

\subsubsection{Complexity Analysis}
The complexity of the exhaustive search algorithm can be determined by calculating the number of the elements in the list $\Omega^{\ast}\{\mathrm{\bf{A}}\}$ and the complexity involved in the feasibility verification for each of them.
The number of the elements in $\Omega^{\ast}\{\mathrm{\bf{A}}\}$, i.e., the number of potential SAs, is the product of the number of potential
SAs for all femtocells, which satisfy the ``fairness condition.'' Therefore, the number of potential SAs can be calculated as
\begin{equation}
\label{Ch4_eq:num_trial}
\begin{array}{*{2}{l}}
T\!\!\!\! &=\!\!\prod \limits_{k=2}^K \sum \limits_{\tau_k=0}^{\eta_k} \prod \limits_{i=0}^{M_k-1} \mathrm{C}_{N-i\tau_k}^{\tau_k}\\
{}\!\!\!\!&=\!\!\prod \limits_{k=2}^K \sum \limits_{\tau_k=0}^{\eta_k}\dfrac{N!}{(\tau_k!)^{M_k}(N-M_k\tau_k)!}\approx O\left(\!(N!)^{(\!K-1\!)}\!\right),
\end{array}
\end{equation}
where $\eta_k=\lfloor N/M_k\rfloor$ represents for the largest integer less than or equal to $ N/M_k$ and $\mathrm{C}_m^n=\frac{m!}{n!(m-n)!}$ denotes the \textit{``m-choose-n''} operation. According to Section~\ref{Ch4_sec:fsb_chk}, the complexity of the feasibility verification for each potential SA mainly 
depends on the eigenvalue calculation of the corresponding matrix and solving the linear system to find the PA solution for each subchannel. It requires $O(K^3)$ to calculate the eigenvalues of $\mathrm{\bf{G}}^n\mathrm{\bf{H}}^n$ \cite{Pan99} and $O(K^3)$ to obtain $\mathrm{\bf{p}}^n$ by solving a  system of linear equations \cite{Bojan84}. Therefore, the complexity of the optimal exhaustive search algorithm is $O\left(K^3 \times N \times (N!)^{(K-1)}\right)$, which
is exponential in the numbers of subchannels and FBSs. This optimal exhaustive search algorithm will serve as a benchmark to evaluate the low-complexity
algorithm developed in the following.

\subsection{Sub-Optimal and Distributed Algorithm}
\label{Ch4_sec:SOA}
To resolve the exponential complexity of the centralized optimal algorithm presented in Section~\ref{Ch4_sec:OA}, we 
develop a low-complexity and distributed resource-allocation algorithm. 
To achieve max-min fairness as required by the objective function (\ref{Ch4_objfun}), our algorithm
aims to assign the maximum equal number of subchannels to FUEs
in each femtocell and to perform Pareto-optimal PA for FUEs and MUEs on each subchannel so that they meet the SINR constraints
in (\ref{Ch4_eq:rate_n}) and (\ref{Ch4_eq:rate_cond}). To achieve this design goal, we propose a novel subchannel allocation weight metric so 
that high weights are assigned for ``bad allocations'' requiring large transmission power values and \textit{vice versa}.
 
The proposed resource-allocation algorithm is described in details in Algorithm~\ref{Ch4_alg:gms1}. The key operation in this algorithm is the iterative weight-based SA that is performed in parallel at all
femtocells. The SA weight for each subchannel and FUE pair
is defined as the multiplication of the estimated transmission
power and a scaling factor capturing the quality of the corresponding allocation. 
Specifically, each UE $i$ in cell $k$ estimates the transmission power on subchannel $n$ in each iteration $l$ of the algorithm
 by using the Foschini-Miljanic power update given in (\ref{Ch4_dispow}) as follows:
\begin{equation}
\label{Ch4_eq:minP}
p_i^{n,\texttt{min}}=  I_i^n(l) \bar{\gamma}_i^n.  
\end{equation}
The effective interference $I_i^n(l)$ given in (\ref{Ch4_eq:I}) can be estimated by UE $i$ as follows.
The serving BS of UE $i$ estimates the power channel gain $h^n_{b_i i}$ on subchannel $n$ from this user to itself  
and transmits this channel state information to UE $i$. In addition, 
the BS of UE $i$ measures the total received power on subchannel $n$ due to both desired and interfering signals, as well as the Gaussian noise,
 which is then sent back to UE $i$. UE $i$ can calculate $I_i^n(l)$ easily by using
the gathered information (i.e., its transmission power, power channel gain to the BS, and the total received power at the BS) according to (\ref{Ch4_eq:I}). 
We propose the following assignment weight for FUE $i$ on subchannel $n$ in cell $k$ as 
\begin{equation}
w_{i}^n = \chi_i^n p_i^{n,\texttt{min}},
\end{equation}
where the scaling factor $\chi_i^n$ is defined as follows: 
\begin{equation}
\label{Ch4_eq:w2}
\chi_i^n  = \left\lbrace \begin{array}{*{5}{l}}
\alpha_i^{n}, & \text{if } p_i^{n,\texttt{min}} \leq \frac{P_i^{\texttt{max}}} {\tau_k}\\
\alpha_i^{n} \theta_i^{n}, & \text{if } \frac{P_i^{\texttt{max}}} {\tau_k} < p_i^{n,\texttt{min}} \leq P_i^{\texttt{max}}\\
\alpha_i^{n}\delta_i^{n}, & \text{if }  P_i^{\texttt{max}} < p_i^{n,\texttt{min}},
\end{array} \right. 
\end{equation}
where $\tau_k$ denotes the number of subchannels assigned for each FUE in femtocell $k$; $\alpha_i^{n} \geq 1$ is a factor, which helps maintain SINR constraints of MUEs (i.e.,
it is increased if assigning subchannel $n$ to FUE $i$ tends to result in violation of the SINR constraint of the corresponding MUE); and $\theta_i^{n},\delta_i^{n} \geq 1$ are another factors that are set higher if the assignment of subchannel $n$ for FUE $i$ tends to require transmission
power larger than the average power per subchannel (i.e., ${P_i^{\texttt{max}}}/{\tau_k}$) and the maximum power budget (i.e., $P_i^{\texttt{max}}$), respectively.
We will set $\delta_i^{n}$ as $\delta_i^{n} =\mu_i \theta_i^{n}$ where  $\mu_i=N$ in Algorithm~\ref{Ch4_alg:gms1}.
Given the weights defined for each FUE $i$, femtocell $k$ finds the SA for its FUEs by
solving the following problem:
\begin{equation}
\label{Ch4_eq:lcl_prob}
\begin{array}{*{5}{c}}
{}&\min \limits_{\mathrm{\bf{A}}_k} \sum \limits_{i \in \mathcal{U}_k}\sum \limits_{n \in \mathcal{N}} a_i^n w_i^n\\
{\text{s.t.}} & \sum_{n \in \mathcal{N}} a_i^n = \tau_k, \:\: \forall i \in \mathcal{U}_k.
\end{array}
\end{equation}
This problem aims to find an assignment matrix $\bf{A}_{\textit{k}}$ for which each FUE in femtocell $k$ is assigned $\tau_k$ subchannels
with minimum total weight (i.e., to attain low-power SAs).
The optimization problem (\ref{Ch4_eq:lcl_prob}) can be transformed into the standard matching problem (for example, between ``jobs'' and ``employees'') as follows. 
Suppose we create $\tau_k$ virtual ``employees'' for each FUE in femtocell $k$, then
we can consider the matching problem between $\tau_k M_k$ virtual ``employees'' (virtual FUEs) and $N$ ``jobs'' (subchannels). In particular, FUE $i$ is equivalent to $\tau_k$ virtual FUEs $\{i_1,...,i_{\tau_k}\}$. Let the edge $v_{i_u}^n$ ($u \in \{1,...,\tau_k\}$) between subchannel $n$ and virtual 
FUE $i_u$ represent the assignment of that subchannel to the corresponding FUE. Then, the weight $w_{i_u}^n$ of the edge $v_{i_u}^n$ is equal to $w_{i}^n$.
After performing this transformation, the SA solution of the problem (\ref{Ch4_eq:lcl_prob}) can be found by using the standard Hungarian algorithm
(i.e., Algorithm 14.2.3 given in \cite{Jungnickel08}). After running the Hungarian algorithm, if there exits a virtual FUE $i_u$, $u \in \{1,...,\tau_k\}$, being matched with subchannel $n$, then we have $a^n_i=1$; otherwise, we set $a^n_i=0$.
For further interpretation of Algorithm~\ref{Ch4_alg:gms1}, let $W_k(l)$ denote the total minimum weight
due to the optimal solution of (\ref{Ch4_eq:lcl_prob}). The main operations of Algorithm~\ref{Ch4_alg:gms1} can be summarized as follows. 

\begin{itemize}

\item 
In steps 2–10, the MBS needs to estimate the effective
interference on all subchannels and calculates the transmission power values for the corresponding MUEs by
using (\ref{Ch4_eq:minP}). Then, the MBS checks the power constraints (\ref{Ch4_powcon}) for its associated MUEs. For MUEs that can maintain the power constraints (\ref{Ch4_powcon}), the MBS will send the
newly calculated transmission power values to them so that MUEs can update their power values accordingly
(step 5). Otherwise, if the power constraint of any MUE is violated, the MBS will scale down the transmission
power values to maintain the power constraints (step 7) and send the corresponding transmission power values to
its MUEs. 
In addition, the MBS will inform all FBSs, which will find the FUE creating the largest interference
to the victim MUE on the subchannel with the largest transmission power; then, we increase the parameter $\alpha_{m_i^{\ast}}^{n_i^{\ast}}$ corresponding to this FUE and subchannel pair by a factor
of 2 (steps 8 and 9). This can be realized by allowing the
MBS to request nearby FBSs to report the transmission power on the victim subchannel based on which the MBS
can identify and request the most interfering FUE $m_i^{\ast}$ to update its parameter $\alpha_{m_i^{\ast}}^{n_i^{\ast}}$. The signaling required by these operations can be accommodated by the wired backhaul
links [e.g., digital subscriber line (DSL) links].

\item 
In steps 15-20, each femtocell which has any FUEs's power constraints being violated in the previous iteration solves the SA
problem (\ref{Ch4_eq:lcl_prob}) with the current weight values to update its SA solution. In addition, FBS $k$ decreases the target number of assigned
 subchannels $\tau_k$ by one if the optimal total weight $W_k(l)$ of the problem (\ref{Ch4_eq:lcl_prob}) satisfies the condition in step 17. Each FBS can complete 
all required tasks in these steps since it knows current transmission power values and scaling factors $\chi_i^n$  of its FUEs based on which it can calculate
all SA weights. 

\item 
In step 12, each FBS estimates the effective interference on all subchannels and calculates the transmission power
values for its FUEs by using (20). Then, each FBS will check the power constraints of its FUEs using the newly
calculated transmission power values. For FUEs that have their power constraints satisfied, the corresponding FBS
will send the newly calculated transmission power values to them so that they can update power values accordingly (steps 22 and 23). For any FUE that has its power constraint violated, its FBS scales down the transmission
power values to meet the power constraints (step 25) and increase the $\theta_i^{n_i^{\ast}}$ parameter for the most interfering
subchannel by a factor of 2 (steps 26 and 27). In both cases, each FBS must send the transmission power levels
on all subchannels to the corresponding FUEs by using available control channels.
\end{itemize}

It can be observed that Algorithm~\ref{Ch4_alg:gms1} can be implemented distributively. In fact, the signaling information
required in steps 8 and 9 must be sent over the wired backhaul link.
In addition, each BS of either tier only needs to collaborate with its associated UEs to conduct all required tasks in other steps
where the required signaling is sent over the air by using available control channels.

\begin{algorithm}[!t]
\caption{\textsc{Distributed Uplink Resource Allocation}}
\label{Ch4_alg:gms1}
\begin{algorithmic}[1]
\STATE Initialization
\begin{itemize}
\item Set $p_i(0)=0$ for all UE $i$, $i \in \mathcal{U}_{\sf f}$, feasible $\mathrm{\bf{A}}_1$.
\item Set $\tau_k=\lfloor \frac{N}{\vert \mathcal{U}_k \vert} \rfloor$ and $\varrho_k=0$ for all $k \in \mathcal{B}_{\sf f}$.
\item Set $\alpha_i^{n}=\theta_i^{n}=1$, $\mu_i=N$, $\forall i \in \mathcal{U}_{\sf f}$, $n \in \mathcal{N}$.
\end{itemize}
\STATE \textbf{For the macrocell:}
\STATE MBS estimates  $I_i^n(l)$ and calculates $p_i^{n,\texttt{min}}$ for each MUE $i$ as in (\ref{Ch4_eq:minP}). Let 
$\beta_i={\sum_{n \in \mathcal{N}} a_i^n(l)p_i^{n,\texttt{min}}}/{P_i^{\texttt{max}}}$.

\IF{$\beta_i \leq 1$} \STATE Set $p_i^n(l+1)=a_i^n(l) p_i^{n,\texttt{min}}$, $\forall n \in \mathcal{N}$.
\ELSIF{$\beta_i > 1$} \STATE Set 
$p_i^n(l+1)\!=\!\frac{a_i^n(l)\! p_i^{n,\texttt{min}}}{\beta_i}\!,\! \forall n\! \in\! \mathcal{N}$
\STATE Find $n_i^{\ast}= \argmax \limits_{n \in \mathcal{N},c_n>1}a_i^n(l)$ $p_i^{n,\texttt{min}}$;
$m_i^{\ast} = \argmax \limits_{m \in \mathcal{U}_{\sf f}} a_m^{n_i^{\ast}}(l)p_m^{n_i^{\ast}}(l)h_{1m}^{\ast}$
\STATE Set  
$\alpha_{m_i^{\ast}}^{n_i^{\ast}}=2\alpha_{m_i^{\ast}}^{n_i^{\ast}}$ and $\varrho_{b_{m^{\ast}_i}}=0$.
\ENDIF

\STATE \textbf{For each femtocell $k \in \mathcal{B}_{\sf f}$:}

\STATE Each FBS $k$ estimates $I_i^n(l)$ and calculates $p_i^{n,\texttt{min}}$ for each FUE $i$ as in (\ref{Ch4_eq:minP}).

\IF{$\varrho_k=1$}  \STATE Keep $\mathrm{\bf{A}}_k(l)=\mathrm{\bf{A}}_k(l-1)$ 
\ELSIF{$\varrho_k=0$}  \STATE Calculate subchannel assignment weights $w_{i_u}^n$ between $\mathcal{N}$ and $\cup_{i \in \mathcal{U}_k}\{i_1,...,i_{\tau_k}\}$, as in (\ref{Ch4_eq:w2}) and run 
Hungarian algorithm with $\left\lbrace w_{i_u}^n \right\rbrace $ to obtain $W_k(l)$ and $\mathrm{\bf{A}}_k(l)$.
\IF{$W_k(l) > V \sum _{i \in \mathcal{U}_k}P_i^{\texttt{max}}$} \STATE Set $\tau_k:=\tau_k-1$. 
\ENDIF

\ENDIF

\STATE Let $\beta_i={\sum_{n \in \mathcal{N}} a_i^n(l)p_i^{n,\texttt{min}}}/{P_i^{\texttt{max}}}$. 

\IF{$\beta_i \leq 1$} \STATE Set $p_i^n(l+1)=a_i^n(l) p_i^{n,\texttt{min}}$, $\forall n \in \mathcal{N}$ and $\varrho_{k,i}=1$. 
\ELSIF{$\beta_i > 1$} \STATE Set $p_i^n(l+1)\!=\!\frac{a_i^n(l)\! p_i^{n,\texttt{min}}}{\beta_i}\!,\! \forall n\! \in\! \mathcal{N}$
\STATE Find $n_i^{\ast}\!=\! \argmax \limits_{n \in \mathcal{N}}a_i^n(l)$ $p_i^{n,\texttt{min}}$
 \STATE Set $\theta_i^{n_i^{\ast}}=2\theta_i^{n_i^{\ast}}$ and $\varrho_{k,i}=0$.
\ENDIF
\STATE Set $\varrho_{k}=\prod_{i \in \mathcal{U}_k}\varrho_{k,i}$.
\STATE Let $l:=l+1$, return to step 2 until convergence.
\end{algorithmic}
\end{algorithm}

\subsection{Convergence and Complexity Analysis of Algorithm \ref{Ch4_alg:gms1}}
\label{Ch4_sec:cvrg_cplx}
\subsubsection{Convergence Analysis}
\label{Ch4_sec:cvrg}
The convergence of Algorithm~\ref{Ch4_alg:gms1} is stated in the following.
\begin{prop}
\label{Ch4_thm03}
Algorithm \ref{Ch4_alg:gms1} converges to a feasible solution $(\mathrm{\bf{A}},\mathrm{\bf{P}})$ of the optimization problem (\ref{Ch4_objfun}) and (\ref{confun}).
\end{prop}

\begin{proof}
To prove this proposition, we will consider two possible scenarios.
For the first scenario, there exists an iteration after which the scaling
 factors $\chi_i^n$ for SA weights given in (\ref{Ch4_eq:w2})
do not change. According to Algorithm~\ref{Ch4_alg:gms1}, after this iteration, the SA solution will remain unchanged.
In addition, Algorithm~\ref{Ch4_alg:gms1} simply 
updates transmission power values $p_i^{n,\texttt{min}}$ for all UEs on their assigned subchannels by using the Foschini-Miljanic power updates. The authors in \cite{foschini93} have shown that these power updates indeed converge, which implies the convergence of Algorithm~\ref{Ch4_alg:gms1}.

Now suppose that the scaling factors $\chi_i^n$ are still changed over iterations, we will prove that the 
system will ultimately evolve into the first scenario previously discussed.
First, it can be verified that the power $p_i^{n,\texttt{min}}$ given in (\ref{Ch4_eq:minP}) is lower bounded by $\bar{\gamma}_i^n \eta_{b_i}^n/h_{b_ii}^n$.
Therefore, if the scaling factors $\chi_i^n$ keep increasing over iterations, then the total weight $W_k(l)$ returned by the assignment
problem (\ref{Ch4_eq:lcl_prob}) will increase over iterations as well. Therefore, there exist some femtocells $k$ that
decrease their number of assigned subchannels $\tau_k$ over iterations (steps 17--19 of Algorithm~\ref{Ch4_alg:gms1}). Since initial values of all $\tau_k$ are 
finite, this process will terminate after a finite number of iterations. Then, the system will be in the first scenario discussed above; therefore, 
Algorithm~\ref{Ch4_alg:gms1} converges. Therefore, we have completed the proof of the proposition.
\end{proof}

\subsubsection{Complexity Analysis and Comparison with Existing Algorithms}
\label{Ch4_sec:cplx}
It is observed that the major complexity of Algorithm~\ref{Ch4_alg:gms1} is involved in solving the SA by using the Hungarian method in step 16. 
According to \cite{Jungnickel08}, the complexity of the Hungarian algorithm is $O(N^3)$. Therefore, the complexity of our proposed
algorithm is $O(K \times N^3)$ for each iteration. However, the local SAs can be performed in parallel at all $(K-1)$ 
femtocells. Therefore, the runtime complexity of our algorithm is $O(N^3)$ multiplied by the number of required iterations, which is 
quite moderate according to our simulation results (i.e., tens of iterations).

For comparison purposes, we summarize how existing resource allocation formulations and algorithms for
two-tier macro-femto networks cover different design aspects in Table~\ref{Ch4_tb:notation}
where we write NA for the convergence and complexity aspects if they are not analyzed in these existing works. As we can
see, the existing works consider different optimization objectives and cover some design issues although none of them
accounts for all aspects. Our proposed algorithm captures all
design aspects except the downlink scenario.
Moreover, only Zhang \textit{et al.} \cite{zhang12} and we consider the power constraints in our
resource allocation. However, Zhang \textit{et al.} aim to maximize the
total throughput of the femtocell network while constraining
the cross-tier interference from FUEs to the MBS. In their
work, the co-tier interference among femtocells is assumed to
be part of the noise power, which is not explicitly managed.
Given their comparable design and our proposed algorithm, we
will conduct performance comparison for the two algorithms in
terms of throughput and fairness in Section~\ref{Ch4_result}.

\begin{table*}[!t] 
\footnotesize
\caption{Summary of Existing Algorithms}
\centering
\begin{tabular}{| c | p{30mm} | p{8mm} | p{8mm} | p{8mm} | p{10mm} | p{10mm} | c | c |} 
\hline 
Papers & Objective function & UL/ DL & SA/ PA & Prt. MUE & Pw. Const. & CoTI/ CrTI & Conv. & Complexity \\ 
\hline \hline
Ref. \cite{chuhan11} & Max-min number of assigned subchannels & Both & SA & Yes & No & No & NA & NA \\
\hline
Ref. \cite{yanzan12} & Max sum rate & UL & SA & Yes & No & Both & NA & NA \\
\hline
Ref. \cite{yushan12} & Max sum PRB & Both & SA & No & No & NA & NA & $O(KN \log N)$\\
\hline
Ref. \cite{SMCheng12} & NA (Analysis) & DL & SA & No & No & Both & NA & NA \\
\hline
Ref. \cite{hoon11} & Min error between interference and predetermined thresholds at MUEs & UL & Both & Yes & No & Both & NA & NA \\
\hline
Ref. \cite{renchao12} & Max energy efficiency & Both & Both & No & No & No & Yes & NA \\
\hline
Ref. \cite{Lu13} & Max-min effective rate & DL & PA & No & No & CoTI & NA & NA \\
\hline
Zhang's paper \cite{zhang12} & Max sum rate & UL & Both & Yes & Yes & CrTI & NA & NA \\
\hline
Our paper & Max-min rate & UL & Both & Yes & Yes & Both & Yes & $O(KN^3)$ \\ 
\hline \hline
\multicolumn{4}{|l|}{UL: Uplink. \; SA: Subchannel assignment.} & \multicolumn{5}{l|}{DL: Downlink. \; PA: Power allocation.}  \\ 
\multicolumn{4}{|l|}{Prt. MUE: Protect the QoS of MUEs} & \multicolumn{5}{l|}{Pw. Const.: Power constraint}  \\
\multicolumn{4}{|l|}{CoTI: Co-tier interference among femtocells } & \multicolumn{5}{l|}{CrTI: Cross-tier interference}  \\
\multicolumn{4}{|l|}{Conv.: Convergence guarantee} & \multicolumn{5}{l|}{PRB: Physical resource block}  \\
\hline
\end{tabular}
\label{Ch4_tb:notation}
\end{table*}
\nomenclature{DL}{Downlink}

\section{Further Extensions}
\label{Ch4_sec:extend}
\subsection{Downlink Resource Allocation} 
\label{Ch4_sec:dwlk}
For the downlink context, we can use the same notations as in the uplink system. 
However, $h_{ij}^n$ and $\eta_i^n$ are the power channel gain from  BS $\mathit{j}$
UE $i$ and noise power at UE $i$, respectively. 
In the downlink system, we need to impose the power constraints for the BSs instead of UEs as follows:
\beq \label{Ch4_eq:pwcon_dlk}
\sum_{i \in \mathcal{U}_k}\sum_{n \in \mathcal{N}}p^n_i \leq P_{\sf BS, \textit{k}}^{\texttt{max}}, \quad k \in \mathcal{B},
\eeq
where $P_{\sf{BS},k}^{\texttt{max}}$ is the maximum power of BS $k$. Then, the downlink resource-allocation problem can be formulated as
\begin{eqnarray}
\label{Ch4_dlk_objfun1}
\mathop {\max} \limits_{(\mathrm{\bf{A}},\mathrm{\bf{P}}) } \sum \limits_{2\leq k \leq K} \mathrm{R}^{(k)}(\mathrm{\bf{A}},\mathrm{\bf{P}}) \\
\text{s. t. constraints} \quad  (\ref{Ch4_eq:c4}), (\ref{Ch4_eq:rate_cond}), (\ref{Ch4_eq:pwcon_dlk}). \label{Ch4_confun1}
\end{eqnarray}

We can employ Algorithm~\ref{Ch4_alg:gms1} to find a feasible solution of this problem with only some changes in the scaling factor $\chi_i^n$
of the assignment weight $w_{i}^n = \chi_i^n p_i^{n,\texttt{min}}$ as follows:
\begin{equation}
\label{Ch4_eq:w3}
\chi_i^n  = \left\lbrace \begin{array}{*{5}{l}}
\alpha_i^{n}, & \text{if } p_i^{n,\texttt{min}} \leq P_{\sf BS, {\it b_i}}^{\texttt{max}}\\
\infty, & \text{if } p_i^{n,\texttt{min}} > P_{\sf BS, {\it b_i}}^{\texttt{max}}.
\end{array} \right. 
\end{equation}

In fact, the structure of power constraints in the downlink
context is less complicated than those in the uplink setting since
we only need to maintain the total power constraint for each BS
instead of all UEs connected with each BS. This explains the
simpler structure of (\ref{Ch4_eq:w3}) compared with that in (\ref{Ch4_eq:w2}). 

The signaling requirements for the downlink case are a bit different from those in the uplink one, which is described
in the following. To complete the tasks in steps 2–10, each MUE needs to estimate the current effective interference on
their allocated subchannels based on which the MUE calculates the transmission power values on the corresponding
subchannels. All MUEs then send the updated transmission power values to the MBS. The MBS updates the transmission
power values on all subchannels accordingly if the power constraint (\ref{Ch4_eq:pwcon_dlk}) is satisfied (step 5). 
Otherwise, if the power constraint of the MBS is violated, the MBS scales
down transmission power values on all subchannels by a factor $\beta^{1}={\sum_{n \in \mathcal{N}} \sum_{i \in \mathcal{U}_1} a_i^n(l)p_i^{n,\texttt{min}}}/{P^{\sf max}_{\sf BS, \textrm{1}}}$ to maintain the power constrailt. In addition, the MBS 
 informs all FBSs, which will find the FUE creating the largest interference to 
the victim MUE and increase the parameter $\alpha_{m_i^{\ast}}^{n_i^{\ast}}$ by a factor of 2 (see steps 8 and 9).
Here, the victim MUE $i$ is
the one that is allocated subchannels $n_i^{\ast}$ satisfying 
 $n_i^{\ast} = \argmax_{n \in \mathcal{N},c_n>1}a_i^n(l)$ $p_i^{n,\texttt{min}}$.

The signaling in these tasks can be realized as follows.
MUEs report the new transmission power values of their subchannels to the MBS via a control channel. In addition, the
MBS requests nearby FBSs to report the transmission power values
 on subchannel $n_i^{\ast}$, and based on this, the MBS can identify and request the most interfering
FUE $m_i^{\ast}$ on subchannel $n_i^{\ast}$ to update its parameter $\alpha_{m_i^{\ast}}^{n_i^{\ast}}$ by using the wired backhaul links (e.g., DSL links). 
In steps 15–20, each FBS can complete all required tasks in these steps by requiring its
FUEs to report newly calculated transmission power values
based on which it can calculate all SA weights. 
Note that
each FBS knows the current values of scaling factor $\chi_i^n$ of the SA weights. 
The required signaling between FUEs and their corresponding FBSs can be accommodated by using available
wireless control channels. Finally, each FBS that has its power constraint satisfied updates its transmission power (see steps 22
and 23). In addition, any FBS $k$ that has its power constraint 
violated scale downs transmission power values by a factor $\beta^{k}={\sum_{n \in \mathcal{N}} \sum_{i \in \mathcal{U}_k} a_i^n(l)p_i^{n,\texttt{min}}}/{{P^{\sf max}_{\sf BS, \textit{k}}}}$
and increases the $\theta_i^{n_i^{\ast}}$ parameter for the subchannel with largest transmission power
by a factor of 2 (see steps 26 and 27).

\subsection{Adaptive-Rate Resource Allocation} \label{Ch4_sec:mulrate}
We can relax the fixed-rate assumption and enable FUEs in each femtocell $k$ to choose any constellation size (modulation scheme) $s_k \in \mathfrak{M}$ 
for transmissions on their assigned subchannels. The target SINRs for different modulation schemes $s_k\text{-QAM}$ $\bar{\gamma}(s_k)$ can be calculated 
as in (\ref{Ch4_tsnr}). Hence, for a given $(s_k,\tau_k)$, the minimum spectral efficiency achieved by each FUE in femtocell $k$ is
\beq \label{Ch4_eq:Rate_rate}
\mathrm{R}^{(k)}(s_k,\tau_k)= r_k \tau_k,
\eeq
where $\tau_k$ is the number of subchannels assigned for each FUE in femtocell $k$ and $r_k={\log_2 s_k}/{N}$.
The resource-allocation problem for this adaptive-rate context is 
the same with (\ref{Ch4_objfun}) and (\ref{confun}) except that different femtocells $k$ can choose different constellation size $s_k \in \mathfrak{M}$.
Hence, we have following SINR constraints for each FUE $i$ connected with FBS $b_i$:
\begin{eqnarray}
\label{Ch4_objfun3}
 \Gamma_i^n(\mathrm{\bf{A}},\mathrm{\bf{P}}) \geq \bar{\gamma}(s_{b_i}), \forall i \in \mathcal{U}_{\sf f}, s_{b_i} \in \mathfrak{M}.
\end{eqnarray}

The equal number of subchannels assigned to each FUE in femtocell $k$ is at most $\lfloor{N}/{M_k}\rfloor$, and the maximum spectral efficiency of each subchannel is  $\bar{r}={\log_2 s_{\sf max}}/{N}$, where $s_{\sf max}$ is the maximum allowable constellation size in $\mathfrak{M}$.
Hence, the maximum spectral efficiency achieved by any FUE is $\bar{r} \lfloor{N}/{M_k}\rfloor$.

Resource-allocation algorithm for this adaptive-rate setting can be performed by slightly adapting Algorithm~\ref{Ch4_alg:gms1}, which
is described in Algorithm~\ref{Ch4_alg:gms2}. 
We list and sort all possible couples $(s_k,\tau_k)$ in the decreasing order of their $\mathrm{R}^{(k)}(s_k,\tau_k)=r_k \tau_k$ to obtain 
a sorted list $\Theta_k$. Then, the weight-based SAs can be performed by iteratively  
updating $(s_k,\tau_k)$ for all femtocells in the same order of the sorted list. In particular, if the current pair $(s_k,\tau_k)$
of any femtocell $k$ results in violation of the condition, as specified in step 10 of Algorithm~\ref{Ch4_alg:gms2}, then we update the
pair $(s_k,\tau_k)$ by the next one in the list $\Theta_k$ (i.e., the element $(t_k+1)$ in the sorted list $\Theta_k$ is chosen), which has 
lower value of $\mathrm{R}^{(k)}(s_k,\tau_k)$. 
These updates are intuitive since the search space for this adaptive-rate setting comprises both the number of subchannels assigned
for each FUE ($\tau_k$) and the modulation scheme represented by $s_k$.
These operations are repeated until convergence.

\begin{algorithm}[!t]
\caption{\textsc{Adaptive-Rate Resource Allocation}}
\label{Ch4_alg:gms2}
\begin{algorithmic}[1]
\STATE Initialization
\begin{itemize}
\item Set $p_i^{(0)}=0$ for all UE $i$, $i \in \mathcal{U}_{\sf f}$, feasible $\mathrm{\bf{A}}_1$.
\item Set $t_k=1$ and $\varrho_k=0$ for all $k \in \mathcal{B}_{\sf f}$.
\item Set $\mu_i=2$, $\alpha_i^{n}=1$, $\forall i \in \mathcal{U}_{\sf f}$, $n \in \mathcal{N}$.
\end{itemize}
\STATE \textbf{For the macrocell:}
\STATE Update power values for MUEs and $\alpha_i^n$ factors as in steps 3--9 of Algorithm \ref{Ch4_alg:gms1}.
\STATE \textbf{For each femtocell $k \in \mathcal{B}_{\sf f}$:}
\STATE FBS $k$ estimates $I_i^n(l)$ and calculates $p_i^{n,\texttt{min}}$ for each FUE $i$ using (\ref{Ch4_eq:minP}) based on the target SINR $\bar{\gamma}(s_k^{(t_k)})$.
\IF{$\varrho_k=1$}  \STATE Keep $\mathrm{\bf{A}}_k(l)=\mathrm{\bf{A}}_k(l-1)$ 
\ELSIF{$\varrho_k=0$} 
 \STATE Calculate subchannel assignment weights $w_{i_u}^n$ between $\mathcal{N}$ and $\cup_{i \in \mathcal{U}_k}\{i_1,...,i_{\tau_k^{(t_k)}}\}$, as in (\ref{Ch4_eq:w2}) and run Hungarian algorithm with $\left\lbrace w_{i_u}^n \right\rbrace $ to obtain $W_k(l)$ and $\mathrm{\bf{A}}_k(l)$.
\IF{$W_k(l) > V\sum _{i \in \mathcal{U}_k}P_i^{\texttt{max}}$} \STATE Set $t_k:=t_k+1$, $\theta_i^{n}=1$, $\forall i \in \mathcal{U}_k$, $n \in \mathcal{N}$. 
\ENDIF
\ENDIF
\STATE Update power values for FUEs and factors as steps 21--29 of Algorithm \ref{Ch4_alg:gms1}.
\STATE Set $l:=l+1$, return to step 2 until convergence.
\end{algorithmic}
\end{algorithm}

\subsection{Resource Allocation with Hybrid Access} \label{sec:hybdaccess}
We can extend the considered problem to implement a hybrid access strategy, which can maintain the required QoS of MUEs
while enhancing the performance of the femtocell tier. It can be observed that in the proposed algorithms
(Algorithms \ref{Ch4_alg:gms2}, \ref{Ch4_alg:gms2}), there may exist femtocells that do not utilize all available
subchannels since doing so may prevent them from maintaining the target SINRs for all FUEs and MUEs. Therefore, it
may be beneficial if we allow MUEs that potentially create and/or suffer from strong cross-tier interference from nearby
femtocells to change their connections from the MBS to nearby FBS. However, this can only be performed if the subchannels
assigned to the underlying MUE are not utilized by the target femtocell.

We present the resource allocation algorithm with rate adaptation under such hybrid access in Algorithm~\ref{Ch4_alg:gms3}.
To interpret the operations of this algorithm, we first give definitions of involved quantities and parameters in the following.
Let $b_i(l)$ and $\mathcal{N}_i(l)$ be the BS of UE $i$ and the set of its assigned subchannels in iteration $l$, respectively. 
Note that we have assumed that each MUE has a fixed set of assigned subchannels and
FUEs have fixed associated FBSs. Therefore, we have $b_i(l)=b_i$ for each FUE $i$ and $\mathcal{N}_i(l)=\mathcal{N}_i$ for each MUE $i$ as specified by $\mathrm{\bf{A}}_1$. In addition,
we can estimate the total required transmission power over the assigned subchannels of MUE $i$ that is connected with BS $k$ as 
\beq \label{Ch4_eq:ttl_rq_pw}
P_{i,k}=\bar{\gamma}(s^{\sf m})\sum \limits_{n \in \mathcal{N}_i} I_{i,k}^n(l)
\eeq
where $I_{i,k}^n(l)$ is the effective interference corresponding to MUE $i$ for its connection with BS $k$ on subchannel $n$, which can be calculated
as in (\ref{Ch4_eq:I}). MUE $i$ can determine $P_{i,k}$ for each potential FBS $k$ by estimating $I_{i,k}^n(l)$ on each assigned subchannel $n$.
This can be realized as we have discussed in the
paragraph below (\ref{Ch4_eq:minP}).
For the SA in any iteration $l$, we define the set of potential FBSs for MUE $i$ as follows:
\beq \label{Ch4_eq:bi}
\begin{array}{*{2}{l}}
\mathcal{B}_i(l)=&\left\lbrace k| k \in \mathcal{B}_{\sf f}, \mathcal{N}_i \subseteq \mathcal{N} / \cup_{j \in \mathcal{U}_k} \mathcal{N}_j(l-1), \right.\\ 
{}& \left. P_{i,k} < \min(P_{i,1}, P_i^{\texttt{max}}), |\mathcal{N}_i| \leq Q_k \right\rbrace 
\end{array}
\eeq
where $Q_k$ is the number of unused subchannels at FBS $k$ in the underlying iteration.
Here, $\mathcal{B}_i(l)$ is the set of FBSs whose set of unused subchannels contains the set $\mathcal{N}_i$
of MUE $i$; connection between MUE $i$ with the underlying FBS requires less power than connection between MUE $i$ and the MBS.
We also define the set $\mathcal{U}_{m,k}(l)=\{i \in \mathcal{U}_{\sf m}| b_i(l)=k\}$, which is the set of MUEs connecting with BS $k$ in iteration $l$.
 
In Algorithm~\ref{Ch4_alg:gms3}, we do not change existing associations between any FBS $k$ and its corresponding MUEs until the chosen pair 
$(s_k,\tau_k)$ in this femtocell
 is updated (i.e., the next pair $(s_k,\tau_k)$ in the sorted list $\Theta_k$ is chosen).
Specifically, if $t_k$ is updated (i.e., in step 13) then all MUEs currently connected with FBS $k$ are forced to change
 their connections back to the MBS (see step 14); otherwise, we maintain the BS associations of all MUEs as they are (see steps 15 and 16).
In addition, only MUEs currently connecting with the MBS (i.e., any MUE $i$ with $b_i(l)$ equal $1$) are allowed to change their BS association in iteration $l$.
In particular, the couple of FBS and MUE requiring the smallest transmission power will be chosen for association in each iteration (see step 7). 
Other operations corresponding to SA, PA, and updates of various parameters are performed as in Algorithm~\ref{Ch4_alg:gms2}.

\begin{algorithm}[!t]
\caption{\textsc{Adaptive-Rate Resource Allocation with Hybrid Access}}
\label{Ch4_alg:gms3}
\begin{algorithmic}[1]
\STATE Initialization: Set parameters as step 1 of Algorithm \ref{Ch4_alg:gms2} and $Q_k=0$, $\forall k \in \mathcal{B}_{\sf f}$.
\STATE \textbf{For the macrocell:}
\STATE Estimate $P_{i,k}$, update $\mathcal{B}_i(l)$, set $\mathcal{B}=\cup_{i \in \mathcal{U}_{m,1}(l)} \mathcal{B}_i(l)$.
\IF{$\mathcal{B} = \varnothing$} 
\STATE Go to step 10.
\ELSIF{$\mathcal{B} \neq \varnothing$} 
\STATE Find $(k^{\star},i^{\star})=\argmin_{ i \in \mathcal{U}_{m,1}(l) } \min_{k \in \mathcal{B}_i } P_{i,k}$.
\STATE Set $b_{i^{\star}}\!(l)\!=\!k^{\star}$, $ Q_{k^{\star}}\!=\!Q_{k^{\star}}\!-\!|\!\mathcal{N}_{i^{\star}}|$, $\varrho_{k^{\star}}\!=\!0$, and go
back to step 3.
\ENDIF
\STATE Update power values and all parameters as in steps 3--10 of Algorithm \ref{Ch4_alg:gms1}.

\STATE \textbf{For each femtocell $k \in \mathcal{B}_{\sf f}$:}
\STATE Update the subchannel, PA, and other parameters for all FUEs in cell $k$ as steps 5--14 of Algorithm~\ref{Ch4_alg:gms2} where
the Hungarian algorithm is run for the set of subchannels  $\mathcal{N}/\cup_{i \in \mathcal{U}_{m,k}(l)}\mathcal{N}_i$ and the set
 of virtual UEs $\cup_{i \in \mathcal{U}_k}\{i_1,...,i_{\tau_k^{(t_k)}}\}$. 
\IF {$t_k$ changes} 
\STATE Set $Q_k=N-\tau_k^{(t_k)}M_k$, $b_i(l+1)=1$, $\forall i \in \mathcal{U}_{m,k}(l)$.
\ELSIF {$t_k$ does not change} 
\STATE Set $Q_k\!=\!N-\!\sum \limits_{i, b_i(l)=k}\!\!|\mathcal{N}_i(l)|$, $b_i(l+1)\!\!=\!k$, $\forall i\! \in \! \mathcal{U}_{m,k}(l)$.
\ENDIF
\STATE Let $l:=l+1$, return to step 2 until convergence.
\end{algorithmic}
\end{algorithm}

Algorithm~\ref{Ch4_alg:gms3} can be implemented as follows. To complete the tasks in steps 3--9, each MUE $i$
upon estimating $P_{i,k}$ for all potential FBSs $k \in \mathcal{B}_i$ will send a connection request to the FBS that requires the minimum power
(i.e., $\min_{k \in \mathcal{B}_i } P_{i,k}$) together with the value of $P_{i,k}$. If a particular FBS $k$ receives several connection requests then the FBS will
grant the connection to only one MUE with the smallest value of $P_{i,k}$. Signaling information needed to exchange the connection request and granting messages
can be accommodated by available control channels. In steps 13 and 14, if the next pair $(s_k,\tau_k)$ with index $t_k$ in the sorted list $\Theta_k$ 
in femtocell $k$ is chosen (i.e., $t_k$ is increased by one), then FBS $k$ simply requests all associated MUEs to change their connections back to
the MBS. Otherwise, if $t_k$ remains unchanged, then FBS $k$ simply 
updates parameter $Q_k$ based on the new SA assignment decisions in step 12. 

\section{Numerical Results}
\label{Ch4_result}

\begin{figure}[!t]
\begin{center}
\includegraphics[width=0.7 \textwidth]{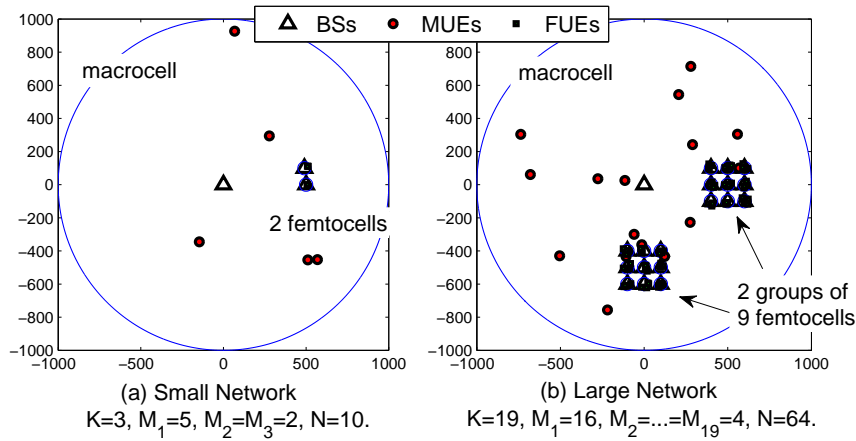}
\end{center}
\caption{Macrocell-femtocell networks used in simulation}
\label{Ch4_fig:system}
\end{figure}

We obtain numerical results for two different networks (with a small and large number of femtocells and UEs) to evaluate the efficacy of our proposed algorithms. The network setting and UE placement for our simulations 
are illustrated in Fig.~\ref{Ch4_fig:system}, where MUEs and FUEs are randomly located inside circles of radii of $r_1 = 1000\:m$ and $r_2 =30\:m$, respectively. The power channel gains $h_{ij}^n$ are generated by considering both 
Rayleigh fading, which is represented by an exponentially distributed random variable with the mean value of one, and the path-loss $L_{ij}=A_i \mathrm{log}_{10}(d_{ij})+B_i+C\mathrm{log}_{10}(\frac{f_c}{5})+ WL \times n_{ij}$, where 
$d_{ij}$ is the distance from UE $j$ to BS $i$; $(A_i,B_i)$ are set as $(36,40)$ and $(25,45)$ for MBS and FBSs, respectively; $C=20$, $f_c=2.5\:GHz$; $WL$ is the wall-loss parameter, $n_{ij}$ is the number of walls between BS $i$ and UE $j$.
This path-loss model is chosen according to the path-loss formula in \cite{winner07}, which is suggested by the WINNER II channel modeling project.
 The noise power is set as $\eta_i=10^{-13} \: W$, $\forall i \in \mathcal{B}$. 

\begin{table}[!t]
\caption{Target SINRs for different constellation sizes with target BER $\overline{P}_e=10^{-3}$.}
\centering
\begin{tabular}{|c| c| c |c |c| c|}
\hline 
$s\text{-QAM}$ & $4$ & $16$ & $64$ & $256$ & $1024$ \\ \hline 
bits/symbol & $2$ & $4$ & $6$ & $8$ & $10$ \\ \hline
$\bar{\gamma}(s)$ &  $9.55$ &  $45.11$& $179.85$ & $694.17$  & $2667.32$ \\ 
\hline 
\end{tabular}
\label{Ch4_tb:tSINR}
\end{table}

To obtain simulation results, we use two modulation schemes ($4\text{-}$ and $16\text{-QAM}$) for MUEs and five modulation schemes ($4\text{-}$, $16\text{-}$, $64\text{-}$, $256\text{-}$ and $1024\text{-QAM}$) for FUEs. Moreover, we choose the target BER $\overline{P}_e=10^{-3}$ whose target SINRs corresponding to 
different modulation schemes, which can be calculated by using (\ref{Ch4_eq:tSINR}), are given in Table~\ref{Ch4_tb:tSINR}. Each simulation result is obtained by taking 
the average over 20 different runs where, for each run, UEs are randomly located and a feasible SA matrix for MUEs $\mathrm{\bf{A}}_1$ is chosen so that each MUE 
is assigned ${N}/{M_1}$ subchannels.

\subsection{Performance of Proposed Algorithms}

\begin{figure}[!t]
\begin{center}
\includegraphics[width=0.7 \textwidth]{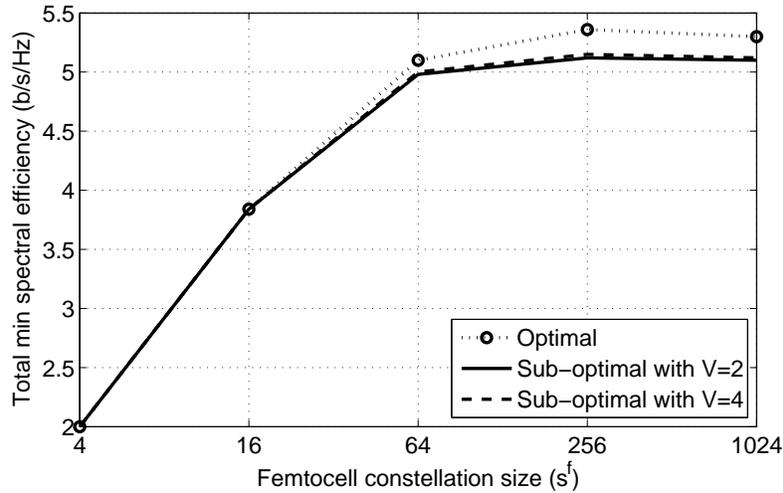}
\end{center}
\caption{Total minimum spectral efficiency for small network under optimal and supoptimal algorithms for $s^{\sf{m}}=4$, $WL=5\:dB$, $P^{\texttt{max}}_{\sf{m}}=P^{\texttt{max}}_{\sf{f}}=0.01\:W$.}
\label{Ch4_fig:compare}
\end{figure}

\begin{figure}[!t]
\begin{center}
\includegraphics[width=0.7 \textwidth]{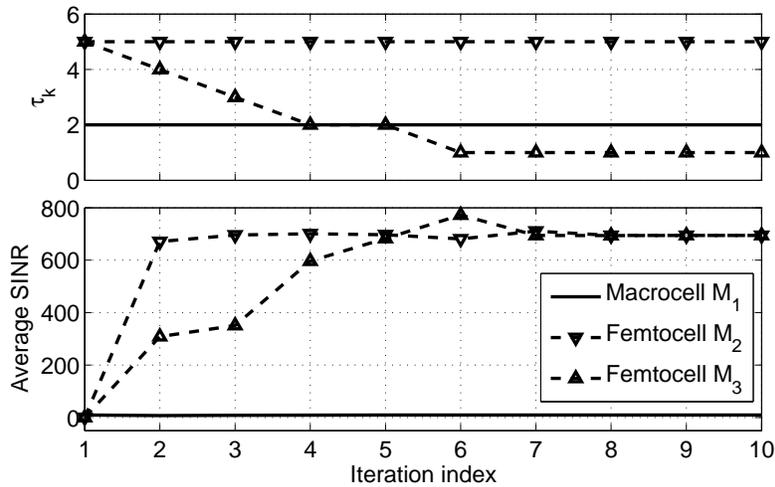}
\end{center}
\caption{Minimum number of subchannels assigned for FUEs and average SINRs versus iteration index where $s^{\sf{m}}=4$, $s^{\sf{f}}=256$, $WL=5\:dB$, 
$P^{\texttt{max}}_{\sf{m}}=P^{\texttt{max}}_{\sf{f}}=0.01\:W$.}
\label{Ch4_fig:cnvg}
\end{figure}

In Fig.~\ref{Ch4_fig:compare}, we show the total minimum spectral efficiency of all femtocells [i.e., the optimal objective value of (\ref{Ch4_objfun})] versus
 the constellation size of FUEs ($s^{\sf{f}}$) for the
 small network due to both optimal and supoptimal algorithms (see Algorithm~\ref{Ch4_alg:gms1}), which are presented in Section~\ref{Ch4_sec:OSOA}. As shown, for the
  low constellation sizes
(i.e., low target SINRs), Algorithm~\ref{Ch4_alg:gms1} can achieve almost
 the same spectral efficiency as the optimal algorithm, whereas for higher values of $s^{\sf{f}}$, Algorithm~\ref{Ch4_alg:gms1} results in just slightly lower spectral efficiency 
 than that due to the optimal algorithm.
Moreover, when we increase the value of parameter $V$, a slightly better performance can be achieved.
We then plot the minimum number of subchannels assigned for FUEs and the average SINRs over assigned subchannels versus iteration index for 
the small network in Fig.~\ref{Ch4_fig:cnvg}. This figure shows the convergence of Algorithm~\ref{Ch4_alg:gms1} in terms of both user SINR and 
number of assigned subchannels per FUE.

\begin{figure}[!t]
\begin{center}
\includegraphics[width=0.7 \textwidth]{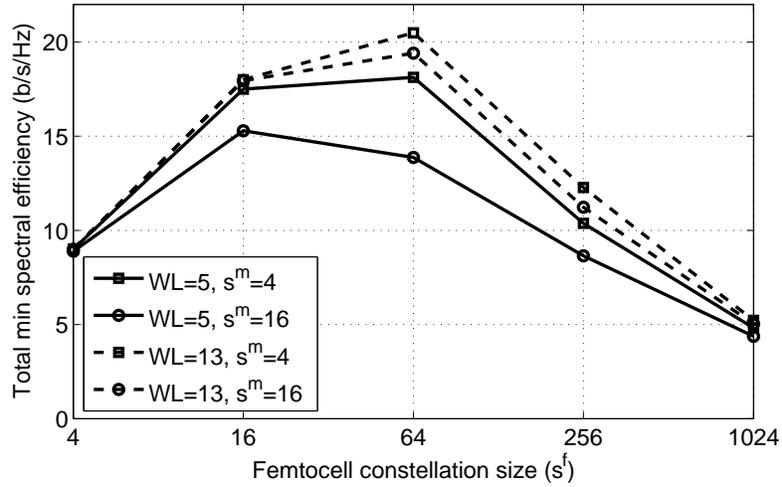}
\end{center}
\caption{Total minimum spectral efficiency versus $QAM$ constellation size of femtocells with $P^{\texttt{max}}_{\sf{m}}=P^{\texttt{max}}_{\sf{f}}=0.01\:W$.}
\label{Ch4_fig:wl5}
\end{figure}

Fig.~\ref{Ch4_fig:wl5} shows the total femtocell minimum spectral efficiency versus the $QAM$ constellation size, which is obtained by running Algorithm~\ref{Ch4_alg:gms1}
 for the large network. This figure shows that the total minimum spectral efficiency increases and then decreases as femtocell constellation size increases; it decreases as the constellation size of MUEs increases. These results can be interpreted as follows. Higher constellation sizes have
 higher target SINRs, which require higher transmission power and produce more interference to other users in the network.
Therefore, the power constraints of FUEs and MUEs are more likely to be violated for higher constellation sizes (e.g., 254
or 1024), which limits the number of subchannels allocated for each FUE. Moreover, for low modulation levels (therefore,
low target SINRs of FUEs), the spectral efficiency on each assigned subchannel increases with the increasing modulation
level, whereas the number of subchannels assigned for each FUE does not decrease too much. Moreover, the total minimum
spectral efficiency of femtocells increases with the increasing wall-loss value $W_l$. This is because the higher wall loss
reduces both co-tier and cross-tier interference to users of both tiers.

\begin{figure}[!t]
\begin{center}
\includegraphics[width=0.7 \textwidth]{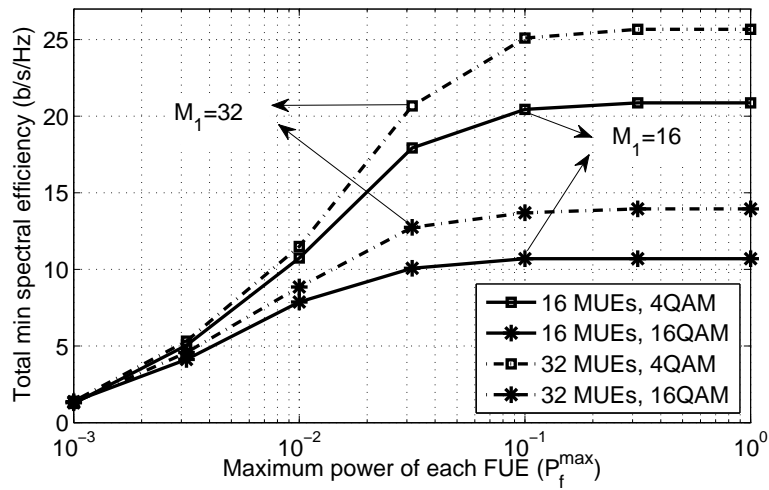}
\end{center}
\caption{Total minimum spectral efficiency versus $P_{\sf{f}}^{\texttt{max}}$ for the large network with $s^{\sf{f}}=256$, $WL=5\:dB$, $P^{\texttt{max}}_{\sf m}=0.01\:W$.}
\label{Ch4_fig:pf}
\end{figure}

In Figs.~\ref{Ch4_fig:pf} and \ref{Ch4_fig:pm}, we plot the total femtocell minimum spectral efficiency versus the maximum power of FUEs ($P_{\sf{f}}^{\texttt{max}}$) and MUEs ($P_{\sf{m}}^{\texttt{max}}$), respectively, for different modulation levels of MUEs (the modulation scheme of FUEs is $256-QAM$). These figures show that the total 
minimum spectral efficiency increases with the increase of maximum power budgets $P_{\sf{f}}^{\texttt{max}}$ or $P_{\sf{m}}^{\texttt{max}}$.
However, this value is
saturated as the maximum power budgets $P_{\sf f}^{\texttt{max}}$ or $P_{\sf m}^{\texttt{max}}$ become sufficiently large.
In addition, as the number of MUEs increases, the total femtocell minimum spectral efficiency increases due to the better diversity gain offered by the macro tier.

\begin{figure}[!t]
\begin{center}
\includegraphics[width=0.7 \textwidth]{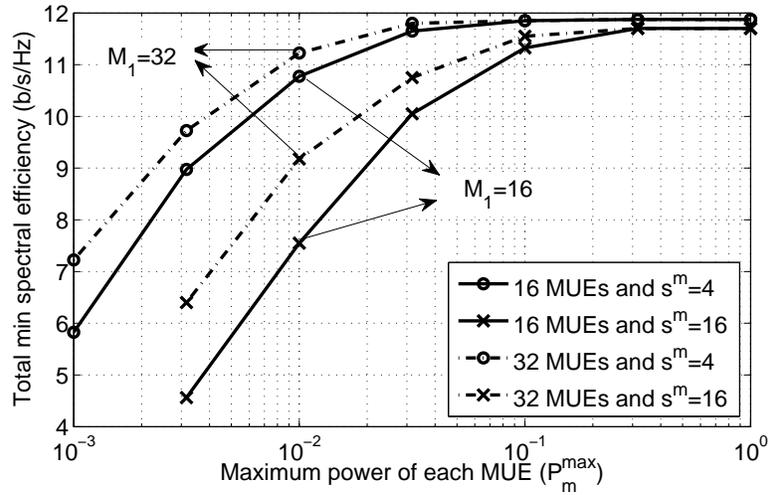}
\end{center}
\caption{Total minimum spectral efficiency versus $P_{\sf m}^{\texttt{max}}$ for the large network with $s^{\sf f}=256$, $WL=5\:dB$, $P^{\texttt{max}}_{\sf f}=0.01\:W$.}
\label{Ch4_fig:pm}
\end{figure}

\begin{figure}[!t]
\begin{center}
\includegraphics[width=0.7 \textwidth]{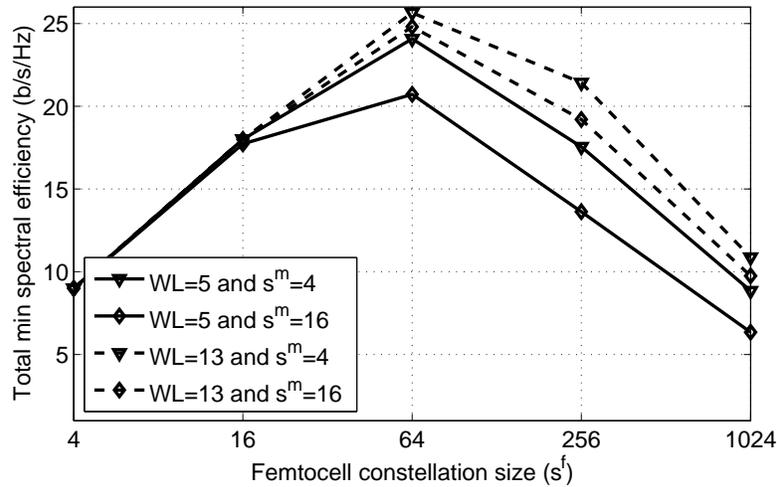}
\end{center}
\caption{Total minimum spectral efficiency of femtocells in the downlink with $WL=5\:dB, P^{\texttt{max}}_{\sf{BS},\it{m}}=0.2\:W, P^{\texttt{max}}_{\sf{BS},\it{f}}=0.05\:W$.}
\label{Ch4_fig:dlk}
\end{figure}

In Fig.~\ref{Ch4_fig:dlk}, we show the total minimum spectral efficiency of
femtocells versus the QAM constellation sizes for the downlink
by running the resource-allocation algorithm with the power
constraints (\ref{Ch4_eq:pwcon_dlk}) and SA weights (\ref{Ch4_eq:w3}).This figure shows that the
proposed algorithm works well for the downlink system where
these results are similar to those for the uplink case presented
in Fig.~\ref{Ch4_fig:wl5}.

Fig.~\ref{Ch4_fig:mRate_WL} presents the total minimum spectral efficiency of
femtocells versus the wall-loss parameter $WL$ achieved by the fixed-rate algorithm (i.e., Algorithm~\ref{Ch4_alg:gms1})
 and the adaptive-rate algorithm
(i.e., Algorithm~\ref{Ch4_alg:gms2}) whose results are indicated by ``Max-Fixed Rate'' and ``Multi-Rate'' in this figure, respectively.
In Algorithm~\ref{Ch4_alg:gms2}, each femtocell can choose one best modulation scheme among five schemes ($s^{\sf f}=4,16,64,256,1024$)
whereas Algorithm~\ref{Ch4_alg:gms1} employs the maximum modulation level for all femtocells.
This figure demonstrates that the great performance gain can be achieved by exploiting the adaptive-rate
feature in performing resource allocation for FUEs.

\begin{figure}[!t]
\begin{center}
\includegraphics[width=0.7 \textwidth]{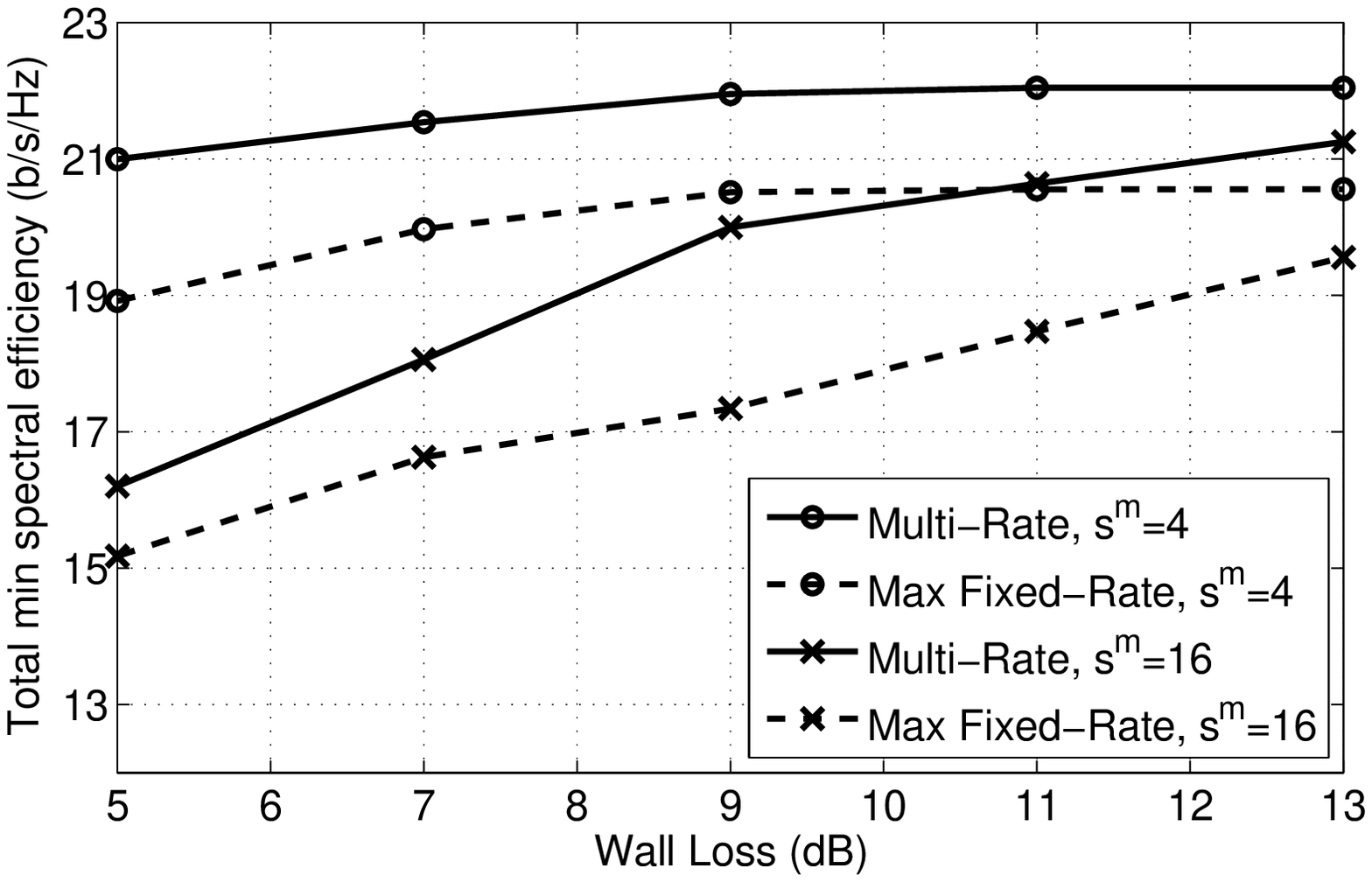}
\end{center}
\caption{Total minimum spectral efficiency versus wall loss for Algorithm~\ref{Ch4_alg:gms2} with $P_{\sf{m}}^{\texttt{max}}=P_{\sf{f}}^{\texttt{max}}=0.01\:W$.}
\label{Ch4_fig:mRate_WL}
\end{figure}

\begin{figure}[!t]
\begin{center}
\includegraphics[width=0.7 \textwidth]{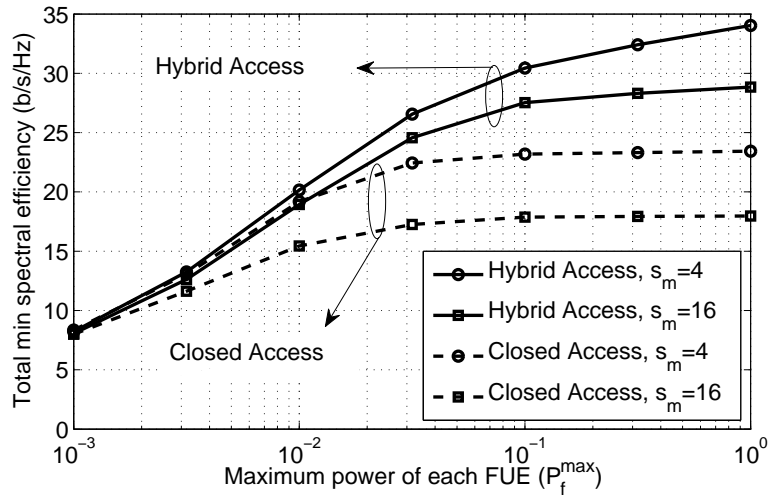}
\end{center}
\caption{Total minimum spectral efficiency versus $P_{\sf f}^{\sf max}$ under the adaptive rate and hybrid access strategy for $P_{\sf m}^{\texttt{max}}=0.01\:W$, $M_1=32$ and $WL=5\:dB$.}
\label{Ch4_fig:Hbdacc}
\end{figure}

In Fig.~\ref{Ch4_fig:Hbdacc}, we show the total minimum spectral efficiency of all femtocells versus $P_{\sf{f}}^{\sf max}$ under the
 hybrid and closed-access strategies, which are obtained by Algorithm~\ref{Ch4_alg:gms3} and Algorithm~\ref{Ch4_alg:gms2}, respectively. These results correspond to the large network with $32$ MUEs and $W_l=5\:dB$. 
 As shown, the hybrid access strategy achieves higher performance
 than that under closed access scheme. 
 Moreover, the performance gap between two strategies becomes larger with increasing maximum power budget of FUEs. 
Under the closed access, the total spectral efficiency of femtocells is indeed limited by MUEs that are located far away from the MBS and close to some
 FBSs. The hybrid access strategy enables these victim MUEs to connect with nearby FBSs, which mitigates this problem and increases the spectral
 efficiency of the femto tier.

\subsection{Performance Comparison with Zhang's Algorithm}

\begin{figure}[!t]
\begin{center}
\includegraphics[width=0.7 \textwidth]{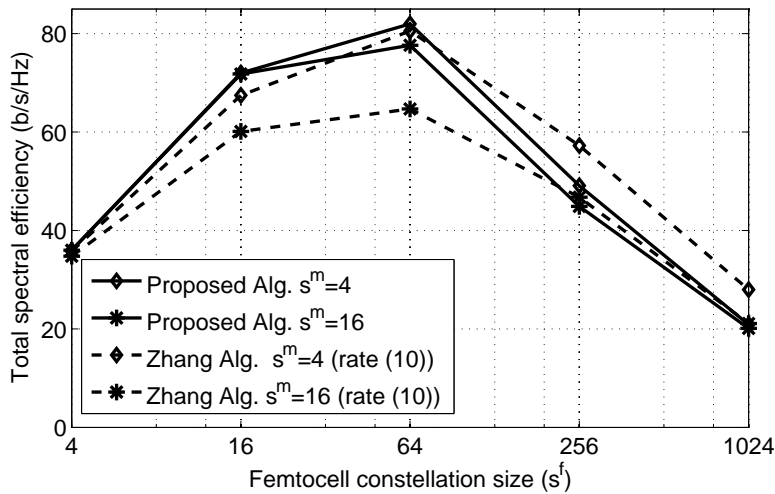}
\end{center}
\caption{Total spectral efficiency versus QAM constellation size of FUEs with $P^{\texttt{max}}_{\sf m}=P^{\texttt{max}}_{\sf f}=0.01\:W$.}
\label{Ch4_fig:c_zh_fQAM}
\end{figure}

Performance comparison between our proposed algorithm and that in \cite{zhang12} is presented in the following.
To obtain the simulation results in Fig.~\ref{Ch4_fig:c_zh_fQAM}, the spectral efficiency of each FUE is obtained by running Zhang's algorithm and 
calculating the spectral efficiency by using the fixed-rate formula (\ref{Ch4_eq:rate_n}).
In Fig.~\ref{Ch4_fig:c_zh_fQAM}, we show the total spectral efficiency of all FUEs in all femtocells due to Zhang's algorithm \cite{zhang12} and Algorithm~\ref{Ch4_alg:gms1} as we vary the QAM constellation size of each FUE $s^{\sf{f}}$.
This figure shows that, for a fixed constellation size employed by MUEs, the total spectral efficiency achieved by Zhang's algorithm is lower than ours for low values of $s^{\sf{f}}$,
whereas Zhang's algorithm results in slightly higher spectral efficiency than our algorithm for larger values of $s^{\sf{f}}$. The slightly better performance of Zhang's algorithm
comes at the cost of severely unfair resource sharing among FUEs, as we will show in Fig.~\ref{Ch4_fig:c_zh_fair}. 
It is worth recalling that Zhang \textit{et al.} aim to maximize the total throughput
without considering user fairness, while our formulation focuses on achieving an equal rate for FUEs in each femtocell (i.e., max-min fairness).

\begin{figure}[!t]
\begin{center}
\includegraphics[width=0.7 \textwidth]{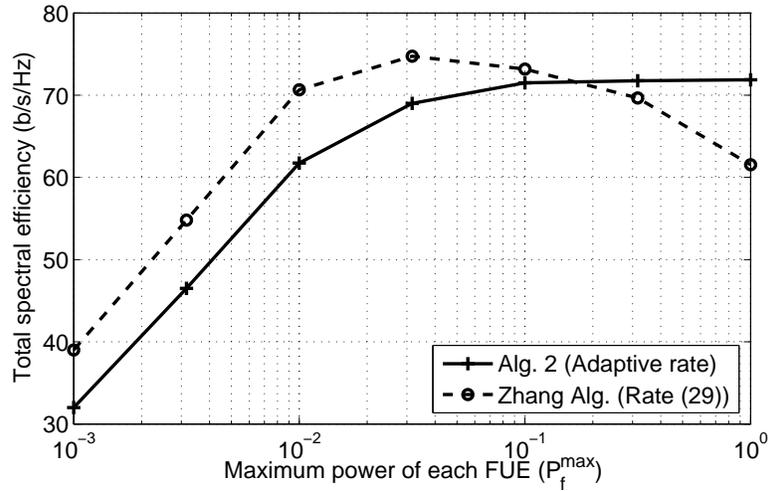}
\end{center}
\caption{Total spectral efficiency versus maximum power of each FUE with $P^{\texttt{max}}_{\sf m}=0.01\:W$.}
\label{Ch4_fig:c_zh_pfmax}
\end{figure}

Fig.~\ref{Ch4_fig:c_zh_pfmax} shows the total spectral efficiency of Zhang's algorithm and our proposed algorithms with rate adaptation (i. e., Algorithm~\ref{Ch4_alg:gms2}) as
we vary the maximum power of each FUE $P^{\texttt{max}}_{\sf{f}}$. Here, the spectral efficiency of FUEs under Zhang's algorithm is calculated by using the adaptive rate formula (\ref{Ch4_objfun3}). 
As evident, the total spectral efficiency due to Zhang's algorithm is higher than ours for low values of $P^{\texttt{max}}_{\sf{f}}$, but Zhang's
algorithm performs worse compared with our algorithm if $P^{\texttt{max}}_{\sf f}$ becomes sufficiently large. The superior performance of
Zhang's algorithm in the low-power regime is again due to
its unfair resource-sharing nature. In fact, Zhang's algorithm
does not account for the co-tier interference among femtocells, which explains why it performs worse compared with
our algorithm in the high-power regime. This is because the transmission power values of FUEs increase with the larger
power budget, which results in higher femtocell co-tier interference. The high co-tier interference degrades the performance
of Zhang's algorithm since this type of interference is not
managed in their work.

\begin{figure}[!t]
\begin{center}
\includegraphics[width=0.7 \textwidth]{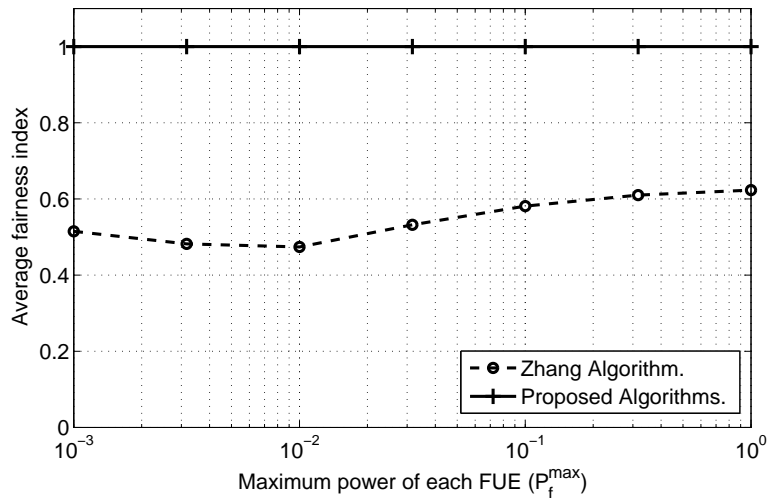}
\end{center}
\caption{Fairness index versus maximum power of each FUE with $P^{\texttt{max}}_{\sf m}=0.01\:W$.}
\label{Ch4_fig:c_zh_fair}
\end{figure}

To compare the fairness of Zhang's and our proposed algorithms, we present the average fairness index achieved by FUEs in each femtocell, which is calculated as \cite{jain99}
\beq \label{eq:fair_index}
FI_k= {\left( \sum_{i=1}^{M_k} R_i \right)^2} \bigg / {M_k \left( \sum_{i=1}^{M_k} R_i^2 \right)}. \eeq 
This fairness index is widely employed in the literature to evaluate the level of fairness achieved by resource allocation algorithms where
an algorithm is fairer if its fairness index is higher and close to the maximum value of one and vice versa. 
Fig.~\ref{Ch4_fig:c_zh_fair} shows the average fairness index achieved by Zhang's and our proposed algorithms.
 As evident, our algorithm provides maximum fairness for FUEs in each femtocell (since the rates of FUEs in each femtocell are equal), while the average fairness index 
of Zhang's algorithm is around $0.5$--$0.6$, which is quite low. This implies that Zhang's algorithm may result in big differences in the rates achieved FUEs in each femtocell, which is 
not very desirable. Considering that the total spectral efficiency due to Zhang's and our proposed algorithms is not much different while Zhang's algorithm is highly unfair,
we would conclude that our proposed algorithms better balance between the throughput and fairness compared to Zhang's algorithm.

\section{Conclusion}
\label{Ch4_ccls}
We have proposed centralized optimal and distributed
resource-allocation algorithms that maximize the total minimum spectral efficiency of the femtocell network while ensuring fairness among FUEs and QoS protection for all MUEs.
Moreover, the proposed algorithm has been extended for three
scenarios, namely, downlink context, adaptive-rate systems
where FUEs in each femtocell can adaptively choose one in
a predetermined set of modulation schemes, and hybrid access
design where MUEs can associate with nearby FBSs to improve
the performance of the femto tier. Extensive numerical results
have been presented to demonstrate the impacts of different
system parameters on the network performance and the efficacy
of our proposed resource-allocation algorithms.

\begin{table} 
\footnotesize
\caption{Summary of Key Notations}
\centering
\begin{tabular}{|l | p{140mm} |}
\hline 
Notations & Description \\ 
\hline
$a^n_i$ & Subchannel assignment variable for subchannel $n$ and UE $i$  \\
$\mathrm{\bf{A}}$ & Subchannel assignment matrix for all $M$ UEs over $N$ subchannels\\
$\mathrm{\bf{A}}_k$ & Subchannel assignment matrix for UEs in cell $k$ over $N$ subchannels\\
$A_i, B_i, C$ & Coefficients in path loss formula\\
$b_i$ & Base station serving UE $i$\\
$BW$ & Total bandwidth of all $N$ subchannels\\
$\mathcal{B}$ & Set of all base stations\\
$\mathcal{B}_{\sf f}$ & Set of all FBSs\\
$\mathcal{B}_i$ & Set of potential FBSs for MUE $i$\\
$c^n=|\mathcal{U}^n|$ & Number of elements in set $\mathcal{U}^n$\\
$d_{ij}$ & Distance from UE $j$ to BS $i$\\
$f_c$ & Frequency\\
$f\left( \gamma(s) \right)$ & Bit error rate (BER) function for SINR $\gamma(s)$ as the constellation size $s$ is adopted \\
$h_{ij}^n$ & power channel gain from user $j$ to BS $\mathit{i}$ over subchannel $n$ \\
$I_{i,k}^n$ & Effective interference corresponding to MUE $i$ for its connection with BS $k$ on subchannel $n$ \\
$I_i^n(\mathrm{\bf{A}},\mathrm{\bf{P}})$ & Effective interference corresponding to UE $i$ on subchannel $n$ \\
$K$ &  Number of cells \\
$L_{ij}$ & Path loss\\
$M_k$ & Number of UEs served by base station $k$ \\
$M_{\sf m}, M_{\sf f}$ & Numbers of MUEs and FUEs, respectively \\
$\mathfrak{M}$ & Set of constellation sizes \\
$\mathcal{N}$ & Set of all subchannels \\
$\mathcal{N}_i$ & Set of subchannels assigned to UE $i$ \\
$N$ &  Number of subchannels \\
$n_{ij}$ & Number of walls between BS $i$ and UE $j$  \\
$p_i^{n,\texttt{min}}$ & Required transmission power of UE $i$ over its assigned subchannel $n$\\
$p^n_i$ & Transmission power of UE $i$ over subchannel $n$ \\
$P_{i,k}$ & Estimated required power for MUE $i$ connecting with BS $k$\\
$P_i^{\texttt{max}}$ & Maximum transmission power of UE $i$ \\
$P_{\sf BS, \textit{k}}^{\texttt{max}}$ & Maximum power of BS $k$ in the downlink \\
$\mathrm{\bf{P}}$ & Power allocation matrix whose elements are $p^n_i$\\
$\mathrm{\bf{P}}_k$ & Power allocation matrix for UEs over $N$ subchannels in cell $k$ \\
$\overline{P}_e$ & Target BER \\
$Q_k$ & Number of unused subchannels at FBS $k$ \\
$\mathbb{Q}(.)$ & Q function \\
$r_i^n(\mathrm{\bf{A}},\mathrm{\bf{P}})$ & Spectral efficiency (bits/s/Hz) achieved by FUE $i$ on subchannel $n$\\
$R_i(\mathrm{\bf{A}},\mathrm{\bf{P}})$ & Spectral efficiency (bits/s/Hz) achieved by FUE $i$\\
$s$ & Constellation size corresponding to the $s\text{-QAM}$ modulation scheme\\
$s^k$ & Constellation size of the $\text{MQAM}$ modulation utilized in cell $k$ \\
$s_{\sf max}$ & Maximum allowable constellation size in the set $\mathfrak{M}$ \\
$s^{\sf m}, s^{\sf f}$ & Constellation size of the $\text{MQAM}$ modulation employed in macrocell and femtocells, respectively\\
$\mathcal{U}_k$ & Set of all UEs served by base station $k$ \\
$\mathcal{U}_{m,k}$ & Set of MUEs connecting with BS $k$ \\
$\mathcal{U},\mathcal{U}_{\sf m}, \mathcal{U}_{\sf f}$ & Set of all UEs, MUEs and FUEs, respectively \\
$\mathcal{U}^n$ & Set of UEs in both tiers that are assigned subchannel $n$ \\
$w_i^n$ & Assignment weight of UE $i$ over subchannel $n$ \\
$W_k$ & Total minimum weight of subchannel assignments in cell $k$\\
$WL$ & Wall-loss parameter \\
$\alpha_i^{n}, \theta_i^{n},\delta_i^{n}, \delta_i^{n}$ & Other factors for calculating scaling factor $\chi_i^n$ \\
$\bar{\gamma}(s)$ & Target SINR corresponding to constellation size $s$ \\
$\bar{\gamma}_i^n$ & Target SINR of UE $i$ over subchannel $n$ \\
$\Gamma_i^n(\mathrm{\bf{A}},\mathrm{\bf{P}})$ & SINR of UE $i$ over subchannel $n$\\
$\eta_i^n$ & Noise power at BS $\mathit{i}$ on subchannel $n$ \\
$\tau_k$ & Number of subchannels assigned for each UE in cell $k$\\
$\chi_i^n$ & Scaling factor in SA weight \\
$\Omega\{\mathrm{\bf{A}}\}$ & List of all potential SA solutions \\
$\Omega^{\ast}\{\mathrm{\bf{A}}\}$ & Sorted list of $\Omega\{\mathrm{\bf{A}}\}$ in decreasing order of $\sum_{k=2}^K \tau_k$ \\
\hline 
\end{tabular}
\label{Ch4_keynotation1}
\end{table}

%% file: chap7/Ha_chap7.tex
\chapter{Coordinated Multipoint Transmission Design for Cloud-RANs with Limited Fronthaul Capacity Constraints} 
\renewcommand{\rightmark}{Chapter 7.  CoMP Transmission Design for C-RAN with Limited FH Constraint}
\label{Ch5}
The content of this chapter was published in IEEE Transactions on Vehicular Technology in the following paper:

Vu N. Ha, Long B. Le, , and Ng\d{o}c-D\~{u}ng \DH\`{a}o, ``Coordinated multipoint transmission design for Cloud-RANs considering limited fronthaul capacity constraints,'' {\em IEEE Trans. Veh. Tech.,} vol. PP, no. 99, pp. 1--16, 2016.

\section{Abstract}
\label{Ch5_Abs}
In this paper, we consider the coordinated multipoint (CoMP) transmission design for the downlink cloud radio access
network (Cloud-RAN). Our design aims to optimize the set of remote radio heads (RRHs) serving each user and the precoding
and transmission power to minimize the total transmission power while maintaining the fronthaul capacity and users' quality-ofservice
(QoS) constraints. The fronthaul capacity constraint involves a nonconvex and discontinuous function that renders the
optimal exhaustive search method unaffordable for large networks.
To address this challenge, we propose two low-complexity algorithms. The first pricing-based algorithm solves the underlying
problem through iteratively tackling a related pricing problem while appropriately updating the pricing parameter. In the
second iterative linear-relaxed algorithm, we directly address the fronthaul constraint function by iteratively approximating it with
a suitable linear form using a conjugate function and solving the corresponding convex problem. For performance evaluation, we
also compare our proposed algorithms with two existing algorithms in the literature. Finally, extensive numerical results are
presented, which illustrate the convergence of our proposed algorithms and confirm that our algorithms significantly outperform
the state-of-the-art existing algorithms.

\section{Introduction}
\label{Ch5_Introduction}
The next-generation wireless cellular network is expected to provide significantly enhanced capacity to support the
exponential growth of mobile data traffic \cite{cisco14, mob_rep_13, Le-Hossain13, Le12, Le_EU15}.
Toward this end, coordinated multipoint (CoMP) transmission/reception techniques provide promising solutions, which have been
adopted in the Long-Term Evolution Advanced standard \cite{Americas12}.
In fact, CoMP employs different forms of base-station (BS) coordination with dynamic sharing of channel state information
(CSI) and/or data information among BSs, as well as efficient transmission, precoding, and resource-allocation algorithms \cite{tanno2010}.
However, implementation of CoMP in wireless cellular networks typically requires costly high-speed backhauls connecting
different BSs for CSI and information exchanges and distributed computation \cite{irmer2011,Fettweis11}.
\nomenclature{CSI}{Channel State Information}

Cloud-RAN has been recently proposed as an alternative network architecture that can achieve the performance gains of
the CoMP techniques while effectively exploiting the computation power of the cloud computing technology. 
In a Cloud-RAN, the cloud comprising a pool of baseband processing units (BBUs) performs most baseband processing tasks while transmission functions are realized by simple remote radio heads (RRHs) using the processed baseband signals received from the
cloud through a fronthaul transport network \cite{chinamobile2011,NGMN2013,checko13,wu12}.
With a centralized cloud processing center (CPC), complex resource-allocation optimization, such as precoding, power control, user
scheduling, interference management algorithms, can be realized, which can translate into significant network performance
improvements. 
In addition, Cloud-RAN with simple RRHs enables us to reduce both capital and operational expenditures of
the network \cite{checko13}. Moreover, this emerging network architecture also allows to deploy dense small cells efficiently \cite{Liu_infocom_13}. 
\nomenclature{CPC}{Cloud Processing Center}

Despite these benefits, there are various technical challenges one must resolve in designing and deploying the Cloud-RAN
architecture. 
In particular, suitable mechanisms that efficiently utilize computing resources in the cloud, fronthaul capacity to
realize advanced communications, baseband signal processing, and resource-allocation schemes must be developed.
Some recent works have addressed some of these issues, as can be described in the following.
In \cite{fan_arxiv}, an efficient clustering algorithm was proposed to reduce the number of computations in
the centralized pool of BBUs where the number of RRHs is very large.
In \cite{shamai13a} and \cite{shamai13b}, compression techniques to minimize the amount of data transmitted over the fronthaul transport network were developed.

Several other papers focus on precoding/decoding design for CoMP considering different design aspects of Cloud-RAN. 
In particular, in \cite{vincent_lau_14}, the joint optimization of antenna selection,
regularization, and power allocation was studied to maximize the average weighted sum rate.
The random matrix theory was utilized to decompose the considered nonconvex
problem into subproblems that can be tackled more easily.
In \cite{letaief14} and \cite{luo_arxiv}, the precoding vectors were optimized for all
RRHs to minimize the total network power consumption; the downlink was considered in \cite{letaief14}, whereas
whereas both downlink and uplink communications were addressed in \cite{luo_arxiv}.
Total power to support radio transmission and operations of fronthaul links and RRHs is accounted for in these research works,
where the authors show how to transform the underlying problems into the sparse beamforming problems. 
Then, these papers combine
the solution techniques employed to address the traditional power minimization problem in \cite{Bengtsson99} and \cite{TomLuo06} and the fronthaul capacity minimization problem in \cite{Quek2013}
to tackle the transformed problem.
This solution technique is also closely related to that employed in \cite{cheng13}. 
In fact, the design problems considered in \cite{Quek2013} and \cite{cheng13} for CoMP and in \cite{letaief14} and \cite{luo_arxiv} for Cloud-RAN with standard convex constraints
do not explicitly model the fronthaul capacity constraints. 
 Consequently, they can be solved directly by employing
 the compressed sensing techniques \cite{donoho06,tao06,bach12}.
 We would like to emphasize that the related work \cite{Quek2013} aims 
at minimizing the number of active links between base stations (RRHs) and users.
Therefore,
the limited fronthaul capacity is not explicitly imposed as constraints in \cite{Quek2013}, but it is rather
considered in the design objective. Note that the greedy principle proposed in \cite{Quek2013} can be employed to solve
our problem (even though our problem is not the same with the problem in \cite{Quek2013}); however, such greedy approach may not achieve very good performance.

This current work aims to study the CoMP joint transmission design problem for Cloud-RAN 
that explicitly considers the fronthaul capacity and users' quality-of-service (QoS) constraints.
In particular, we make the following contributions.

\begin{itemize}
\item
We formulate the joint transmission design problem for Cloud-RAN, which optimizes the set of RRHs serving each user
together with their precoding and power allocation solutions to minimize the total transmission power considering
users' QoS and fronthaul capacity constraints.
In particular, we assume that certain fronthaul capacity, which depends on the user's required QoS, is consumed 
to transfer the processed baseband signal from the cloud to one particular RRH serving that user and that the fronthaul transport network has limited capacity.
The considered problem is indeed nonconvex, and it comprises the discontinuous fronthaul capacity constraint and other traditional power and users' QoS constraints.
Therefore, addressing this problem through exhaustive search requires exponential computation complexity, which is not affordable for practically large networks.

\item
To tackle the considered problem, we then develop two different low-complexity algorithms. 
For the first algorithm, which is called \textit{pricing-based algorithm}, we
consider the related pricing problem and devise a novel mechanism to iteratively update the pricing parameter
to obtain a good feasible solution.
The second algorithm, which is called \textit{iterative linear-relaxed algorithm}, is 
developed by regularizing the fronthaul constraint function into an approximated linear form and
iteratively solving the approximated optimization problem. 

\item 
For performance comparison, we also describe two existing algorithms proposed in \cite{Quek2013} and \cite{dai14}
(with appropriate modifications due to different considered settings), which can also solve our considered problem.
We then study the extended setting where there are multiple individual fronthaul capacity constraints.
We show that it is possible to adapt the proposed low-complexity algorithms to solve the corresponding
problem in this setting. In addition, we indicate the need to introduce multiple pricing parameters, each
of which is required to deal with one corresponding fronthaul capacity constraint in the extended \textit{iterative linear-relaxed algorithm}.

\item 
Numerical results are presented to demonstrate the efficacy of the proposed algorithms and their relative performance 
compared to the optimal algorithm and existing algorithms. We show the convergence and the impacts of number of users, number of antennas, 
users' required QoS on the network performance. Specifically, while both two proposed algorithms perform remarkably well
in most cases, the pricing-based algorithm achieves slightly better performance than the iterative linear-relaxed algorithm at the cost of longer convergence time. 
In addition, both proposed algorithms significantly outperform the existing algorithms in \cite{Quek2013} and \cite{dai14} for all considered simulation scenarios.
\end{itemize}

The remaining of this paper is organized as follows. We describe the system model and problem formulation in Section~\ref{Ch5_sysmod}. In Section~\ref{Ch5_solving}, we present two low-complexity algorithms and
two existing algorithms. 
Next, we extend our proposed algorithms to the setting with multiple individual fronthaul capacity
constraints and the multiple-input–multiple-output (MIMO) systems where multiple data streams can be transmitted to each
user in Section~\ref{Ch5_sec:furex}.
Numerical results are presented in Section~\ref{Ch5_results} followed by conclusion in Section~\ref{Ch5_sec:cls}. Some preliminary results of this paper
have been published in \cite{vuha_ciss_2014} and \cite{vuha_globecom_2014}.

For notation, we use $\mathbf{X}^T$, $\mathbf{X}^H$, $\mathsf{Tr}(\mathbf{X})$, and $\mathsf{rank}(\mathbf{X})$ to denote transpose, Hermitian transpose, trace, and
 rank of matrix $\mathbf{X}$, respectively. 
$\mathbf{1}_{x \times y}$ and $\mathbf{0}_{x \times y}$ denote the matrix of ones and the matrix of zeros whose dimension are $x \times y$, respectively.
$|\mathcal{S}|$ denotes the cardinality of set $\mathcal{S}$, and $\mathsf{diag}(\mathbf{x})$ is the diagonal matrix constructed from the elements of vector $\mathbf{x}$.
 
\section{System Model and Problem Formulation}
\label{Ch5_sysmod}

\subsection{System Model}
The general architecture of Cloud-RAN under consideration, which is shown in Fig.~\ref{Ch5_fig:CRAN}, consists of three main components: 1) CPC consisting BBUs pool; 2) the optical fronthaul transport network (i.e., fronthaul links); and 3) RRH access units with antennas located at remote sites.
Specifically, the CPC comprising a certain number of BBUs is the heart of this architecture where BBUs act as virtual base stations that process baseband signals for users and optimize the network-resource-allocation tasks.
The fronthaul transport network connecting the CPC and distributed RRHs is usually deployed by using optical fibers.
In addition, RRHs communicate with users in the downlink via RF signals, which are formed by using the baseband signals and the precoding vectors received from the BBU pool.

\begin{figure}[!t]
\begin{center}
\includegraphics[width=100mm]{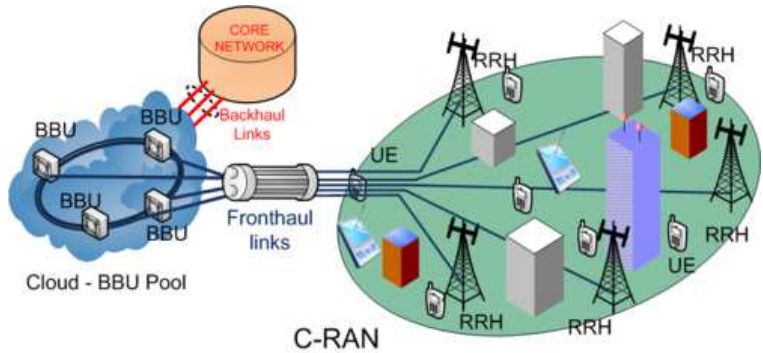}
\end{center}
\caption{Cloud-RAN architecture.}
\label{Ch5_fig:CRAN}
\end{figure}

In this paper, we consider the joint transmission design for CoMP downlink communications in the Cloud-RAN with $K$ RRHs and $M$ users.
Let $\mathcal{K}$ and $\mathcal{U}$ be the sets of RRHs and users in the network, respectively.
Suppose that RRH $k$ is equipped with $N_k$ antennas ($k \in \mathcal{K}$) and each user has a single antenna. 
We assume that each user is served by a specific group of assigned RRHs, and one RRH can serve a number of users.
When RRH $k$ is assigned to serve user $u$, this RRH receives the user's processed baseband signal from the cloud. 
Then, the RRH converts and transmits the corresponding RF signal using a suitably designed precoding vector. 
Moreover, denote $\mathbf{v}_u^{k} \in \mathbb{C}^{N_k \times 1}$ as
 the precoding vector at RRH $k$ corresponding to the signal transmitted to user $u$.
Then, the transmission power used by RRH $k$ 
to serve user $u$ can be expressed as
\beq \label{Ch5_eq:pv}
p^k_u=\mathbf{v}_u^{kH}\mathbf{v}_u^{k}.
\eeq 
Let $(k,u)$ represent the communication link between RRH $k$ and user $u$.
For simplicity, we assume that each user is virtually served
by all RRHs in all expressions and problem formulations.
However, a user is actually served by one particular RRH if the
corresponding transmission power is strictly larger than zero.
Let $x_u \in \mathbb{C}$ denote the signal symbols with unit power for user $u$, which are transmitted by RRHs in set $\mathcal{K}$ upon receiving the processed baseband signals from the cloud.
Then, the baseband signal $y_u$ received at user $u$ can be written as
\beq
\label{Ch5_eq:rx_sig}
y_u= \underbrace{\sum \limits_{k \in \mathcal{K}} \mathbf{h}_u^{kH} \mathbf{v}_u^{k} x_u}_\text{desired signal} 
 + \underbrace{ \sum \limits_{i =1, \neq u}^{M} \sum \limits_{l \in \mathcal{K}} \mathbf{h}_u^{lH} \mathbf{v}_i^{l} x_{i}}_\text{interference} + \eta_u,
\eeq
where $\mathbf{h}_u^k \in \mathbb{C}^{N_k \times 1}$ denotes the channel vector between RRH $k$ and user $u$, and $\eta_u$ describes the noise at user $u$. 
Then, the signal-to-interference-plus-noise ratio (SINR) achieved by $u$ can be described as
\beq \label{Ch5_eq:SINR}
\Gamma_u =\dfrac{ \left|  \sum \limits_{k \in \mathcal{K}} \mathbf{h}_u^{kH} \mathbf{v}_u^{k}\right| ^2 }
{\sum \limits_{i =1, \neq u}^{M} \left| \sum \limits_{l \in \mathcal{K}}\mathbf{h}_u^{lH} \mathbf{v}_i^{l}\right|^2 + \sigma^2 },
\eeq
where $\sigma^2$ is the noise power.
In this paper, we are interested in design the precoding in order to satisfy users' QoS. This QoS constraint of user $u$ is expressed as
\beq \label{Ch5_eq:SINR_constraint}
\Gamma_u \geq \bar{\gamma}_u, \qquad \forall u \in \mathcal{U},
\eeq
where $\bar{\gamma}_u$ denotes the target SINR of user $u$.

To manage the computation complexity in large networks, we can impose additional constraints
where each user $u$ can only be served by set of near RRHs $\mathcal{R}_u$ (e.g., clustering constraints).
This set of RRHs can be determined based on the distance or channel gain from them to each user.\footnote{For example, each user is potentially associated
 with $r$ nearest RRHs as $r$th-order Voronoi tessellation \cite{RWHeath15} or RRHs with strong channel gain as \cite{gong11}.}
In the following, we assume that only RRH-user communications links in predetermined set $\mathcal{L}$ are allowed, where $\mathcal{L} = \lbrace (k,u)\vert u \in \mathcal{U}, k \in \mathcal{R}_u \rbrace$.
When no such constraint exists, set $\mathcal{L}$ comprises all possible $MK$ links between $K$ RRHs
and $M$ users. We can express the relation between transmission power and communication link $(k,u)$ as follows: 
\beq \label{Ch5_cnt:noL}
p^k_u = \mathbf{v}_u^{kH}\mathbf{v}_u^{k} = 0 \text{ if } (k,u) \notin \mathcal{L}.
\eeq

\begin{remark}
In this paper, we assume that perfect CSI between all RRHs and users is available at the CPC.
In practice, the CSI can be estimated by the corresponding users and RRHs and then transferred to the CPC through the fronthaul 
network.\footnote{Detailed realization of CSI channel estimation varies depending if the time-division duplex (TDD) or frequency-division duplex (FDD) strategy is employed.
In particular, in the TDD system, RRHs perform CSI estimation; then, they transfer the estimated CSI to the CPC.}
In general, the transfer of CSI to the CPC over the fronthaul network consumes much smaller fraction of fronthaul capacity compared with the transfer of users' data (i.e., I/Q data
sequences) from the cloud to the RRHs. Moreover, the transfer of CSI can be performed periodically since the CSI would change slowly over time, whereas the data transfer
over the fronthaul network occurs continuously.  
Therefore, we do not explicitly consider the fronthaul capacity consumption due to CSI transfer in this paper, which is left for our future works.   
\end{remark}
\nomenclature{TDD}{Time-Division Duplex}
\nomenclature{FDD}{Frequency-Division Duplex}

\subsubsection{Fronthaul Capacity}
Let $\mathbf{p}^k=[p^k_1 ... p^k_M]^T$ be the transmission power vector of RRH $k$ whose elements $p^k_u \geq 0$ represent transmission power $p^k_u$ given in (\ref{Ch5_eq:pv}).
We also define the vector $\mathbf{p}=[\mathbf{p}^{1 T} ... \mathbf{p}^{K T}]^T$ to describe the transmission powers of all RRHs for all users.
Recall that $p^k_u=0$ implies that RRH $k$ does not serve user $u$.
In contrast, if $p^k_u > 0$, the fronthaul link from RRH $k$ to user $u$ is activated for carrying the baseband signal to serve user $u$ 
at its required target SINR. 
Therefore, the total capacity of the fronthaul links can be indicated by the value of vector $\mathbf{p}$ and the target SINR of users, which can be written mathematically as
\beq \label{Ch5_eq:Ck}
G(\mathbf{p})= \sum \limits_{k \in \mathcal{K}} \sum \limits_{u \in \mathcal{U}}  \delta( p_u^k ) R_u^{k,\sf{fh}},
\eeq 
where $\delta(\cdot) $ denotes the step function, and $R_u^{k,\sf{fh}}$ represents the required capacity for transferring the signal of user $u$ over
the fronthaul network to the RRH $k$. 
The required fronthaul capacity $R_u^{k,\sf{fh}}$ depends on the actual required data rate corresponding to the target SINR $\bar{\gamma}_u$, the specific design of Cloud-RAN, and the specific
quantization technique employed to process the baseband signals, which are discussed in the following remark.

\begin{remark}
In the Cloud-RAN downlink system, the CPC processes the baseband signals and optimizes the precoding vectors. As described in \cite{ParkICC13} and \cite{shamai15}, there 
are two strategies to realize such design, namely ``compression-after-precoding'' (CAP) and ``compression-before-precoding'' (CBP).
In the CAP, which is employed in \cite{shamai13a} and \cite{shamai13b}, 
 the CPC precodes the data streams of users with their corresponding precoding vectors, then compresses the resulting signals, and forwards them to the corresponding RRHs over 
the fronthaul network. 
For the CBP, the CPC directly compresses the precoding vectors and forwards the compressed vectors as well as users' data streams to the corresponding RRHs \cite{ParkICC13}.
Then, the RRHs precode the signals and transmit to users.
\nomenclature{CAP}{Compression-After-Precoding}
\nomenclature{CBP}{Compression-Before-Precoding}

In this paper, we assume that the CBP strategy is employed, which is motivated by its advantages for the Cloud-RAN system with preliminary clustering \cite{shamai15}.
In this strategy, information streams and precoding vectors are compressed and transmitted from the cloud to RRHs separately. 
The fronthaul capacity consumed by each RRH depends on the group of users it serves which is given in (\ref{Ch5_eq:Ck}).
In particular, the required fronthaul capacity to support the communication link between user $u$ to RRH $k$  can be calculated as
\beq \label{Ch5_R_rmk2}
R_u^{k,\sf{fh}} = R_u^{k,\sf{pr}} + R_u^{\sf{dt}}(\bar{\gamma}_u),
\eeq
where $R_u^{\sf{dt}}(\bar{\gamma}_u)$ and $R_u^{k,\sf{pr}}$
correspond to the fronthaul capacity consumption for transferring the information stream and precoding vector $\mathbf{v}_u^k$ for user $u$, as given in Eq.~3 of \cite{ParkICC13}, respectively. 
In fact, $R_u^{\sf{dt}}(\bar{\gamma}_u)$ is a function of the target SINR (e.g., it can be expressed $R_u^{\sf{dt}}(\bar{\gamma}_u) = \log_2(1+\bar{\gamma}_u)$ if
the capacity can be achieved), and $R_u^{k,\sf{pr}}$ can be predetermined based on the desirable quantization quality of precoding vectors.
When the precoding vectors are quantized, and then employed at RHHs, there are the quantization errors.
However, if the corresponding quantization noise is sufficiently small compared with the interference and noise at receivers, the quantization errors can be omitted, which is what to be assumed in this paper. 
\end{remark}

Note that our model allows to capture the scenario where users demand services with different rate requirement (e.g., voice and video services).
We are now ready to describe the considered fronthaul-constrained power minimization (FCPM) problem 
in the following.

\subsection{Problem Formulation}
We are interested in determining the set of RRHs serving each user and the corresponding precoding vectors 
 (i.e.,  $\lbrace \mathbf{v}_u^k \rbrace$) 
to minimize the total transmission power considering the constraints on fronthaul capacity limit,
transmission power, and users' QoS.  
In addition to the SINR constraints in (\ref{Ch5_eq:SINR_constraint}), we also impose the constraint on total transmission power for each RRH $k$ as
\beq \label{Ch5_eq:pwc}
\sum \limits_{u \in \mathcal{U}} p^k_u = \Vert \mathbf{p}^k \Vert_{\mathbf{1}}=\sum \limits_{u \in \mathcal{U}} \mathbf{v}_u^{kH}\mathbf{v}_u^{k} \leq P_k, \;\;\; \forall k \in \mathcal{K},
\eeq 
where $P_k$ ($k \in \mathcal{K}$) denotes the maximum power of RRH $k$, $\Vert \mathbf{x} \Vert_{\mathbf{1}}$ represents
 the $\ell1$-norm of vector $\mathbf{x}$. Furthermore, we assume that the capacity of fronthaul transport network between the 
cloud and all RRHs is limited, and we denote $C$ as the fronthaul capacity limit. Then, we have to impose the following fronthaul capacity constraint:
\beq \label{Ch5_eq:C_cons}
G(\mathbf{p})= \sum \limits_{k \in \mathcal{K}} \sum \limits_{u \in \mathcal{U}}  \delta( p_u^k ) R_u^{k,\sf{fh}}  \leq C.
\eeq

This indeed represents the sum fronthaul capacity constraint. We will discuss the extended setting with multiple
 individual fronthaul capacity constraints in Section~\ref{Ch5_sec:furex}. Now, it is ready to state the FCPM problem as follows:
\begin{align}
(\mathcal{P}_{\sf{FCPM}}) \;\;\;\; \min \limits_{\lbrace \mathbf{v}_u^k \rbrace, \mathbf{p} } & \;\;\;\;\;\;\;\; \Vert \mathbf{p} \Vert_{\mathbf{1}}  \label{Ch5_obj_1} \\
\;\;\;\;\;\;\;\; \text{s. t. } & \; \text{ constraints (\ref{Ch5_eq:pv}), (\ref{Ch5_eq:SINR_constraint}), (\ref{Ch5_cnt:noL}), (\ref{Ch5_eq:pwc}), (\ref{Ch5_eq:C_cons})}. \nonumber  
\end{align}
We will describe how to transform this problem into an appropriate form, based on which we can determine its optimal solution in the following.

\subsubsection{Optimal Exhaustive Search Algorithm}
The main challenge involved in solving problem $\mathcal{P}_{\sf{FCPM}}$ comes from the fronthaul capacity
 constraint (\ref{Ch5_eq:C_cons}). Let us now define variables $a_u^k \in \left\{0, 1\right\}$, where
$a_u^k=1$ if RRH $k$ serves user $u$ and $a_u^k=0$ otherwise. 
Denote $\mathfrak{a}$ as the allocation vector $[a_1^1,...,a_M^K]$ and $\mathcal{S}_{\mathfrak{a}}$ as the set of all possible $\mathfrak{a}$. Then, the consumed fronthaul capacity
 corresponding to a given vector $\mathfrak{a}$ can be written as
\beq
G(\mathfrak{a})= \sum \limits_{k \in \mathcal{K}} \sum \limits_{u \in \mathcal{U}}  a_u^k  R_u^{k,\sf{fh}}.
\eeq
Then, we have to optimize over
 both integer variables $\mathfrak{a}$ in $\mathcal{S}_{\mathfrak{a}}$ and continuous variables $\lbrace \mathbf{v}_u^k \rbrace$ to find the optimal solution
of the considered problem so that $G(\mathfrak{a}) \leq C$. Note, however that the number of elements in set $\mathcal{S}_{\mathfrak{a}}$ is finite. Moreover,
suppose that the values of $\mathfrak{a}$ which satisfy constraint 
(\ref{Ch5_eq:C_cons}) are given, then we only need to solve following precoding optimization problem 
\begin{align} 
(\mathcal{P}_{\mathfrak{a}}) \;\;\;\; \min \limits_{\lbrace \mathbf{v}_u^k \rbrace} & \sum \limits_{k \in \mathcal{K}}  \sum \limits_{u \in \mathcal{U}} \mathbf{v}_u^{kH}\mathbf{v}_u^{k} \nonumber \\
\;\;\;\;\;\;\;\; \text{s. t. } & \; \text{ constraints  (\ref{Ch5_eq:SINR_constraint}), (\ref{Ch5_eq:pwc})}, \nonumber \\
{} & \;  \mathbf{v}_u^{kH}\mathbf{v}_u^{k} = 0 \text{ if } a_u^k =0, \;\; \forall (u,k).
\end{align}

This problem is indeed the sum-power minimization problem (SPMP), which can be transformed
into a solvable convex semi-definite program (SDP). This transformation is described in Appendix~\ref{Ch5_sdp_app}.
Here, we refer to such a SPMP for given $\mathfrak{a}$ as a sparse SPMP.
It is now clear that the optimal solution for problem  $\mathcal{P}_{\sf{FCPM}}$ can be determined as 
follows. First, we enumerate all possible $\mathfrak{a}$ that satisfy the fronthaul capacity constraints (\ref{Ch5_eq:C_cons}).
Second, for each such feasible $\mathfrak{a}$, we solve the corresponding problem $(\mathcal{P}_{\mathfrak{a}})$ and obtain the optimal precoding vectors and objective value.
Finally, the feasible $\mathfrak{a}$ that achieves the minimum total transmission power together with the corresponding
optimal precoding vectors are the optimal solution of problem $\mathcal{P}_{\sf{FCPM}}$.
\nomenclature{SPMP}{Sum-Power Minimization Problem}

Such exhaustive search method, however, has exponentially high complexity. This motivates us to develop low-complexity algorithms to 
solve problem $\mathcal{P}_{\sf{FCPM}}$, which is the focus of the following section.

\section{Low-Complexity Algorithms}
\label{Ch5_solving}
Here, we describe two low-complexity algorithms with different levels of complexity to solve problem $\mathcal{P}_{\sf{FCPM}}$.
For the performance evaluation purposes, we also present two existing algorithms in \cite{Quek2013} and \cite{dai14}.

\subsection{Pricing-based Algorithm}
\label{Ch5_Pricing}
The first low-complexity algorithm is developed by employing the penalty method to deal with the step-function fronthaul capacity constraint.
Specifically, we consider the so-called \textit{pricing fronthaul capacity and power minimization} (PFCPM) problem, which is given in the following:
\begin{align}
(\mathcal{P}_{\sf{PFCPM}}) \;\;\;\; \min \limits_{\lbrace \mathbf{v}_u^k \rbrace, \mathbf{p} } & \Vert \mathbf{p} \Vert_{\mathbf{1}} + q G(\mathbf{p})   \nonumber \\
\;\;\;\;\;\;\;\; \text{s. t. } & \; \text{ constraints (\ref{Ch5_eq:pv}), (\ref{Ch5_eq:SINR_constraint}), (\ref{Ch5_cnt:noL}), (\ref{Ch5_eq:pwc}).}   \label{Ch5_obj_2}
\end{align}
In this problem, the consumed fronthaul capacity $G(\mathbf{p})$ scaled by a pricing parameter $q$ is added to the objective function of problem $\mathcal{P}_{\sf{FCPM}}$,
and we have removed the fronthaul capacity constraint (\ref{Ch5_eq:C_cons}). We can obtain a good feasible solution for problem $\mathcal{P}_{\sf{FCPM}}$
by iteratively solving problem  $\mathcal{P}_{\sf{PFCPM}}$ while adaptively adjusting the pricing parameter $q$ to satisfy the fronthaul capacity constraint (\ref{Ch5_eq:C_cons}).
We describe this procedure in details in the following.

\subsubsection{Relationship between FCPM and PFCPM Problems}
We now establish the relationship between problems $\mathcal{P}_{\sf{FCPM}}$ and $\mathcal{P}_{\sf{PFCPM}}$, based on which we can develop an efficient algorithm to solve problem
 $\mathcal{P}_{\sf{FCPM}}$.
Let $G_{\sf{PFCPM}}(q)$ be the consumed fronthaul capacity resulted from solving problem  $\mathcal{P}_{\sf{PFCPM}}$ for a given price parameter $q$. 
Then, it can be verified that the optimal precoding vectors and the transmission powers of problems $\mathcal{P}_{\sf{FCPM}}$ and $\mathcal{P}_{\sf{PFCPM}}$ are the same if 
 $G_{\sf{PFCPM}}(q)=C$.
This fact enables us to develop an iterative algorithm to solve problem $\mathcal{P}_{\sf{FCPM}}$ through tackling problem $\mathcal{P}_{\sf{PFCPM}}$ 
while adjusting $q$ iteratively in attempting to attain $G_{\sf{PFCPM}}(q)=C$.  
Our proposed algorithm is developed based on the assumption that problem $\mathcal{P}_{\sf{PFCPM}}$ can be solved, which will be addressed in the next subsection.
In the following, we establish some theoretical results based on which we develop the mechanism to update the pricing parameter.

\begin{proposition} \label{Ch5_lm1} Let $\sigma_{\sf{min}}$ be the smallest nonzero value of $\vert G(\mathfrak{a})-G(\mathfrak{a}^{\prime})\vert$, where ${\lbrace\mathfrak{a},\mathfrak{a}^{\prime}\rbrace \subset \mathcal{S}_\mathfrak{a}}$, and $\bar{q} = \sum \limits_{k \in \mathcal{K}} P_k/\sigma_{\sf{min}}$. We have the following.
\begin{enumerate} 
\item $G_{\sf{PFCPM}}(q)$ is a decreasing function of $q$.
\item If we increase $q \geq \bar{q}$ then $G_{\sf{PFCPM}}(q)$ cannot be further decreased.
\item When $q \geq \bar{q}$, if $G_{\sf{PFCPM}}(q) > C$, problem $\mathcal{P}_{\sf{FCPM}}$  is infeasible.
\end{enumerate}
\end{proposition}
\begin{proof} 
The proof is given in Appendix~\ref{prf_Ch5_lm1}. 
\end{proof}

These results form the foundation based on which we can develop an iterative algorithm presented in Algorithm~\ref{Ch5_alg:gms2}.
In fact, we have described how the pricing parameter $q$ is updated over iterations, which can be summarized as follows.
At the first iteration, $q^{(0)}$ is set equal to $\bar{q}$ to verify the feasibility of problem $\mathcal{P}_{\sf{FCPM}}$.
If $G_{\sf{PFCPM}}(\bar{q}) > C$, we can conclude that problem  $\mathcal{P}_{\sf{FCPM}}$ is infeasible.
Otherwise, we apply the bisection search method to update $q$ until $G_{\sf{PFCPM}}(\bar{q}) = C$.

\begin{algorithm}[t]
\caption{\textsc{Pricing-based Algorithm for FCPM Problem}}
\label{Ch5_alg:gms2}
\begin{algorithmic}[1]
\STATE Solve PFCPM problem using Alg.~\ref{Ch5_alg:gms3} with $q^{(0)}=\bar{q}$.
\IF{$G_{\sf{PFCPM}}(\bar{q}) > C$} \STATE Stop, the FCPM problem is infeasible.
\ELSIF{$G_{\sf{PFCPM}}(\bar{q}) = C$} \STATE Stop, the solution is achieved.
\ELSIF{$G_{\sf{PFCPM}}(\bar{q}) < C$} 
\STATE Set $l=0$, $q_{\sf{U}}^{(l)}=\bar{q}$ and $q_{\sf{L}}^{(l)}=0$.
\REPEAT 
\STATE Set $l=l+1$ and $q^{(l)}=\left(q_{\sf{U}}^{(l-1)} + q_{\sf{L}}^{(l-1)} \right)/2$.
\STATE Solve PFCPM problem using Alg.~\ref{Ch5_alg:gms3} with $q^{(l)}$.
\IF{$G_{\sf{PFCPM}}(q^{(l)}) > C$} \STATE Set $q_{\sf{U}}^{(l)}=q_{\sf{U}}^{(l-1)}$ and $q_{\sf{L}}^{(l)}=q^{(l)}$.
\ELSIF{$G_{\sf{PFCPM}}(q^{(l)}) < C$} \STATE Set $q_{\sf{U}}^{(l)}=q^{(l)}$ and $q_{\sf{L}}^{(l)}=q_{\sf{L}}^{(l-1)}$.
\ENDIF
\UNTIL{$G_{\sf{PFCPM}}(q^{(l)}) = C$ or $q_{\sf{U}}^{(l)} - q_{\sf{L}}^{(l)}$ is too small}.
\ENDIF
\end{algorithmic}
\end{algorithm}

\subsubsection{PFCPM Problem Solution}
We now develop an efficient algorithm to solve problem $\mathcal{P}_{\sf{PFCPM}}$  based on concave approximation of step function.
Note that the step function makes the objective function of problem $\mathcal{P}_{\sf{PFCPM}}$ nonsmooth, which is difficult to solve. To overcome this challenge, 
 the step function $\delta(x)$ for $x \geq 0$ can be approximated by a suitable concave function.
Denote $f_{\mathsf{apx}}^{(k,u)}(p^k_u)$ as the concave penalty function that approximates the step function $\delta(p^k_u)$ corresponding to link $(k,u)$. Then,
problem  $\mathcal{P}_{\sf{PFCPM}}$ can be approximated by the following problem:
\begin{align}
 \min \limits_{\lbrace \mathbf{v}_u^k \rbrace,\mathbf{p}} & \sum \limits_{k \in \mathcal{K}} \sum \limits_{u \in \mathcal{U}} p_u^k + q \sum \limits_{k \in \mathcal{K}} \sum 
\limits_{u \in \mathcal{U}} f_{\mathsf{apx}}^{(k,u)}\left( p_u^k \right) R_u^{k,\sf{fh}} \nonumber \\
 \text{s. t. } & \; \text{ constraints (\ref{Ch5_eq:pv}), (\ref{Ch5_eq:SINR_constraint}), (\ref{Ch5_cnt:noL}), (\ref{Ch5_eq:pwc}).}  \label{Ch5_objfun2} 
\end{align}
The objective function of this approximated problem is concave and the feasible region determined by all the constraints is convex.
This problem is of the following general form:
\beq \label{Ch5_eq:grlp}
\min \limits_{\mathbf{x}} g(\mathbf{x}) \; \text{s.t.} \; \mathbf{x} \in \mathcal{F}
\eeq
where $g(\mathbf{x})$ is concave with respect to $\mathbf{x}$, and $\mathcal{F}$ is the corresponding feasible region.
This problem can be solved by the standard gradient method as follows. 
We start with an initial solution $\mathbf{x}^{(0)} \in \mathcal{F}$ and then iteratively update its solution as
\beq \label{Ch5_eq:grm}
\mathbf{x}^{(n+1)}=\text{arg}\min \limits_{\mathbf{x}} g(\mathbf{x}^{(n)}) + \nabla g(\mathbf{x}^{(n)})(\mathbf{x}-\mathbf{x}^{(n)}) \text{ s.t. } \mathbf{x} \in \mathcal{F}.
\eeq
In particular, $\mathbf{x}^{(n+1)}$ needs to be determined from the following problem:
\beq \label{Ch5_eq:grm2}
\min \limits_{\mathbf{x}}  \nabla g(\mathbf{x}^{(n)})\mathbf{x} \text{ s.t. } \mathbf{x} \in \mathcal{F}.
\eeq
Return to our problem, and let $\mathbf{x}$ represent the precoding vectors and powers, and $g(\mathbf{x})$ denote
 the objective function of problem (\ref{Ch5_objfun2}). Then, we have $\nabla g(\mathbf{x}) \mathbf{x} = \sum_{(k,u) \in \mathcal{L} }\left( 1 + q \nabla f_{\mathsf{apx}}\left( p_u^k \right) R_u^{k,\sf{fh}}\right)p_u^k$. 
By applying the gradient method, we can solve problem (\ref{Ch5_objfun2}) by iteratively solving the following problem until convergence:
\beq
 \min \limits_{\lbrace \mathbf{v}_u^k \rbrace} \; \sum \limits_{k \in \mathcal{K}} \sum \limits_{u \in \mathcal{U}} \alpha_u^{k(n)} \mathbf{v}_u^{kH}\mathbf{v}_u^{k} \;
 \text{s. t.} \; \text{ constraints (\ref{Ch5_eq:SINR_constraint}), (\ref{Ch5_cnt:noL}), (\ref{Ch5_eq:pwc})},  \label{Ch5_objfun3} 
\eeq
where 
\beq \label{Ch5_eq:alp}
\alpha_u^{k(n)}=  1 + q \nabla f_{\mathsf{apx}}^{(k,u)}\left( p_u^k \right)R_u^{k,\sf{fh}}.
\eeq
Problem (\ref{Ch5_objfun3}) is a weighted SPMP, which can be transformed into the convex SDP as presented in Appendix~\ref{Ch5_sdp_app}.

\subsubsection{Algorithm Design}
\begin{algorithm}[!t]
\caption{\textsc{SDP-Based Algorithm for PFCPM Problem}}
\label{Ch5_alg:gms3}
\begin{algorithmic}[1]

\STATE Initialization: Set $n=0$, and $\alpha_u^{k(0)}=1$ for all RRH-user links $(k,u)$.

\STATE Iteration $n$:
\begin{description}
\item[a.] Solve problem (\ref{Ch5_objfun3}) with $\left\lbrace \alpha_u^{k(n-1)} \right\rbrace $ to obtain $(\mathbf{p}^{(n)},\lbrace \mathbf{v}_u^k \rbrace^{(n)})$.
\item[b.] Update $\left\lbrace \alpha_u^{k(n)} \right\rbrace $ as in (\ref{Ch5_eq:alp}).
\end{description} 
\STATE Set $n:=n+1$, and go back to Step 2 until convergence.
\end{algorithmic}
\end{algorithm}
The algorithm to solve (\ref{Ch5_objfun2}) is given in Algorithm~\ref{Ch5_alg:gms3}. 
We have the following proposition which states the convergence property of Algorithm~\ref{Ch5_alg:gms3}.
\begin{proposition} \label{Ch5_cvg_alg2}
Algorithm~\ref{Ch5_alg:gms3} converges to a local optimal solution of problem (\ref{Ch5_objfun2}). 
\end{proposition}
\begin{proof} 
The proof is given in Appendix~\ref{prf_Ch5_cvg_alg2}. 
\end{proof}
Algorithm~\ref{Ch5_alg:gms2}, which is proposed to solve problem $\mathcal{P}_{\sf{FCPM}}$ is based on the solution of problem $\mathcal{P}_{\sf{PFCPM}}$, which can be 
obtained by using Algorithm~\ref{Ch5_alg:gms3}. 

\subsection{Iterative Linear-Relaxed Algorithm}
\label{Ch5_IterRelax}
Here, we propose an iterative linear-relaxed algorithm to directly deal with the step-function fronthaul capacity constraint.
Toward this end, we first propose to approximate the step-function by an approximation function. 
By doing so, the noncontinuous constraint with step functions becomes a continuous but still nonconvex one.
To convexify the obtained problem, the nonconvex approximated constraint function is further relaxed to a linear form by using
 the concave duality method \cite{rockafellar70}. The linear-relaxed constraint function can be made sufficiently close to the 
original nonconvex function by iteratively updating its parameters.
First, problem $\mathcal{P}_{\sf{FCPM}}$ can be approximated by the following problem:
\begin{align}
 \min \limits_{\lbrace \mathbf{v}_u^k \rbrace, \mathbf{p}} & \;\;\; \Vert \mathbf{p} \Vert_{\mathbf{1}}  \label{Ch5_obj_4} \\
 \text{s. t. } & \; \text{ constraints (\ref{Ch5_eq:pv}), (\ref{Ch5_eq:SINR_constraint}), (\ref{Ch5_cnt:noL}), (\ref{Ch5_eq:pwc}),}  \nonumber  \\
 {} & \; \sum \limits_{k \in \mathcal{K}} \sum \limits_{u \in \mathcal{U}} f_{\mathsf{apx}}^{(k,u)}( p_u^k ) R_u^{k,\sf{fh}} \leq C. \label{Ch5_const4}
\end{align}
As mentioned earlier, the function $f_{\mathsf{apx}}^{(k,u)}(p_u^k)$ is concave with respect to $p_u^k$; hence, the constraint (\ref{Ch5_const4}) is indeed nonconvex.
Hence, we approximate it by the corresponding linear form based on the duality properties of conjugate of convex functions \cite{rockafellar70} as follows.
First, it can be verified that the function $f_{\mathsf{apx}}^{(k,u)}( p_u^k )$ can be rewritten by using its concave duality as
\beq \label{Ch5_eq:g2}
f_{\mathsf{apx}}^{(k,u)}( p_u^k )= \inf \limits_{z_u^k} \left[ z_u^k p_u^k - f_{\mathsf{apx}}^{(k,u)\ast}(z_u^k) \right],
\eeq
where $f_{\mathsf{apx}}^{(k,u)\ast}(z)$ is the conjugate function of $f_{\mathsf{apx}}^{(k,u)}( w )$, which can be expressed as
\beq \label{Ch5_f_ast}
f_{\mathsf{apx}}^{(k,u)\ast}(z)= \inf \limits_{w} \left[ z w -f_{\mathsf{apx}}^{(k,u)}(w) \right].  \\
\eeq
According to \cite{rockafellar70}, function $f_{\mathsf{apx}}^{(k,u)\ast}(z)$ can be determined by the optimal value of $w$ obtained from the right-hand side of (\ref{Ch5_f_ast}) for a given $z$. 
After substituting the results of (\ref{Ch5_f_ast}) into (\ref{Ch5_eq:g2}), it can be verified that the optimization problem in the right-hand side of (\ref{Ch5_eq:g2}) achieves its minimum at
\beq \label{Ch5_eq:z}
\hat{z}_u^k=\nabla f_{\mathsf{apx}}^{(k,u)}( w )\vert_{w=p_u^k}.
\eeq
With the representation of $f_{\mathsf{apx}}^{(k,u)}(p_u^k)$ as in (\ref{Ch5_eq:g2}), the constraints (\ref{Ch5_eq:C_cons}) can be rewritten in a linear form for a given $\left\lbrace \hat{z}_u^k \right\rbrace $ as 
\beq \label{Ch5_eq:C_cons3}
\sum \limits_{k \in \mathcal{K}}\sum \limits_{u \in \mathcal{U}}  \hat{z}_u^k R_u^{k,\sf{fh}} p_u^k \leq C + \sum \limits_{k \in \mathcal{K}} \sum \limits_{u \in \mathcal{U}} R_u^{k,\sf{fh}} f_{\mathsf{apx}}^{(k,u)\ast}(\hat{z}_u^k).
\eeq
In summary, for a given value of $\left\lbrace \hat{z}_u^k \right\rbrace$, the problem (\ref{Ch5_obj_4})-(\ref{Ch5_const4}) can be reformulated to
\begin{align}
 {} & \min \limits_{\lbrace \mathbf{v}_u^k \rbrace} \:\:\: \sum \limits_{k \in \mathcal{K}} \sum \limits_{u \in \mathcal{U}} \mathbf{v}_u^{kH}\mathbf{v}_u^{k}  \label{Ch5_obj_1sdp} \\
 \text{s. t. } & \; \text{ constraints (\ref{Ch5_eq:SINR_constraint}), (\ref{Ch5_cnt:noL}), (\ref{Ch5_eq:pwc}),}  \nonumber  \\
  {} & \!\!\!\!\! \sum \limits_{k \in \mathcal{K}} \sum \limits_{u \in \mathcal{U}} \! \hat{z}_u^k  R_u^{k,\sf{fh}}  \mathbf{v}_u^{kH}\mathbf{v}_u^{k} \leq   C \!  + \! \sum \limits_{k \in \mathcal{K}} \sum \limits_{u \in \mathcal{U}} \! R_u^{k,\sf{fh}} f_{\mathsf{apx}}^{(k,u)\ast}(\hat{z}_u^k).  \label{Ch5_const5}
\end{align}

Now problem (\ref{Ch5_obj_1sdp})-(\ref{Ch5_const5}) is the well-known SPMP, which can be solved by transforming it
into the SDP, which is described in Appendix~\ref{Ch5_sdp_app} where the additional constraint (\ref{Ch5_const5}) can be rewritten in (\ref{Ch5_eq:zc_mtx}), shown in Appendix~\ref{Ch5_sdp_app} below.
In summary, we can fulfill our design objectives by updating $\left\lbrace \hat{z}_u^k \right\rbrace$ iteratively, based 
on which we repeatedly solve the precoding problem (\ref{Ch5_obj_1sdp})-(\ref{Ch5_const5}).
This complete procedure is described in Algorithm~\ref{Ch5_alg:gms4} whose properties are stated in the following proposition.

\begin{algorithm}[!t]
\caption{\textsc{Iterative Linear-Relaxed Algorithm}}
\label{Ch5_alg:gms4}
\begin{algorithmic}[1]
\STATE Start with a feasible solution.
\STATE Set $l=0$.
\REPEAT 
\STATE Calculate $\left\lbrace \hat{z}_u^{k,(l)} \right\rbrace $ as in (\ref{Ch5_eq:z}) for all $(k,u)$.
\STATE Solve problem (\ref{Ch5_obj_1sdp})-(\ref{Ch5_const5}) with $\left\lbrace \hat{z}_u^{k,(l)}\right\rbrace $.
\STATE Update $l=l+1$.
\UNTIL Convergence.
\end{algorithmic}
\end{algorithm}

\begin{proposition} \label{Ch5_lm2}
Algorithm~\ref{Ch5_alg:gms4} has the following properties.
\begin{enumerate}
\item If the FCPM problem is feasible, Algorithm~\ref{Ch5_alg:gms4} converges.
\item The feasible solution achieved by Algorithm~\ref{Ch5_alg:gms4} at convergence satisfies all constraints of problem (\ref{Ch5_obj_4})-(\ref{Ch5_const4}).
\end{enumerate}
\end{proposition}
\begin{proof} The proof is given in Appendix~\ref{prf_Ch5_lm2}. \end{proof}

\begin{remark}
Some careful study of $\hat{z}_u^k$ given in (\ref{Ch5_eq:z}) yields $f_{\mathsf{apx}}^{(k,u)}( p_u^k )= \hat{z}_u^k p_u^k - f_{\mathsf{apx}}^{(k,u)\ast}(\hat{z}_u^k) $,
which can be considered as the first-order approximation (first-order Taylor expansion) of $f_{\mathsf{apx}}^{(k,u)}( w )$ at $w = p_u^k$. 
In our paper, we prefer using the conjugate function concept in approximating the concave function. This is because the property of the conjugate function, i.e., $f_{\mathsf{apx}}^{(k,u)}( w ) \leq z w - f_{\mathsf{apx}}^{(k,u)\ast}(z) $ for all $z$, enables us to prove the convergence of Algorithm~\ref{Ch5_alg:gms4}, which is
stated in Proposition~\ref{Ch5_lm2}.
\end{remark}

\subsection{Adjusting Precoding and Power Solution}

After running Algorithms \ref{Ch5_alg:gms2} and \ref{Ch5_alg:gms4},
the obtained solution may have several small but nonzero power elements in the power vector $\mathbf{p}^{\ast}$.
Since the approximation function is smooth, there may be significant difference between $C(\mathbf{p}^{\ast})$ and
$\sum \limits_{k \in \mathcal{K}} \sum \limits_{u \in \mathcal{M}} R_u^{k,\sf{fh}} f_{\mathsf{apx}}^{(k,u)}( p_u^{k\ast})$. To address
this problem, we force the power value and precoding vector to zero for any link $(k,u)$ as follows:
\beq \label{Ch5_eq_rnd0}
p^{k \ast}_u = 0 \text{ if } f_{\mathsf{apx}}^{(k,u)}(p^{k \ast}_u)  < {1}/{2}.
\eeq
We will show how to design the approximation function so that this adjustment only results in tolerable
performance degradation in the next subsection. 

\subsection{Design of Approximation Function}
Here, we discuss the design of an approximation function that is employed to approximate the step function in the fronthaul
capacity constraint.
Let us denote $\mathbf{p}^{\ast}$ as the transmission vector achieved by any of our proposed algorithms at convergence. 
It can be verified that we have $\sum \limits_{k \in \mathcal{K}} 
\sum \limits_{u \in \mathcal{M}} R_u^{k,\sf{fh}} f_{\mathsf{apx}}^{(k,u)}( p_u^{k\ast}) = C$ and the SINR
of any user $u$ satisfies $\Gamma_u^{\ast} = \bar{\gamma}_u$.

First, we want the value of the approximation function to be close to one for large power values so that it well approximates
the step function. Specifically, it is required that
\beq \label{Ch5_eq_rnd1}
f_{\mathsf{apx}}^{(k,u)}(P_k) = 1.
\eeq
In addition, we would like the adjusting procedure described earlier to
result in the deviation of at most $\epsilon$ for the achieved SINR of any user.
Let $\Gamma_u^{\ast}\vert_{p^{k \ast}_u=0}$ denote the SINR achieved by user $u$ after we apply the adjustment
procedure given in (\ref{Ch5_eq_rnd0}) for the obtained power vector $\mathbf{p}^{\ast}$. 
Then, the results in the following proposition provide the guideline to achieve this design goal.

\begin{proposition} \label{Ch5_tlr_f}
If the approximation concave function satisfies the following condition:
\beq \label{Ch5_fapx_cdt}
f_{\mathsf{apx}}^{(k,u)}\left( \epsilon \beta^k_u \right) \geq {1}/{2}
\eeq
where $\beta^k_u=\bar{\gamma}_u \sigma^2/\vert h^k_u\vert^2$, then we have 
\beq
{\bar{\gamma}_u - \Gamma_u^{\ast}\vert_{p^{k \ast}_u=0}}/{\bar{\gamma}_u} < \epsilon.
\eeq
\end{proposition}
\begin{proof} The proof is given in Appendix~\ref{prf_Ch5_tlr_f}. \end{proof}

We provide an example of an exponential approximation function and the required conditions on its parameters in Table.~\ref{Ch5_tb:cv_fct}.
Note that other choice of the approximation functions is possible as long as it satisfies the conditions stated in Proposition~\ref{Ch5_tlr_f}.

\begin{table*}[!t]
\caption{Step Function Approximation}
\centering
\begin{tabular}{| p{30mm} | c | c | c | }
\hline 
$f_{\mathsf{apx}}^{(k,u)}(p^k_u)$ - \textbf{Exponential function} & \textbf{Conjugate function} & \textbf{Condition in (\ref{Ch5_eq_rnd1})} & \textbf{Condition in (\ref{Ch5_fapx_cdt})} \\
\hline
$\lambda^k_u \left( 1-e^{-\Psi^k_u p^k_u}\right)$;  $(\Psi^k_u \gg 1)$ & $\dfrac{z^k_u}{\Psi^k_u}\left[ 1 - \log(\dfrac{z^k_u}{\lambda^k_u\Psi^k_u}) \right] -\lambda^k_u$ & $ \lambda_u^k = \left( 1-e^{-\Psi^k_u P_k}\right)^{-1}$ & $1 + e^{-\Psi^k_u P_k} \geq 2 e^{-\Psi^k_u \epsilon \beta^k_u}$  
 \\ 
 \hline 
\end{tabular}
\label{Ch5_tb:cv_fct}
\end{table*}

\subsection{Complexity Analysis}
The complexity of the exhaustive search and our proposed
algorithms are now investigated based on the number of required
computations. As can be observed, all these algorithms
require us to solve SPMPs several times (i.e., $\mathcal{P}_{\mathfrak{a}}$ for exhaustive search algorithm,
 (\ref{Ch5_objfun3}) for the pricing-based algorithm, and (\ref{Ch5_obj_1sdp}) for the linear-relaxed algorithm) by transferring them into the SDP problems as described in Appendix~\ref{Ch5_sdp_app}.
Hence, we
first study the complexity of the SDP program corresponding
to each algorithm. We assume that all RRHs are equipped
with the same number of antennas $N$,  and we consider the
worst case where each UE can be served by all RRHs in the
network. Then, the numbers of variables of all SDP problems
are the same, which is equal to $M N K$.
On other hand, the numbers of constraints of ($\mathcal{P}_{\mathfrak{a}}$) and (\ref{Ch5_objfun3}) are $M+K$, whereas the corresponding number of (\ref{Ch5_obj_1sdp}) is $M+K+1$.
As given in \cite{lou10}, the computational complexity involved in solving the SDP is $\mathcal{O}(\max(m,n)^4n^{1/2}\log(\zeta_{\sf{SDP}}^{-1}))$ where $n$ is the number of variables, $m$ is the number of constraints, and $\zeta_{\sf{SDP}}$ represents the solution accuracy.
In practice, the number of users is greater than that of RRHs, and they are larger than 1, i.e., $M > K > 1$; hence, we have $M N K > M + K + 1$.
Thus, we can express the complexity of solving problem ($\mathcal{P}_{\mathfrak{a}}$), (\ref{Ch5_objfun3}) and (\ref{Ch5_obj_1sdp}) as
\beq
X_{\sf{SDP}} = \mathcal{O}\left( \Pi^{4.5} \log(\zeta_{\sf{SDP}}^{-1}) \right),
\eeq
where $\Pi=MNK$. The complexity of the exhaustive search algorithm and our proposed algorithms can be calculated based on the number of iterations and $X_{\sf{SDP}}$ as follows.
First, let us define $R_{\sf{min}}^{\sf{fh}} = \min_{u \in \mathcal{U}} R_u^{k,\sf{fh}}$ and $\bar{C}=C/R_{\sf{min}}^{\sf{fh}}$. Then, the number of possible $\mathfrak{a}$ (i.e., the number 
of elements of set $\mathcal{S}_{\mathfrak{a}}$) is upper bounded by ${MK \choose \bar{C}}$. Hence, the complexity of the exhaustive search method can be expressed as
\beq
X_{\sf{exh}}= {MK \choose \bar{C}} \times X_{\sf{SDP}} = \mathcal{O}\left( {MK \choose \bar{C}}\Pi^{4.5} \log(\zeta_{\sf{SDP}}^{-1}) \right).
\eeq
We now quantify the complexity of solving FCPM problem using Algorithm~\ref{Ch5_alg:gms3}, which is adopted in Algorithm~\ref{Ch5_alg:gms2}. 
Note that, in each iteration of Algorithm~\ref{Ch5_alg:gms3}, problem (\ref{Ch5_objfun3}) is solved with certain values of $\lbrace \alpha_u^k \rbrace$ and $q$.
For a given value of $q$, the number of iterations for solving PFCPM problem is $\mathcal{O}\left( \zeta_2^{-2}\right)$ because the gradient method is applied \cite{Cartis10}, where $\zeta_2$ 
represents the solution accuracy of Algorithm~\ref{Ch5_alg:gms3}.
In addition, using the bisection searching method, the value of $q$ is typically determined  
after $\mathcal{O}\left( \log(\bar{q}\zeta_1^{-1})\right)$ searching steps where $\zeta_1$ denotes the solution accuracy of Algorithm~\ref{Ch5_alg:gms2}.
Hence, the complexity of Algorithm~\ref{Ch5_alg:gms2} can be calculated as
\beq
X_{\sf{PBA}}  = \mathcal{O}\left( \zeta_2^{-2}\right) \times \mathcal{O}\left( \log(\bar{q}/\zeta_1)\right) \times X_{\sf{SDP}} 
 = \mathcal{O}\left(  \Pi^{4.5} \zeta_2^{-2} \log(\zeta_{\sf{SDP}}^{-1}) \log(\bar{q}\zeta_1^{-1})\right).
\eeq
Finally, if we let $I_{\sf{LRA}}$ be the number of iterations required in Algorithm~\ref{Ch5_alg:gms4} to solve the FCPM problem, the complexity of the Algorithm~\ref{Ch5_alg:gms4} can be given as
\beq
X_{\sf{LRA}} = I_{\sf{LRA}} \times X_{\sf{SDP}} = \mathcal{O}\left(  I_{LRA} \Pi^{4.5} \log(\zeta_{\sf{SDP}}^{-1}) \right).
\eeq 
As illustrated in the numerical results, Algorithm~\ref{Ch5_alg:gms4} converges very fast after around 10 iterations for the small system and 30 iterations for the large system. Hence, 
Algorithm~\ref{Ch5_alg:gms4} is less complex than Algorithm~\ref{Ch5_alg:gms2}. Moreover, this complexity study also shows that the complexities of our proposed algorithms are much 
lower than that of the exhaustive searching method.

\subsection{Existing Algorithms}
\label{Ch5_ex_Alg}
We present two existing algorithms in \cite{Quek2013} and \cite{dai14} to evaluate the relative performance of our proposed algorithms in Section~\ref{Ch5_results}. 
These works also attempt to determine the set
of RRH–user links for each user, although their considered
problems are different from ours. Therefore, these existing
algorithms are modified so that they can be applied to solve
our considered problem. This is described in the following.

\subsubsection{Zhao, Quek, and Lei (ZQL) algorithm}

We first describe the greedy algorithm proposed by Zhao, Quek, and Lei (ZQL), which is called \textit{iterative link removal algorithm} in \cite{Quek2013}.
In this paper, the authors aim to minimize the number of active BS-user links (i.e., fronthaul capacity minimization) subject to the constraints on BSs' maximum transmission power 
and users' target SINRs. This greedy iterative algorithm works as follows.
In each iteration, the power-minimization precoding problem
with a specific set of active BS–user links is solved where
the authors allow all possible BS–user links initially. If this
problem is feasible (i.e., all power and QoS constraints can
be supported), then the active BS–user links are sorted in the
increasing order of their required transmission power. Then,
the set of removed BS–user links is updated before the whole
procedure is repeated in the next iteration.

If the power-minimization precoding problem is not feasible
in a particular iteration, then the number of removed BS–user
links is reduced compared to that in the previous iteration.
Otherwise, we increase the number of removed BS–user links.
These updates for the number of removed BS–user links follow
the bisection method. This iterative procedure is repeated until the largest number of removed BS–user links can be determined
while all constraints can still be maintained.

This greedy link removal principle can be applied to solve
our problem as follows. We initially solve the precoding problem
for power minimization where all possible RRH–user links
are allowed. Then, the required power values of all RRH–user
links are sorted in increasing order based on which we remove
the minimum number of links with the lowest power to maintain
the fronthaul capacity constraint. Finally, we solve the precoding
problem again to determine the precoding and power-allocation
solution. 

\subsubsection{Dai and Yu (DY) algorithm}
The second algorithm is proposed by Dai and Yu (DY) \cite{dai14}, where 
the authors consider the sum-rate maximization problem subject to the constraints on RRHs' transmission power and fronthaul capacity similar to (\ref{Ch5_eq:ivd_C_cons}).
The main idea of this algorithm is to approximate the
step function in the fronthaul constraint by a simple function
and solve the corresponding optimization problem iteratively.
Specifically, the fronthaul capacity constraint is approximated
as follows \cite{dai14}:
\beq \label{Ch5_fappx}
\sum \limits_{k \in \mathcal{K}}\sum \limits_{u \in \mathcal{U}}\beta_u^k  R_u^{k,\sf{fh}}  p_u^k \leq C,
\eeq
where $\beta_u^k$ is a parameter associated with the link between RRH $k$ and user $u$, which
is updated iteratively as
\beq 
\beta_u^k = \dfrac{1}{p_u^k + \tau},
\eeq
where $\tau$ is a small number. 

This approximation can be employed to solve our problem by iteratively solving the problem (\ref{Ch5_obj_1sdp})-(\ref{Ch5_const5}) 
where the fronthaul capacity constraint is approximated by (\ref{Ch5_fappx}).
The limitation of this approximation is that convergence cannot be established.
In addition, we will show later in Section~\ref{Ch5_results} that our proposed algorithms
achieve better performance than this algorithm in all investigated simulation settings.

\section{Further extension} \label{Ch5_sec:furex}
\subsection{Individual Fronthaul Capacity Constraints}
Here, we discuss the extension of the considered CoMP transmission design where there are multiple
individual fronthaul capacity constraints. In particular, each of these constraints captures the limited capacity of one particular fronthaul link between the
 CPC and the corresponding RRH. In this scenario, we have
to replace the sum fronthaul capacity constraint (\ref{Ch5_eq:C_cons}) with the following set of constraints:
\beq \label{Ch5_eq:ivd_C_cons}
G_k(\mathbf{p}^k) = \sum \limits_{u \in \mathcal{M}} R_u^{k,\sf{fh}} \delta( p_u^{k\ast}) \leq C_k, \;\;\; \forall k \in \mathcal{K},
\eeq
where $C_k$ denotes the capacity limit corresponding to the fronthaul link of RRH $k$.
Then, we have the following individual fronthaul constraint power minimization (IFCPM) problem
\beq
(\mathcal{P}_{\sf{IFCPM}}) \;\;\;\; \min \limits_{\lbrace \mathbf{v}_u^k \rbrace, \mathbf{p} } \;\;\; \Vert \mathbf{p} \Vert_{\mathbf{1}}  \;\;\; \text{s. t. }  \; \text{ constraints (\ref{Ch5_eq:pv}), (\ref{Ch5_eq:SINR_constraint}), (\ref{Ch5_cnt:noL}), (\ref{Ch5_eq:pwc}), and (\ref{Ch5_eq:ivd_C_cons})}. \label{Ch5_obj_IFCPM}  
\eeq

Obviously, problem $\mathcal{P}_{\sf{IFCPM}}$  is more challenging to address than problem $\mathcal{P}_{\sf{FCPM}}$.
However, we show in the following that the same design principles adopted in the previous section can be employed to develop 
low-complexity algorithms in this scenario as well.

\subsubsection{Pricing-Based Algorithm}
To address multiple fronthaul capacity constraints in this case, we have to consider the pricing problem with multiple corresponding
pricing parameters $q_k$ whose objective is given as follows:
\beq
 \min \limits_{\lbrace \mathbf{v}_u^k \rbrace, \mathbf{p} }  \Vert \mathbf{p} \Vert_{\mathbf{1}} + \sum_{k \in \mathcal{K} } q_k G_k(\mathbf{p}^k).
\eeq
For given pricing parameters $q_k, \: k \in \mathcal{K}$, the pricing problem can also be transformed into the convex SDP form, based on
which we can determine the precoding vectors. The key challenge is how to iteratively adjust the pricing
parameters $q_k$ to satisfy all fronthaul capacity constraints. This procedure has
been presented in our previous conference work \cite{vuha_ciss_2014}, which is omitted here due to the space constraint.

\subsubsection{Iterative Linear-Relaxed Algorithm} 
We can apply the same linearization technique to the sum fronthaul capacity constraint in Section~\ref{Ch5_IterRelax}
to deal with each fronthaul capacity constraint in (\ref{Ch5_eq:ivd_C_cons}) for problem $\mathcal{P}_{\sf{IFCPM}}$. 
In particular, these constraints can be relaxed into the following form:
\beq \label{Ch5_mulcon}
\sum \limits_{u \in \mathcal{U}} \!  \hat{z}_u^k  R_u^{k,\sf{fh}} \mathbf{v}_u^{kH}\mathbf{v}_u^{k} \leq   C_k \!  + \! \sum \limits_{u \in \mathcal{U}} \!   R_u^{k,\sf{fh}} {\mathsf{apx}}^{(k,u)\ast}(\hat{z}_u^k), \;\; \forall k \in \mathcal{K}.
\eeq
Then, these constraints can be transformed into the convex SDP form as in (\ref{Ch5_eq:zc_mtx_IFCPM}), shown in Appendix~\ref{Ch5_sdp_app} below.

\subsection{MIMO Systems with Multi-Stream Communications}

We now discuss how our proposed algorithms for multiple-input–single-output
(MISO) systems can be extended for the MIMO systems where multiple data streams can be transmitted to each user.
Let us consider a multi-stream communication setting where user $u$ is equipped with $T_u$ antennas and $D_u$ data streams are transmitted to each user $u$, where $D_u 
\leq T_u \leq \min_{k \in \mathcal{K}} N_k$.
Denote $\mathbf{x}_u = [x_{u,1} x_{u,2} ... x_{u,D_u}]^T \in \mathbb{C}^{D_u \times 1}$ as the vector describing $D_u$ data streams for user $u$, where $\mathbb{E}\lbrace \mathbf{x}_u \mathbf{x}_u^H\rbrace = \mathbf{I}_{D_u \times D_u}$. Moreover,
let $\mathbf{V}^k_u=[\mathbf{v}^k_{u,1} \mathbf{v}^k_{u,2} ... \mathbf{v}^k_{u,D_u}] \in \mathbb{C}^{N_k \times D_u}$ denote the precoding matrix at RRH $k$ corresponding to the signal transmitted to user $u$ as in \cite{Choi03}, $\mathbf{U}_u=[\mathbf{u}_{u,1} \mathbf{u}_{u,2} ... \mathbf{u}_{u,D_u}] \in \mathbb{C}^{T_u \times D_u}$ be the decoding matrix of user $u$, and $\mathbf{H}^k_u \in \mathbb{C}^{T_u \times N_k}$ be the channel matrix between RRH $k$ and user $u$.
Then, the baseband signal $\mathbf{y}_u$ received at user $u$ can be written as
\nomenclature{MISO}{Multiple-Input–Single-Output}
\beq
\label{Ch5_eq:rx_sig_mimo}
\mathbf{y}_u = \mathbf{U}_u^H \left( \sum \limits_{i \in \mathcal{U}} \sum \limits_{l \in \mathcal{K}}  \mathbf{H}_u^{l} \mathbf{V}_i^{l} \mathbf{x}_i  +  \Xi_u \right),  
\eeq
where $\Xi_u \in \mathbb{C}^{T_u \times 1}$ denotes the additive noise vector and $\mathbb{E}\lbrace \Xi_u \Xi_u^H\rbrace = \sigma^2 \mathbf{I}_{D_u \times D_u}$. Then, the received baseband signal for stream $m$ of user $u$ can be expressed as
\beq
y_{u,m}  = \sum \limits_{k \in \mathcal{K}} \sum \limits_{t = 1}^{D_u} \mathbf{u}_{u,m}^{H} \mathbf{H}_u^{k} \mathbf{v}_{u,t}^k x_{u,t} + \sum \limits_{i =1, \neq u}^{M} \sum \limits_{l \in \mathcal{K}} \sum \limits_{t=1}^{M_i} \mathbf{u}_{u,m}^{H} \mathbf{H}_u^{l} \mathbf{v}_{i,t}^{l} x_{i,t} + \mathbf{u}_{u,m}^H \Xi_u. 
\eeq
Moreover, the SINR of stream $m$ of user $u$ can be given as
\beq \label{Ch5_eq:SINR2}
\Gamma_{u,m} = \dfrac{ \left|  \sum \limits_{k \in \mathcal{K}} \mathbf{u}_{u,m}^{H} \mathbf{H}_u^{k} \mathbf{v}_{u,m}^{k} \right| ^2 }
{\sum \limits_{(i,t) \neq (u,m)} \left| \sum \limits_{l \in \mathcal{K}} \mathbf{u}_{u,m}^{H} \mathbf{H}_u^{l} \mathbf{v}_{i,t}^{l}\right|^2 + \sigma^2 \left| \mathbf{u}_{u,m} \right|^2 }.
\eeq
Assume that each data stream $m$ of every user $u$ has a QoS requirement, which is represented by target SINR $\gamma_{u,m}$. Our design problem for the MIMO system
 with multi-stream communications $(\mathcal{P}_{\sf{MIMO}})$ can be written as
\begin{align}
  \min \limits_{ \lbrace \mathbf{V}_u^k \rbrace, \lbrace \mathbf{U}_u \rbrace } & \;\;\; \sum \limits_{\forall(k,u)} \mathsf{Tr} \left( \mathbf{V}_u^{kH} \mathbf{V}_u^k \right)  \label{Ch5_obj_MIMO} \\
 \text{s. t. } & \sum \limits_{u \in \mathcal{U}} \mathsf{Tr} \left( \mathbf{V}_u^{kH} \mathbf{V}_u^k \right) \leq P_k, \;\;\; \forall k, \\
& \Gamma_{u,m} \geq \bar{\gamma}_{u,m}, \;\;\; \forall (u,m), \\
& \mathsf{Tr} \left( \mathbf{V}_u^{kH} \mathbf{V}_u^k \right) = 0 \text{ if } (k,u) \notin \mathcal{L},\\ 
& \sum \limits_{\forall (k,u,m)} \delta(\mathbf{v}^{kH}_{u,m}\mathbf{v}^k_{u,m}) R_{u,m}^{k,\sf{fh}} \leq C,
\end{align}
where these constraints capture the power limitation of RRHs, the QoS requirement of each data stream, the clustering constraint, and the limited fronthaul capacity, respectively.
Here, $R_{u,m}^{k,\sf{fh}}$ represents the required fronthaul capacity of RRH $k$ for transmitting the precoding vector and data stream $m$ of user $u$. 
As described in \cite{Heath11}, for a given fixed precoding matrix, the optimal decoding matrix based on the MMSE criterion can be expressed as follows:
\nomenclature{MMSE}{Minimum Mean Square Error}
\beq \label{Ch5_decoding}
 \mathbf{U}_u^{\star} = \left( \sum \limits_{i \in \mathcal{U}} \sum \limits_{l \in \mathcal{K}} \mathbf{H}_u^{l} \mathbf{V}_i^{l} \mathbf{V}_i^{lH}\mathbf{H}_u^{lH}  + 
\sigma^2 \mathbf{I}_{D_u \times D_u} \right)^{-1} \mathbf{H}_u^{k} \mathbf{V}_u^{k}.
\eeq
Then, if $\mathbf{U}_u$ is given, we can treat each user $u$ as $D_u$ virtual users $\lbrace u_1, u_2, ..., u_{D_u}\rbrace$. Specifically, each virtual user $u_m$ ($m \in \lbrace 1, 2, ..., D_u \rbrace$) receives a single data stream represented by data symbol $x_{u,m}$, the virtual channel vector between RRH $k$ and user $u_m$ can be expressed as $\mathbf{h}^k_{u,m} =  \mathbf{H}_u^{kH} \mathbf{u}_{u,m}$, and the power noise $\sigma_{u,m}^2 = \sigma^2 \left| \mathbf{u}_{u,m} \right|^2$. Therefore, the multi-stream design problem can be reformulated by the corresponding virtual single-stream one. Consequently, our proposed algorithms for the
single-stream MISO setting can be applied to determine the precoding matrices for this multi-stream MIMO system.
Based on these results, we can easily develop an iterative algorithm, which is similar to the one in \cite{Heath11}, to solve the multi-stream problem where we
alternatively determine the precoding matrices and decoding matrices in each iteration.

\begin{algorithm}[!t]
\caption{\textsc{MIMO Fronthaul Constraint Algorithm}}
\label{Ch5_alg:gms5}
\begin{algorithmic}[1]
\STATE Start with any value of $\mathbf{U}_u^{(0)}$ for all $u \in \mathcal{U}$.
\STATE Set $l=0$.
\REPEAT 
\STATE Use our proposed algorithm (Algorithm~\ref{Ch5_alg:gms2} or \ref{Ch5_alg:gms4}) to find the precoding matrices $\mathbf{V}_u^{k,(l)}$ by determining each precoding vector $\mathbf{v}_{u,m}^{k,(l)}$.
\STATE Update the decoding matrices $\mathbf{U}_u^{(l+1)}$ as in (\ref{Ch5_decoding}) for all $u \in \mathcal{U}$.
\STATE Update $l=l+1$.
\UNTIL Convergence.
\end{algorithmic}
\end{algorithm}

\section{Numerical Results}
\label{Ch5_results}

\begin{figure}[!t]
        \centering
        \begin{subfigure}[c]{0.34\textwidth}
                \begin{center}
\includegraphics[height=50mm]{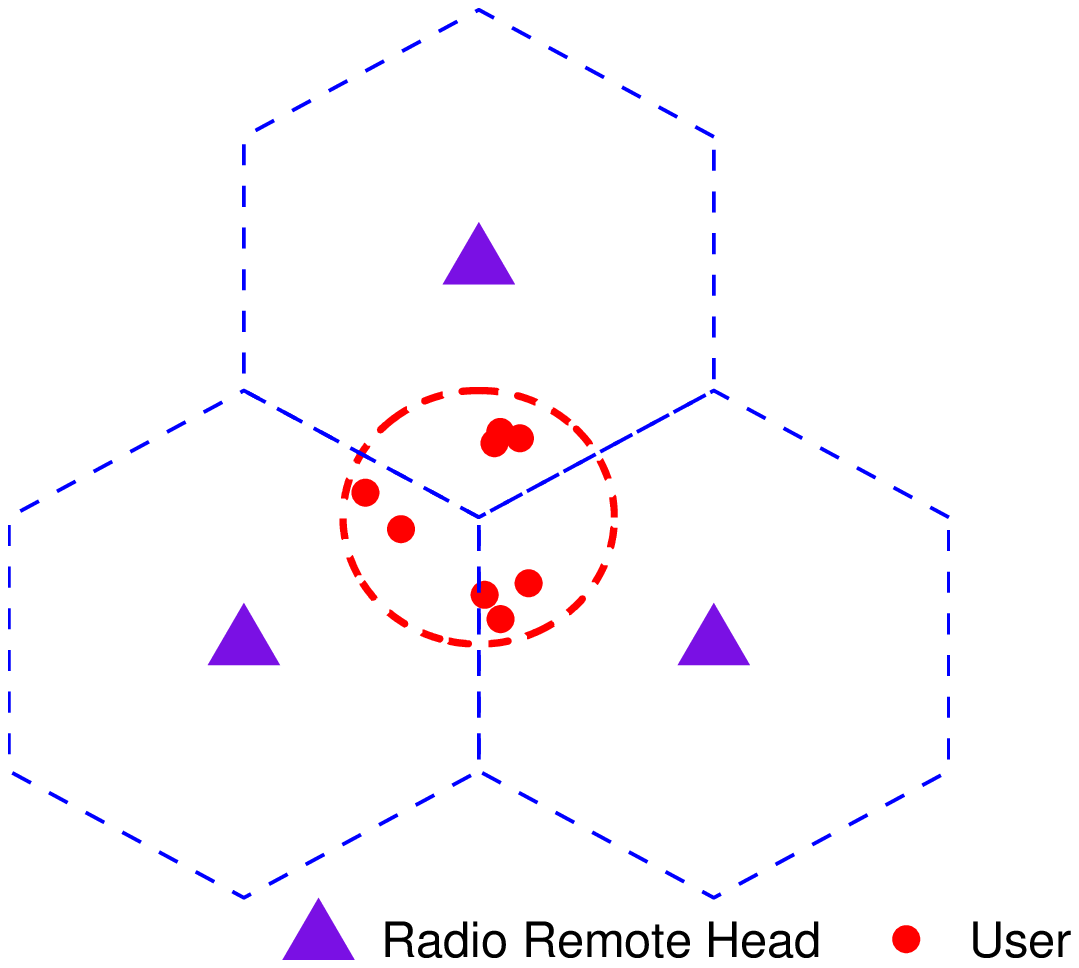}
\end{center}
\caption{Small network} \label{Ch5_small_system}
        \end{subfigure}%
        \qquad  
        ~ 
        \begin{subfigure}[c]{0.46\textwidth}
\begin{center}
\includegraphics[height=50mm]{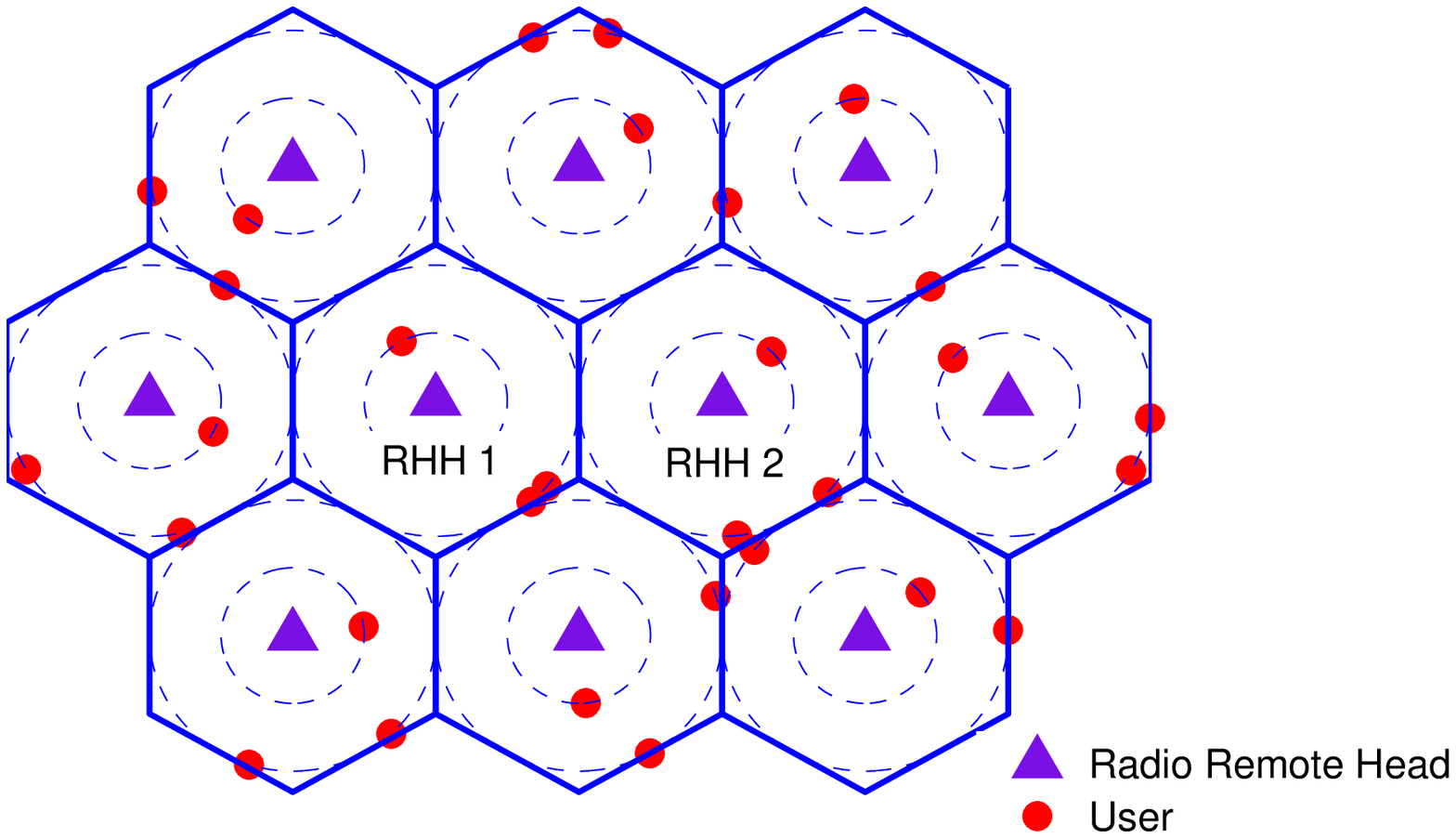}
\end{center}
\caption{Large network}
\label{Ch5_system_model_10BS}
        \end{subfigure}
         \caption{The network simulation setting}
\end{figure}


We present numerical results to illustrate the performance achieved by our proposed algorithms. 
To obtain all numerical results for our proposed algorithms, we employ
the exponential approximation function given in Table.~\ref{Ch5_tb:cv_fct}.
To demonstrate the effectiveness of our design, we consider a small network simulation setting as illustrated in Fig.~\ref{Ch5_small_system},
which allows us to obtain the global optimal solution based on exhaustive search. We then further study the network performance
for a larger network setting shown in Fig.~\ref{Ch5_system_model_10BS}. For ease of exposition, we will present
all results using normalized fronthaul capacity, which is defined as $\bar{C}=C/R_{\sf{min}}^{\sf{fh}}$, where $R_{\sf{min}}^{\sf{fh}} = \min_{u \in \mathcal{U}} R_u^{k,\sf{fh}}$. 
 Furthermore, the initial solution for our iterative linear-relaxed algorithm and DY algorithm \cite{dai14} is obtained by running Algorithm~\ref{Ch5_alg:gms3} with the pricing parameter $\bar{q}$.

Unless stated otherwise, the following parameter and simulation setup are adopted for both network settings.
The channel gains are generated by considering both Rayleigh fading and path loss, which is modeled 
as $L_u^{k}=36.8 \mathsf{log}_{10}(d_{u}^k)+43.8+20\mathsf{log}_{10}(\frac{f_c}{5})$,
where $d_u^k$ is the distance from user $u$ to RRH $k$; $f_c=2.5\:GHz$. 
The noise power is set equal to $\sigma^2=10^{-13} \; W$.  We set other parameters as follows: $\epsilon=10^{-6}$, $\tau=10^{-8}$, RRH power
 $P_k=3 \; W$ for all $k \in \mathcal{K}$, the number of antennas of each RRH equal to 4, except 
for the results in Fig.~\ref{Ch5_3BS_PvsN}. We denote our proposed algorithms
and two existing algorithms as ``Pricing-Based Alg.'', ``Iterative Linear-Relaxed Alg.'', 
``ZQL Alg.'' \cite{Quek2013}, and ``DY Alg.'' \cite{dai14}, respectively, in relevant figures.

\subsection{Small Network Simulation Setting}

\begin{figure}[!t]
\begin{center}
\includegraphics[width=0.7 \textwidth]{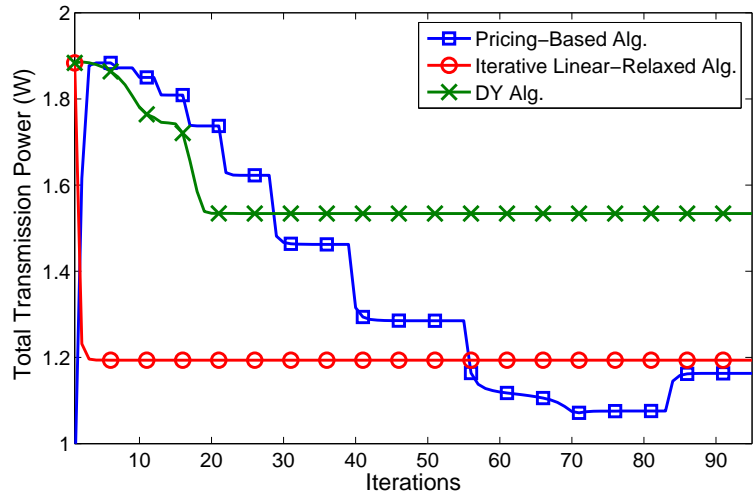}
\end{center}
\caption{Variations of total power under three algorithms.}
\label{Ch5_3BS_3func}
\end{figure}
We consider three RRHs in this setting where the distance between their centers is equal to $500 \; m$.
Users are randomly placed inside a circle at the center of three RRHs whose radius is $125 \; m$. 
First, we examine the convergence of our proposed algorithms and DY algorithm in Fig.~\ref{Ch5_3BS_3func}.
Three different curves represent the variations of total transmission powers due to the pricing-based algorithm, the iterative linear-relaxed
algorithm, and DY algorithm, respectively.
To obtain these simulation results, we set the number of users equal to 7 ($M=7$), and user target SINR $\bar{\gamma}_u=0\;dB$ for all $u \in \mathcal{M}$.
It can be observed that
all algorithms converge, although the iterative linear-relaxed
algorithm is the fastest and the pricing-based algorithm is
the slowest but achieves the lowest transmission power. This
figure also illustrates that the DY algorithm converges after
20 iterations; however, the transmission power achieved by this
algorithm is much higher than those achieved by ours.

\begin{figure}[!t]
\begin{center}
\includegraphics[width=0.7 \textwidth]{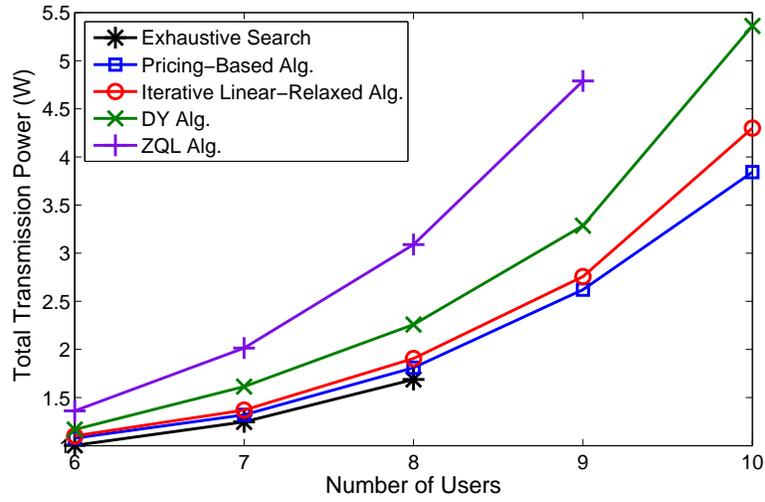}
\end{center}
\caption{Total power versus number of users in a small network.}
\label{Ch5_3BS_PvsM}
\end{figure}

In Fig.~\ref{Ch5_3BS_PvsM}, we show total transmission powers of all RRHs achieved by the
exhaustive searching method, our proposed algorithms, and two reference algorithms, versus the number of users where $\bar{\gamma}_u = 0 \; dB$ $\forall u$.
As can be
seen, our proposed algorithms result in lower total transmission
power compared to DY and ZQL algorithms. Moreover, the
pricing-based algorithm is slightly better than the iterative
linear-relaxed algorithm, and both proposed algorithms require
marginally higher total transmission power than that due to
the optimal exhaustive search algorithm. In addition, the DY
algorithm outperforms the ZQL algorithm, but these existing
algorithms demand considerably higher power compared with
our algorithms. Moreover, this figure shows that the total transmission power increases as the number of users becomes
larger under all algorithms, as expected.

\begin{figure}[!t]
\begin{center}
\includegraphics[width=0.7 \textwidth]{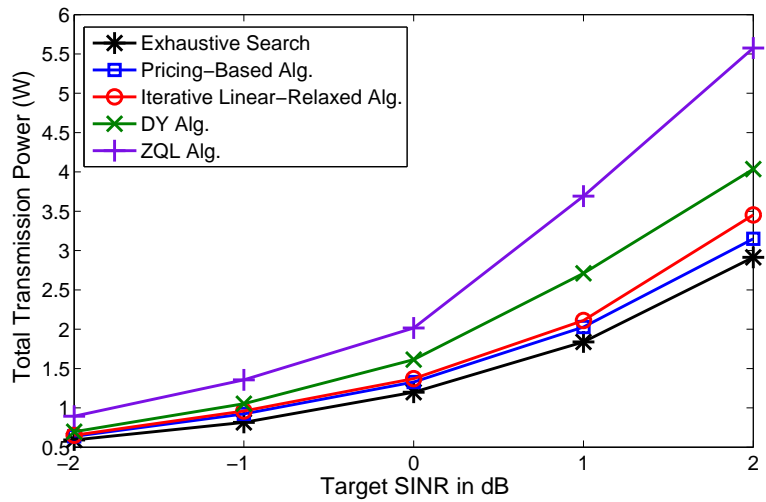}
\end{center}
\caption{Total power versus user target SINR in a small network.}
\label{Ch5_3BS_PvsSINR}
\end{figure}

\begin{figure}[!t]
\begin{center}
\includegraphics[width=0.7 \textwidth]{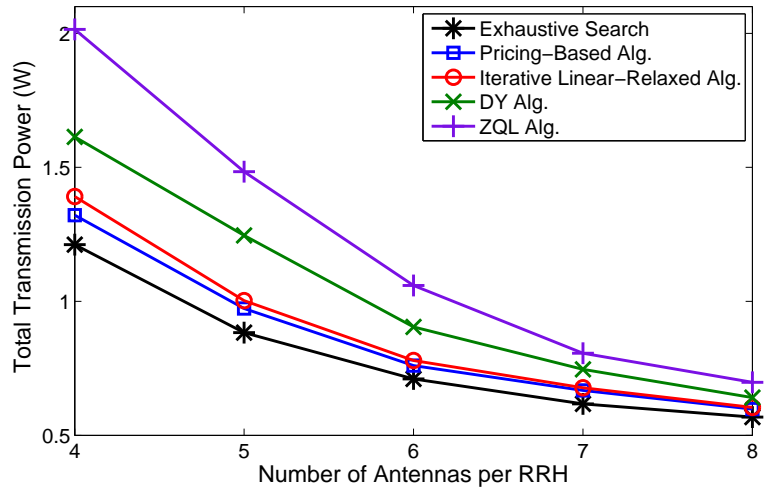}
\end{center}
\caption{Total power versus number of antennas per RRH in a small network.}
\label{Ch5_3BS_PvsN}
\end{figure}

Fig.~\ref{Ch5_3BS_PvsSINR} and Fig.~\ref{Ch5_3BS_PvsN} shows the variations of total transmission power versus user target SINR and  the number of antennas per RRH for $M = 7$, respectively. The target SINR of all users for the scheme in Fig.~\ref{Ch5_3BS_PvsN} is set as $\bar{\gamma}_u = 0 \; dB$.
 As expected, larger
power is needed as users require a higher target SINR, and a
larger number of equipped antennas result in reduction of total
transmission power due to the increasing network degrees of
freedom. In addition, our proposed algorithms can achieve the
total transmission power that is close to that obtained by the
exhaustive search method. Moreover, as can be observed, these
figures also confirm the superiority of the proposed algorithms
compared to the reference algorithms. Once again, the pricing-based
algorithm slightly outperforms the iterative linear-relaxed
algorithm.

\subsection{Large Network Simulation Setting}

We now study the performance of the proposed algorithms for the large network simulation setting where there are ten hexagonal cells 
 ($K=10$) and the distance between the centers of two nearest RRHs is $500 \, m$.
In this setting, three users are randomly placed inside each cell so that the distance between each of them and their nearest RRH is either $250 \, m$ or $125 \, m$. Overall, we 
generate 30 users ($M = 30$) in this network, and we set $\bar{\gamma}_u=0\;dB$ for all users, except the simulation in Fig.~\ref{Ch5_10BS_PvsSINR}. Fig.~\ref{Ch5_C_cv_10BS_Cb70} shows the convergence of two proposed algorithms and DY algorithm where it confirms that our proposed algorithms also perform as expected 
in this large network setting.
Moreover, the pricing-based algorithm requires more
convergence time compared with the iterative linear-relaxed
algorithm and DY algorithm, although the former achieves
smaller transmission power than the latter. Once again, the
iterative linear-relaxed algorithm converges faster than the DY
algorithm.

\begin{figure}[!t]
\begin{center}
\includegraphics[width=0.7 \textwidth]{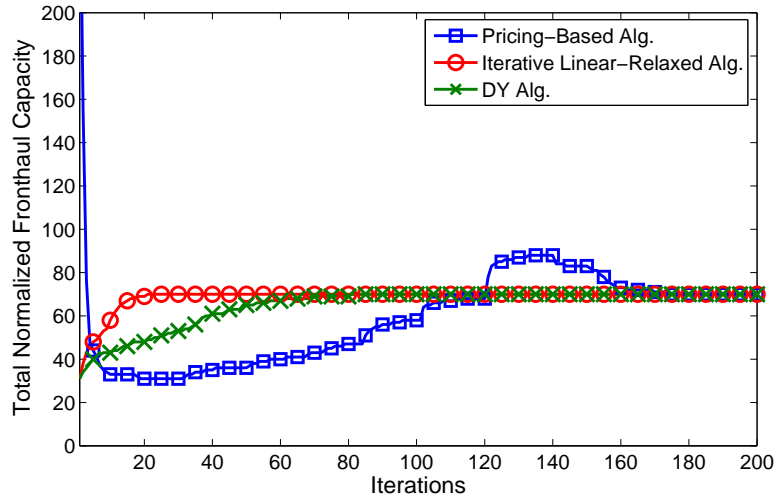}
\end{center}
\caption{Variations of utilized fronthaul capacity due to Algorithms~\ref{Ch5_alg:gms2}, Algorithms~\ref{Ch5_alg:gms4}, and DY Algorithm over iterations.}
\label{Ch5_C_cv_10BS_Cb70}
\end{figure}

\begin{figure}[!t]
\begin{center}
\includegraphics[width=0.7 \textwidth]{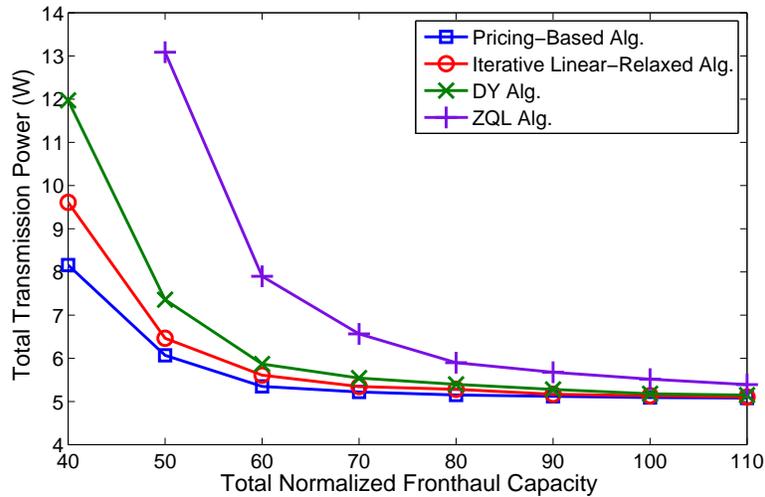}
\end{center}
\caption{Total power versus total fronthaul capacity.}
\label{Ch5_10BS_PvsC}
\end{figure}

\begin{figure}[!t]
\begin{center}
\includegraphics[width=0.7 \textwidth]{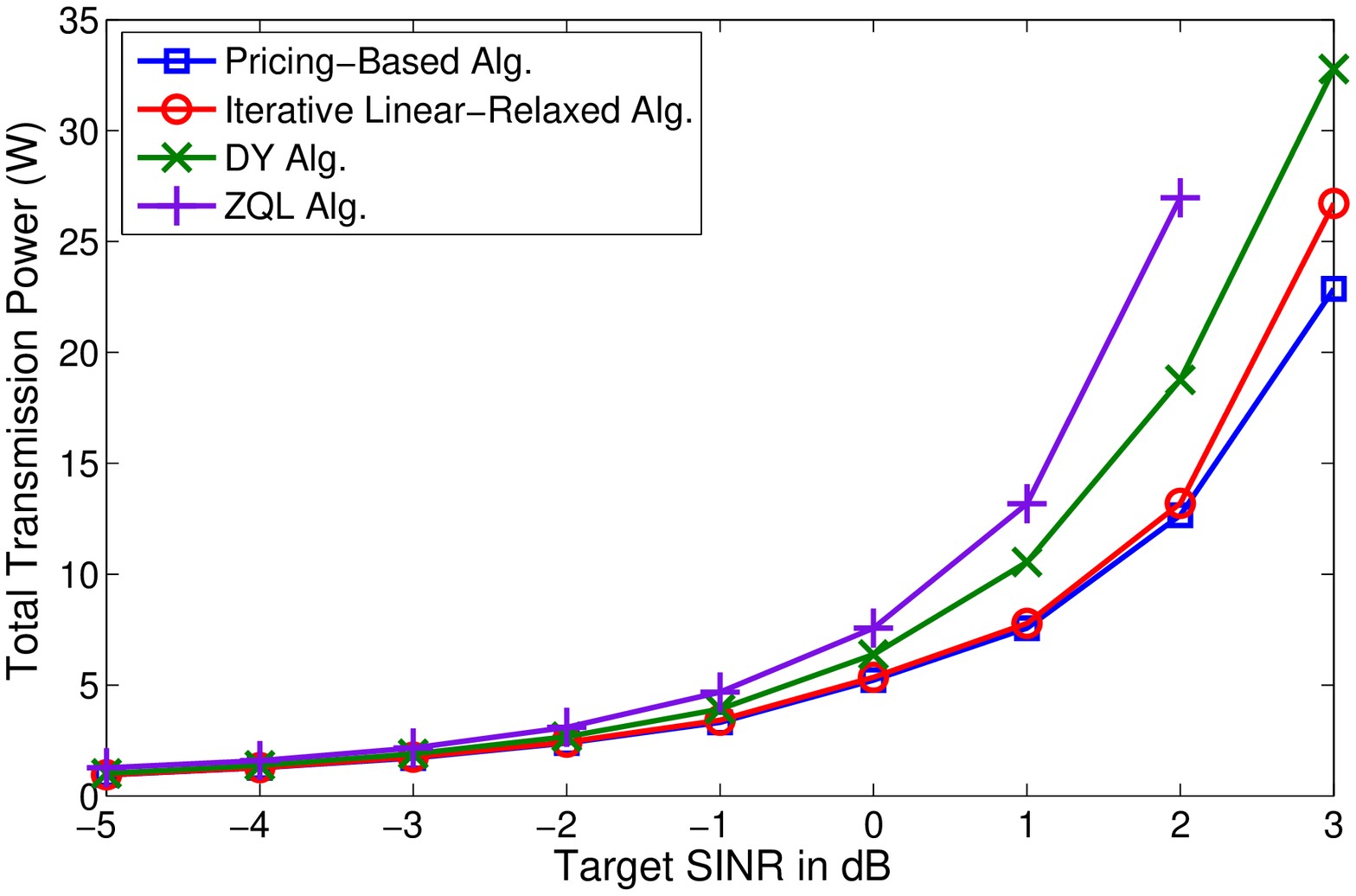}
\end{center}
\caption{Total power versus user target SINR.}
\label{Ch5_10BS_PvsSINR}
\end{figure}

Figs.~\ref{Ch5_10BS_PvsC},~\ref{Ch5_10BS_PvsSINR} show the total transmission powers versus the normalized fronthaul capacity and users' target SINR, respectively.
Our proposed algorithms again outperform the other existing algorithms.
The iterative linear-relaxed algorithm results in the total transmission power which is slightly higher than that due to the
pricing-based algorithm.
Moreover, as shown in Fig.~\ref{Ch5_10BS_PvsC}, the performance gap between these algorithms are more significant
if the available fronthaul capacity is smaller. 
Furthermore, the three algorithms except ZQL algorithm achieve similar performance for
sufficiently large fronthaul capacity or small target SINR values.

In Fig.~\ref{Ch5_10BS_PvsClusterSize}, we present the variations of total transmission power versus the cluster size.
Here, the cluster size is equal to the number of nearest RRHs, which are allowed to serve each user. 
Specifically, the cluster size of $m$ means that each user can be served by at most $m$ nearest RRHs.
In this simulation, we set $\bar{C} = 70$ and the target SINR $\bar{\gamma} = 0 \; dB$.
As can be seen, the larger cluster size results in the lower total transmission power.
However, total transmission powers achieved by all algorithms become flat as the cluster size is sufficiently large. 
These results imply that activating weak links between users and RRHs does not significantly improve the network performance.
This figure again confirms that our proposed algorithms significantly outperform other existing algorithms, and the pricing-based algorithm is slightly better than iterative linear-relaxed algorithm.

\begin{figure}[!t]
\begin{center}
\includegraphics[width=0.7 \textwidth]{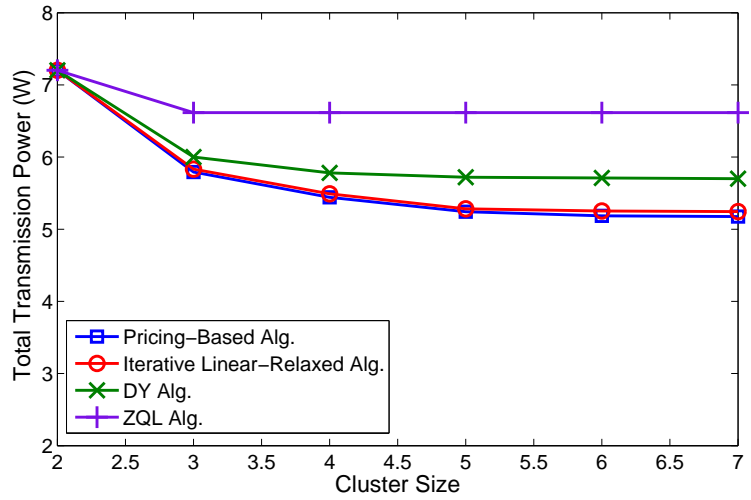}
\end{center}
\caption{Total power versus cluster size.}
\label{Ch5_10BS_PvsClusterSize}
\end{figure}

\subsection{Individual Fronthaul Capacity Constraints}

\begin{figure}[!t]
\begin{center}
\includegraphics[width=0.7 \textwidth]{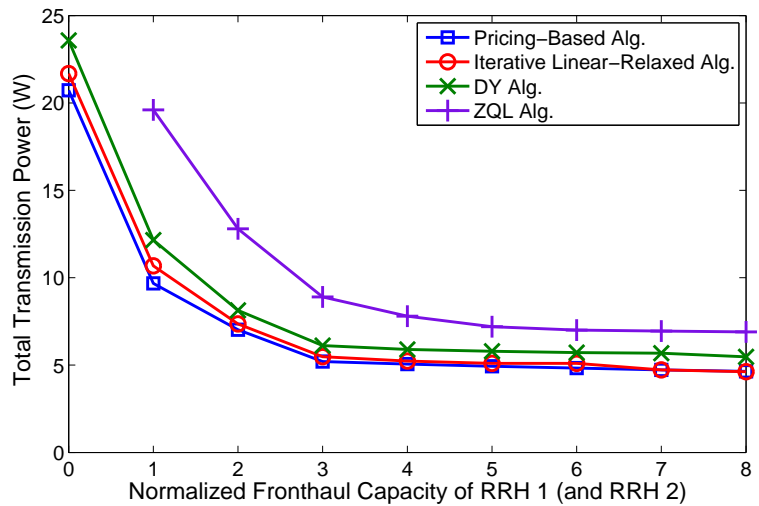}
\end{center}
\caption{Total power versus fronthaul capacity per RRH 1/2.}
\label{Ch5_10BS_PvsC1_C2}
\end{figure}

Fig.~\ref{Ch5_10BS_PvsC1_C2} shows variations of total transmission powers achieved by our proposed algorithms and two reference algorithms as there are multiple individual 
fronthaul capacity constraints. In this simulation, we keep the normalized fronthaul capacity $\bar{C}_k = 8$ for all $ k > 2$ while varying the 
normalized capacity of fronthaul links corresponding RRHs 1 and 2, which is denoted as $\bar{C}_{1-2}$ (i.e., $\bar{C}_1=\bar{C}_2=\bar{C}_{1-2}$).
Again, the individual normalized fronthaul capacity is defined as $\bar{C}_k=C_k/R_{\sf{min}}^{\sf{fh}}$.
In fact, if we set $\bar{C}_{1-2} = 0$, then RRHs 1 and 2 are turned off, which cannot serve any user in the network.
As can be observed, the total transmission power decreases with the increase of $\bar{C}_{1-2}$, which
shows a similar trend as in Fig.~\ref{Ch5_10BS_PvsC}.
Interestingly, all users in cells 1 and 2 can still be supported at their target SINR by running pricing-based, iterative linear-relaxed and DY algorithms
even when RRHs 1 and 2 are turned off, which is not the case with the ZQL algorithm.    
Overall, our proposed algorithms still outperform the two existing algorithms.

\subsection{MIMO Systems with Multi-stream Communications}
 
\begin{figure}[!t]
\begin{center}
\includegraphics[width=0.7 \textwidth]{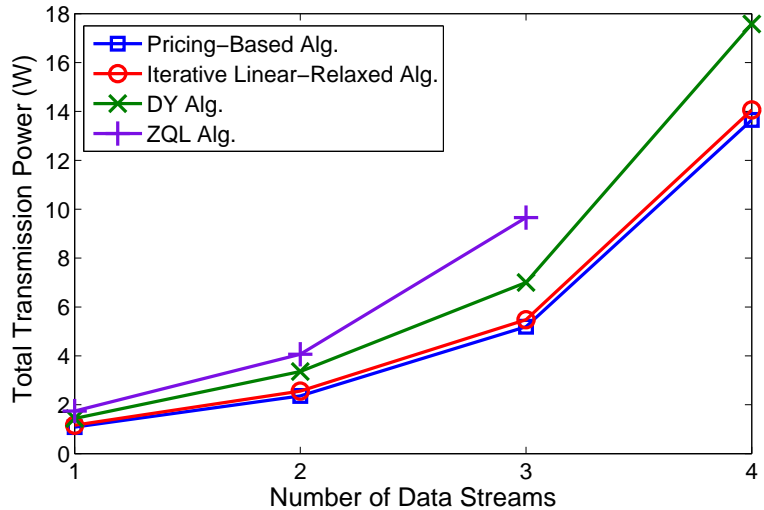}
\end{center}
\caption{Total power versus the number of data streams per user.}
\label{Ch5_mul_streams}
\end{figure}

For this simulation, we consider the large network, but we allocate only two users in each cell to ensure the feasibility of the network. Each RRH is equipped with eight antennas and each user has four antennas.
The total normalized fronthaul capacity is set at $100$, and the target SINR for each stream is $0 \; dB$.
Fig.~\ref{Ch5_mul_streams} shows the total transmission power of all RRHs versus the number of data streams for each users.
As shown, the proposed algorithms still work well for the MIMO stream and also achieve better solution than that of the algorithms in literature.
Interestingly, the higher number of streams per user requires the higher transmission power of all RRHs.

\section{Conclusion} \label{Ch5_sec:cls}

We have proposed efficient and low-complexity algorithms to solve the downlink joint transmission problem in Cloud-RAN that aims 
to minimize the total transmission power subject to constraints on transmission powers, fronthaul capacity, and  users' QoS. 
We have considered both scenarios with sum and individual fronthaul capacity constraints.
Numerical results have illustrated the efficacy of our proposed algorithms and the impacts of different parameters on the network performance.
In particular, the pricing-based and iterative linear-relaxed algorithms achieve the performance very close to that due to the optimal exhaustive search.
In addition, our proposed algorithms significantly outperform two existing algorithms proposed in \cite{Quek2013} and \cite{dai14} in all investigated scenarios.
There are still several interesting directions for future study. First, joint design and optimization of CoMP techniques and fronthaul signal
compression is an interesting problem. Furthermore, consideration of CoMP design, CSI estimation, and feedback deserves further efforts.

\section{Appendices}

\subsection{Weighted Sum-Power Minimization Solution}
\label{Ch5_sdp_app}
Denote $\mathcal{K}_u$ as the set of RRHs allowed to serve user $u$, i.e., $\mathcal{K}_u=\left\lbrace k \vert (k,u) \in \mathcal{L} \right\rbrace$.
Let $\mathbf{v}_u$ denote the precoding solution over all RRHs in $\mathcal{K}_u$,
which is defined as $\mathbf{v}_u=[\mathbf{v}_u^{u_1T}, \mathbf{v}_u^{u_2T} ... \mathbf{v}_u^{u_{a_u}T}]^T$ where $\{u_1,...,u_{a_u}\} = \mathcal{K}_u$ and $a_u=|\mathcal{K}_u|$.
Therefore, we have $\mathbf{v}_u \in \mathbb{C}^{N_u \times 1}$ where $N_u=\sum_{k \in \mathcal{K}_u} N_k$.
Here, we are interested in determining $\mathbf{v}_u$ because we have $\mathbf{v}_u^{i}=\mathbf{0}$ for all $i \notin \mathcal{K}_u$
(i.e., the precoding vector of any RRH that does not serve the underlying user is equal to zero).

Let us define $\mathbf{W}_u=\mathbf{v}_u\mathbf{v}_u^H$, we have $\mathbf{W}_u \in \mathbb{C}^{N_u \times N_u}$. 
It is positive semi-definite ($\mathbf{W}_u \succeq \mathbf{0}$) and has rank one because it is generated from vector $\mathbf{v}_u$.
We also define the channel vector $\mathbf{h}_{u,i}=[\mathbf{h}_u^{i_1T}, \mathbf{h}_u^{i_2T}, ... \mathbf{h}_u^{i_{a_i}T}]^T$ and $\mathbf{H}_{u,i}=\mathbf{h}_{u,i}\mathbf{h}_{u,i}^{H}$ for all ${u,i} \in \mathcal{U}$. Then, the SINR constraint for user $u$ in (\ref{Ch5_eq:SINR_constraint}), and the power constraint for RRHs can be rewritten as 
\beqn
\mathsf{Tr}(\mathbf{H}_{u,u}\mathbf{W}_u) - \bar{\gamma}_u \sum \limits_{i \in \mathcal{U}/u} \mathsf{Tr}(\mathbf{H}_{u,i}\mathbf{W}_i) \geq \bar{\gamma}_u \sigma^2, \;\; \forall u \in \mathcal{U}, \label{Ch5_eq:SINR_mtx} \\
\sum \limits_{u \in \mathcal{U}} \mathsf{Tr}(\mathbf{E}_u^k \mathbf{W}_u) \leq P_k, \;\; \forall k \in \mathcal{K}, \label{Ch5_eq:pw_mtx}
\eeqn
where $\mathbf{E}_u^k=\mathsf{diag}(\mathbf{0}_{N_{u_1} \times 1},..., \mathbf{1}_{N_{u_i} \times 1},...,\mathbf{0}_{N_{u_{a_u}} \times 1})$ if $u_i=k$.
Therefore, the weighted sum-power minimization can be formulated as the following SDP problem:
\begin{align}
 \min \limits_{\left\lbrace \mathbf{W}_u\right\rbrace_{u=1}^{M} } & \sum \limits_{u=1}^M \mathsf{Tr}(\mathbf{F}_u \mathbf{W}_u) \label{Ch5_SDP} \\
 \text{s.t. } & \text{ constraints (\ref{Ch5_eq:SINR_mtx}), (\ref{Ch5_eq:pw_mtx}),} \nonumber\\
 {} & \mathbf{W}_u \succeq \mathbf{0}, \mathsf{rank}(\mathbf{W}_u)=1, \; \forall u, \label{Ch5_W_rk1}
\end{align}
where $\mathbf{F}_u=\mathsf{diag}(\alpha_u^{u_1(n)} \mathbf{1}_{N_{u_1} \times 1}, ...,\alpha_u^{u_{a_u}(n)} \mathbf{1}_{N_{u_{a_u}} \times 1})$.
This transformation reveals a special structure of the precoding design problem. Specifically, if we remove the rank-one constraints in (\ref{Ch5_W_rk1}) from (\ref{Ch5_SDP}) then the new 
problem is convex. In fact, this relaxed problem is the convex semi-definite program (SDP), which can be therefore solved easily by using  standard tools such as CVX solver \cite{Boyd2009}.
As given in \textit{Theorem 3.1} of \cite{Bengtsson99} and \textit{Lemma 2} of \cite{Ottersten01}, if (\ref{Ch5_SDP}) is feasible, then it has at least one solution, where $\mathsf{rank}(\mathbf{W}_u)=1$, for all $u \in \mathcal{U}$.

Because the relaxed SDP problem is convex, $\mathbf{v}_u$ can be calculated as the eigenvector of $\mathbf{W}_u$ 
if the optimum solution is unique. Such unique optimum solution satisfies the rank-one constraint. 
As discussed in \cite{Bengtsson99,Ottersten01}, the relaxed SDP may have more than one optimum solution, which
means that the CVX solver cannot ensure to return the rank-one solution in general.
However, this situation almost surely never happens in practice, except for cases where the channels from two groups are exactly symmetric. 
Nevertheless, if the algorithm does produce one of such solutions where $\mathbf{W}_u$ does not have rank one, we can still obtain rank-one optimum solution from that solution by
using the method described in \textit{Lemma 5} of \cite{Ottersten01}, or by applying the solution method for solving the rank constrained problem presented in \textit{Algorithm~2} of
 \cite{Huang10}, or by utilizing the best rank-one approximation based on the largest eigenvalue and the corresponding eigenvector as discussed in Section~II of \cite{lou10}.

In addition, this SDP-based transformation can be applied to solve the 
problems (\ref{Ch5_obj_1sdp})-(\ref{Ch5_const5}) and (\ref{Ch5_obj_IFCPM}) where the additional linear-relaxed fronthaul constraints (\ref{Ch5_const5}) and 
(\ref{Ch5_mulcon}) can also be rewritten respectively into the SDP-form as follows:
\beq
\sum \limits_{u \in \mathcal{U}} R_u^{k,\sf{fh}} \mathsf{Tr}(\mathbf{Z}_u \mathbf{W}_u) \leq C + \! \sum \limits_{k \in \mathcal{K}} \sum \limits_{u \in \mathcal{U}} \! R_u^{k,\sf{fh}} f_{\mathsf{apx}}^{(k,u)\ast}(\hat{z}_u^k), \label{Ch5_eq:zc_mtx} 
\eeq
and
\beq
\sum \limits_{u \in \mathcal{U}} R_u^{k,\sf{fh}} \mathsf{Tr}(\mathbf{Z}^k_u \mathbf{W}_u) \leq C_k + \! \sum \limits_{u \in \mathcal{U}} \!  R_u^{k,\sf{fh}} f_{\mathsf{apx}}^{(k,u)\ast}(\hat{z}_u^k), \label{Ch5_eq:zc_mtx_IFCPM} 
\eeq
where $\mathbf{Z}_u =\mathsf{diag}(z_u^{u_1}\mathbf{1}_{N_{u_1} \times 1},...,z_u^{u_{a_u}}\mathbf{1}_{N_{u_{a_u}} \times 1})$ and $\mathbf{Z}^k_u =\mathsf{diag}(\mathbf{0}_{N_{u_1} \times 1},..., z_u^{u_i}\mathbf{1}_{N_{u_i} \times 1}, ..., \mathbf{0}_{N_{u_{a_u}} \times 1})$ if $u_i=k$.

\subsection{Proof of Proposition~\ref{Ch5_lm1}}
\label{prf_Ch5_lm1}
\subsubsection{Proof of Statement 1}

Let us consider two values of the pricing parameter $q$ and $q^{\prime}$ where $q^{\prime} > q$.
Let $(\mathbf{p}, \left\{ \mathbf{v}_u^k \right\} )$ and $(\mathbf{p}^{\prime}, \left\{ \mathbf{v}_u^{k \prime} \right\})$ be the solutions of PFCPM problem with $q$ and $q^{\prime}$, respectively.
Then, we have
\beqn
\sum \limits_{\forall (u,j)} p_u^j + q G_{\sf{PFCPM}}(q) \leq \sum \limits_{\forall (u,j)} p_u^{j\prime} + q G_{\sf{PFCPM}}(q^{\prime}),  \\ 
\sum \limits_{\forall (u,j)} p_u^{j \prime} + q^{\prime} G_{\sf{PFCPM}}(q^{\prime}) \leq \sum \limits_{\forall (u,j)} p_u^{j} + q^{\prime} G_{\sf{PFCPM}}(q).
\eeqn
After combining and simplifying these two inequalities, we have
$(q^{\prime}-q)G_{\sf{PFCPM}}(q^{\prime}) \leq (q^{\prime}-q)G_{\sf{PFCPM}}(q)$. Therefore, we have finished the proof for the first statement.

\subsubsection{Proof of Statement 2}

We prove that, for any $(\mathbf{p}, \left\{ \mathbf{v}_u^k \right\})$ satisfying constraints (\ref{Ch5_eq:pv}), (\ref{Ch5_eq:SINR_constraint}), (\ref{Ch5_eq:pwc}), we always
 have $G(\mathbf{p}) \geq G_{\sf{PFCPM}}(q)$ if $q \geq \bar{q}$. 
Let $P^{\sf{opt}}(q)$ be the total power of PFCPM problem corresponding to $q$.
Because $(\mathbf{p}, \left\{ \mathbf{v}_u^k \right\} )$ satisfy constraints (\ref{Ch5_eq:pv}), (\ref{Ch5_eq:SINR_constraint}), (\ref{Ch5_eq:pwc}), we have
\beqn
{} & P^{\sf{opt}}(q)+ q G_{\sf{PFCPM}}(q) \leq \sum \limits_{(k,u) \in \mathcal{L}} p_u^k + q G(\mathbf{p}),  \\
\rightarrow  & G_{\sf{PFCPM}}(q) - G(\mathbf{p}) \leq  \frac{ \sum \limits_{(k,u) \in \mathcal{L}} p_u^k - P^{\sf{opt}}(q)}{q} < \dfrac{\sum \limits_{k \in \mathcal{K}} P_k}{\bar{q}}. 
\eeqn

Thus, we have $G_{\sf{PFCPM}}(q)-G(\mathbf{p})<\sigma_{\sf{min}}$ if $q \geq \bar{q}$. 
Based on the definition of $\sigma_{\sf{min}}$, we must have that $G(\mathbf{p})$ cannot be smaller than $G_{\sf{PFCPM}}(q)$.

\subsubsection{Proof of Statement 3}
From the proof of statement 2, we can see that there exists no $(\mathbf{p}, \left\{ \mathbf{v}_u^k \right\})$ satisfying constraints (\ref{Ch5_eq:pv}), (\ref{Ch5_eq:pwc}), (\ref{Ch5_eq:SINR_constraint}) so that 
$G(\mathbf{p}) \leq C <G_{\sf{PFCPM}}(q)$. Hence, the last statement of Proposition~\ref{Ch5_lm1} is proved.

\subsection{Proof of Proposition~\ref{Ch5_cvg_alg2}}
\label{prf_Ch5_cvg_alg2}

To prove the convergence of Algorithm~\ref{Ch5_alg:gms3}, we will prove that the value of the objective function of problem (\ref{Ch5_objfun2}) (which is represented by $g(\mathbf{x})$) decreases after each step of Algorithm~\ref{Ch5_alg:gms3}.
For any vector $\mathbf{x}$ representing all precoding vectors and powers, the following inequality holds because of the concavity of function $g(\mathbf{x})$
\beq \label{Ch5_inqlt1}
g(\mathbf{x}) \leq g(\mathbf{x}^{(n)}) + \nabla g(\mathbf{x}^{(n)}) \left( \mathbf{x} - \mathbf{x}^{(n)}  \right). 
\eeq
In addition, the solution achieved in the $(n+1)$th iteration is 
\beq
\mathbf{x}^{(n+1)} = \arg \min \nabla g(\mathbf{x}^{(n)}) \mathbf{x} \text{ s.t. } \: \text{constraints (\ref{Ch5_eq:SINR_constraint}), (\ref{Ch5_cnt:noL}), (\ref{Ch5_eq:pwc})}. \nonumber
\eeq
Hence, we have $\nabla g(\mathbf{x}^{(n)}) \mathbf{x}^{(n+1)} \leq  \nabla g(\mathbf{x}^{(n)}) \mathbf{x}^{(n)}$. 
Substituting $\mathbf{x}^{(n+1)}$ into the inequality (\ref{Ch5_inqlt1}) yields
\beq
g(\mathbf{x}^{(n+1)}) \leq g(\mathbf{x}^{(n)}) + \nabla g(\mathbf{x}^{(n)}) \left( \mathbf{x}^{(n+1)} - \mathbf{x}^{(n)}  \right)   \leq g(\mathbf{x}^{(n)}). 
\eeq
Thus, the objective function of problem (\ref{Ch5_objfun2}) corresponding to the solution in each iteration of Algorithm~\ref{Ch5_alg:gms3} is monotonically decreasing. 
Therefore, Algorithm~\ref{Ch5_alg:gms3} must converge to a local optimum point. 

\subsection{Proof of Proposition~\ref{Ch5_lm2}}
\label{prf_Ch5_lm2}

\subsubsection{Proof of Statement 1}

Denote $\Omega_l$ as the optimum objective value of problem (\ref{Ch5_obj_1sdp})-(\ref{Ch5_const5}) in iteration $l$,
 which is obtained after performing steps 8-12. We will prove that $\Omega_l$ decreases over each iteration; hence, our proposed algorithm converges. 
Denote $\mathcal{F}_l$ as the feasible set of (\ref{Ch5_obj_1sdp})-(\ref{Ch5_const5}) in iteration $l$, which corresponds to the value of $\left\lbrace \hat{z}_u^{k,(l)} \right\rbrace$.
Let us define $\left(\mathbf{p}^{(l)}, \left\lbrace \mathbf{v}_u^{k,(l)} \right\rbrace  \right)$ as the optimal solution of problem (\ref{Ch5_obj_1sdp})-(\ref{Ch5_const5}) corresponding to
 $\left\lbrace \hat{z}_u^{k,(l)} \right\rbrace $ as we run steps 8-12 in Algorithm~\ref{Ch5_alg:gms4} . 
By setting $\left\lbrace \hat{z}_u^{k,(l+1)} \right\rbrace $ as in (\ref{Ch5_eq:z}), we always have $\left(\mathbf{p}^{(l)}, \left\lbrace \mathbf{v}_u^{k,(l)}\right\rbrace \right) \in \mathcal{F}_{l+1}$.
Hence, we must have
\beq
\Omega_l \geq \Omega_{l+1} \;\; \forall l>0
\eeq
which completes the proof.

\subsubsection{Proof of Statement 2}
Denote $\mathcal{F}$ as the feasible set of problems (\ref{Ch5_obj_4})-(\ref{Ch5_const4}). Because of (\ref{Ch5_eq:g2}), if $\sum \limits_{k \in \mathcal{K}} \sum \limits_{u \in \mathcal{U}}   \hat{z}_u^k R_u^{k,\sf{fh}}p_u^{k} \leq   C   +  \sum \limits_{k \in \mathcal{K}} \sum \limits_{u \in \mathcal{U}} R_u^{k,\sf{fh}} f_{\mathsf{apx}}^{(k,u)\ast}(\hat{z}_u^k)$, we have $\sum \limits_{k \in \mathcal{K}} \sum \limits_{u \in \mathcal{U}} R_u^{k,\sf{fh}} f_{\mathsf{apx}}^{(k,u)}(p_u^k) \leq C$ for any value of $\left\lbrace \hat{z}_u^{k,(l+1)} \right\rbrace $.
Therefore, we have $\mathcal{F}_l \subseteq \mathcal{F}$ for any iteration $l$, which means the optimal solution of problem (\ref{Ch5_obj_1sdp})-(\ref{Ch5_const5}) in any iteration $l$ 
satisfies all constraints of problem (\ref{Ch5_obj_4})-(\ref{Ch5_const4}). 
Hence, Algorithm~\ref{Ch5_alg:gms4} returns the solution that satisfies all constraints of problem (\ref{Ch5_obj_4})-(\ref{Ch5_const4}).

\subsection{Proof of Proposition~\ref{Ch5_tlr_f}}
\label{prf_Ch5_tlr_f}

Let $\mathbf{p}^{\ast}$ and $\left\lbrace \mathbf{v}_u^{k \ast}\right\rbrace $ be the power vector and precoding vectors achieved after running our proposed algorithms. Because these results are obtained by solving a weighted SPMP with SINR constraint for each user, we have 
\beq \label{Ch5_prof41}
\bar{\gamma}_u = \dfrac{ \left|  \sum \limits_{j \in \mathcal{K}} \mathbf{h}_u^{jH} \mathbf{v}_u^{j \ast}\right| ^2 }
{\sum \limits_{i =1, \neq u}^{M} \left| \sum \limits_{l \in \mathcal{K}}\mathbf{h}_u^{lH} \mathbf{v}_i^{l \ast}\right|^2 + \sigma^2 }.
\eeq 
Applying the triangle inequality yields
\beq \label{Ch5_prof42}
\dfrac{ \left|  \sum \limits_{j \in \mathcal{K}} \mathbf{h}_u^{jH} \mathbf{v}_u^{j \ast}\right| ^2 }
{\sum \limits_{i =1, \neq u}^{M} \left| \sum \limits_{l \in \mathcal{K}}\mathbf{h}_u^{lH} \mathbf{v}_i^{l \ast}\right|^2 + \sigma^2 } \leq \dfrac{ \left|  \sum \limits_{j \in \mathcal{K}/k} \mathbf{h}_u^{jH} \mathbf{v}_u^{j \ast}\right| ^2 + \left| \mathbf{h}_u^{kH} \mathbf{v}_u^{k \ast}\right| ^2}
{\sum \limits_{i =1, \neq u}^{M} \left| \sum \limits_{l \in \mathcal{K}}\mathbf{h}_u^{lH} \mathbf{v}_i^{l \ast}\right|^2 + \sigma^2 }.
\eeq
Now, if we suppose that the power of user $u$ and RRH $k$ is forced to be zero (i.e., $p^{k \ast}_u=0$), then its
 updated SINR becomes
\beq \label{Ch5_prof43}
\Gamma_u^{\ast}\vert_{p^{k \ast}_u=0} = \dfrac{ \left|  \sum \limits_{j \in \mathcal{K}/k} \mathbf{h}_u^{jH} \mathbf{v}_u^{j \ast}\right| ^2}
{\sum \limits_{i =1, \neq u}^{M} \left| \sum \limits_{l \in \mathcal{K}}\mathbf{h}_u^{lH} \mathbf{v}_i^{l \ast}\right|^2 + \sigma^2 }.
\eeq
Consider the contribution of the underlying power $p^{k \ast}_u$, we have
\beq \label{Ch5_prof44}
\dfrac{ \left| \mathbf{h}_u^{kH} \mathbf{v}_u^{k \ast}\right| ^2}
{\sum \limits_{i =1, \neq u}^{M} \left| \sum \limits_{l \in \mathcal{K}}\mathbf{h}_u^{lH} \mathbf{v}_i^{l \ast}\right|^2 + \sigma^2 } 
 \leq \dfrac{ \left| \mathbf{h}_u^{k} \right| ^2 \left| \mathbf{v}_u^{k \ast} \right| ^2}{\sigma^2 } = \dfrac{ \left| \mathbf{h}_u^{k} \right| ^2 p^{k \ast}_u}{\sigma^2 }
\eeq
Combining all these results in (\ref{Ch5_prof41})--(\ref{Ch5_prof44}) yields
\beq
\dfrac{\bar{\gamma}_u - \Gamma_u^{\ast}\vert_{p^{k \ast}_u=0}}{\bar{\gamma}_u} \leq  \dfrac{ \left| \mathbf{h}_u^{k} \right| ^2 p^{k \ast}_u}{\bar{\gamma}_u \sigma^2 }.
\eeq

Because we have assumed that $p^{k \ast}_u$ is forced to be zero as in (\ref{Ch5_eq_rnd0}) and $f^{(k,u)}_{\mathsf{apx}}(x)$ is an increasing function, 
 we have $p^{k \ast}_u < f^{(k,u)-1}_{\mathsf{apx}}(1/2)$, where $f^{(k,u)-1}_{\mathsf{apx}}(x)$ is the inverse function of $f^{(k,u)}_{\mathsf{apx}}(x)$.
Therefore, if $f^{(k,u)}_{\mathsf{apx}}(x)$ satisfies (\ref{Ch5_fapx_cdt}), which 
implies $f^{(k,u)-1}_{\mathsf{apx}}(1/2) \leq \epsilon \beta^k_u = \dfrac{\epsilon \bar{\gamma}_u \sigma^2   }{ \left| \mathbf{h}_u^{k} \right| ^2 }$., we have
\beq
\dfrac{\bar{\gamma}_u - \Gamma_u^{\ast}\vert_{p^{k \ast}_u=0}}{\bar{\gamma}_u} < \dfrac{ \left| \mathbf{h}_u^{k} \right| ^2}{\bar{\gamma}_u \sigma^2 } \dfrac{\epsilon \bar{\gamma}_u \sigma^2   }{ \left| \mathbf{h}_u^{k} \right| ^2 } = \epsilon.
\eeq
This completes the proof of Proposition~\ref{Ch5_tlr_f}.

%% file: chap8/Ha_chap8.tex
\chapter{Resource Allocation for Wireless Virtualization of OFDMA-Based Cloud Radio Access Networks} 
\renewcommand{\rightmark}{Chapter 8.  RA for Wireless Virtualization of OFDMA-Based Cloud Radio Access Networks}
\label{Ch8}
The content of this chapter was submitted in IEEE Transactions on Vehicular Technology in the following paper:

Vu N. Ha and Long B. Le, ``Resource allocation for wireless virtualization of OFDMA-based Cloud Radio Access Networks,'' {\em IEEE Trans. Veh. Tech.,} 2016 (Under review).

\section{Abstract}
We consider the resource allocation for the virtualized OFDMA uplink Cloud Radio Access Network (C-RAN) where 
multiple wireless operators (OPs) share the C-RAN infrastructure and resources owned by an infrastructure provider (InP). 
\nomenclature{InP}{Infrastructure Provider}
The resource allocation is designed through studying tightly coupled problems at two different levels. 
The upper-level problem aims at slicing the fronthaul capacity and cloud computing
resources for all OPs to maximize the weighted profits of OPs and InP considering practical constraints on
the fronthaul capacity and cloud computation resources. Moreover, the lower-level problems maximize individual OPs'
 sum rates by optimizing users' transmission rates and quantization bit allocation for the compressed I/Q baseband
signals. We develop a two-stage algorithmic framework to address this two-level resource allocation design. In the first stage, we transform
both upper-level and lower-level problems into corresponding problems
by relaxing underlying discrete variables to the continuous ones. We show that these
relaxed problems are convex and we develop fast algorithms to attain their optimal solutions. In the second stage,
we propose two methods to round the optimal solution of the relaxed problems and
achieve a final feasible solution for the original problem.
Numerical studies confirm that the proposed algorithms outperform a greedy resource allocation
algorithm and their achieved sum rates are very close to sum rate upper-bound obtained by solving relaxed problems.
Moreover, we study the impacts of different parameters on the system sum rate, various performance tradeoffs,
and illustrate insights on a potential system operating point and resource provisioning issues.

\section{Introduction}

Next-generation wireless cellular systems are expected to provide significantly
higher capacity in a cost-efficient manner to support the tremendous growth of wireless traffic and services  \cite{cisco14, mob_rep_13}.
Some recent studies have indicated that the traditional model of single ownership of network architecture can be inefficient because 
the average load demand is usually much lower than the designed peak demand \cite{cisco14,mob_rep_13,chinamobile2011, Costa-Perez_CM13, Liang_CST15}.
Therefore, advanced access techniques for C-RAN and wireless virtualization to support multiple OPs (also called
\textit{``wireless network virtualization (WNV)''}) have attracted a lot of attention from both industry and academia recently 
\cite{cisco14,mob_rep_13,chinamobile2011,Checko15,Wubben14,Suryaprakash15,Costa-Perez_CM13,Liang_CST15,Wang_JSAC}.
\nomenclature{WNV}{Wireless Network Virtualization}

By realizing various communications and processing functions in the cloud, C-RAN enables more efficient utilization of 
network resources, which results in better network throughput and reduced network deployment and operation costs.
With WNV, multiple OPs can efficiently share various network resources such as radio spectrum, computation resources, backhaul/fronthaul capacity; 
hence, the capital expenditures (CAPEX) and operating expenses (OPEX) can be reduced significantly \cite{Liang_CST15}, \cite{Costa-Perez_CM13}. 
\nomenclature{CAPEX}{Capital Expenditures}
\nomenclature{OPEX}{Operating Expenditures}
To attain the potential benefits of the C-RAN and WNV technologies, one has to address many technical challenges \cite{Checko15,Liang_CST15}.
Resource allocation, which determines the allocation of centralized computation resource \cite{Wubben14,Suryaprakash15}, fronthaul capacity \cite{Peng14}, 
radio spectrum and power allocation in C-RAN \cite{Checko15}, and the slicing/allocation of infrastructure resources for 
different OPs to optimize the desired design objectives in WNV are among the major research challenges \cite{Liang_CST15}.

\subsection{Related Works}

Recent literature on C-RAN and WNV has tackled some of these technical problems which are described in the following. 
In particular, authors in \cite{Luoto_TWC15} consider the spectrum sharing problem in a heterogeneous wireless 
network where small-cell base stations are optimally matched with OPs. 
In \cite{Kokku_ACM10}, the virtual resource slicing problem, which aims at maximizing the system utility as a function of
achieved cumulative rates and assigned resource slots while meeting the requirement of each slice is studied.
The slicing problem for different OPs is also considered in \cite{Kamel_VTC14} where a fairness-based dynamic resource allocation scheme 
is proposed. Recently, a novel energy-efficient resource allocation strategy is proposed for C-RAN virtualization with optical fronthaul \cite{Wang_JSAC}
where the matching problem between cells, users, baseband units (BBUs), and a number of wavelengths in optical links is investigated. 
Two heuristic methods, namely static and dynamic ones, are presented to solve this problem.
However, these works do not study potential strategies to share computation resources among OPs and efficiently utilize
the fronthaul transport network.

Resource allocation for efficient utilization of the C-RAN fronthaul capacity has been studied 
in \cite{letaief14,luo_arxiv,VuHa_TVT16,shamai13b,rui_zhang_Tcom15,XRao15}.
The works \cite{VuHa_TVT16,letaief14,luo_arxiv} address the precoding problem for remote radio heads (RRHs) to minimize the total 
power consumption where \cite{VuHa_TVT16,letaief14} considers the downlink while \cite{luo_arxiv} addresses both downlink and uplink communications. 
In \cite{letaief14,luo_arxiv}, the authors optimize the utilization of fronthaul links accounting for the power consumption of fronthaul links 
while we consider the downlink joint transmission design for RRHs to minimize the total power consumption under the fronthaul capacity constraint 
in \cite{VuHa_TVT16}. Moreover, the data compression issue for fronthaul capacity reduction
has been addressed in \cite{shamai13b,rui_zhang_Tcom15} where
\cite{shamai13b} focuses on minimizing the amount of data transmitted over the fronthaul transport network while
 \cite{rui_zhang_Tcom15} jointly designs signal quantization and power control to maximize the system sum rate.
The authors in \cite{XRao15} study the joint fronthaul signal compression and signal recovery 
in the uplink C-RAN; however, this work does not consider transmission design aspects. 

A few existing works have considered the cloud computation complexity for processing users' data, which is a 
major challenge in large-scale C-RAN deployment. 
The works \cite{Werthmann13,sabella14,Rost_TWC15,Rost_GBC15} study the optimization of C-RAN computational resources.
Specifically, \cite{Werthmann13} models the computation complexity in downlink communication considering computational requirement
 of electrical circuits for processing the base-band signals where the computation complexity is a non-linear function of different parameters including the 
number of antennas, modulation bits corresponding to FFT blocks, coding rate, and the number of data streams.

The work \cite{sabella14} then applies this model to quantify the C-RAN energy-efficiency benefits.
The authors in \cite{Rost_TWC15} propose a different computation complexity model for the uplink C-RAN system, which accounts mainly for 
the decoding process of turbo-encoded uplink data streams of all users.
This work is motivated by the fact that the power utilized in the decoding process is much higher than that in the encoding process.
Based on this model, the authors address the rate allocation for uplink transmissions to maximize the system sum rate considering
the cloud computational capacity constraint \cite{Rost_GBC15}. This work, however, assumes unlimited fronthaul capacity and does not 
address the resource allocation.

\subsection{Research Contributions}

To our best knowledge, uplink C-RAN design considering constraints on limited fronthaul capacity and cloud computation resources has not
been studied except our previous work  \cite{VuHa_WCNC14}, which, however, relies on an over-simplified computational complexity model
and does not consider the C-RAN virtualization issue.
This paper aims to fill this gap in the existing C-RAN literature where the virtualized OFDMA-based uplink C-RAN design
is addressed. In particular, we make the following contributions.

\begin{itemize}

\item
We consider the virtualized resource allocation design for the uplink OFDMA-based C-RAN where an InP
leases its resources to different OPs to support their mobile users. This design boils down to 
solving two-level coupled problems where the upper-level problem aims to determine the resource slicing solution for the computation resource and fronthaul capacity
to maximize the weighted profits of the InP and OPs while the lower-level problems model the resource allocation of individual 
OPs for their users. Specifically, each lower-level problem must be solved by the corresponding OP to maximize 
the sum rate of its users by determining the optimal rate and quantization bit allocations for the resource slicing solution 
given by the upper-level problem. 

\item
We develop a two-stage solution framework to solve the tightly coupled two-level problems. 
In stage one, we study the relaxed problems of the corresponding upper-level and lower-level problems to deal with
 the discrete rate and quantization bit allocation variables. We show that the relaxed problems in 
both levels are convex and we describe how to solve these problems optimally. 
Specifically, by employing the dual-based approach, we derive the optimal rate and quantization bit allocation solution 
for a given dual point, which enables us to develop a fast algorithm to solve the relaxed lower-level (RLL) problem optimally.
Importantly, the optimal solution of the RLL problem is employed to tackle the relaxed upper-level (RUL) problem.
In the second stage, we propose two rounding methods which are applied to the optimal solutions of the relaxed problems to
attain a feasible solution for the original problem.  


\item 
For performance evaluation of the developed algorithms, we also describe a greedy resource allocation algorithm. 
Extensive numerical studies are conducted where we examine the convergence and efficiency of the proposed algorithms 
as well as the impacts of different system parameters on the system sum rate. In addition, we also study different 
tradeoffs, which characterize the relations of available resources, resource provisioning, and the corresponding 
 benefits achieved by the InP and OPs.
\end{itemize}

The remaining of this paper is organized as follows. We describe the system model and formulations of two-level problems
 in Section~\ref{Ch8_secII}. In Section~\ref{Ch8_secIII}, we characterize the convexity of these relaxed problems.
Then, we develop the optimal algorithms to solve these relaxed problems and present the proposed rounding methods in Section~\ref{Ch8_secIV}. 
The greedy resource allocation algorithm is presented in Section~\ref{Ch8_secV}. Numerical results are presented in Section~\ref{Ch8_secVI} followed by conclusion in Section~\ref{Ch8_secVII}.

\section{System Model}
\label{Ch8_secII}
We consider the C-RAN which consists of BBUs in the cloud, $K$ RRHs, and the fronthaul transport network 
connecting RRHs to the cloud. The C-RAN is owned by an InP which leases network resources to $O$ OPs to serve their own users.
The uplink OFDMA transmission based on the Third Generation Partnership Project (3GPP) LTE standard with full frequency reuse (frequency reuse factor of one) is assumed. 
\nomenclature{3GPP}{Third Generation Partnership Project}
Specifically,
 each cell utilizes the whole spectrum comprising $S$ physical resource blocks (PRBs) where
each PRB corresponds to $12$ sub-carriers (180kHz) in the frequency domain and a slot duration of $t_s=0.5~ms$, which
is equivalent to $7$ OFDM symbols \cite{Suryaprakash15}. We denote the set of all PRBs as $\mathcal{S}$.

Based on certain long-term contracts between the InP and OPs, we assume that the PRB allocations to individual OPs in each cell have been predetermined. 
Moreover, let $\mathcal{S}_k^o$ denote the set of PRBs assigned to OP $o$ in cell $k$. 
We assume that $\cup_{o \in \mathcal{O}} \mathcal{S}_k^o = \mathcal{S}$ and $\mathcal{S}_k^o \cap \mathcal{S}_k^m = \varnothing$ for any 
two different OPs $o$ and $m$. Therefore, there is only inter-cell interference on a specific PRB if concurrent transmissions 
from different cells occur on the underlying PRB. We assume that  each RRH upon receiving users' baseband signals of different OPs
quantizes these signals and forwards them to the cloud for decoding. 
In the following, we refer to RRH $k$ and its corresponding coverage area as cell $k$.
We further assume that both RRHs and users are equipped with single antenna.

Let $x_{k}^{(s)} \in \mathbb{C}$ represent the baseband signal transmitted on PRB $s$ in cell $k$ and we assume that
the signal $x_{k}^{(s)}$ has the unit power. Then, the signal received at RRH $k$ on PRB $s$ can be written as
\beq
y^{(s)}_{k} = \sum_{j \in \mathcal{K}} h_{k,j}^{(s)} \sqrt{p_{j}^{(s)}} x_{j}^{(s)} + \eta_k^{(s)},
\eeq
where $\mathcal{K}$ denotes the set of cells, $p_{j}^{(s)}$ represents the transmission power corresponding to $x_{j}^{(s)}$, 
$h_{k,j}^{(s)}$ is the channel gain from the user assigned PRB $s$ in cell $j$ to RRH $k$, and $\eta_k^{(s)} \sim \mathcal{CN}\left( 0,\sigma_k^{(s)2}\right)$ 
denotes the Gaussian thermal noise.

\subsection{Signal Quantization and Processing}

We assume that all baseband signals $y^{(s)}_{k}$ must be quantized and then forwarded to
the cloud for further processing and decoding, which are performed by the BBU.
Moreover, RRH $k$ uses $b_{k}^{(s)}$ bits to quantize the real and imaginary parts of the received symbol $y^{(s)}_{k}$.
Then, according to the results in \cite{Bucklew79}, the quantization noise power can be approximated as
\beq
q^{(s)}_{k}(b^{(s)}_{k}) \simeq  {2 Q(y^{(s)}_{k})}/{2^{b_{k}^{(s)}}},
\eeq
where $ Q(y) = {\left( \int_{-\infty}^{\infty} f(y)^{1/3} dy \right)^3 }/{12}$, and $f(y^{(s)}_{k})$ is the probability density function of both the real and imaginary parts of $y^{(s)}_{k,u}$. Assuming a Gaussian distribution of the signal to be quantized, we have \cite{Baracca14}
\beq \label{eq:q_kmn}
q^{(s)}_{k}(b^{(s)}_{k}) \simeq   \dfrac{\sqrt{3} \pi}{2^{2 b_{k}^{(s)} +1}}  Y^{(s)}_{k},
\eeq
where $ Y^{(s)}_{k}$ is the power of received signal $y^{(s)}_{k}$, which is equal to 
\beq
Y^{(s)}_{k} = \sum \limits_{j \in \mathcal{K}} \vert h_{k,j}^{(s)} \vert ^2 p_{j}^{(s)} + \sigma_k^{(s)2} = D_k^{(s)} + I_k^{(s)},
\eeq
where $D_k^{(s)} = \vert h_{k,k}^{(s)} \vert ^2 p_{k}^{(s)}$ and $I_k^{(s)} = \sum_{j \in \mathcal{K}/k} \vert h_{k,j}^{(s)} \vert ^2 p_{j}^{(s)} + \sigma_k^{(s)2}$.
Then, the total number of quantization bits required by all users of OP $o$ which are forwarded from RRH $k$ to the cloud in
 one second (measured in bit-per-second (bps)) can be expressed as
\beq
B_k^o = 2N_{\sf{RE}} \sum_{s \in \mathcal{S}_k^o} b_{k}^{(s)},
\eeq 
where $N_{\sf{RE}}$ is the number of resource elements (REs) in one second and one subchannel of a PRB (corresponding to 12 subcarriers).  
Then, $N_{\sf{RE}}$ can be calculated as $N_{\sf{RE}}= 12 \times 7 / t_s$ \cite{Suryaprakash15} since each LTE slot
spans an interval $t_s$ = 0.5ms with 7 symbols. 
Let $\mb{b}$ denote the vector whose elements represent the number of quantization bits allocated to all users in the network.
For a given $\mb{b}$, the quantized version of $y^{(s)}_{k}$ can be written as 
\beq
\tilde{y}^{(s)}_{k} = y^{(s)}_{k} + e^{(s)}_{k},
\eeq
where $e^{(s)}_{k}$ represents the quantization error for $y^{(s)}_{k}$, which has zero mean and variance $q^{(s)}_{k}(b^{(s)}_{k})$. The SINR of 
the signal corresponding to PRB $s$ in cell $k$ can be expressed as
\beq \label{eq:SINR}
\gamma^{(s)}_{k}(b^{(s)}_{k}) = \dfrac{D_k^{(s)}}{I_k^{(s)}+ q^{(s)}_{k}(b^{(s)}_{k})} \simeq \dfrac{D_k^{(s)}}{I_k^{(s)} + \frac{\sqrt{3} 
\pi Y^{(s)}_{k}}{2^{2 b_{k}^{(s)} +1}} }.
\eeq

\subsection{Decoding Computational Complexity}

We assume that a data rate $r^{(s)}_{k}$ (in \textit{``bits per channel use (bits pcu)''}) on 
PRB $s$ is chosen from a given discrete set with $M_R$ different rates (i.e., different predetermined modulation and coding schemes) 
$\mathcal{R} = \lbrace R_1, R_2, ..., R_{M_R} \rbrace$. 
Moreover, the chosen rate is set smaller than the link capacity to ensure satisfactory communication reliability and manageable decoding complexity, i.e., 
$r^{(s)}_{k} \leq \log_2(1 + \gamma^{(s)}_{k}(b^{(s)}_{k}))$. We assume that the capacity-achieving turbo code is
employed, then the computation effort required to successfully decode information bits depends on the number of turbo-iterations. 
According to the results \cite{Rost_TWC15}, the required computation effort expressed in \textit{``bit-iterations (bi)''} for decoding the 
transmitted  signal on PRB $s$ in cell $k$ is a function of $\gamma^{(s)}_{k}(b^{(s)}_{k})$ and $r^{(s)}_{k}$ can be expressed as 
\beq \label{eq:cpt}
C^{(s)}_{k} = \chi^{(s)}_{k}(r^{(s)}_{k},b^{(s)}_{k})  =  A r^{(s)}_{k} \!\! \left[B - 2 \log_2 \! \left( \log_2 \! \left( \! 1 +\gamma^{(s)}_{k}(b^{(s)}_{k}) \! \right) \! - r^{(s)}_{k} \! \right) \! \right], 
\eeq
where $A = 1/\log_2(\zeta - 1)$, $B = \log_2\left( {(\zeta -2)}/{(\zeta T(\epsilon_{\sf{ch}}))}\right)$, $\zeta$ is a parameter 
related to the connectivity of the decoder, $T(\epsilon_{\sf{ch}})  =  -{T'}/{\log_{10}(\epsilon_{\sf{ch}})}$, $T'$ is another model parameter,
 and $\epsilon_{\sf{ch}}$ is the target computational outage probability. Note that the computational outage probability occurs when there is not
 sufficient computational resources to correctly decode the received signal. 
The set $\left\lbrace T', \zeta \right\rbrace $ can be selected by calibrating (\ref{eq:cpt}) with an actual turbo-decoder 
implementation or a message-passing decoder.
Let $\mb{r}^o$ and $\mb{b}^o$ denote the vectors whose elements represent
the data rates and numbers of quantization bits selected for all users from OP $o$, respectively.
Then, the total computation effort required by the cloud to successfully decode the signals for all users from OP $o$ (calculated
in \textit{``bi per second (bips)''}) can be expressed as \cite{Suryaprakash15}
\beq
C_{\sf{tt}}^o(\mb{r}^o,\mb{b}^o) = N_{\sf{RE}} \sum_{k \in \mathcal{K}} \sum_{s \in \mathcal{S}^o_k} C^{(s)}_{k}.
\eeq

\subsection{Bi-level Resource Allocation  Formulation}

We now present the bi-level problem formulation that models the interactions among
 the C-RAN InP, the OPs, and mobile users. Specifically, the OPs 
must pay the InP to rent network resources, which are then utilized to provide services to the mobile users.
Moreover, such problem formulation must account for limited cloud computational effort and fronthaul capacity. 

In the considered bi-level problem formulation, the performance measure of interest is the profits achieved by the InP and OPs
which are modeled as follows. Let $C^o$ and $B^o_k$ denote the computational effort and fronthaul capacity corresponding to RRH $k$ that OP $o$ requires 
from the C-RAN InP. Moreover, the prices corresponding to one unit of computational effort and fronthaul capacity are denoted as 
$\psi^o$ (\textcent$/bips$) and $\beta^o_k$ (\textcent$/bps$), respectively. We assume that
the profit obtained by the InP  is equal to the total payment from all OPs (i.e., we omit the operation cost since it simply adds
a constant term to underlying optimization objective), which can be calculated as
\beq
G^{\sf{InP}} = \sum \limits_{o \in \mathcal{O}} G^{\sf{InP}}_o = \sum \limits_{o \in \mathcal{O}} \left( \psi^o C^o + \sum \limits_{k \in \mathcal{K}} \beta^o_k B^o_k\right),
\eeq
where $G^{\sf{InP}}_o$ is the payment from OP $o$ to the InP for using the amount of cloud computational resource $C^o$ and the amount of fronthaul capacity
 $B^o_k$ for RRH $k$, i.e., $G^{\sf{InP}}_o = \psi^o C^o + \sum \limits_{k \in \mathcal{K}} \beta^o_k B^o_k$. For convenience, let us define $\mb{B}^o = [B^o_1, ..., B^o_K]$.

Let us define $R_o(C^o,\mb{B}^o) = \sum_{k \in \mathcal{K}} \sum_{s \in \mathcal{S}^o_k} r^{(s)}_{k}$ as the total rate of 
all users from OP $o$, which is achieved by using the amount of
 computational resource $C^o$ and the amount of fronthaul  capacity $\mb{B}^o$.
Then, the profit achieved by OP $o$, which is equal to revenue minus cost, can be expressed as
\beq
G^{\sf{OP}}_o= \rho^o N_{\sf{RE}} R_o(C^o,\mb{B}^o) - G^{\sf{InP}}_o
\eeq   
where  $\rho^o$ (\textcent$/bps$) is the price per rate unit that OP $o$ obtains by providing services to its users.
The upper-level problem aims to maximize the weighted sum profit of the InP and OPs, which can be mathematically stated as
\begin{subequations} \label{obj_U}
\begin{align}
 \max \limits_{\lbrace C^o, \mb{B}^o \rbrace} & \;\;\;\;\; \upsilon^{\sf{InP}} G^{\sf{InP}} + \sum \limits_{o \in \mathcal{O}} \upsilon^o G^{\sf{OP}}_o, \label{obj_Ua} \\
\text{s. t. } & \;\;\; \sum_{o \in \mathcal{O}} C^{o} \leq \bar{C}_{\sf{cloud}}, \label{obj_Ub} \\
{} & \;\;\; \sum_{o \in \mathcal{O}} B^{o}_k \leq \bar{B}_k, \forall k \in \mathcal{K}, \label{obj_Uc} 
\end{align}
\end{subequations}
where $\upsilon^{\sf{InP}}$ and $\upsilon^o$ represent the weights corresponding to the InP and OP $o$, respectively,
which can be used to control the desirable profit sharing between the InP and OPs. Moreover, $\bar{C}_{\sf{cloud}}$
describes the available computational resource in the cloud and $\bar{B}_k$ denotes the capacity of the fronthaul link
connecting RRH $k$ with the cloud. 


In the lower level, each OP is interested in maximizing the sum rate through optimizing transmission rates 
and the number of quantization bits allocated to individual users. The SINR $\gamma^{(s)}_{k}(b^{(s)}_{k})$
is a complicated function of quantization bits $b^{(s)}_{k}$. In the following Proposition, we present
a tight lower bound of the SINR for a good signal quantization design, which is incorporated in the lower-level problems. Moreover,
 the SINR $\gamma^{(s)}_{k}(b^{(s)}_{k})$ can be upper-bounded by $\bar{\gamma}^{(s)}_{k}={D_k^{(s)}}/{I_k^{(s)}}$,
 which can be obtained by setting the quantization noise $q^{(s)}_{k}(b^{(s)}_{k}) = 0$ in the SINR expression.
\begin{proposition} \label{R8_prt1}
If the number of quantization bits $b^{(s)}_{k}$ satisfies $q^{(s)}_{k}(b^{(s)}_{k}) \leq \sqrt{Y^{(s)}_{k} I^{(s)}_{k}}$, then
we have the followig lower-bound for the SINR $\gamma^{(s)}_{k}(b^{(s)}_{k})$ 
\beq  \label{eq:qrq1}
\gamma^{(s)}_{k}(b^{(s)}_{k}) \geq \underline{\gamma}^{(s)}_{k} = \sqrt{\bar{\gamma}^{(s)}_{k} + 1} - 1.
\eeq
In addition, we have the following relations between the SINR upper and lower bounds
\beqn
\underline{\gamma}^{(s)}_{k} {\simeq} \sqrt{\bar{\gamma}^{(s)}_{k}} \text{ when } \bar{\gamma}^{(s)}_{k} \gg 1, \\
\underline{\gamma}^{(s)}_{k} {\simeq} \bar{\gamma}^{(s)}_{k}/2 \text{ when } \bar{\gamma}^{(s)}_{k} \ll 1.
\eeqn
\end{proposition}
\begin{proof} 
The proof is given in Appendix \ref{prf_R8_prt1}. 
\end{proof}

The requirement $q^{(s)}_{k}(b^{(s)}_{k}) \leq \sqrt{Y^{(s)}_{k} I^{(s)}_{k}}$ in Proposition~\ref{R8_prt1} is indeed
equivalent to $b^{(s)}_{k} \geq \lceil \underline{b}^{(s)}_{k} \rceil$ where 
\beq \label{blbound}
\underline{b}^{(s)}_{k} = \dfrac{1}{2}\left( \log_2 \left( \sqrt{3}\pi \sqrt{\dfrac{ Y^{(s)}_{k}}{I^{(s)}_{k}}} \right) -1\right),
\eeq
and $\lceil * \rceil$ stands for the ceiling operation.

The lower-level problem for OP $o$ $(\mathcal{P}^o)$ can be formally stated as 
\begin{subequations} \label{obj_1}
\begin{align}
 \max \limits_{\mb{r}^o, \mb{b}^o} & \;\;\;\;\;  \sum_{k \in \mathcal{K}} \sum_{s \in \mathcal{S}^o_k} r^{(s)}_{k}  \label{obj_1a} \\
\text{s. t. } & \; \sum_{k \in \mathcal{K}} \sum_{s \in \mathcal{S}^o_k} C^{(s)}_{k} \leq C^o/N_{\sf{RE}}, \label{obj_1b} \\
 {} & \;  r^{(s)}_{k} \leq \log_2 \left(1 + \gamma^{(s)}_{k}(b^{(s)}_{k}) \right), \forall k \in \mathcal{K}, \forall s \in \mathcal{S}^o_k, \label{obj_1c}\\
 {} & \; \sum_{s \in \mathcal{S}^o_k} b_{k}^{(s)} \leq B^o_k/(2N_{\sf{RE}}), \forall k \in \mathcal{K}, \label{obj_1d} \\
  {} & \; b^{(s)}_{k} \geq \lceil \underline{b}^{(s)}_{k} \rceil, \forall k \in \mathcal{K}, \forall s \in \mathcal{S}^o_k, \label{obj_1e} \\
 {} & \; \text{$b^{(s)}_{k}$ is integer}, \forall k \in \mathcal{K}, \forall s \in \mathcal{S}^o_k, \label{obj_1f} \\
 {} & \; r^{(s)}_{k} \in \mathcal{M}_R, \forall k \in \mathcal{K}, \forall s \in \mathcal{S}^o_k. \label{obj_1g}
\end{align}
\end{subequations}
Constraint (\ref{obj_1b}) requires that the total computation effort required by all users of OP $o$ should not exceed the sliced computational resource 
for this OP, $C^o$, which is determined from the upper-level problem.
The second constraint (\ref{obj_1c}) is the standard capacity constraint for PRB $s$ while
constraint (\ref{obj_1d}) ensures that the amount of fronthaul capacity allocated for RRH $k$ of OP $o$ is upper bounded by $B^o_k$,
which is determined from the upper-level problem. Constraints (\ref{obj_1e}) capture the good quantization regime with
the lower bound of quantization bits $\underline{b}^{(s)}_{k}$ given in (\ref{blbound}).

This two-level resource allocation design is difficult to tackle because we have to optimize 
the discrete variables related to the rate and quantization bit allocation  $\mb{r}^o, \mb{b}^o$ in the lower-level problem
as well as the continuous variables  $C^o, \mb{B}^o$ in the upper-level problem. Moreover, 
computational complexity $C^{(s)}_{k} = \lbrace\chi^{(s)}_{k}(r^{(s)}_{k},b^{(s)}_{k})\rbrace$ in the lower-level problems
is a complex function of the optimization variables $r^{(s)}_{k},b^{(s)}_{k}$. Finally, the
lower-level and upper-level problems are tightly coupled since the variables $C^o, \mb{B}^o$ in the later are
the parameters in the former.

\section{Problem Transformation and Convexity Characterization}
\label{Ch8_secIII}
We propose a two-stage solution framework to solve the bi-level resource allocation formulation where
 we solve the relaxed problems in stage one and develop rounding methods to find an efficient
and feasible solution for the original design problems in the second stage. 

\subsection{Problem Relaxation}

The lower-level problem with discrete optimization variables $\mb{b}^o$ and  $\mb{r}^o$
plays an important role in our design. Since optimization discrete variables are highly complex,
we adopt the natural relaxation approach to tackle the lower-level problems where the discrete variables are relaxed into the continuous ones. 
Specifically, the constraint (\ref{obj_1f}) is relaxed to
\beq \label{obj_1f2}
R_{\sf{min}} \leq r^{(s)}_{k} \leq R_{\sf{max}}, \forall k \in \mathcal{K}, \forall s \in \mathcal{S},
\eeq
where $R_{\sf{min}}$ and $R_{\sf{min}}$ represent the lowest and highest rates in the rate set $\mathcal{M}_R$, respectively.
With this relaxation, we study the following relaxed lower-level (RLL) problem
\begin{align}
 \max \limits_{\mb{r}^o, \mb{b}^o} \;\;\;  \sum_{k \in \mathcal{K}} \sum_{s \in \mathcal{S}} r^{(s)}_{k} \;\;\; \text{s. t. } \text{ (\ref{obj_1b})-(\ref{obj_1e}) 
and (\ref{obj_1f2})}. \label{obj_2}
\end{align}
Let $\bar{R}_o(C^o,\mb{B}^o)$ denote the optimal total rate of all users from OP $o$ obtained
by solving this RLL problem. 
It is clear that $\bar{R}_o(C^o,\mb{B}^o)$ is the upper bound of $R_o(C^o,\mb{B}^o)$. 
Based on the obtained rates $\bar{R}_o(C^o,\mb{B}^o)$, we consider the following relaxed upper-level (RUL) problem 
\begin{align}
 \max \limits_{\lbrace C^o, \mb{B}^o \rbrace} \;\;  \sum \limits_{o \in \mathcal{O}} \Psi^o(C^o, \mb{B}^o) = \upsilon^{\sf{InP}} G^{\sf{InP}} + \sum \limits_{o \in \mathcal{O}} \upsilon^o\left( \rho^o N_{\sf{RE}} \bar{R}_o(C^o,\mb{B}^o) -  G^{\sf{InP}}_o \right) \;\;
\text{s. t. } \;\;\; \text{(\ref{obj_Ub}), (\ref{obj_Uc})}, \label{obj_U2}
\end{align}
where $\Psi^o(C^o, \mb{B}^o)=(\upsilon^{\sf{InP}}-\upsilon^o) G^{\sf{InP}}_o + \upsilon^o \rho^o N_{\sf{RE}} \bar{R}_o(C^o,\mb{B}^o)$.

\subsection{Convexity Characterization}

\subsubsection{Convexity of RLL Problem}

To solve the RLL problem, we first characterize the convexity of the computational complexity function $C^{(s)}_{k} = \lbrace\chi^{(s)}_{k}(r^{(s)}_{k},b^{(s)}_{k})\rbrace$ in the following theorem.
\begin{theorem} \label{R8_thr1} $\chi^{(s)}_{k,u}(r^{(s)}_{k},b^{(s)}_{k})$ is a jointly convex function with respect to variables $(r^{(s)}_{k},b^{(s)}_{k})$ if 
\beq \label{eq:qrq2}
q^{(s)}_{k}(b^{(s)}_{k}) \leq \sqrt{Y^{(s)}_{k} I^{(s)}_{k}}.
\eeq
\end{theorem}
\begin{proof} 
The proof is given in Appendix \ref{prf_R8_thr1}. 
\end{proof}
Based on the result in this theorem, we state the convexity of the RLL problem in
 the following proposition.
\begin{proposition} \label{R8_prt2}
The RLL problem (\ref{obj_2}) is convex.
\end{proposition}
\begin{proof} 
Note that the condition required to have the SINR lower bound (\ref{eq:qrq1}) in Proposition~\ref{R8_prt1}, which is captured
in constraint (\ref{obj_1e}), is exactly the requirement in (\ref{eq:qrq2}). Hence, the constraint (\ref{obj_1b}) is convex if $b^{(s)}_k$ satisfies the constraint (\ref{obj_1e}). In addition, the constraint function in (\ref{obj_1c}) is convex due to the fact that $\log_2 \left(1 + \gamma^{(s)}_{k}(b^{(s)}_{k}) \right)$ 
is a concave function with respect to $b^{(s)}_k$. Moreover, the objective function and other constraints of the RLL problem (\ref{obj_2}) are in linear form. 
Therefore, the RLL problem (\ref{obj_2}) is convex.
\end{proof}

\subsubsection{Convexity of RUL problem}

We characterize the convexity of the RUL problem in the following theorem and proposition.
\begin{theorem} \label{R8_thr2}
$\bar{R}_o(C^o,\mb{B}^o)$ is a concave function with respect to $C^o$ and $\mb{B}^o$.
\end{theorem}
\begin{proof} 
The proof is given in Appendix \ref{prf_R8_thr2}. 
\end{proof}
Based on the result in this theorem, we have the following proposition.
\begin{proposition} \label{R8_prt5}
The RUL problem (\ref{obj_U2}) is convex.
\end{proposition}
\begin{proof} 
Due to the result in Theorem~\ref{R8_thr2}, the objective function of the RUL problem is concave with respect to variables 
$c^o$ and $\mb{b}^o$. In addition, all the constraint functions are in linear form. Therefore, the RUL problem (\ref{obj_U2}) is convex.
\end{proof}


\section{Resource Allocation Algorithms}
\label{Ch8_secIV}
\subsection{Proposed Algorithm to Solve RLL Problems}

According to the result in Proposition~\ref{R8_prt2}, the RLL problem is convex; 
hence, it can be solved optimally by tackling the corresponding dual problem. 
Specifically, the dual function $g(\lambda)$ of the  RLL problem can be defined as
\begin{align}
g^o(\lambda^o) \! = \! \max \limits_{\mb{r}^o, \mb{b}^o} \Phi^o \! ( \! \lambda^o \! ,\mb{r}^o \!, \mb{b}^o \! ) \text{ s. t. (\ref{obj_1c})-(\ref{obj_1e}) and (\ref{obj_1f2}),} \label{obj_3}
\end{align}
where $\Phi^o \! ( \! \lambda^o \! ,\mb{r}^o \!, \mb{b}^o \! )$ is the Lagrangian obtained by relaxing the constraint (\ref{obj_1b}),
which can be expressed as
\beq \label{LargRLL}
\Phi^o(\lambda^o,\mb{r}^o, \mb{b}^o) \! = \! \sum_{k \in \mathcal{K}} \! \sum_{s \in \mathcal{S}^o_k} \! r^{(s)}_{k} \! - \! \lambda^o \!\! \left( \! \sum_{k \in \mathcal{K}} \! \sum_{s \in \mathcal{S}^o_k} \! C^{(s)}_{k} \! - \! \dfrac{C^o}{N_{\sf{RE}}} \! \right)
\eeq
where $\lambda_o$ denotes the Lagrange multiplier. Then, the dual problem can be written as
\begin{align} \label{obj_4}
\min \limits_{\lambda^o} \: g^o(\lambda^o) \text{ s. t. } \lambda^o \geq 0. 
\end{align}
Since the dual problem is always convex, $g^o(\lambda^o)$ can be minimized by using the standard sub-gradient method
where the dual variable $\lambda^o$ can be iteratively updated as follows:
\beq \label{eq:updld}
\lambda^o_{(l+1)} = \left[\lambda^o_{(l)} + \delta^o_{(l)}\left( \sum_{k \in \mathcal{K}} \sum_{s \in \mathcal{S}^o_k} C^{(s)}_{k} - \dfrac{C^o}{N_{\sf{RE}}} \right)  \right]^{+},
\eeq
where $l$ denotes the iteration index, $\delta^o_{(l)}$ represents the step size, and $[x]^+$ is defined as $\max(0, x)$. This
sub-gradient update guarantees to converge to the optimal value of $\lambda^o$ for given primal point $(\mb{r}^o, \mb{b}^o)$ if the 
step-size $\delta^o_{(l)}$ is chosen appropriately so that $\delta^o_{(l)} \rightarrow 0$ when $l \rightarrow \infty$ 
such as $\delta^o_{(l)} = 1/\sqrt{l} $ \cite{Bertsekas99}.
%
%
%
%

To solve the RLL problem optimally, one can iteratively solve problem (\ref{obj_3}) for a given dual point $\lambda_o$
and employ the sub-gradient method to update $\lambda_o$ as in (\ref{eq:updld}). Therefore,
the remaining step is to solve the optimization problem in the right-hand-side of (\ref{obj_3}).
We will show that this can be accomplished by decoupling this problem into $K$ sub-problems corresponding to $K$ cells 
and iteratively solving these sub-problems optimally.
It can be verified that the Lagrangian function (\ref{LargRLL}) can be rewritten as
\beq
\Phi^o(\lambda^o,\mb{r}^o, \mb{b}^o) = \sum_{k \in \mathcal{K}}\Phi^o_k(\lambda^o,\mb{r}^o_k, \mb{b}^o_k) + \lambda^o \dfrac{C^o}{N_{\sf{RE}}},
\eeq
where $\Phi^o_k(\lambda^o,\mb{r}^o_k, \mb{b}^o_k)$ is expressed in
\beqn
\label{eq:Phik}
\Phi^o_k(\lambda^o,\mb{r}^o_k, \mb{b}^o_k) & = & \sum_{s \in \mathcal{S}^o_k} r^{(s)}_{k} - \lambda^o \sum_{s \in \mathcal{S}^o_k} C^{(s)}_{k} \nonumber \\
& = & \sum_{s \in \mathcal{S}^o_k} \left[ (1-\lambda^o AB) r^{(s)}_{k} + 2\lambda^o A r^{(s)}_{k} \log_2\left( \log_2 \left( 1 +\gamma^{(s)}_{k}(b^{(s)}_{k}) \right) - r^{(s)}_{k} \right)\right],
\eeqn
and $\mb{r}^o_k, \mb{b}^o_k$ represent the vectors of all rates and quantization bits corresponding to cell $k$ and OP $o$. 
Therefore, we can decouple the RLL problem into $K$ independent sub-problems, $(\mathcal{P}^o_{k})$'s, which are given as
\begin{subequations}
\begin{align}
(\mathcal{P}^o_{k}) \;\;\;\; \max \limits_{\mb{r}^o_k, \mb{b}^o_k} & \;\;\;\; \Phi^o_k(\lambda^o,\mb{r}^o_k, \mb{b}^o_k)  \label{obj_4a} \\
\;\;\;\;\;\;\;\; \text{s. t. } & \;\;\;  r^{(s)}_{k} \leq \log_2 \left(1 + \gamma^{(s)}_{k}(b^{(s)}_{k}) \right), \forall s \in \mathcal{S}^o_k, \label{obj_4c}\\
 {} & \;\;\; \sum_{s \in \mathcal{S}^o_k} b_{k}^{(s)} \leq B^o_k/(2N_{\sf{RE}}),  \label{obj_4d} \\
 {} & \;\;\; b^{(s)}_{k} \geq \lceil \underline{b}^{(s)}_{k} \rceil, \forall s \in \mathcal{S}^o_k, \label{obj_4e} \\
 {} & \;\;\; R_{\sf{min}} \leq r^{(s)}_{k} \leq R_{\sf{max}}, \forall s \in \mathcal{S}^o_k. \label{obj_4f}
\end{align}
\end{subequations}
This problem is still convex due to the result in Proposition~\ref{R8_prt2}.
In the following, we solve problem $(\mathcal{P}^o_{k})$ optimally by alternately optimizing over one variable in $\mb{r}^o_k$ and $\mb{b}^o_k$ while 
keeping the other fixed. 

\subsubsection{Solving $(\mathcal{P}^o_{k})$ for given $\mb{b}^o_k$}

For a given $\mb{b}^o_k$, problem $(\mathcal{P}^o_{k})$ becomes
\begin{subequations} \label{obj_5}
\begin{align}  
 \max \limits_{\mb{r}^o_k} & \sum_{s \in \mathcal{S}^o_k} \! \left[ E^{(s)}_{k} r^{(s)}_{k} + 2\lambda^o A r^{(s)}_{k} \log_2 \left( 1 - \dfrac{r^{(s)}_{k}}{t(b^{(s)}_{k})} \right)\right]   \label{obj_5a} \\
 \text{s. t. } &  \;  R_{\sf{min}} \leq r^{(s)}_{k} \leq \min\left(t(b^{(s)}_{k}),R_{\sf{max}} \right), \forall s \in \mathcal{S}^o_k, \label{obj_5c}
\end{align}
\end{subequations}
where $E^{(s)}_{k} = \left( 1-\lambda^o AB + 2\lambda^o A \log_2t(b^{(s)}_{k})\right) $ and $t(b^{(s)}_{k})=\log_2 \left(1 + \gamma^{(s)}_{k}(b^{(s)}_{k}) \right)$.

The optimal solution of this problem is described in the following Proposition.
\begin{proposition} \label{R8_prt3} The optimal solution to problem (\ref{obj_5}) can be expressed as
\beq \label{eq:r_opt}
r^{(s)\star}_{k} \!\! = \max \! \left[ \! R_{\sf{min}},  \min \! \left( \!\! t(b^{(s)}_{k}),R_{\sf{max}}, r\vert_{\frac{\partial w(r)}{\partial r} = - E^{(s)}_{k}} \! \right) \! \right],
\eeq
where $w(r)= 2\lambda^o A r \log_2 \left( 1 - r/{t(b^{(s)}_{k})} \right)$.
\end{proposition}
\begin{proof} 
The proof is given in Appendix~\ref{prf_R8_prt3}.
\end{proof}

\subsubsection{Solving $(\mathcal{P}^o_{k})$ for given $\mb{r}^o_k$}

For given $\mb{r}^o_k$, problem $(\mathcal{P}^o_{k})$ becomes equivalent to
\begin{subequations} \label{obj_6}
\begin{align}
 \max \limits_{\mb{b}^o_k} & \;\;\;\; \sum_{s \in \mathcal{S}^o_k} z\left( b^{(s)}_k \right)  \label{obj_6a} \\
 \text{s. t. } &  \;\;\;\;  b^{(s)}_{k} \geq \max \left(\lceil \underline{b}^{(s)}_{k} \rceil, t^{-1}\left( r^{(s)}_{k}\right) \right)  , \forall s \in \mathcal{S}^o_k, \label{obj_6c}\\
 {} & \;\;\;\; \sum_{s \in \mathcal{S}^o_k} b_{k}^{(s)} \leq B^o_k/(2N_{\sf{RE}}),  \label{obj_6d} 
\end{align}
\end{subequations}
where $z(b^{(s)}_{k})=G^{(s)}_{k}  \log_2  \left(  \log_2 \left(  1 + \gamma^{(s)}_{k}(b^{(s)}_{k}) \right) \!\! - r^{(s)}_{k} \!\! \right)$, 
$G^{(s)}_{k} = 2 \lambda A r^{(s)}_{k}$, and $t^{-1}(r)$ is the inverse function of $t(b)$.
The objective function of this problem is concave with respect to ${\mb{b}^o_k}$. Hence, this problem is a convex one whose optimal
solution can be obtained by studying the Karush-Kuhn-Tucker optimality conditions. The optimal solution of this problem is 
summarized in the following proposition.

\begin{proposition} \label{R8_prt4} The optimal solution of problem (\ref{obj_6}) can be expressed as
\beq \label{eq:b_opt}
b^{(s)\star}_k =  \max \left(\lceil \underline{b}^{(s)}_{k} \rceil, t^{-1}\left( r^{(s)}_{k}\right), b\vert_{\frac{\partial z(b)}{\partial b} = \mu} \right), 
\eeq
where $\mu$ is a constant so that $\sum_{s \in \mathcal{S}^o_k} b_{k}^{(s)\star} = B^o_k/(2N_{\sf{RE}})$.
\end{proposition}
\begin{proof} 
The proof is given in Appendix~\ref{prf_R8_prt4}. 
\end{proof}

\subsubsection{Proposed Algorithm}

We summarize how to solve the RLL problem in Algorithm~\ref{R8_alg:gms1}.
In this iterative algorithm, we alternatively update one of the two variables $\mb{b}^o_k$ and $\mb{r}^o_k$ while keeping the other fixed
until convergence. Because the optimal solution for each variable can be obtained, the objective value increases over iterations
which ensures fast convergence for this algorithm. We then update the dual variable $\lambda^o$ as in (\ref{eq:updld}) in the
outer loop. 

\begin{algorithm}[!t]
\caption{\textsc{Algorithm to Solve RLL Problem}}
\label{R8_alg:gms1}
\begin{algorithmic}[1]
\STATE Initialization: Set $r_k^{(s)} = R_{\sf{min}}$ for all $(k,s) \in \mathcal{K} \times \mathcal{S}$, $\lambda^o_{(0)}=0$ and $l=0$. Choose a tolerance parameter $\varepsilon$ for convergence.
\REPEAT
\FOR{$k \in \mathcal{K}$} 
\REPEAT
\STATE Fix $\mb{r}^o_k$ and update $\mb{b}^o_k$ as in (\ref{eq:b_opt}) with $\lambda^o_{(l)}$.
\STATE Fix $\mb{b}^o_k$ and update $\mb{r}^o_k$ as in (\ref{eq:r_opt}) with $\lambda^o_{(l)}$.
\UNTIL Convergence.
\ENDFOR 
\STATE Calculate all $C_k^{(s)}$ with new $\mb{r}^o_k$ and $\mb{b}^o_k$.
\STATE Update $\lambda^o_{(l+1)}$ as in (\ref{eq:updld}).
\STATE Set $l=l+1$.
\UNTIL $\vert \lambda^o_{(l)} - \lambda^o_{(l-1)} \vert < \varepsilon$.
\end{algorithmic}
\end{algorithm}

\subsection{Proposed Algorithm to Solve RUL Problem}

Since the RUL problem is convex according to Proposition~\ref{R8_prt5}, its optimal solution can be found efficiently by standard convex optimization techniques.  
It appears non-tractable to derive the closed-form optimal solution of the RLL problem. 
Therefore, we employ the sub-gradient method to solve this problem based on the sub-gradient $\nabla \bar{R}_o(C^o,\mb{B}^o)$  of $\bar{R}_o(C^o,\mb{B}^o)$.
Specifically, the sub-gradient method to iteratively update $C^o,\mb{B}^o$ can be performed as follows:
\beq \label{eq:updcb}
[C^o, \mb{B}^o]_{(l+1)} = \mathcal{P} \left[[C^o,\mb{B}^o]_{(l)} + \tau^o_{(l)} \nabla \Psi^o(C^o,\mb{B}^o)  \right]
\eeq
where $[C^o,\mb{B}^o]_{(l)}$ denotes the vector formed from the optimization variables $C^o$ and $\mb{B}^o$, $\tau^o_{(l)}$ represents
step size in the iteration $l$, $\nabla \Psi^o(C^o,\mb{B}^o)=\left[ \dfrac{\partial\Psi^o(C^o,\mb{B}^o)}{\partial C^o} \dfrac{\partial \Psi^o(C^o,\mb{B}^o)}{\partial B^o_1}  ... \right.$ $\left. \dfrac{\partial \Psi^o(C^o,\mb{B}^o)}{\partial B^o_K} \right] ^T$. Moreover, $\mathcal{P}\left[  C^o,\mb{B}^o \right]$ represents
 the projection of $C^o,\mb{B}^o$  to the feasible region, which is achieved by solving the following quadratic problem
\begin{align}
	\min_{[C^o,\mb{B}^o]} \;\; \Vert[C^o,\mb{B}^o] - [\hat{C}^o,\hat{\mb{B}}^o] \Vert^2 \;\;	
	\text{s. t. } \text{(\ref{obj_Ub}), (\ref{obj_Uc})}, \label{quadprob}
\end{align}
where $\hat{C}^o = C^o_{(l)} + \delta^o_{(l)} {\partial \Psi^o(C^o,\mb{B}^o)}/{\partial C^o}$ and $\hat{B}^o_k = B^o_{k,(l)} + \delta^o_{(l)} {\partial \Psi^o(C^o,\mb{B}^o)}/{\partial B^o_k}$, $\forall k \in \mathcal{K}$ are updated $[C^o,\mb{B}^o]$ given by (\ref{eq:updcb}).

The sub-gradient based updates can guarantee to converge to the optimal values of $C^o,\mb{B}^o$ if the step-size $\tau^o_{(l)}$ is chosen 
appropriately to satisfy $\tau^o_{(l)} \rightarrow 0$ when $l \rightarrow \infty$ such as $\tau^o_{(l)} = 1/\sqrt{l}$ \cite{Bertsekas99}.

The remaining issue is to determine the value of $\nabla \Psi^o(C^o,\mb{B}^o)$, which can be expressed as
\begin{align} \label{eq:nlPsi}
\nabla \Psi^o(C^o,\mb{B}^o) = \nu^o \rho^o N_{\sf{RE}} \! \left[ \! \begin{matrix} {\partial \bar{R}_o(C^o,\mb{B}^o)}/{\partial C^o} \\ {\partial \bar{R}_o(C^o,\mb{B}^o)}/{\partial B^o_1} \\ ... \\ {\partial \bar{R}_o(C^o,\mb{B}^o)}/{\partial B^o_K} \end{matrix} \! \right] \! + \! (\upsilon^{\sf{InP}} \! - \! \upsilon^o) \!\! \left[\begin{matrix} \psi^o \\ \beta^o_1 \\ ... \\ \beta^o_K \end{matrix} \right] \! ,
\end{align}
where ${\partial \bar{R}_o(C^o,\mb{B}^o)}/{\partial C^o}$ and ${\partial \bar{R}_o(C^o,\mb{B}^o)}/{\partial B^o_k}$, $\forall k \in \mathcal{K}$, can be approximated as
\begin{subequations} \label{eq:nlR}
\begin{align} 
\dfrac{\partial \bar{R}_o(C^o,\mb{B}^o)}{\partial C^o} \simeq \dfrac{\bar{R}_o(C^o+\Delta C^o,\mb{B}^o)-\bar{R}_o(C^o,\mb{B}^o)}{\Delta C^o} ,\\
\dfrac{\partial \bar{R}_o(C^o,\mb{B}^o)}{\partial B^o_k} \simeq \dfrac{\bar{R}_o(C^o,\mb{B}^o+\Delta \mb{B}^o_k)-\bar{R}_o(C^o,\mb{B}^o)}{\Delta B^o_k},
\end{align}
\end{subequations}
where $\Delta \mb{B}^o_k$ is the vector of size $K \times 1$ whose elements are zero except that the $k^{th}$ element equals to $\Delta B^o_k$. 
In (\ref{eq:nlR}), the values of $\Delta C^o$ and $\Delta B^o_k$, $\forall k \in \mathcal{K}$ are chosen sufficiently small. We summarize 
the procedure to update $[C^o,\mb{B}^o]$'s in Algorithm~\ref{R8_alg:gms2} which is employed to solve the RUL problem.

\begin{algorithm}[!t]
\caption{\textsc{Algorithm To Solve RUL Problem}}
\label{R8_alg:gms2}
\begin{algorithmic}[1]
\STATE Initialization: Set $C^o_{(0)} = \bar{C}_{\sf{cloud}}/O$, and $B^o_{k,(0)} = \bar{B}_k/O$ for all $(o,k) \in \mathcal{O} \times \mathcal{K}$, $\nu^o_{(0)}=0$ for all $o \in \mathcal{O}$, and $l=0$. 
\REPEAT
\STATE Run Algorithm~\ref{R8_alg:gms1} to obtain $\lbrace \bar{R}_o(C^o,\mb{B}^o)\rbrace $ with $\lbrace [C^o,\mb{B}^o]_{(l)}\rbrace$ for all $o \in \mathcal{O}$.

\STATE Run Algorithm~\ref{R8_alg:gms1} to obtain $\bar{R}_o(C^o+\Delta C^o,\mb{B}^o)$ and $\bar{R}_o(C^o,\mb{B}^o+\Delta \mb{B}^o_k)$.

\STATE Calculate $\nabla \Psi^o(C^o,\mb{B}^o)$ as in (\ref{eq:nlPsi}) by using $\bar{R}_o(C^o+\Delta C^o,\mb{B}^o)$ and $\bar{R}_o(C^o,\mb{B}^o+\Delta \mb{B}^o_k)$  to determine ${\partial \bar{R}_o(C^o,\mb{B}^o)}/{\partial C^o}$, and ${\partial \bar{R}_o(C^o,\mb{B}^o)}/{\partial B^o_k}$ 
for all $(k,o) \in \mathcal{K} \times \mathcal{O}$ as in (\ref{eq:nlR}).
\STATE Update $[C^o,\mb{B}^o]_{(l+1)}$ for all $o \in \mathcal{O}$ as in (\ref{eq:updcb}).
\STATE Set $l=l+1$.
\UNTIL Convergence.
\end{algorithmic}
\end{algorithm}

\subsection{Rounding Design}

After running Algorithm \ref{R8_alg:gms2},
we obtain a feasible solution $\lbrace C^o,\mb{B}^o \rbrace$ and their corresponding 
$\left\lbrace r^{(s) \bigstar}_k \right\rbrace $ and $\left\lbrace b^{(s) \bigstar}_k \right\rbrace $ of the 
relaxed problems, which take real values. 
To obtain a feasible and discrete solution that satisfies the constraints (\ref{obj_1e}) and (\ref{obj_1f}), the continuous variables must be 
appropriately rounded to the corresponding discrete values.
This rounding design must be conducted carefully because the resulting discrete results may not satisfy the original cloud computation and 
fronthaul constraints. Toward this end, we propose two rounding methods which are described in the following.

\noindent \textit{Iteratively Rounding (IR) Method:} For each OP, we iteratively run Algorithm \ref{R8_alg:gms2} and perform 
the following task in each iteration. We choose one value of $r^{(s) \bigstar}_k$ (or $b^{(s) \bigstar}_k$), which is closest
to one rate value in $\mathcal{M}_R$ (or an integer)
and fix it to that value in following iterations. This one-by-one rounding process is repeated until convergence.

\noindent \textit{One-time Rounding and Adjusting (RA) Method:} This method has two phases for each OP. First, we round all $\left\lbrace r^{(s) \bigstar}_k \right\rbrace $ and $\left\lbrace b^{(s) \bigstar}_k \right\rbrace $ 
to their closest values in $\mathcal{M}_R$ and the set of integers, respectively in the first phase. Then, if any 
constraints are violated, we round down the corresponding variables one-by-one where the variable that affects the violated 
constraints the most is chosen in each rounding-down step in the second phase. 

\section{Greedy Resource Allocation Algorithm}
\label{Ch8_secV}
\begin{algorithm}[!t]
\caption{\textsc{Greedy Resource Allocation Algorithm }}  
\label{R8_alg:gms3}
\begin{algorithmic}[1]
\STATE Define the network resource for each OP as in (\ref{eq:R8_fga_C}) and (\ref{eq:R8_fga_B}).
\FOR{OP $o$}
\STATE Calculate $\lbrace b^{(s)\prime}_k \rbrace_{k \in \mathcal{K}, s \in \mathcal{S}_k^o}$ as in (\ref{eq:b_opt2}).
\STATE Set $b^{(s)}_k = \lfloor b^{(s)\prime}_k\rfloor$, for all $(k,s) \in \mathcal{K} \times \mathcal{S}_k^o$.
\STATE Set $r^{(s)}_k = \max_{r \in \mathcal{M}_R} r \text{ s. t. } r \leq t(b^{(s)}_k)$, for all $(k,s)$.
\WHILE{$C_{\sf{tt}}^o(\mb{r}^o,\mb{b}^o) > C^o_{\sf{ga}}$}
\STATE Find $(k^*,s^*)=\arg\max C^{(s)}_k$.
\STATE Reduce $r^{(s^*)}_{k^*}$ to the nearest value in $\mathcal{M}_R$.
\ENDWHILE
\ENDFOR
\end{algorithmic}
\end{algorithm}

For comparison purposes, we present a two-stage greedy resource allocation algorithm. 
In stage one, we calculate the SINR upper bounds $\bar{\gamma}^{(s)}_{k}={D_k^{(s)}}/{I_k^{(s)}}$
 on every PRB $s$ and cell $k$ and the corresponding upper bound of the rate $r_k^{(s)}$ as $\log_2\left(1 + \bar{\gamma}_k^{(s)}\right) $.
We then allocate the computational and fronthaul capacity resources for different OPs based on the upper bound of the sum rate
of each OP $o$ as follows:
\begin{align}
C^o_{\sf{ga}} &={\hat{R}^o  \bar{C}_{\sf{cloud}}} /{\sum_{o \in \mathcal{O}} \hat{R}^o} , \label{eq:R8_fga_C} \\
B^o_{k,\sf{ga}} &={\hat{R}^o_k \bar{B}_k}/{\sum_{o \in \mathcal{O}} \hat{R}^o_k}, \forall k \in \mathcal{K}, \label{eq:R8_fga_B}
\end{align} 
where $\hat{R}^o_k = \sum_{s \in \mathcal{S}^o_k}\log_2\left(1 + \bar{\gamma}_k^{(s)}\right)$ and $\hat{R}^o=\sum_{k \in \mathcal{K}} \hat{R}^o_k$.
In stage two, we propose a simple method to solve the RLL problem for OP $o$ with $C^o_{\sf{ga}}$ and $\mb{B}^o_{k,\sf{ga}}$. 
Specifically, we optimize the quantization bit allocation to maximize the sum rate of all users
by relaxing the underlying variables to continuous ones and solving the following problem
\begin{align}
 \max \limits_{\mb{b}} & \;\;\;  \sum_{k \in \mathcal{K}} \sum_{s \in \mathcal{S}} \log_2 \left(1 + \gamma^{(s)}_{k}(b^{(s)}_{k}) \right) \nonumber \\
 \text{s. t.} & \;\;\; \sum_{s \in \mathcal{S}^o_k} b_{k}^{(s)} \leq B^o_{k,\sf{ga}}/(2N_{\sf{RE}}), \forall k \in \mathcal{K}.  \label{obj_6ga} 
\end{align}
Similar to problem (\ref{obj_5}), it can be shown that this problem is convex because its objective function is concave. 
By studying the KKT optimality conditions, we can obtain the optimal solution as follows:
\beq \label{eq:b_opt2}
b^{(s)\prime}_k =  \max \left(0, b\vert_{\frac{\partial t(b)}{\partial b} = \nu} \right), 
\eeq
where $\nu$ is a constant which must be set to satisfy $\sum_{s \in \mathcal{S}} b_{k}^{(s)\prime} = B^o_{k,\sf{ga}}/(2N_{\sf{RE}})$.
Then, we can obtain the feasible vector $\mb{b}$ by applying the flooring operation to $\mb{b}^{\prime}$.

The remaining task is to determine the users' rates that satisfy constraints (\ref{obj_1b}) and (\ref{obj_1c}) which  
can be accomplished by applying the ``Complexity Cut-Off'' method \cite{Rost_GBC15}.
Specifically, we start by setting each $r^{(s)}_k$ to the highest value in the rate set $\mathcal{M}_R$ which is smaller 
than $\log_2 \left(1 + \gamma^{(s)}_{k}(b^{(s)}_k) \right)$.
Then, we iteratively reduce the rate variable that requires the highest computation effort if the cloud computation constraint is violated 
(i.e., the required computation effort of OP $o$, $C_{\sf{tt}}^o(\mb{r}^o,\mb{b}^o)$, is greater than the assigned value $C^o_{\sf{ga}}$).
This iterative process is performed until all cloud computation constraints are satisfied. We summarize this two-stage solution in Algorithm~\ref{R8_alg:gms3}.

\section{Numerical Results}
\label{Ch8_secVI}
\begin{figure}[!t]
\begin{center}
\includegraphics[width=0.4 \textwidth]{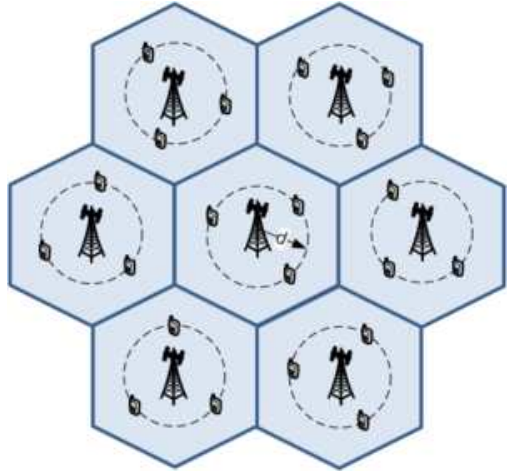}
\end{center}
\caption{Simulation model.}
\label{Fig01}
\end{figure}

We consider the 7-cell network for performance evaluation where the distance between the centers of
 any two nearest RRHs is $400 \, m$ as shown in Fig.~\ref{Fig02}. 
In each cell, we randomly place users so that the distance from the cell center to every user is $d \: (m)$.
The channel gains are generated by considering both Rayleigh fading and path loss. The path loss is modeled as 
$L_{j,u}^{k}=36.8 \mathsf{log}_{10}(d_{j,u}^k)+43.8+20\mathsf{log}_{10}(\frac{f_c}{5})$
where $d_{j,u}^k$ denotes the distance from user $u$ in cell $j$ to RRH $k$ and $f_c=2.5\:GHz$. 
We set the noise power $\sigma^2=10^{-13} \; W$ and the power $p^{(s)}_k=0.1 \;W$.
Moreover, we set $T'=0.2$, $\zeta=6$ and $\epsilon_{\sf{ch}}=10\%$ for the computation complexity model.
The rate set $\mathcal{M}_R$ is chosen corresponding to $27$ distinct MCSs with turbo coding as in the LTE standard \cite{LTE}.
In fact, the data rate corresponding to each MCS can be calculated as $TBS \times 10^3/N_{\sf{RE}}$ where the transport block size ($TBS$) for each MCS can be determined as in \textit{Table 7.1.7.2.1} in \cite{LTE} with $N_{\sf{PRB}}=1$.

We assume that there are three OPs ($O=3$) to obtain results in all simulations.
Except for the results in Fig.~\ref{Fig04}, the numbers of PRBs assigned for these three OPs are set equal to
 $5$, $10$, $15$ randomly in each cell, which means $S=30$.
In all simulations, we set the same fronthaul capacity limit for different cells. 
To obtain the results in Figs.~\ref{Fig02}-\ref{Fig05}, we set $\upsilon^{\sf{InP}}=\upsilon^o$ for all $o \in \mathcal{O}$. 


\begin{figure}[!t]
\begin{center}
\includegraphics[width=0.7 \textwidth]{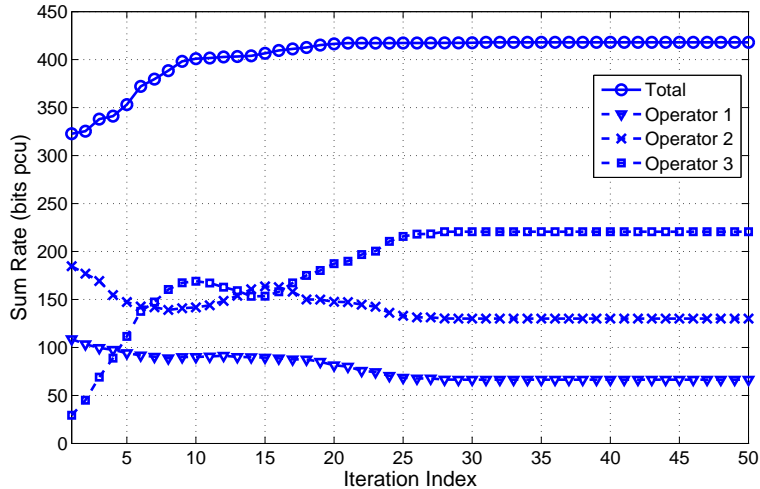}
\end{center}
\caption{Variations of sum rates over iterations}
\label{Fig02}
\end{figure}

We illustrate the convergence of our proposed algorithms in Fig.~\ref{Fig02}
where the variations of system sum rate and the rates of individual OPs over iterations 
by using Algorithm~\ref{R8_alg:gms1} and Algorithm~\ref{R8_alg:gms2} to solve the relaxed problem are shown.
To obtain results in this figure, $\bar{B}_k$ is set equal to $120~Mbps$ for each cell $k \in \mathcal{K}$ and $\bar{C}_{\sf{cloud}} = 90~Mbips$. 
As can be seen, the system sum rate increases over the first 25 iterations before settling down at the maximum value.
The third OP, who is assigned the largest number of PRBs, achieves low sum rate at the beginning
then reaches the higher rate in convergence compared to other OPs. This is because OP 3 is assigned more network resources 
than those for OPs $1$ and $2$.

\begin{figure}[!t]
\begin{center}
\includegraphics[width=0.7 \textwidth]{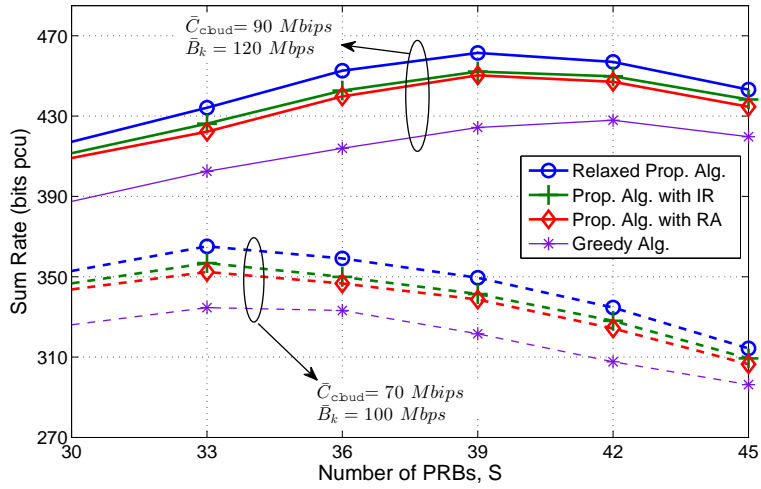}
\end{center}
\caption{Sum rate vs number of PRBs ($S$).}
\label{Fig04}
\end{figure}

In Fig.~\ref{Fig04}, we show the system sum rate obtained by different schemes, namely the our proposed algorithms without rounding (Relaxed Prop. Alg.), 
with IR-rounding and RA-rounding methods (Prop. Alg. with IR and Prop. Alg. with RA), and the Fast Greedy Algorithm (Greedy Alg.), versus the number
 of PRBs in each cell. To obtain these results, we sequentially add one more PRB to each OP in each cell to obtain different points on each curve.
The Relaxed Prop. Alg. gives the upper-bound of the system sum rate of any resource allocation algorithm. The fact that the sum rates
achieved by the Prop. Alg. with IR and Prop. Alg. with RA are very close to that achieved by the Relaxed Prop. Alg. confirms
the efficiency of our proposed two-stage algorithms. Moreover, our proposed algorithm outperforms the greedy algorithm in all studied scenarios.
In addition, the proposed algorithm with IR rounding results in slightly better sum rate than the RA rounding based counterpart.
Interestingly, the system sum rate increases and then decreases as number of PRBs in each cell increases. This means that
limited computation and fronthaul capacity resources can indeed hurt the system performance if the bandwidth provisioning
is not properly provisioned.

\begin{figure}[!t]
\begin{center}
\includegraphics[width=0.7 \textwidth]{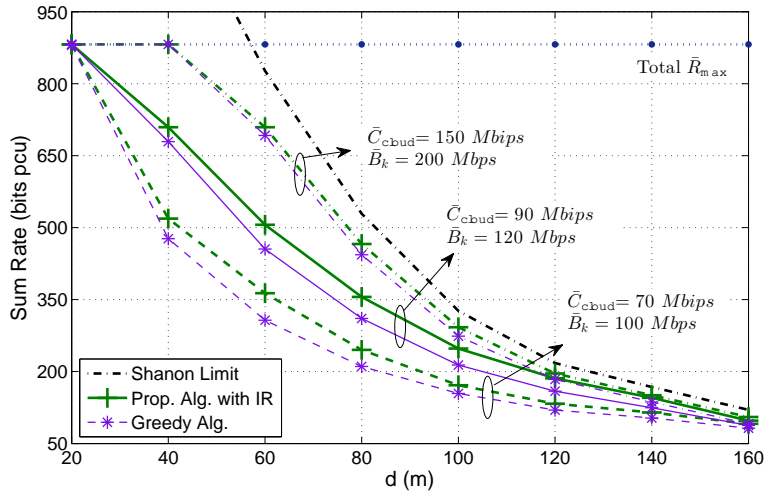}
\end{center}
\caption{Sum rate vs distance from RRHs to their users.}
\label{Fig05}
\end{figure}

Fig.~\ref{Fig05} shows the variations of the system sum rate due to the Prop. Alg. with IR and Greedy Alg. versus the user-RRH distance $d$
under three different parameter settings, namely $\bar{C}_{\sf{cloud}} = 150~Mbips$ and $\bar{B}_k=200~Mbps$;
 $\bar{C}_{\sf{cloud}} = 90~Mbips$ and $\bar{B}_k=120~Mbps$; and $\bar{C}_{\sf{cloud}} = 70~Mbips$ and $\bar{B}_k=100~Mbps$.
We also present the upper bound of the system sum rate, which is obtained by using the SINR upper bound $\bar{\gamma}_k^{(s)}$ by
setting the quantization noise on all PRB $s$ and cell $k$ to zero. This rate upper bound is equal to
 $\sum_{\forall (k,s)} \log_2(1+\bar{\gamma}_k^{(s)})$ and it is denoted as ``Shannon Limit'' in this figure.
For smaller $d$, the received signal becomes stronger in combating the multi-cell interference leading to higher link SINR $\bar{\gamma}_k^{(s)}$,
which explains the higher sum rate for smaller $d$. It can also be observed that the achieved system sum rate tends to the rate upper bound (i.e., the ``Shannon Limit'')
as the cloud computation and fronthaul capacity limits increase. Moreover, the Prop. Alg. with IR outperforms the Greedy Alg. in all studied scenarios,
which confirms the excellent performance of our proposed design.

\begin{figure}[!t]
\begin{center}
\includegraphics[width=0.7 \textwidth]{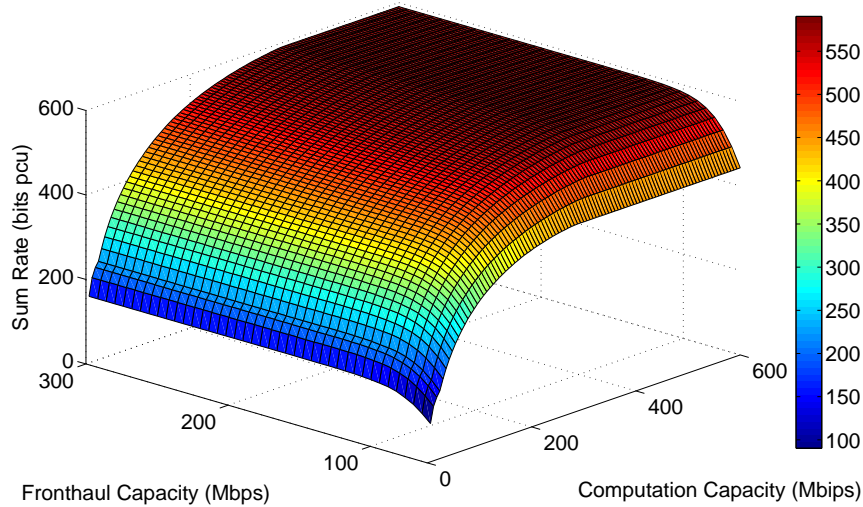}
\end{center}
\caption{Sum rate vs $\bar{C}_{\sf{cloud}}$ and $\bar{B}_k$.}
\label{Fig03}
\end{figure}

To illustrate the impacts of limited network resources, we show the system sum rate upper-bound, which is the outcome of 
Algorithm~\ref{R8_alg:gms2}, versus the computation limit ($\bar{C}_{\sf{cloud}}$) and fronthaul capacity from the cloud to 
each cell $\bar{B}_k$ in Fig.~\ref{Fig03}. 
We can see that higher cloud computation limit and larger fronthaul capacity result in the greater sum rate as expected.
In addition, the sum rate becomes saturated as the cloud computation limit or fronthaul capacity become sufficiently large.
These results imply that the proposed design framework can be employed for provisioning the cloud computation limit or fronthaul capacity
and for analyzing the provisioned network resources and performance tradeoffs.

\begin{figure}[!t]
\begin{center}
\includegraphics[width=0.7 \textwidth]{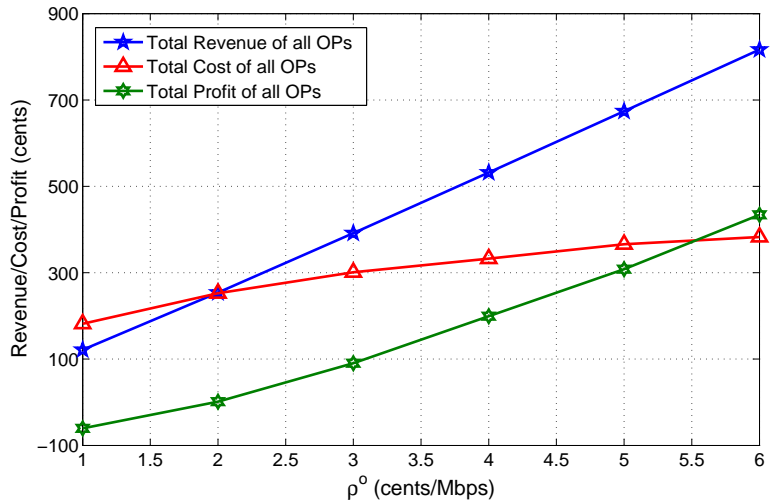}
\end{center}
\caption{OPs' revenue/cost/profit vs service price $\rho^o$ for users where $\upsilon^1=\upsilon^2=\upsilon^3=1$.}
\label{Fig06}
\end{figure}

\begin{figure}[!t]
\begin{center}
\includegraphics[width=0.7 \textwidth]{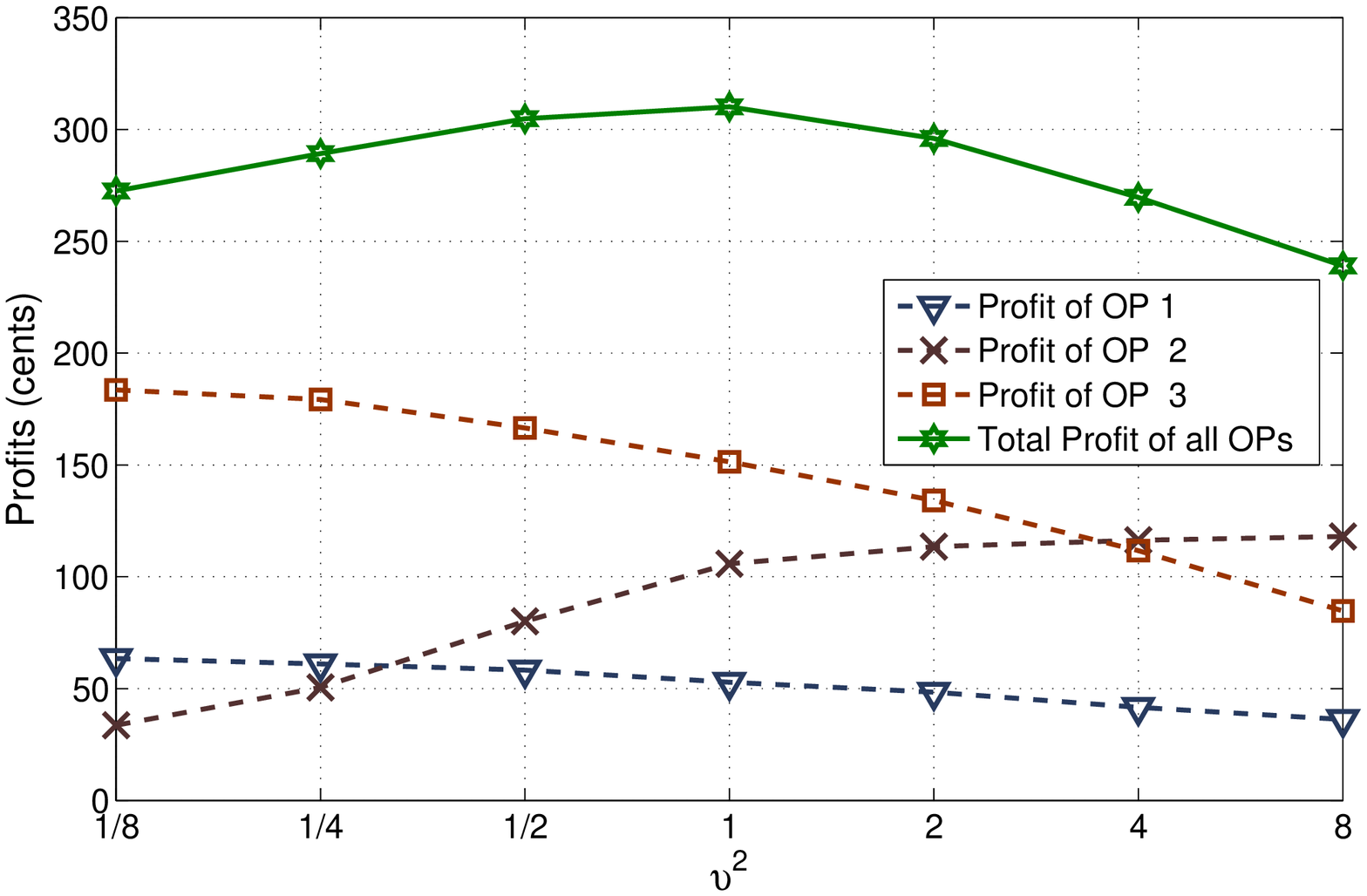}
\end{center}
\caption{OPs' profits vs weighting parameter $\upsilon^2$ for $\upsilon^1=\upsilon^3=1$.}
\label{Fig07}
\end{figure}

In Fig.~\ref{Fig06} and \ref{Fig07}, we study the profits achieved by the OPs by setting $\upsilon^{\sf{InP}}=0$. 
By setting $\upsilon^o=1$ for all $o \in \mathcal{O}$, the upper-level problem becomes the profit maximization problem for
 all OPs whose results are illustrated in Fig.~\ref{Fig06}.
In this figure, the OPs' cost (payment from all OPs to the InP), the OPs' revenue (payment of all users to the OPs), and OPs' profit 
(revenue minus cost) are shown versus $\rho^o$ while $\phi^o$ and $\beta^o$ are set equal to one for all $o \in \mathcal{O}$.
As can be seen, the OPs can attain higher profit as the price per data rate unit $\rho^o$ increases. This is because the OPs' cost tends to
saturate at sufficiently high $\rho^o$ while the revenue scales linearly with the service price.

In Fig.~\ref{Fig07}, we study the impact of the weighting parameters on the OPs' profits.
Specifically, we fix two weighting parameters as $\upsilon^1 = \upsilon^3 = 1$ while
varying the value of $\upsilon^2$ to obtain the curves in this figure.
As expected, the profit of OP $2$ increases while those of remaining OPs decrease with increasing $\upsilon^2$.
In addition, the total profit of all OPs is also presented and this figure indicates that the
maximum profit can be achieved when $\upsilon^1=\upsilon^2=\upsilon^3=1$.

\begin{figure}[!t]
\begin{center}
\includegraphics[width=0.7 \textwidth]{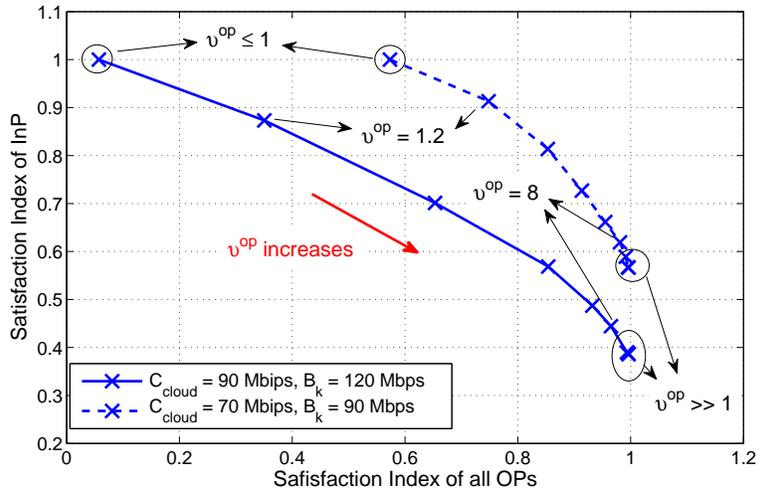}
\end{center}
\caption{Tradeoff between satisfaction indexes of InP and OPs.}
\label{Fig08}
\end{figure}

We now study a satisfaction index for InP and OPs which is defined as the ratio between the achieved profit and the maximum potential profit.
In Fig.~\ref{Fig08}, we plot the satisfaction index  of the InP versus that of all OPs
for $\rho^o=5$ (\textcent/$Mbps$) and $\phi^o=\beta^o=1$ (\textcent/$Mbips$ and \textcent/$Mbps$) for all OPs.
The maximum InP's profit can be determined as $\phi^o \bar{C}_{\sf{cloud}}+\beta^o\sum_{k \in \mathcal{K}} \bar{B}_k$ (\textcent) while the maximum 
profit of all OPs can be obtained by the same method used to obtain results in Fig.~\ref{Fig06}.
 Fig.~\ref{Fig08} illustrates trade-off between the satisfaction indices of the InP and all OPs for
 $\upsilon^{\sf{InP}}=1$ and $\upsilon^1=\upsilon^2=\upsilon^3=\upsilon^{\sf{op}}$ 
as we vary the value of $\upsilon^{\sf{op}}$. 
This figure suggests that as $\upsilon^{\sf{op}} \leq 1$ the OPs utilize all network resources,
which leads to the minimum OPs' satisfaction index. 
The presented tradeoff results indicate that one can determine an operating point of the two-level 
design framework where both the InP and OPs find it satisfactory.

\section{Conclusion}
\label{Ch8_secVII}
We have proposed a novel algorithmic framework for uplink wireless virtualization of the C-RAN supporting multiple OPs 
via joint rate and quantization bit allocation for users served by each OP.
This design aims to maximize the weighted sum profits of the InP and OPs considering practical constraints on the fronthaul capacity and cloud 
computation limits. 
Numerical results have illustrated that our proposed algorithms outperform the greedy resource allocation algorithm and
achieve the sum rate very close to the sum rate upper-bound obtained by solving relaxed problems. We have also studied
the impacts of various parameters on the system sum-rate and relevant performance tradeoffs.


\section{Appendices}
\label{Ch8_secVIII}

\subsection{Proof of Proposition~\ref{R8_prt1}}
\label{prf_R8_prt1}

If $b^{(s)}_{k}$ is selected so that $q^{(s)}_{k}(b^{(s)}_{k}) \leq \sqrt{Y^{(s)}_{k} I^{(s)}_{k}}$, we have,
\beq
\gamma^{(s)}_{k}(b^{(s)}_{k}) = \dfrac{D_k^{(s)}}{I_k^{(s)}+ q^{(s)}_{k}(b^{(s)}_{k})} \geq \dfrac{D_k^{(s)}}{I_k^{(s)}+ \sqrt{Y^{(s)}_{k} I^{(s)}_{k}}}. 
\eeq
Substituting $Y^{(s)}_{k} = D^{(s)}_{k}+I_k^{(s)}$ into this equality, we attain 
\beq
\gamma^{(s)}_{k}(b^{(s)}_{k}) \geq \dfrac{D^{(s)}_{k}/I_k^{(s)}}{1+ \sqrt{D^{(s)}_{k}/I_k^{(s)} +1 }} = \dfrac{\bar{\gamma}^{(s)}_{k}}{\sqrt{\bar{\gamma}^{(s)}_{k}+1}+1} = \sqrt{\bar{\gamma}^{(s)}_{k}+1}-1=\underline{\gamma}^{(s)}_{k}. 
\eeq
In addition, it can be verified that 
\begin{itemize}
\item When $\bar{\gamma}^{(s)}_{k} \gg 1$, we have $\sqrt{\bar{\gamma}^{(s)}_{k}+1} {\simeq} \sqrt{\bar{\gamma}^{(s)}_{k}} \gg 1$. Thus, we have $\underline{\gamma}^{(s)}_{k} {\simeq} \sqrt{\bar{\gamma}^{(s)}_{k}}$.
\item When $\bar{\gamma}^{(s)}_{k} \ll 1$, we have $\sqrt{\bar{\gamma}^{(s)}_{k}+1} {\simeq} 1+\bar{\gamma}^{(s)}_{k}/2$. Hence, it can be 
implied that $\underline{\gamma}^{(s)}_{k} {\simeq} \bar{\gamma}^{(s)}_{k}/2$.
\end{itemize}
This concludes the proof for Proposition~\ref{R8_prt1}.

\subsection{Proof of Theorem~\ref{R8_thr1}}
\label{prf_R8_thr1}

To prove that $\chi^{(s)}_{k}(r,b)$ is a convex function with respect to variables $(r,b)$, we will show that the Hessian matrix of $\chi^{(s)}_{k}(r,b)$ is 
positive definite.
For simplicity, we omit the superscripts and subscripts in all notations, i.e., $\chi(r,b)$, $D$, $I$ and $Y$ stand for $\chi^{(s)}_{k}(r,b)$, $D^{(s)}_{k}$, $I^{(s)}_{k}$ and $Y^{(s)}_{k}$, respectively.
First, we derive the Hessian matrix of $\chi(r,b)$. Let $\mb{H}=[H_{11} H_{12}, H_{21} H_{22}]$ be the Hessian matrix of $\chi(r,b)$, its elements can be
written as
\beqn
H_{11} \!\!\!\! &=& \!\!\!\! \dfrac{\partial^2\chi(r,b)}{\partial r^2}= \! \dfrac{2Ar}{(\ln 2) \left( Z -r\right)^2  },\\
H_{12} \!\!\!\! &=& \!\!\!\! \dfrac{\partial^2\chi(r,b)}{\partial r \partial b} = \! - \dfrac{2\sqrt{3}\pi ADY r}{(\ln2)2^{2b} X(X+D) (Z-r)^2}, \\
H_{21} \!\!\!\! &=& \!\!\!\! \dfrac{\partial^2\chi(r,b)}{\partial b \partial r} = \! - \dfrac{2\sqrt{3}\pi ADY r}{(\ln2)2^{2b} X(X+D) (Z-r)^2}, \\
H_{22} \!\!\!\! &=& \!\!\!\! \dfrac{\partial^2\chi(r,b)}{\partial b^2}  = \! \dfrac{6 \pi^2 A D^2 Y^2 r}{(\ln 2)2^{4b}X^2(X+D)^2 (Z-r)^2} \!\!\!\! + \!\!\!\!\!\! \dfrac{2\sqrt{3} \pi ADY \! r \! \left[ \! 2^{2b+1}\! X \!( \! X \!\! + \!\! D \! ) \! - \!\! \sqrt{3} \pi Y\! (2X \!\! + \!\! D) \! \right]  }{2^{4b}X^2(X+D)^2 (Z-r)}, \label{eq:H22}
\eeqn
where $X = I + \frac{\sqrt{3} \pi Y}{2^{2 b +1}}$ and $Z = \log_2 \left( 1 +\dfrac{D}{I + \frac{\sqrt{3} \pi Y}{2^{2 b +1}}} \right)$. 
Let $q=\frac{\sqrt{3} \pi Y}{2^{2 b +1}}$, then we have $X=I+q$, $X+D=Y+q$, and $\sqrt{3} \pi Y = 2^{2 b +1} q$.
Substituting these results into (\ref{eq:H22}) yields
\beq
H_{22}  =  \dfrac{6 \pi^2 A D^2 Y^2 r}{(\ln 2)2^{4b}X^2(X+D)^2 (Z-r)^2} + \dfrac{4\sqrt{3} \pi ADYr\left(IY - q^2 \right)  }{2^{2b}X^2(X+D)^2 (Z-r)}.
\eeq
Then, if $IY \geq q^2$, we will have $H_{11}, H_{22} > 0$, and 
\beq
\det\left| 	\mb{H} \right| = H_{11}H_{22}-H_{12}H_{21}  =  \dfrac{8\sqrt{3} \pi A^2DYr^2\left(IY - q^2 \right)  }{\ln(2)2^{2b}X^2(X+D)^2 (Z-r)^3} \geq 0. 
\eeq
Hence, the Hessian matrix of $\chi(r,b)$ is positive definite. Thus, we can conclude that $\chi(r,b)$ is jointly convex
with respect to $(r, b)$ if $IY \geq q^2$. 

\subsection{Proof of Theorem~\ref{R8_thr2}}
\label{prf_R8_thr2}

We will prove this theorem by using the definition of a concave function, i.e., $f\left(\phi x_1 +(1-\phi) x_2 \right) \geq \phi f(x_1) + (1-\phi) f(x_2)$ for all $0 \leq \phi \leq 1$.
Let us consider two possible values of the involved variables, $(C^o_1,\mb{B}^o_1)$ and $(C^o_2,\mb{B}^o_2)$.
We assume that there exists the optimum solutions for the lower-level problems (\ref{obj_1}) corresponding to these two values of variables. 
Moreover, the optimal sum rates for these cases are denoted as $\bar{R}(C^o_1,\mb{B}^o_1)$ and $\bar{R}(C^o_2,\mb{B}^o_2)$, respectively
with $\lbrace r_{k,1}^{(s)}, b_{k,1}^{(s)} \rbrace$ and $\lbrace r_{k,2}^{(s)}, b_{k,2}^{(s)} \rbrace$ being the optimal rates and number of
 quantization bits, respectively.
Then, $\lbrace r_{k,i}^{(s)}, b_{k,i}^{(s)} \rbrace$ must satisfy all the constraints of (\ref{obj_1}) corresponding 
to $(C^o_i,\mb{B}^o_i)$ for $i=1 \text{ or } 2$.
For any value of $\phi$ such that $0 \leq \phi \leq 1$, we define $\lbrace r_{k,3}^{(s)}, b_{k,3}^{(s)} \rbrace$ as
\beqn
r_{k,3}^{(s)} = \phi r_{k,1}^{(s)} + (1-\phi) r_{k,2}^{(s)}, \forall (k,s) \in \mathcal{K} \times \mathcal{S}^o_k, \\
b_{k,3}^{(s)} = \phi b_{k,1}^{(s)} + (1-\phi) b_{k,2}^{(s)}, \forall (k,s) \in \mathcal{K} \times \mathcal{S}^o_k.
\eeqn
Since all constraint functions of problem (\ref{obj_1}) are convex, it is easy to see that $\lbrace r_{k,3}^{(s)}, b_{k,3}^{(s)} \rbrace$ satisfy all the constraints of (\ref{obj_1}) corresponding to $\left(C^o_3 ,\mb{B}^o_3\right)$, where $C^o_3 = \phi C^o_1 + (1-\phi)C^o_2 $ and $\mb{B}^o_3 = \phi \mb{B}^o_1 + (1-\phi) \mb{B}^o_2$.
Therefore, $\lbrace r_{k,3}^{(s)}, b_{k,3}^{(s)} \rbrace$ is a feasible solution of problem (\ref{obj_1}) corresponding to $(C^o_3,\mb{B}^o_3)$.
Consequently, we have
\beqn
\bar{R}(C^o_3,\mb{B}^o_3) & \geq & \sum_{k \in \mathcal{K}} \sum_{s \in \mathcal{S}^o_k} r^{(s)}_{k,3} \nonumber \\
{} & = & \sum_{k \in \mathcal{K}} \sum_{s \in \mathcal{S}^o_k} \left( \phi r_{k,1}^{(s)} + (1-\phi) r_{k,2}^{(s)} \right) \nonumber \\
{} & = & \phi \bar{R}(C^o_1,\mb{B}^o_1) + (1-\phi) \bar{R}(C^o_2,\mb{B}^o_2),
\eeqn
for any value of $\phi$ such that $0 \leq \phi \leq 1$.
Hence, $\bar{R}(C^o,\mb{B}^o)$ must be a concave function with respect to $(C^o,\mb{B}^o)$ and we
 have completed the proof for Theorem~\ref{R8_thr2}.

\subsection{Proof of Proposition~\ref{R8_prt3}}
\label{prf_R8_prt3}

It can be verified that the second derivative of the objective function of (\ref{obj_5}) is the same as that of $w(r^{(s)}_{k})$ which 
can be expressed as
\beq
\dfrac{\partial^2 w(r^{(s)}_{k})}{\partial r^{(s)2}_{k} } = -\dfrac{2\lambda A r^{(s)}_{k}+ 4\lambda A\left(t(b^{(s)}_{k}) -r^{(s)}_{k} \right)}{(\ln 2)\left(t(b^{(s)}_{k}) -r^{(s)}_{k} \right)^2 },
\eeq
which is less than zero if $t(b^{(s)}_{k}) \geq r^{(s)}_{k}$.
Hence, this function is concave. 
Therefore, the optimum rate $r^{(s)}_{k}$ can be obtained by studying the KKT conditions. Taking the first derivative of
the objective function and setting it to zero results in
\beq
\frac{\partial w(r)}{\partial r} = - E^{(s)}_{k}.
\eeq
Using the constraint (\ref{obj_5c}), the optimal solution to problem (\ref{obj_5}) can be written as
\beq \label{eq:r_opt_prf}
r^{(s)\star}_{k} \!\! = \max \! \left[ \! R_{\sf{min}},  \min \! \left( \!\! t(b^{(s)}_{k}),R_{\sf{max}}, r\vert_{\frac{\partial w(r)}{\partial r} = - E^{(s)}_{k}} \! \right) \! \right].
\eeq
Therefore, we have completed the proof of Proposition~\ref{R8_prt3}.

\subsection{Proof of Proposition~\ref{R8_prt4}}
\label{prf_R8_prt4}

The Lagrangian of problem (\ref{obj_6}) can be expressed as
\beq
\mathcal{L}\left(\mb{b}^o_k,\mu \right) =  \sum_{s \in \mathcal{S}^o_k} z\left( b^{(s)}_k \right) - \mu\left( \sum_{s \in \mathcal{S}^o_k} b_{k}^{(s)} - B^o_k/(2N_{\sf{RE}}) \right), 
\eeq
where $\mu$ is the Lagrangian multiplier associated with the fronthaul capacity constraint of problem (\ref{obj_6}).
In addition, the dual function of problem (\ref{obj_6}) can be written as
\beq
g(\mu) = \max \limits_{\mb{b}^o_k} \mathcal{L}\left(\mb{b}^o_k,\mu \right) \text{ s. t. }  b^{(s)}_{k} \geq J^{(s)}_{k}, \forall s \in \mathcal{S}^o_k,
\eeq
where $J^{(s)}_{k} = \max \left(\lceil \underline{b}^{(s)}_{k} \rceil, t^{-1}\left( r^{(s)}_{k}\right)\right)$.
This problem can be decoupled into $S$ parallel sub-problems each of which corresponds to one PRB.
In addition, all these sub-problems have the same structure. Since its objective function is concave, each sub-problem can be solved by 
using the KKT condition ${\partial \mathcal{L}\left(\mb{b}^o_k,\mu \right)}/{\partial b^{(s)}_k} =0,$ which is equivalent to
\beq
{\partial z(b^{(s)}_k)}/{\partial b^{(s)}_k} = \mu.
\eeq
Using the constraint (\ref{obj_6c}), the optimal solution of $b^{(s)}_k$ must satisfy (\ref{eq:b_opt}). In addition, the objective function is 
an increasing function with respect to $\mb{b}^o_k$; hence, the fronthaul capacity constraint (\ref{obj_6d}) must be met with equality. Therefore, 
$\mu$ can be determined to satisfy $\sum_{s \in \mathcal{S}^o_k} b^{(s)}_k = B_k/(2N_{\sf{RE}})$. Therefore, we have completed the proof of Proposition~\ref{R8_prt4}.

%% file: chap9/Ha_chap9.tex
\chapter{Conclusions and Further Works} 
\renewcommand{\rightmark}{Chapter 9. Conclusions and Future Research Directions}
\label{Ch9}

In this chapter, we summarize our research contributions and discuss some potential directions for further research.

\section{Major Research Contributions}

In the first contribution \cite{VuHa_TVT14_BSA_PC,VuHa_VTC12}, we have studied the design of joint base station association and power control for single-carrier-based HetNets, presented in Chapter \ref{Ch3}. 
In particular, we have developed a generalized BSA and PC algorithm and proved its convergence if the underlying power update function satisfies the so-called two-sided-scalable property. 
In addition, we have proposed an hybrid power control adaptation algorithm that effectively adjusts key design parameters to maintain the SINR requirements of all users 
whenever possible while enhancing the system throughput. This proposed algorithm has been proved to perform better than other state-of-the-art design in the literature.
Finally, we have also described the application of the proposed framework to the design of the two-tier macrocell-femtocell wireless network.

In the second contribution \cite{VuHa_TVT14_PC_SA,VuHa_WCNC13}, we have investigated the fair resource allocation problem for OFDMA-based HetNets, presented in Chapter \ref{Ch4}.
Specifically, we have presented a resource allocation formulation for the two-tier macrocell-femtocell network that aims to maximize the total minimum rate of all femtocells subject to QoS protection constraints for macrocell users. We have proposed both exhaustive optimal search method as well as
 low-complexity distributed joint subchannel and power allocation algorithm to solve the problem. The low-complexity algorithm  
 estimates transmission power on each subchannel based on which FBSs can make subchannel allocations for FUEs in the distributed manner.
The proposed algorithm has been proved to converge. We have also described how to extend the design into various settings 
including downlink scenario, adaptive multiple-rate transmission,
and open access strategy for femtocells.

In the third contribution \cite{VuHa_TVT16,VuHa_WCNC14,vuha_ciss_2014,vuha_globecom_2014,VuHa_WCNC15}, we have considered the joint cooperative 
transmission design for the downlink communications that
aims to minimize the total power consumption of all RRHs subject to various QoS and system constraints in Chapter \ref{Ch5}. 
The proposed formulations capture the fact that fronthaul links connecting RRHs with BBU pool have limited
capacity, which is translated into the new fronthaul capacity constraint involving a non-convex and discontinuous function. 
To deal with this difficult problem, we have proposed two low-complexity algorithms.
The first one, so-called pricing-based algorithm, has solved the underlying problem through iteratively tackling a related pricing
problem while appropriately updating the pricing parameter. 
In the second one, namely iterative linear-relaxed algorithm, we have directly addressed the fronthaul constraint function by iteratively approximating 
it in a suitable linear form using a conjugate function and solving the corresponding convex problem.

In the fourth contribution \cite{VuHa_sTWC16,VuHa_ICC16}, 
we have considered the resource allocation for virtualized uplink C-RAN in Chapter \ref{Ch8}.
In this work, multiple OPs are assumed to share the C-RAN infrastructure and resources to serve their users. 
We have developed a novel slicing strategy for OPs to maximize the weighted sum profit of both infrastructure provider and OPs under 
the limited fronthaul capacity and cloud computation. The design requires to solve the upper-level and lower-level problems. 
The upper-level problem focuses on slicing the fronthaul capacity and cloud computing resources for all OPs while
the lower-level problems maximize the OPs' sum rates by optimizing users' transmission rates and quantization bit allocation 
for the compressed I/Q baseband signals. A two-stage algorithmic framework has been proposed to solve these problems. 
In the first stage, we have relaxed the underlying discrete variables 
to the continuous variables to transform both problems into
the corresponding convex optimization problems which can be solved optimally.
In the next stage, we propose two methods to round the solution obtained by solving the relaxed problems
to achieve a final feasible solution for the original problem.

\section{Further Research Directions}

Our research work in this dissertation focuses on the resource allocation for wireless HetNets and C-RANs for performance enhancements
 of the future wireless cellular networks, i.e., enhancing the networks throughput and reducing the network transmission power.
The following research direction are of importance and deserve further investigation.

\subsection{CSI Processing and Feedback Strategy for HetNets and C-RANs}

In general, the performance of the wireless network depends on the quality of CSI available at 
the transmitter and receiver sides. 
Moreover, estimation and exchanges of CSI and user information over the wireless channel (e.g., CSI feedback) in HetNets
or fronthaul links consume non-trivial network resources. 
Therefore, there is tradeoff between the achievable network performance and network resources utilized for estimation, exchanges of CSI and user information. 
We plan to consider the  CSI design and data processing issues for both downlink and uplink communications in HetNets and C-RANs. 
Specifically, we seek to develop efficient quantization strategies for CSI and baseband signals, and study the underlying performance tradeoffs.

\subsection{Heterogeneous C-RANs with Multiple Local Clouds}

The wireless HetNet architecture can bring many performance benefits in enhancing the network coverage and improving networks performance. 
On another hand, the advantages of C-RANs come from the centralized processing and joint transmission/reception designs.
Both solutions can, therefore, improve the network capacity significantly and reduce the deployment and operation costs.
However, interference management in HetNets may require non-negligible network resources for signalling and information exchanges. 
Additionally, centralized processing for many cell sites in a large network deployment area in C-RAN may require significant computation resources and 
high deployment cost of the fronthaul transport network. 
Therefore, appropriate design of the heterogeneous C-RANs to leverage benefits of both HetNets and C-RANs can be a good option,
which will be studied in our future works.

\section{List of Publications}
\subsection{Journals} 
\renewcommand{\theenumi}{[J\arabic{enumi}]}
\begin{etaremune}
\item Vu N. Ha and Long B. Le, ``Resource allocation for wireless virtualization of OFDMA-based cloud radio access networks,'' submitted to {\em IEEE Transactions on Vehicular Technology}.
\item Vu N. Ha, Long B. Le, and Ng\d{o}c-D\~{u}ng \DH\`{a}o, ``Coordinated multipoint transmission design for Cloud-RANs considering limited fronthaul capacity constraints,'' {\em IEEE Transactions on Vehicular Technology}, to appear.
\item Tri Nguyen, Vu N. Ha, and Long B. Le, ``Resource allocation optimization in multi-user multi-cell massive MIMO networks considering pilot contamination,'' {\em IEEE Access} 3, 2015.
\item Vu N. Ha and Long B. Le, ``Fair resource allocation for OFDMA femtocell networks with macrocell protection,'' {\em IEEE Transactions on Vehicular Technology,} vol. 63, no. 3, Mar. 2014.
\item Vu N. Ha and Long B. Le, ``Distributed base station association and power control for heterogeneous cellular networks,'' {\em IEEE Transactions on Vehicular Technology,} vol. 63, no. 1, Jan. 2014.
\end{etaremune}

\subsection{Conferences}
\renewcommand{\theenumi}{[C\arabic{enumi}]}
\begin{etaremune}
\item Vu N. Ha and Long B. Le, ``Resource allocation for uplink OFDMA C-RANs with limited computation and fronthaul capacity,'' in Proceeding of {\em IEEE ICC 2016}, May. 2016.
\item Vu N. Ha and Long B. Le, ``Computation capacity constrained joint transmission design for C-RANs,'' in Proceeding of {\em IEEE WCNC 2016}, Apr. 2016.
\item Vu N. Ha, Duy H. N. Nguyen, and Long B. Le, ``Sparse precoding designs for Cloud-RANs sum-rate maximization,'' in Proceeding of {\em IEEE WCNC 2015}, Mar. 2015.
\item Vu N. Ha and Long B. Le, ``Joint coordinated beamforming and admission control for fronthaul constrained Cloud-RANs,'' in Proceeding of {\em IEEE GLOBECOM 2014}, Dec. 2014. (in Top 50 Best Paper of Globecom'2014)
\item Vu N. Ha, Long B. Le, and Ng\d{o}c-D\~{u}ng \DH\`{a}o, ``Energy-efficient coordinated transmission for Cloud-RANs: Algorithm design and tradeoff,'' in Proceeding of {\em IEEE CISS 2014}, Mar. 2014.
\item Vu N. Ha, Long B. Le, and Ng\d{o}c-D\~{u}ng \DH\`{a}o, ``Cooperative transmission in Cloud-RAN considering fronthaul capacity and cloud processing constraints,'' in Proceeding of {\em IEEE WCNC 2014}, Apr. 2014.
\item Vu N. Ha and Long B. Le, ``Resource management for two-tier femtocell networks using interference alignment,'' in Proceeding of {\em Workshop - Heterogeneous and Small Cell Networks (HetSNets), IEEE GLOBECOM 2013}, Dec. 2013. 
\item Vu N. Ha and Long B. Le, ``Distributed resource allocation for OFDMA femtocell networks with macrocell protection,'' in Proceeding of {\em IEEE WCNC 2013}, Apr. 2013. 
\item Vu N. Ha and Long B. Le, ``Hybrid access design for femtocell networks with dynamic user association and power control,'' in Proceeding of {\em IEEE VTC2012-Fall}, Sept. 2012.
\end{etaremune}